\documentclass[12pt]{article}
\usepackage{epsfig,lscape}
\usepackage{amssymb,amsfonts,amsmath,amsthm}
\usepackage{rotating}
\usepackage{mathtools}
\usepackage{bbm, bm}
\usepackage{natbib}

\usepackage{multibib}

\RequirePackage[colorlinks,citecolor=blue,urlcolor=blue]{hyperref}

\usepackage{threeparttable}
\usepackage{booktabs}
\usepackage{booktabs}
\usepackage{lmodern}
\usepackage{multirow}
\usepackage{verbatim}

\usepackage{hyperref}

\usepackage{subfig}
\usepackage{siunitx}
\usepackage{textpos, tikz}
\usepackage{subfig}
\usepackage{graphicx}
\usepackage{float}
\usepackage{amssymb}
\usepackage{pdflscape}
\usepackage{pdfpages}
\usepackage{booktabs}
\usepackage{lmodern}
\usepackage{multirow}
\usepackage{verbatim}
\usepackage{hyperref}
\usepackage[titletoc,title]{appendix}
\usepackage{minibox}
\usepackage{algorithm}
\usepackage{algpseudocode}

\newcites{append}{Appendix References}

\usepackage[left=1in,right=1in,top=.9in,bottom=1.1in]{geometry}

\def\me{\mathrm e}

\def\dif{\mathrm d}

\def\var{\mathrm{var}}
\def\cov{\mathrm{cov}}
\def\N{\mathrm{N}}

\def\DB{\mathrm{DB}}
\def\T{ {\mathrm{\scriptscriptstyle T}} }

\def\bbR{\mathbb R}

\def\argmin{\mathrm{argmin}}

\newenvironment{prf}
{\noindent \textbf{Proof.}}{\hfill $\Box$ \vspace{.05in}}

\newtheorem{thm}{Theorem}
\newtheorem{lem}{Lemma}
\newtheorem{pro}{Proposition}
\newtheorem{cor}{Corollary}
\newtheorem{ass}{Assumption}

\theoremstyle{definition}
\newtheorem{eg}{Example}

\theoremstyle{definition}
\newtheorem{rem}{Remark}

\newcommand{\expit}{\mathrm{expit}}

 \newcommand{\norm}[1]{\left\lVert#1\right\rVert}
 \newcommand{\abs}[1]{\left|#1\right|}

\begin{document}

\begin{titlepage}

\begin{center}
{\Large Doubly Robust Semiparametric Inference Using Regularized Calibrated Estimation with High-dimensional Data}

\vspace{.1in} {\large Satyajit Ghosh\footnotemark[1] \& Zhiqiang Tan\footnotemark[1]}

\vspace{.1in}
\today
\end{center}

\footnotetext[1]{Satyajit Ghosh is postdoc and
and Zhiqiang Tan is professor, Department of Statistics, Rutgers University,
Piscataway, NJ 08854, USA (Email: ztan@stat.rutgers.edu). }

\paragraph{Abstract.}
Consider semiparametric estimation where a doubly robust estimating function for a low-dimensional parameter is available, depending on two working models.
With high-dimensional data, we develop regularized calibrated estimation as a general method for estimating the parameters in the two working models,
such that valid Wald confidence intervals can be obtained for the parameter of interest under suitable sparsity conditions if either of the two working models is correctly specified.
We propose a computationally tractable two-step algorithm and provide rigorous theoretical analysis which justifies
sufficiently fast rates of convergence for the regularized calibrated estimators in spite of  sequential construction
and establishes a desired asymptotic expansion for the doubly robust estimator.
As concrete examples, we discuss applications to partially linear, log-linear, and logistic models and estimation of average treatment effects.
Numerical studies in the former three examples demonstrate superior performance of our method, compared with debiased Lasso.


\end{titlepage}

\section{Introduction}

Semiparametric modeling and estimation aims to draw inference about low-dimensional parameters of interest, while
allowing flexible specification for nuisance parameters, which are often in the form of smooth functions of covariates (\citealt{bickel1993efficient, tsiatis2007semiparametric}).
With low-dimensional covariates, various methods and theory have been developed, using nonparametric smoothing techniques
to estimate those unknown functions. There are increasing difficulties,
as the complexity of functions increases with a fixed number of covariates, or the number of covariates increases with parametric specifications for the unknown functions.
These two problems are fundamentally related. For concreteness, we focus on the latter setting,
where the number of covariates is large, while the unknown functions are modeled using known basis functions, for example, main effects or interactions.
This setting also allows connections to high-dimensional statistics (\citealt{buhlmann2011statistics}).

In this article, we study a broad class of semiparametric problems, where a doubly robust estimating function $\tau(U;\theta,\alpha,\gamma)$
for the parameter of interest $\theta$ is available as follows. Here $U$ denotes a data vector including a possibly high-dimensional covariate vector $X$, and
$(\alpha,\gamma)$ are two nuisance parameters defined through
working models $g(x;\alpha)$ and $f(x;\gamma)$ for unknown functions $g^*(x)$ and $f^*(x)$.
The estimating function $\tau$ is assumed to be unbiased, $E\{ \tau(U;\theta,\alpha,\gamma)\}=0$, when $\theta$ is set to the true value $\theta^*$, and
either $\alpha$ or $\gamma$, but not necessarily both, is set to the true value $\alpha^*$ or $\gamma^*$ defined respectively such that
$g(x;\alpha^*) \equiv g^*(x)$ or $f(x;\gamma^*) \equiv f^*(x)$ if model $g(\cdot;\alpha)$ or $f(\cdot;\gamma)$ is correctly specified.
In general, doubly robust estimation using $\tau$ consists of two stages: some estimators $(\hat\alpha,\hat\gamma)$ are first defined,
and then $\hat\theta$ is defined by solving the estimating equation $\tilde E\{ \tau(U;\theta,\hat\alpha,\hat\gamma)\}=0$,
where $\tilde E()$ denotes a sample average.
Conventionally, the estimators $(\hat\alpha, \hat\gamma)$ are derived by maximum likelihood or variations associated with models $g(\cdot;\alpha)$ and $f(\cdot;\gamma)$ for $g^*$ and $f^*$.

While such doubly robust estimation is perhaps most extensively studied in missing-data problems and estimation of average treatment effects (\citealt{scharfstein1999adjusting, kang2007demystifying, tan2010bounded}),
doubly robust methods have been developed in various semiparametric problems, including
partially linear and log-linear models (\citealt{robins2001comment}), instrumental variable analysis (\citealt{tan2006regression, Okui2012}),
mediation analysis (\citealt{tchetgen2012semiparametric}), and dimension reduction (\citealt{ma2012semiparametric}) among others.
As a somewhat under-appreciated result, we point out  that the familiar least-squares estimator for each individual coefficient in linear regression
is doubly robust in the context of a partially linear model.
This result is also closely related to debiased Lasso estimation in high-dimensional linear regression (\citealt{zhang2014confidence, SVD, javanmard2014confidence}). See Examples~\ref{eg:PLM-LS} and \ref{eg:PLM-RCAL} for further discussion.


The main contribution of our work can be summarized as follows. Given a doubly robust estimating function $\tau$, we develop a general method as an alternative to maximum likelihood
for constructing estimators $(\hat\alpha,\hat\gamma)$
of nuisance parameters, which are used to define an estimator $\hat\theta$ as a solution to $\tilde E\{ \tau(U;\theta,\hat\alpha,\hat\gamma)\}=0$.
For this method, the limit values $(\bar\alpha, \bar\gamma)$ of $(\hat\alpha, \hat\gamma)$ are designed to satisfy
a pair of population estimating equations, called calibration equations. If either model $g(\cdot;\alpha)$ or
$f(\cdot;\gamma)$ is correctly specified, then
the resulting estimator $\hat\theta$ can be shown to be not only consistent for $\theta^*$, but also
achieve an asymptotic expansion in the following manner under suitable conditions with a sample size $n$.
\begin{itemize} \addtolength{\itemsep}{-.1in}
\item In low-dimensional settings, the expansion of $\hat\theta$ is in the usual order $O_p (n^{-1/2})$, but
not affected by the variation of $(\hat\alpha, \hat\gamma)$, which is also of order $O_p(n^{-1/2})$.

\item In high-dimensional settings, the expansion of $\hat\theta$ remains in the order  $O_p (n^{-1/2})$,
even though the convergence of  $(\hat\alpha, \hat\gamma)$ to $(\bar\alpha,\bar\gamma)$ is slower than $O_p(n^{-1/2})$.
\end{itemize}
In fact, with high-dimensional data, we propose a computationally tractable two-step algorithm using Lasso regularized estimation.
We provide rigorous theoretical analysis which justifies sufficiently fast convergence rates for $(\hat\alpha,\hat\gamma)$ in spite of sequential construction
and establishes the desired asymptotic expansion and variance estimation for $\hat\theta$.
Doubly robust Wald confidence intervals can be obtained, based on $\hat\theta$ and consistent variance estimation.
As concrete examples, we discuss applications to partially linear, log-linear, and logistic models and a missing-response problem related to estimation of average treatment effects.

\vspace{.1in}
{\bf Related work.}
There is an extensive literature related to our work. In low-dimensional settings, estimating equations similar to our calibration equations are proposed by
\cite{vermeulen2015bias}, where a similar asymptotic expansion similar as described above is obtained.
The two methods are equivalent in some problems such as estimation of average treatment effects, where a similar method is also proposed in \cite{kim2014doubly}.
However, there exists a general difference: estimating equations in \cite{vermeulen2015bias}
are defined from the influence function of a doubly robust estimator as originally motivated to achieve bias reduction,
whereas our calibration equations are defined from a doubly robust estimating function to achieve a desired asymptotic expansion.
For instance, see Examples~\ref{eg:PLM-log} and \ref{eg:PLM-log-EE} for differences of the two methods in partially log-linear models.

In high-dimensional settings, doubly robust estimating functions are used with regularized likelihood (or quasi-likelihood) estimators of $(\alpha,\gamma)$
in \cite{belloni2014high} and \cite{farrell2015robust}. Valid confidence intervals are established
under suitable sparsity conditions, when both models $g(\cdot;\alpha)$ and $f(\cdot;\gamma)$ are correctly specified.
For inference about average treatment effects, doubly robust confidence intervals are obtained in \cite{tan2020model} if
either a propensity score model or a linear outcome model is correctly specified.
In this case, regularized calibration estimators of $\alpha$ and then $\gamma$ are determined sequentially, independent of $\theta$.
For a nonlinear outcome model, only model-assisted confidence intervals are established, provided a propensity score model is correctly specified but
the outcome model may be misspecified. In this case and other problems (see Examples \ref{eg:PLM-log-EE}--\ref{eg:ATE-EE}),
there are computational and theoretical complications due to coupled calibration equations.
To tackle these issues, we develop the two-step algorithm and appropriate high-dimensional analysis, to obtain
doubly robust confidence intervals which are not only computationally tractable but also
theoretically justified in general settings where doubly robust estimating functions are available.

For estimating average treatment effects, \cite{avagyan2017honest} proposed a regularized version of
estimating equations in \cite{vermeulen2015bias}. But their theoretical analysis appears to
presume standard convergence rates for the estimators of $(\alpha,\gamma)$ without handling additional data-dependency in loss functions.
\cite{Ning2020} proposed doubly robust confidence intervals, but their method is operationally more complicated than \cite{tan2020model} and our work.
With a nonlinear outcome model, the method in \cite{Ning2020} involves
first three steps which yield the same estimators $(\hat\alpha_2, \hat\gamma_2)$ as in our Example~\ref{eg:ATE-RCAL},
but then performs a fourth step to adjust the fitted propensity score before applying the augmented inverse probability weighted estimator.
In addition, the fourth step relies on variable selection properties, which may require stronger technical conditions than convergence of estimation errors used in our method.
An artificial constraint on the parameter set is also added in the proofs of Theorems 1 and 2 in \cite{Ning2020}.

\cite{smucler2019unifying} made a distinction between two types of doubly robust estimation in high-dimensional settings.
Our method as well as those in \cite{avagyan2017honest}, \cite{tan2020model} and \cite{Ning2020} achieves model double robustness:
an estimator of $\theta$ is obtained of order $O_p(n^{-1/2})$ if either of models $g(\cdot;\alpha)$ and $f(\cdot; \gamma)$ is correctly specified,
under sparsity conditions $s_{\bar\alpha} = o(n^{1/2})$ and $s_{\bar\gamma} = o(n^{1/2})$ up to $\log(p)$ terms, where $s_{\bar\alpha}$ or $s_{\bar\gamma}$
is the number of nonzero elements of the target values $\bar\alpha$ or $\bar\gamma$.
By comparison, several methods have been proposed to achieve rate double robustness:
an estimator of $\theta$ is obtained of order $O_p(n^{-1/2})$ if both models $g(\cdot;\alpha)$ and $f(\cdot; \gamma)$ are correctly specified,
under sparsity conditions $s_{\bar\alpha} s_{\bar\gamma} = o(n)$ or weaker (\citealt{chernozhukov2018double, smucler2019unifying, bradic2019minimax}).
All of these methods appear to rely on sample splitting and cross fitting, which is not pursued in our work.
In particular, the method of \cite{smucler2019unifying} is shown to achieve rate and model double robustness simultaneously in general settings
where the parameter $\theta$ has an influence function in a certain bilinear form.
However, our method is applicable to any doubly robust estimating function including that in a partially logistic model in our Example~\ref{eg:PLM-logit},
which does not satisfy the bilinear condition.

Finally, our work is also connected to debiased Lasso mentioned earlier and extensions (\citealt{neykov2018unified})
to obtain confidence intervals and tests for low-dimensional coefficients in high-dimensional models.
These methods in general do not achieve double robustness. See Examples~\ref{eg:PLM-RCAL}--\ref{eg:PLM-logit-RCAL} on partially linear models for further discussion.


\section{Double robustness and calibrated estimation} \label{sec:dr-cal}

\subsection{Doubly robust estimation}

Let $\{U_i : i=1,\ldots,n\}$ be independent and identically distributed observations as $U$, which is assumed to
include a covariate vector $X$ taking values $x$ in a space $\mathcal X$.
Consider semiparametric estimation based on an estimating equation
\begin{align}
0 = \tilde E \{\tau(U; \theta, g, f)\} = \frac{1}{n} \sum_{i=1}^n \tau(U_i; \theta,g,f),  \label{eq:EE}
\end{align}
where $\tilde E()$ denotes a sample average, $\tau(U;\theta,g,f)$ is an estimating function, $\theta$ is a scalar parameter of interest in $\Theta$,
and $g$ and $f$ are two variation-independent nuisance parameters, defined in some function spaces $\mathcal G$ and $\mathcal F$ on $\mathcal X$.
Denote by $(\theta^*, g^*, f^*)$ the true values (i.e., data-generating values) of $(\theta,g,f)$.
Assume that the estimating function $\tau(U;\theta,g,f)$ is doubly robust in satisfying the following two properties:
\begin{align}
& 0 = E \{\tau(U; \theta^*, g^*, f)\} \mbox{ for any } f \in \mathcal F, \label{eq:DR1} \\
& 0 = E \{\tau(U; \theta^*, g, f^*)\} \mbox{ for any } g \in \mathcal G . \label{eq:DR2}
\end{align}
In other words, $\tau(U;\theta,g,f)$ is unbiased for estimation of $\theta^*$ if either $g=g^*$ or $f=f^*$.
Several examples of doubly robust estimating functions are as follows.
Construction of doubly robust estimating functions is problem-dependent and not discussed here. See \cite{robins2001comment}, \cite{tchetgen2010doubly} and \cite{tan2019doubly} among others.


\begin{eg} \label{eg:PLM}
Suppose that an outcome $Y$ is related to a covariate $Z$ and additional covariates $X$ in a partially linear model
\begin{align}
E(Y | Z, X) &=\theta^* Z + g^*(X), \label{eq:PLM}
\end{align}
where $\theta^*$ is the true value of a coefficient $\theta$ and $g^*(x)$ is the true value of a function $g(x)$.
In addition to $g(\cdot)$, define a nuisance parameter $f(\cdot)$ such that $f^*(X) = E( Z | X)$.
Then the following estimating function is doubly robust (\citealt{robins2001comment}),
\begin{align}
\tau( U; \theta, g, f) = \{Y - \theta Z - g(X)\} \{Z - f(X)\} , \label{eq:PLM-tau}
\end{align}
where $U = (Y,Z,X)$. The true value $\theta^*$ can be regarded as a homogeneous additive treatment effect, in the setting where
$Z$ is a treatment variable.
\end{eg}

\begin{eg} \label{eg:PLM-log}
Consider a partially log-linear model
\begin{align}
E(Y | Z, X) &=\exp\{ \theta^* Z + g^*(X)\}, \label{eq:PLM-log}
\end{align}
where $\theta^*$ is the true value of a coefficient $\theta$ and $g^*(x)$ is the true value of a function $g(x)$.
The nuisance parameter $f(x)$ is still defined such that $f^*(x) = E(Z|X)$.
Then the following estimating function is doubly robust (\citealt{robins2001comment}),
\begin{align}
\tau( U; \theta, g, f) = \{Y \me^{-\theta Z} - \me^{g(X)} \} \{Z - f(X)\} , \label{eq:PLM-log-tau}
\end{align}
where $U = (Y,Z,X)$. The true value $\theta^*$ can be regarded as a homogeneous multiplicative treatment effect, in the setting where
$Z$ is a treatment variable.
\end{eg}

\begin{eg} \label{eg:PLM-logit}
Consider a partially logistic model with binary $Y$,
\begin{align}
E(Y | Z, X) &=\expit \{ \theta^* Z + g^*(X)\}, \label{eq:PLM-logit}
\end{align}
where $\expit(c)=(1+\me^{-c})^{-1}$, $\theta^*$ is the true value of  $\theta$ and $g^*(x)$ is the true value of $g(x)$.
In contrast with Examples~\ref{eg:PLM}--\ref{eg:PLM-log}, define a nuisance parameter $f(\cdot)$ such that $f^*(X) = E( Z | Y=0,X)$.
Then a doubly robust estimating function is (\citealt{tan2019doubly})
\begin{align}
\tau( U; \theta, g, f) = \me^{-\theta Z Y} \{Y- \expit(g(X)) \} \{Z - f(X)\} , \label{eq:PLM-logit-tau}
\end{align}
where $U = (Y,Z,X)$. The true value $\theta^*$ can be regarded as a homogeneous treatment effect in the scale of log odds, in the setting where
$Z$ is a treatment variable.
\end{eg}

\begin{eg} \label{eg:ATE}
Let $Y$ be an outcome variable, $X$ a covariate vector, and $Z$ a binary variable such that $Z=1$ or 0 if $Y$ is observed or missing respectively.
Assume that the missing data mechanism is ignorable: $Y$ and $Z$ are conditionally independent given $X$ (\citealt{rubin1976inference}).
It is of interest to estimate the mean $\theta^* = E(Y)$. The nuisance parameters $g(\cdot)$ and $f(\cdot)$ are defined such that the true values are
\begin{align*}
g^*(X) = E(Y | Z=1, X) , \quad f^*(X) = P(Z=1 | X) ,
\end{align*}
which are called outcome regression function and propensity score.
Then the following estimating function is doubly robust (\citealt{scharfstein1999adjusting}),
\begin{align}
\tau( U; \theta, g, f) = \frac{Z Y}{f(X)} - \left\{ \frac{Z}{f(X)} -1 \right\} g(X) - \theta, \label{eq:ATE-tau}
\end{align}
where $U = (ZY,Z,X)$. The true value $\theta^*$ represents the mean of a potential outcome associated with a treatment
when $Z$ encodes the receipt of the treatment.
\end{eg}

Typically, estimating equation (\ref{eq:EE}) is used in the form of two-stage semiparametric estimation, depending on
some modeling restrictions, $g(x; \alpha)$ and $f(x;\gamma)$ with parameters $\alpha$ and $\gamma$, postulated on $(g^*,f^*)$.
For concreteness, consider the following two models,
\begin{align}
g^*(x)=  g(x; \alpha) = \psi_g \{ \alpha^\T \xi(x) \}, \label{eq:g-model} \\
f^*(x) = f(x; \gamma) = \psi_f \{ \gamma^\T \xi(x) \} , \label{eq:f-model}
\end{align}
where $\psi_g$ and $\psi_f$ are inverse link functions similarly as in generalized linear models (\citealt{mccullagh1989}), $\xi(x)$ is a $p \times 1$ vector of known functions on $\mathcal X$
such as $\xi(x) = (1,x^\T)^\T$, and $\alpha$ and $\gamma$ are $p\times 1$ vectors of unknown coefficients.
Models (\ref{eq:g-model}) and (\ref{eq:f-model}) may be misspecified.
We say that model (\ref{eq:g-model}) is correctly specified if there exists a true value $\alpha^*$ such that $g^*(x) \equiv g (x; \alpha^*)$, or misspecified otherwise.
Similarly, model (\ref{eq:f-model}) is correctly specified if there exists a true value $\gamma^*$ such that $f^*(x) \equiv f(x; \gamma^*) $, or misspecified otherwise.
By definition, a true value $\alpha^*$ or $\gamma^*$ exists only if model (\ref{eq:g-model}) or (\ref{eq:f-model}) is correctly specified.

Given working models (\ref{eq:g-model})--(\ref{eq:f-model}), the first-stage estimation involves constructing some estimators $\hat\alpha$ and $\hat\gamma$
and setting $\hat g= g(x; \hat\alpha)$ and $\hat f = f(x; \hat\gamma)$.
Then an estimator for $\theta^*$, denoted as $\hat\theta(\hat\alpha,\hat\gamma)$, is
defined as a solution to (\ref{eq:EE}) with $(g,f)$ replaced by $(\hat g, \hat f)$, i.e.,
\begin{align}
0 = \tilde E \{\tau(U; \theta, \hat g, \hat f)\}  \label{eq:EE2}.
\end{align}
Conventionally, $(\hat\alpha, \hat\gamma)$ are defined by maximum likelihood (or quasi-likelihood) including least squares
in generalized linear models associated with (\ref{eq:g-model})--(\ref{eq:f-model}).
Our main subject is, however, calibrated estimation as an alternative approach.
To facilitate discussion in Section~\ref{sec:cal}, we describe some general asymptotic results about $(\hat\alpha,\hat\gamma)$
and $\hat\theta(\hat\alpha,\hat\gamma)$, based on theory of estimation with possibly misspecified models (\citealt{white1982maximum, manski1988analog}),
in the classical setting where $\alpha$ and $\gamma$ are fixed-dimensional as the sample size $n$ grows.
To focus on main issues, assume that
$\hat\alpha$ is consistent for $\alpha^*$ if model (\ref{eq:g-model}) is correctly specified,
and $\hat\gamma$ is consistent for $\gamma^*$ if model (\ref{eq:f-model}) is correctly specified.

With possible model misspecification, $\hat\alpha$ can be shown to converge at rate $O_p(n^{-1/2})$ to a target value $\bar \alpha$, which
coincides with the true value $\alpha^*$ (i.e., $\hat\alpha$ is consistent) if model (\ref{eq:g-model}) is correctly specified, but remains well-defined even though $\alpha^*$ is undefined
if model (\ref{eq:g-model}) is misspecified. Similarly,
$\hat\gamma$ can be shown to converge  at rate $O_p(n^{-1/2})$ to a target value $\bar \gamma$, which coincides
with the true value $\gamma^*$ (i.e., $\hat\gamma$ is consistent) if model (\ref{eq:g-model}) is correctly specified, but remains well-defined even though $\gamma^*$ is undefined
if model (\ref{eq:f-model}) is misspecified.
As a result, unbiasedness properties (\ref{eq:DR1})--(\ref{eq:DR2}) can be used to show that $\hat\theta ( \hat\alpha,  \hat\gamma)$ is doubly robust, i.e., remains
consistent for $\theta^*$ if either model (\ref{eq:g-model}) or (\ref{eq:f-model}) is correctly specified.
Moreover, it can be shown that
if model (\ref{eq:g-model}) is correctly specified with $\bar\alpha= \alpha^*$ or model (\ref{eq:f-model}) is correctly specified with $\bar\gamma = \gamma^*$, then
$\hat\theta ( \hat\alpha,  \hat\gamma)$ admits the asymptotic expansion,
\begin{align}
\hat \theta (\hat\alpha, \hat\gamma) - \theta^* =  E^{-1} \left( \frac{\partial \tau}{\partial\theta} \right)
\left\{ \tilde E (\tau) + E^\T \left( \frac{\partial \tau}{\partial \alpha} \right) (\hat \alpha - \bar \alpha) +
E^\T \left( \frac{\partial \tau}{\partial \gamma} \right) (\hat \gamma - \bar \gamma) \right\} + o_p (n^{-1/2}), \label{eq:general-expansion}
\end{align}
where $\tau = \tau( U; \theta, g(x;\alpha), f(x;\gamma))$, and $\tau$ and its partial derivatives
$(\partial\tau/ \partial\theta, \partial\tau/ \partial \alpha, \partial\tau/ \partial \gamma)$
are evaluated above at $(\theta^*,\bar\alpha,\bar \gamma)$.
The preceding expansion (\ref{eq:general-expansion}) indicates how the asymptotic behavior of $\hat\theta(\hat\alpha, \hat\gamma)$ is affected by
the estimators $(\hat\alpha,\hat\gamma)$ through the second and third terms in the curly brackets.
In fact, removing these two terms in (\ref{eq:general-expansion}) yields the asymptotic expansion of the infeasible estimator $\hat\theta( \bar \alpha, \bar \gamma)$,
with $(\hat\alpha,\hat\gamma)$ replaced by $(\bar\alpha, \bar \gamma)$.

\begin{eg} \label{eg:PLM-LS}
We point out a somewhat under-appreciated result that the familiar least squares estimator for each individual coefficient in linear regression is doubly robust
in the context of a partially linear model in Example~\ref{eg:PLM}. Let $\psi_g(\cdot)$ be an identity function in model (\ref{eq:g-model}).
For $\tau$ in (\ref{eq:PLM-tau}), the estimator $\hat\theta(\hat\alpha,\hat\gamma)$ as a solution to (\ref{eq:EE2}) is of closed form with $\xi=\xi(X)$,
\begin{align*}
\hat\theta (\hat\alpha,\hat\gamma) = \frac{\tilde E\{ (Y - \hat\alpha^\T \xi) (Z - \psi_f(\hat\gamma^\T \xi)) \}}{ \tilde E\{Z(Z - \psi_f(\hat\gamma^\T \xi)) \}},
\end{align*}
depending on some estimators $(\hat\alpha,\hat\gamma)$. Suppose that $\psi_f(\cdot)$ is also an identity function, i.e., a linear model is specified for $E(Z|X)$.
Let $(\hat\theta_0, \hat\alpha)$ be the least-squares estimators of $(\theta,\alpha)$ in the linear regression of $Y$ on $Z$ and $\xi(X)$,
and $\hat\gamma$ be that of $\gamma$ in the linear regression of $Z$ on $\xi(X)$. Then
$\hat\theta(\hat\alpha,\hat\gamma)$ is identical to $\hat\theta_0$, the least squares estimator of $\theta$:
\begin{align*}
\hat\theta (\hat\alpha,\hat\gamma) -\hat\theta_0 = \frac{\tilde E\{ (Y -\hat\theta_0 Z - \hat\alpha^\T \xi) (Z -\hat\gamma^\T \xi) \}}{ \tilde E\{Z(Z - \hat\gamma^\T \xi) \}} =0,
\end{align*}
because $\tilde E \{(Y -\hat\theta_0 Z - \hat\alpha^\T \xi) Z \} =0$ and $\tilde E \{(Y -\hat\theta_0 Z - \hat\alpha^\T \xi) \xi \} =0$.
Hence the least-squares estimator $\hat\theta_0$ is doubly robust for $\theta^*$ in the partially linear model (\ref{eq:PLM}), if either a linear model
for $g^*(x)$ or a linear model for $f^*(x) =E(Z|X=x)$ is correctly specified.
Furthermore, the sandwich variance estimator for $\hat\theta_0$ (\citealt{White1980}) can be written as $n^{-1} \hat V$ with
\begin{align*}
\hat V =  \frac{\tilde E\{ (Y -\hat\theta_0 Z - \hat\alpha^\T \xi)^2 (Z -\hat\gamma^\T \xi)^2 \}}{\tilde E^2\{Z(Z - \hat\gamma^\T \xi) \}} .
\end{align*}
By Corollary~\ref{cor:DR-expansion} later, an asymptotic $(1-c)$-confidence interval for $\theta^*$ is $\hat\theta_0 \pm z_{c/2} \sqrt{\hat V/n}$ if
either a linear model for $E(Y|Z,X)$ or that for $E(Z|X)$ is correctly specified.
A high-dimensional version of this result is Corollary~\ref{cor:debiased} later on debiased Lasso for least-squares estimation.
\end{eg}

\subsection{Calibrated estimation} \label{sec:cal}

We derive and discuss implications of basic mean-zero identities for a doubly robust estimating function $\tau (U;\theta,g,f)$.
In particular, we study calibrated estimation converting these identities into estimating equations in $(\alpha,\gamma)$.
Here we assume the classical setting where asymptotic expansion (\ref{eq:general-expansion}) directly holds.
See Section~\ref{sec:rcal} for high-dimensional development.

For a function $h(x)$ and a constant $\delta >0$, denote by $h+ L_2 (\delta) $ the set $\{ h(x) + c(x) : E ( c^2(X) ) \le \delta^2 \}$.
Denote by $\partial\tau /\partial g$ and $\partial\tau /\partial f$ the partial derivatives of $\tau=\tau(U;\theta,g,f)$ with respect to  $G=g(x)$ and $F=f(x)$
as free arguments.
Whenever the dependency of $\tau$ on $(\alpha,\gamma)$ is mentioned, $\tau$ is parameterized as $\tau(U;\theta,\alpha,\gamma)=\tau(U; \theta, g(x;\alpha), f(x;\gamma))$.
For differentiation of $\tau$ with respect to $(\alpha,\gamma)$, it is convenient to introduce linear predictors $(\eta_g,\eta_g)$ such that
$ g(x) = \psi_g (\eta_g(x))$ and $f(x) = \psi_f (\eta_f(x))$.
Hence models (\ref{eq:g-model}) and (\ref{eq:f-model}) can be stated as $\eta_g(x;\alpha )= \alpha^\T \xi(x)$ and $\eta_f(x;\gamma) = \gamma^\T \xi(x)$.
Denote by $\partial\tau /\partial \eta_g$ and $\partial\tau /\partial \eta_f$ the partial derivatives of $\tau$ with respect to
$\eta_g(x)$ and $\eta_f(x)$ as free arguments.
By the chain rule, $\partial\tau /\partial \alpha = (\partial\tau/\partial \eta_g )\xi =  (\partial\tau /\partial g) \psi^\prime_g (\alpha^\T \xi) \xi $
and $\partial\tau /\partial \gamma =( \partial\tau/\partial\eta_f) \xi = (\partial\tau /\partial f) \psi^\prime_f (\gamma^\T \xi) \xi $,
where $\psi^\prime_g$ or $\psi^\prime_f$ denotes the derivative of $\psi_g$ or $\psi_f$.

\begin{pro}
	Under suitable regularity conditions, property (\ref{eq:DR1}) implies that
	\begin{align}
	0 = E \left\{ \frac{\partial \tau}{\partial f} (U; \theta^*, g^*, f) \Big| X \right\} , \label{eq:DR-conseq1}
	\end{align}
	for any $f$ such that $f+L_2(\delta_1) \subset \mathcal F$ for some $\delta_1 >0$.
	Similarly, property (\ref{eq:DR2}) implies that
	\begin{align}
	0 = E \left\{ \frac{\partial \tau}{\partial g} (U; \theta^*, g, f^*) \Big| X \right\} , \label{eq:DR-conseq2}
	\end{align}
	for any $g$ such that $g+L_2(\delta_2) \subset \mathcal G$  for some $\delta_2 >0$.
\end{pro}

\begin{proof}
	For $f$ such that $f+L_2(\delta_1) \subset \mathcal F$, (\ref{eq:DR1}) implies that for any $h\in L_2(1)$ and $a \in [-\delta_1, \delta_1]$,
	\begin{align*}
	0 = E \{f (U; \theta^*, g^*, f + a h  )\} .
	\end{align*}
	Taking the derivative of the above with respect to $a$ with $f$ and $h$ fixed, and assuming the differentiation and expectation are interchangeable, we have
	\begin{align*}
	0 = E \left\{ \frac{\partial \tau}{\partial f} (U; \theta^*, g^*, f) h(X) \right\} .
	\end{align*}
	Hence (\ref{eq:DR-conseq1}) follows because $h\in L_2(1)$ is arbitrary. Similarly, (\ref{eq:DR-conseq2}) can be proved.
\end{proof}

Similar reasoning as above can be applied to the derivatives of $\tau$ with respect to $(\alpha,\gamma)$, given  models (\ref{eq:g-model})--(\ref{eq:f-model}).
Differentiation of (\ref{eq:DR1}) or (\ref{eq:DR2}) with respect to $\gamma$ or $\alpha$ respectively and interchanging differentiation and expectation shows that
for any $(\alpha,\gamma)$,
\begin{align}
0 = E \left\{ \frac{\partial \tau}{\partial \gamma} (U; \theta^*, g^*, f(x;\gamma))  \right\}
= E \left\{  \xi(X)  \frac{\partial \tau}{\partial \eta_f} (U; \theta^*, g^*, f(x;\gamma)) \right\} , \label{eq:DR-EE1} \\
0 = E \left\{ \frac{\partial \tau}{\partial \alpha} (U; \theta^*, g(x;\alpha), f^*)  \right\}
= E \left\{ \xi(X) \frac{\partial \tau}{\partial \eta_g} (U; \theta^*,  g(x;\alpha), f^*) \right\}  . \label{eq:DR-EE2}
\end{align}
Equivalently, (\ref{eq:DR-EE1})--(\ref{eq:DR-EE2}) can also be deduced from the more general identities (\ref{eq:DR-conseq1})--(\ref{eq:DR-conseq2}),
which involve conditional expectations given $X$.
Model (\ref{eq:g-model}) with $g(x;\alpha)$ may be misspecified in (\ref{eq:DR-EE2}), and model (\ref{eq:f-model}) with $f(x;\gamma)$ may be misspecified in (\ref{eq:DR-EE1}).

We stress that identities (\ref{eq:DR-conseq1})--(\ref{eq:DR-conseq2}) and (\ref{eq:DR-EE1})--(\ref{eq:DR-EE2})
are derived from double-robustness properties (\ref{eq:DR1})--(\ref{eq:DR2}) in a general manner.
To some extent, identities (\ref{eq:DR-EE1})--(\ref{eq:DR-EE2}) are intriguingly reminiscent of the score identity in likelihood inference with a parametric model:
the expectation of the gradient of the log-likelihood, evaluated at the true parameter value, is zero.
However, $\tau$ is an estimating function in $\theta$, not a log-likelihood function in $\alpha$ or $\gamma$.

There are various implications of basic identities (\ref{eq:DR-EE1})--(\ref{eq:DR-EE2}).
First, these identities show that $E(\partial\tau/\partial\gamma)$ or $E(\partial\tau /\partial\alpha)$ reduces to 0 in
asymptotic expansion (\ref{eq:general-expansion}) for $\hat\theta(\hat\alpha,\hat\gamma)$, depending on whether model (\ref{eq:g-model})
or (\ref{eq:f-model}) is correctly specified.
If model (\ref{eq:g-model}) with $g(x;\alpha)$ is correctly specified and $\hat\alpha$ is consistent,
then, by (\ref{eq:DR-EE1}), asymptotic expansion (\ref{eq:general-expansion}) reduces to
\begin{align}
\hat \theta (\hat\alpha, \hat\gamma) - \theta^* = - E^{-1} \left( \frac{\partial \tau}{\partial\theta} \right)
\left\{ \tilde E (\tau) +  E^\T \left( \frac{\partial \tau}{\partial \alpha} \right) (\hat \alpha - \alpha^*) \right\} + o_p (n^{-1/2}), \label{eq:g-cor-expansion}
\end{align}
where $\tau$ and its partial derivatives are evaluated at $(\theta,\alpha,\gamma) = (\theta^*,\alpha^*,\bar \gamma)$.
As the term associated with $\hat\gamma-\bar \gamma$ vanishes in (\ref{eq:g-cor-expansion}), the asymptotic behavior of $\hat\theta(\hat\alpha,\hat\gamma)$ does not depend on the definition of $\hat\gamma$,
as long as model (\ref{eq:g-model}) is correctly specified and $\hat\alpha$ is consistent. Similarly,
if model (\ref{eq:f-model}) with $f(x;\gamma)$ is correctly specified and $\hat\gamma$ is consistent,
then, by (\ref{eq:DR-EE2}), the asymptotic behavior of $\hat\theta(\hat\alpha,\hat\gamma)$ does not depend on the definition of $\hat\alpha$:
\begin{align}
\hat \theta (\hat\alpha, \hat\gamma) - \theta^* = - E^{-1} \left( \frac{\partial \tau}{\partial\theta} \right)
\left\{ \tilde E (\tau) + E^\T \left( \frac{\partial \tau}{\partial \gamma} \right) (\hat \gamma - \gamma^*) \right\} + o_p (n^{-1/2}), \label{eq:f-cor-expansion}
\end{align}
where $\tau$ and its partial derivatives are evaluated at $(\theta,\alpha,\gamma) = (\theta^*,\bar \alpha,\gamma^*)$.
Combining the preceding arguments leads to Corollary 1: if both models (\ref{eq:g-model}) and (\ref{eq:f-model}) are correctly specified,
then the asymptotic behavior of $\hat\theta(\hat\alpha,\hat\gamma)$ remains the same for all consistent estimators $(\hat\alpha,\hat\gamma)$.
This result, related to local efficiency in specific examples (e.g., \citealt{robins1994estimation, tan2006distributional}), is obtained here as a general consequence of double robustness of $\tau$.

\begin{cor} \label{cor:corr-expansion}
If both models (\ref{eq:g-model}) and (\ref{eq:f-model}) are correctly specified and $(\hat\alpha,\hat\gamma)$ are consistent,
then as $p$ is fixed and $n\to\infty$, asymptotic expansion (\ref{eq:general-expansion}) reduces to
\begin{align}
\hat \theta (\hat\alpha, \hat\gamma) - \theta^* = - E^{-1} \left( \frac{\partial \tau}{\partial\theta} \right)
\tilde E (\tau) + o_p (n^{-1/2}), \label{eq:desired-expansion}
\end{align}
where $\tau$ and $\partial\tau/\partial\theta$ are evaluated at $(\theta,\alpha,\gamma) = (\theta^*,\alpha^*,\gamma^*)$.
\end{cor}

Second, methodologically, identities (\ref{eq:DR-EE1})--(\ref{eq:DR-EE2}) can also be exploited to construct specific estimators $(\hat \alpha,\hat\gamma)$,
for which the simple expansion (\ref{eq:desired-expansion}) is valid with the true values $(\alpha^*,\gamma^*)$ replaced by target values $(\bar\alpha,\bar\gamma)$
if either model (\ref{eq:g-model}) or (\ref{eq:f-model}), but not necessarily both, is correctly specified.
Suppose that  estimators $(\hat\alpha_{\mbox{\tiny CAL}}, \hat\gamma_{\mbox{\tiny CAL}})$ are defined such that they converge in probability to target values
$(\bar \alpha_{\mbox{\tiny CAL}}, \bar \gamma_{\mbox{\tiny CAL}})$ satisfying the simultaneous equations
\begin{align}
0 &=  E \left\{ \frac{\partial \tau}{\partial \gamma} (U; \theta^*, \alpha,\gamma)  \right\} =  E \left\{ \xi \frac{\partial \tau}{\partial \eta_f} (U; \theta^*, \alpha,\gamma)  \right\}, \label{eq:CAL1} \\
0 &=  E \left\{ \frac{\partial \tau}{\partial \alpha} (U; \theta^*, \alpha,\gamma)  \right\} =  E \left\{ \xi \frac{\partial \tau}{\partial \eta_g} (U; \theta^*, \alpha,\gamma)  \right\}, \label{eq:CAL2}
\end{align}
that is, the coefficients of $\hat\gamma-\bar\gamma$ and $\hat\alpha-\bar\alpha$ are set to 0 in expansion (\ref{eq:general-expansion}) for $\hat\theta(\hat\alpha,\hat\gamma)$.
Assume that there exists at most one value $\alpha$ satisfying (\ref{eq:CAL1}) for each fixed $\gamma$, and
 at most one value $\gamma$ satisfying (\ref{eq:CAL2}) for each fixed $\alpha$.
From our discussion below, this implies that $(\bar \alpha_{\mbox{\tiny CAL}}, \bar \gamma_{\mbox{\tiny CAL}})$ is
a unique solution to (\ref{eq:CAL1})--(\ref{eq:CAL2}) if model (\ref{eq:g-model}) or (\ref{eq:f-model}) is correctly specified.

If model (\ref{eq:g-model}) with $g(x;\alpha)$ is correctly specified,
then by (\ref{eq:DR-EE1}), $\bar\alpha_{\mbox{\tiny CAL}}$ coincides with $\alpha^*$ as a solution to (\ref{eq:CAL1}) for fixed $\gamma=\bar\gamma_{\mbox{\tiny CAL}}$,
i.e., $\hat\alpha_{\mbox{\tiny CAL}}$ is consistent.
In this case, (\ref{eq:CAL1}) can be seen as an unbiased population estimating equation for $\alpha^*$ with fixed $\gamma$.
Similarly, if model (\ref{eq:f-model}) with $f(x;\gamma)$ is correctly specified, then
by comparison of (\ref{eq:DR-EE2}) and (\ref{eq:CAL2}), $\bar\gamma_{\mbox{\tiny CAL}}$ coincides with $\gamma^*$, i.e., $\hat\gamma_{\mbox{\tiny CAL}}$ is consistent.
In this case, (\ref{eq:CAL2}) can be seen as an unbiased population estimating equation for $\gamma^*$ with fixed $\alpha$.
  (An interesting asymmetry is that
	differentiation of $\tau$ with respect to $\gamma$ leads to an estimating equation in $\alpha$, whereas
	that of $\tau$ with respect to $\alpha$ leads to an estimating equation in $\gamma$.)
Combining the two cases and applying asymptotic expansion (\ref{eq:general-expansion}) leads to the Corollary~\ref{cor:DR-expansion}, where, due to (\ref{eq:CAL1})--(\ref{eq:CAL2}) again,
the two terms associated with $\hat\gamma-\bar\gamma$ and $\hat\alpha-\bar\alpha$ are dropped from the expansion (\ref{eq:general-expansion}).
Alteratively, to help understanding, asymptotic expansion (\ref{eq:desired-expansion2}) for $\hat \theta (\hat\alpha_{\mbox{\tiny CAL}}, \hat\gamma_{\mbox{\tiny CAL}}) $ can also be
obtained from expansion (\ref{eq:g-cor-expansion}) with $E(\partial\tau/\partial\alpha)=0$ due to (\ref{eq:CAL2}) if model (\ref{eq:g-model}) is correctly specified,
or from expansion (\ref{eq:f-cor-expansion}) with $E(\partial\tau/\partial\gamma)=0$ due to (\ref{eq:CAL1}) if model (\ref{eq:f-model}) is correctly specified.

\begin{cor} \label{cor:DR-expansion}
If model (\ref{eq:g-model}) or (\ref{eq:f-model}) is correctly specified, then
$\hat\alpha_{\mbox{\tiny CAL}}$ or $\hat\gamma_{\mbox{\tiny CAL}}$ is consistent for $\alpha^*$ or $\gamma^*$ respectively.
In either case, the estimator
$\hat \theta (\hat\alpha_{\mbox{\tiny CAL}}, \hat\gamma_{\mbox{\tiny CAL}}) $ satisfies
\begin{align}
\hat \theta (\hat\alpha_{\mbox{\tiny CAL}}, \hat\gamma_{\mbox{\tiny CAL}}) - \theta^* = - E^{-1} \left( \frac{\partial \tau}{\partial\theta} \right)
\tilde E (\tau) \Big|_{ (\theta,\alpha,\gamma) =(\theta^*,\bar \alpha_{\mbox{\tiny CAL}}, \bar \gamma_{\mbox{\tiny CAL}})}+ o_p (n^{-1/2}),  \label{eq:desired-expansion2}
\end{align}
provided that expansion (\ref{eq:general-expansion}) holds for $(\hat\alpha,\hat\gamma)=(\hat\alpha_{\mbox{\tiny CAL}}, \hat\gamma_{\mbox{\tiny CAL}})$ as $p$ is fixed and $n\to\infty$.
 \end{cor}

We refer to equations (\ref{eq:CAL1})--(\ref{eq:CAL2}) as population calibration equations and
$(\hat\alpha_{\mbox{\tiny CAL}}, \hat\gamma_{\mbox{\tiny CAL}})$ as calibrated estimators for two reasons, following \cite{tan2020model}.
For the missing-data problem in Example~\ref{eg:ATE}, related to estimation of average treatment effects, this method leads to calibrated estimation
for fitting propensity score models $f(x;\gamma)$, which can be traced to the literature on survey calibration (\citealt{folsom1991}). See Example~\ref{eg:ATE-EE} below.
More generally, as indicated by Corollaries~\ref{cor:corr-expansion}--\ref{cor:DR-expansion},
using estimating equations (\ref{eq:CAL1})--(\ref{eq:CAL2}) can be seen as carefully choosing (or calibrating) estimators $(\hat\alpha,\hat\gamma)$
for the nuisance parameters $(\alpha,\gamma)$,
such that the resulting estimator $\hat\theta(\hat\alpha,\hat\gamma)$ behaves as if both models (\ref{eq:g-model}) and (\ref{eq:f-model}) were correctly specified,
while it is only assumed that either model (\ref{eq:g-model}) or (\ref{eq:f-model}) is correctly specified.

A benefit of achieving asymptotic expansion (\ref{eq:desired-expansion2}) is to allow simple variance estimation for
$\hat \theta (\hat\alpha_{\mbox{\tiny CAL}}, \hat\gamma_{\mbox{\tiny CAL}})$, without the need to account for the variations of
$(\hat\alpha_{\mbox{\tiny CAL}},\hat\gamma_{\mbox{\tiny CAL}})$. This benefit is mainly computationally in the setting of
low-dimensional $(\alpha,\gamma)$, where variance estimation can in general be performed for $\hat\theta(\hat\alpha,\hat\gamma)$ by using asymptotic expansion (\ref{eq:general-expansion})
and usual influence functions for $(\hat\alpha,\hat\gamma)$, allowing for model misspecification (\citealt{white1982maximum, manski1988analog}).
However, the influence-function based approach is not applicable in the high-dimensional setting where regularized estimation is involved.
In Section~\ref{sec:rcal}, we develop regularized calibration estimation to achieve a simple expansion similar to (\ref{eq:desired-expansion2}) for the resulting estimator of $\theta^*$,
so that valid variance estimation and confidence intervals can be obtained.

\begin{rem} \label{rem:expansion}
It is important to distinguish the two expansions (\ref{eq:desired-expansion}) and (\ref{eq:desired-expansion2}), although they appear similar to each other.
The expansion (\ref{eq:desired-expansion}) holds for any consistent estimators $(\hat\alpha,\hat\gamma)$ provided that both models (\ref{eq:g-model}) and (\ref{eq:f-model}) are correctly specified.
The two terms $E(\partial\tau /\partial \alpha)$ and $E(\partial\tau / \partial\gamma)$ in (\ref{eq:general-expansion}) reduce to 0 by the assumption of
both models (\ref{eq:g-model}) and (\ref{eq:f-model}) being correctly specified, while appealing to the two identities (\ref{eq:DR-EE1})--(\ref{eq:DR-EE2}) simultaneously.
In contrast, the expansion (\ref{eq:desired-expansion2}) is valid for estimators $(\hat\alpha_{\mbox{\tiny CAL}}, \hat\gamma_{\mbox{\tiny CAL}}) $ constructed such that
(\ref{eq:CAL1})--(\ref{eq:CAL2}) are satisfied, if either model (\ref{eq:g-model}) or (\ref{eq:f-model}), but not necessarily both, is correctly specified.
The two terms $E(\partial\tau /\partial \alpha)$ and $E(\partial\tau / \partial\gamma)$ in (\ref{eq:general-expansion}) reduce to 0 by the construction
of population estimating equations (\ref{eq:CAL1})--(\ref{eq:CAL2}).
Identity (\ref{eq:DR-EE1}) is involved to show consistency of $\hat\alpha_{\mbox{\tiny CAL}}$
if model (\ref{eq:g-model}) is correct or, separately, identity (\ref{eq:DR-EE2}) is involved to show consistency of $\hat\gamma_{\mbox{\tiny CAL}}$
if model (\ref{eq:f-model}) is correct, whereas consistency of $(\hat\alpha,\hat\gamma)$ is presumed in Corollary~\ref{cor:corr-expansion}.
\end{rem}

Our preceding discussion leaves open the question how calibrated estimators $(\hat\alpha_{\mbox{\tiny CAL}}, \hat\gamma_{\mbox{\tiny CAL}})$ can be defined such that
(\ref{eq:CAL1})--(\ref{eq:CAL2}) are satisfied. A direct approach would be to take $(\hat\alpha_{\mbox{\tiny CAL}}, \hat\gamma_{\mbox{\tiny CAL}})$ as a solution
to the sample version of calibration equations (\ref{eq:CAL1})--(\ref{eq:CAL2}), where the expectation $E(\cdot)$ is replaced by
the sample average $\tilde E(\cdot)$. However, there are various complications for this approach even in the classical setting with low-dimensional $(\alpha,\gamma)$.
First, equations (\ref{eq:CAL1})--(\ref{eq:CAL2}) and the sample version may depend on $\theta^*$ to be estimated.
A preliminary doubly robust estimator can be substituted for $\theta^*$.
But the resulting sample version of (\ref{eq:CAL1})--(\ref{eq:CAL2}) remains a system of nonlinear equations in $(\alpha,\gamma)$.
Numerical solution of such equations with finite data may suffer the issue of no solution or multiple solutions (\citealt{Small}).
Theoretical analysis of estimators from nonlinear estimating equations may require cumbersome regularity conditions which would be avoided
when using conventional estimators of $(\alpha,\gamma)$. These issues can be illustrated with the following examples.

\begin{eg} \label{eg:PLM-EE}
For Example~\ref{eg:PLM} with a partially linear model, let $\psi_g(\cdot)$ be an identity function.
The calibration equations (\ref{eq:CAL1})--(\ref{eq:CAL2}) based on $\tau$ in (\ref{eq:PLM-tau}) are
\begin{align}
& 0 =  E \left( \frac{\partial \tau}{\partial \gamma} \right) = - E \left\{(Y- \theta^*Z - \alpha^\T \xi ) \psi^\prime_f ( \gamma^\T \xi ) \xi \right\}, \label{eq:PLM-EE1} \\
& 0 =  E \left( \frac{\partial \tau}{\partial \alpha} \right) = - E \left\{(Z - \psi_f(\gamma^\T \xi) ) \xi \right\} , \label{eq:PLM-EE2}
\end{align}
where $\xi=\xi(X)$ and $\tau$ is evaluated at $\theta=\theta^*$.
Because (\ref{eq:PLM-EE2}) does not depend on $\alpha$, the sample version of  simultaneous equations (\ref{eq:PLM-EE1})--(\ref{eq:PLM-EE2}) can be solved sequentially:
the sample version of (\ref{eq:PLM-EE2}) can be first solved, and then that of (\ref{eq:PLM-EE1}) be solved,
provided that $\theta^*$ is replaced by a preliminary doubly robust estimator.
\end{eg}

\begin{eg} \label{eg:PLM-log-EE}
For Example~\ref{eg:PLM-log} with a partially log-linear model, let $\psi_g(\cdot)$ be an identity function.
The calibration equations (\ref{eq:CAL1})--(\ref{eq:CAL2}) based on $\tau$ in (\ref{eq:PLM-log-tau}) are
\begin{align}
& 0 =  E \left( \frac{\partial \tau}{\partial \gamma} \right) = - E \left\{(Y \me^{-\theta^*Z} - \me^{\alpha^\T \xi} ) \psi^\prime_f ( \gamma^\T \xi ) \xi \right\}, \label{eq:PLM-log-EE1} \\
& 0 =  E \left( \frac{\partial \tau}{\partial \alpha} \right) = - E \left\{(Z - \psi_f(\gamma^\T \xi) ) \me^{\alpha^\T \xi} \xi \right\} , \label{eq:PLM-log-EE2}
\end{align}
where $\xi=\xi(X)$ and $\tau$ is evaluated at $\theta=\theta^*$.
Unlike (\ref{eq:PLM-EE1})--(\ref{eq:PLM-EE2}) in Example~\ref{eg:PLM-EE}, the sample version of (\ref{eq:PLM-log-EE1})--(\ref{eq:PLM-log-EE2})
cannot be solved sequentially even after $\theta^*$ is appropriately estimated. Therefore, algorithms for solving nonlinear equations need to be used.
We point out that calibration equations (\ref{eq:PLM-log-EE1})--(\ref{eq:PLM-log-EE2}) are simpler than estimating equations proposed in (\citealt{vermeulen2015bias}, Section 5.2),
$0 = E( \partial \tau_{\mbox{\tiny IF}} /\partial \gamma )$ and $0 = E( \partial \tau_{\mbox{\tiny IF}} /\partial \alpha)$,
where $\tau_{\mbox{\tiny IF}}$ is the influence function, $\tau_{\mbox{\tiny IF}}(U; \theta,\alpha,\gamma) =
- E^{-1}( \partial \tau /\partial\theta ) \tau(U; \theta, \alpha,\gamma)$, evaluated at $\theta=\theta^*$.
\end{eg}

\begin{eg} \label{eg:PLM-logit-EE}
For Example~\ref{eg:PLM-logit} with a logistic partially linear model, let $\psi_g(\cdot)$ be an identity function.
The calibration equations (\ref{eq:CAL1})--(\ref{eq:CAL2}) based on $\tau$ in (\ref{eq:PLM-logit-tau}) are
\begin{align}
& 0 =  E \left( \frac{\partial \tau}{\partial \gamma} \right) = - E \left\{ \me^{-\theta^* Z Y} (Y- \expit(\alpha^\T \xi)) \psi^\prime_f(\gamma^\T \xi) \xi \right\}, \label{eq:PLM-logit-EE1} \\
& 0 =  E \left( \frac{\partial \tau}{\partial \alpha} \right) = - E \left\{ \me^{-\theta^* Z Y} \expit_2(\alpha^\T \xi)) (Z - \psi_f(\gamma^\T \xi) ) \xi \right\} , \label{eq:PLM-logit-EE2}
\end{align}
where $\expit_2(c) = \expit(c)(1-\expit(c))$ and $\tau$ is evaluated at $\theta=\theta^*$.
Similarly as (\ref{eq:PLM-log-EE1})--(\ref{eq:PLM-log-EE2}), the sample version of (\ref{eq:PLM-logit-EE1})--(\ref{eq:PLM-log-EE2})
cannot be solved sequentially, due to dependency on both $\alpha$ and $\gamma$,  even after $\theta^*$ is appropriately estimated.
\end{eg}

\begin{eg} \label{eg:ATE-EE}
For the missing-data problem in Example~\ref{eg:ATE}, the calibration equations (\ref{eq:CAL1})--(\ref{eq:CAL2}) based on $\tau$ in (\ref{eq:ATE-tau}) are
\begin{align}
& 0 =  E \left( \frac{\partial \tau}{\partial \gamma} \right) = - E \left\{  \frac{\psi^\prime_f(\gamma^\T \xi)}{\psi_f^2(\gamma^\T \xi)} Z (Y - \psi_g(\alpha^\T \xi)) \xi \right\}, \label{eq:ATE-EE1} \\
& 0 =  E \left( \frac{\partial \tau}{\partial \alpha} \right) = - E \left\{ \left(\frac{Z}{\psi_f(\gamma^\T \xi)}-1 \right) \psi_g^\prime(\alpha^\T \xi) \xi \right\} , \label{eq:ATE-EE2}
\end{align}
where $\xi=\xi(X)$ and $\tau$ is evaluated at $\theta=\theta^*$.
In the case where $\psi_g(\cdot)$ is an identity function, i.e., a linear model (\ref{eq:g-model}) is specified for $E(Y|Z=1,X)$,
the sample version of (\ref{eq:ATE-EE1})--(\ref{eq:ATE-EE2}) can be solved sequentially, because (\ref{eq:ATE-EE2}) does not depend on $\alpha$.
But such sequential solution is infeasible with a nonlinear function $\psi_g(\cdot)$, because equations (\ref{eq:ATE-EE1})--(\ref{eq:ATE-EE2}) are intrinsically coupled, each depending on both
$\alpha$ and $\gamma$ (\citealt{tan2020model}, Section 3.5).
\end{eg}

\section{Regularized calibrated estimation} \label{sec:rcal}

We develop regularized calibrated estimation for $(\alpha,\gamma)$, such that
the resulting estimator of $\theta^*$ achieves an asymptotic expansion similar to (\ref{eq:desired-expansion2}), hence allowing
valid confidence intervals, under suitable conditions in high-dimensional settings.
The estimators of $(\alpha,\gamma)$ are derived from a numerically tractable two-step algorithm.
Moreover, high-dimensional analysis is provided to demonstrate the desired asymptotic expansion and consistent variance estimation, which lead to valid Wald confidence intervals.

Conceptually, regularized calibrated estimation involves constructing regularized estimators of $(\alpha,\gamma)$,
which converge in probability to the target values $(\bar\alpha_{\mbox{\tiny CAL}}, \bar\gamma_{\mbox{\tiny CAL}})$ satisfying population calibration equations (\ref{eq:CAL1})--(\ref{eq:CAL2}).
As discussed in Section~\ref{sec:cal} in low-dimensional settings, there may be numerical and theoretical complications with
directly using the sample version of  (\ref{eq:CAL1})--(\ref{eq:CAL2}) as estimating equations.
With high-dimensional data, estimating equations can be regularized by generalizing the Dantzig selector (\citealt{candes2007dantzig}), which seeks to minimize
$\|\alpha\|_1 + \|\gamma\|_1$ subject to
\begin{align*}
& \left\|  \tilde E \left\{ \frac{\partial \tau}{\partial \gamma} (U; \hat\theta_1, g(x;\alpha), f(x;\gamma))  \right\} \right\|_\infty \le \lambda, \\
& \left\| \tilde E \left\{ \frac{\partial \tau}{\partial \alpha} (U; \hat\theta_1, g(x;\alpha), f(x;\gamma))  \right\}  \right\|_\infty \le \lambda,
\end{align*}
where $\hat\theta_1$ is a preliminary doubly robust estimator, $\lambda$ is a tuning parameter, and $\|\cdot\|_1$ or $\|\cdot\|_\infty$ denotes $L_1$ or $L_\infty$ norm.
While theoretical analysis of generalized Dantzig selectors can be performed, this approach is not pursued here mainly because
the required optimization problem seems numerically difficult to solve with complex nonlinear estimating functions.
The generalized Dantzig-selector algorithm in \cite{radchenko2011improved} can potentially be modified for the above problem, but its effectiveness seems uncertain.
Further investigation of the Dantzig-selector approach can be of interest in future work.

\subsection{Two-step algorithm}

We propose a two-step algorithm, shown as Algorithm~\ref{alg:two-step}, for regularized calibrated estimation.
The algorithm is facilitated by exploiting the following convexity assumption, which is satisfied
in various settings including Examples~\ref{eg:PLM}--\ref{eg:ATE} as shown in Section \ref{sec:applications}.
In principle, our approach can also be applied without the convexity assumption, provided that a solution to equation (\ref{eq:CAL1}) or (\ref{eq:CAL2})
is unique in $\alpha$ or $\gamma$, while $\gamma$ or $\alpha$ is fixed respectively.
Such an assumption is used earlier in the discussion leading to Corollary~\ref{cor:DR-expansion}.

\begin{ass} \label{ass:convex}
There exist two loss functions $\ell_1(U; \theta, \alpha,\gamma)$ and $\ell_2(U; \theta, \alpha,\gamma)$
such that $E\{\ell_1(U; \theta, \alpha,\gamma)\}$ is strictly convex in $\alpha$, $E\{ \ell_2(U; \theta, \alpha,\gamma)\}$ is strictly convex in $\gamma$, and
\begin{align}
& \frac{\partial \ell_1}{\partial \alpha} = \frac{\partial \tau}{\partial \gamma}, \quad
\frac{\partial \ell_2}{\partial \gamma} =\frac{\partial \tau}{\partial \alpha}, \label{eq:convex-loss}
\end{align}
where $\tau$ is parameterized as $\tau(U; \theta, \alpha, \gamma)=\tau(U; \theta, g(x;\alpha), f(x;\gamma))$.
\end{ass}

From Assumption~\ref{ass:convex}, various equations in Section~\ref{sec:cal} can be restated in terms of minimization of convex loss functions.
The basic identities (\ref{eq:DR-EE1})--(\ref{eq:DR-EE2}) can be translated to minimization properties.
If model (\ref{eq:g-model}) with $g(x;\alpha)$ is correctly specified, then
(\ref{eq:DR-EE1}) amounts to $ E\{ (\partial /\partial\alpha) \ell_1 (U; \theta^*, \alpha, \gamma) \}_{\alpha=\alpha^*} = 0$
and hence for fixed $\gamma$, the expected loss $E\{ \ell_1 (U; \theta^*, \alpha, \gamma) \}$, convex in $\alpha$, attains a minimum at $\alpha^*$ with zero gradient
 under interchangeability of the differentiation and expectation.
Similarly,  if model (\ref{eq:f-model}) with $f(x;\gamma)$ is correctly specified, then
(\ref{eq:DR-EE2}) amounts to $ E\{ (\partial /\partial\gamma) \ell_2 (U; \theta^*, \alpha, \gamma) \}_{\gamma=\gamma^*} = 0$
and hence for fixed $\alpha$, the expected loss $E\{ \ell_2 (U; \theta^*, \alpha, \gamma) \}$, convex in $\gamma$, is minimized at $\gamma^*$.

\begin{algorithm}[t] 
	\caption{Two-step algorithm}\label{alg:two-step}
	\begin{algorithmic}[1]
		\Procedure{Initial estimation}{}
		\State Compute $(\hat\alpha_1,\hat\gamma_1)$ as model-based estimators of $(\alpha,\gamma)$;
		\State Compute $\hat{\theta}_1 = \hat\theta(\hat\alpha_1,\hat\gamma_1)$ as a solution to $\tilde{E}\{\tau(U; \theta,\hat\alpha_1, \hat\gamma_1)\}=0$.
		\EndProcedure
		\Procedure{Calibrated estimation}{}
		\State Compute $\hat{\gamma}_2 = \argmin_{\gamma} \, [ \tilde E \{\ell_2(U;  \hat{\theta}_1, \hat{\alpha}_1, \gamma)\} + \lambda_1 \|\gamma\|_1  ]$, also denoted as $\hat\gamma_{\mbox{\tiny RCAL}}$;
		\State Compute $\hat{\alpha}_2 = \argmin_{\alpha} \, [ \tilde E \{\ell_1 (U; \hat{\theta}_1, \alpha, \hat{\gamma}_2)\} + \lambda_2 \|\alpha\|_1  ]$, also denoted as $\hat\alpha_{\mbox{\tiny RCAL}}$;
		\State Compute $\hat{\theta}_2 = \hat\theta(\hat\alpha_2,\hat\gamma_2)$ as a solution to $ \tilde{E}\{ \tau(U; \theta, \hat\alpha_2,\hat\gamma_2) \}=0$, also denoted as $\hat\theta_{\mbox{\tiny RCAL}}$.
		\EndProcedure
	\end{algorithmic}
\end{algorithm}

The population calibration equations (\ref{eq:CAL1})--(\ref{eq:CAL2}) can be expressed in the form of alternating minimization:
$E\{ \ell_1 (U; \theta^*, \alpha, \gamma) \}$ is minimized at $\alpha = \bar\alpha_{\mbox{\tiny CAL}}$  for fixed $\gamma= \bar\gamma_{\mbox{\tiny CAL}}$,
and  $E\{ \ell_2 (U; \theta^*, \alpha, \gamma) \}$ is minimized at $\gamma = \bar\gamma_{\mbox{\tiny CAL}}$  for fixed $\alpha= \bar\alpha_{\mbox{\tiny CAL}}$.
This reasoning would suggest the following iterative algorithm for computing
$(\bar\alpha_{\mbox{\tiny CAL}}, \bar\gamma_{\mbox{\tiny CAL}})$ at a population level.

\textit{Population calibration algorithm}.\vspace{-.1in}
\begin{itemize} \addtolength{\itemsep}{-.1in}
\item Determine initial target values $(\bar\alpha_1,\bar\gamma_1)$;
\item For $t=2,3,\cdots$, determine $\bar\gamma_{t}$ as a solution to $E \{ (\partial /\partial \alpha) \tau(U; \theta^*, \bar\alpha_{t-1}, \gamma) \}=0$ or
a minimizer of $E\{ \ell_2 (U; \theta^*, \bar\alpha_{t-1}, \gamma) \}$ in $\gamma$, and then determine
$\bar\alpha_{t}$ as a solution to $E \{ (\partial /\partial \gamma) \tau(U; \theta^*, \alpha, \hat\gamma_t) \}=0$ or
a minimizer of $E\{ \ell_1 (U; \theta^*, \alpha, \bar\gamma_{t}) \}$ in $\alpha$.
\end{itemize}
The limit $(\bar\alpha_\infty,\bar\gamma_\infty) =\lim_{t \to \infty} (\bar\alpha_t, \bar\gamma_t)$, if exists, can be shown to satisfy (\ref{eq:CAL1})--(\ref{eq:CAL2}).
However, remarkably, we show in Proposition~\ref{pro:two-step} that if the initial target values $(\bar\alpha_1,\bar\gamma_1)$ are determined from model-based estimators of $(\alpha,\gamma)$
which are consistent in the case of model (\ref{eq:g-model}) or (\ref{eq:f-model}) being correctly specified, then
the iterative process can be terminated by the second step (i.e., by $t=2$), as far as doubly robust estimation is concerned.
It should also be mentioned that if both models (\ref{eq:g-model}) and (\ref{eq:f-model}) are misspecified, then
the second-step target values $(\bar\alpha_2, \bar\gamma_2)$ may in general not satisfy calibration equations  (\ref{eq:CAL1})--(\ref{eq:CAL2}).

\begin{pro}  \label{pro:two-step}
If model (\ref{eq:g-model}) is correctly specified and $\bar\alpha_1 = \alpha^*$ but $\bar\gamma_1$ is arbitrary, or if model (\ref{eq:f-model}) is correctly specified and
$\bar\gamma_1 = \gamma^*$ but $\bar\alpha_1$ is arbitrary,
then $\bar\alpha_2=\alpha^*$ or $\bar\gamma_2=\gamma^*$ respectively, and $(\bar\alpha_2,\bar\gamma_2)$ jointly satisfy calibration equations  (\ref{eq:CAL1})--(\ref{eq:CAL2}).
\end{pro}

\begin{prf}
By definition, $(\bar\alpha_2,\bar\gamma_2)$ satisfy the equations
\begin{align}
& E \left\{ \frac{\partial \tau}{\partial \alpha} (U;\theta^*, \bar\alpha_1, \bar\gamma_2) \right\} =0, \label{eq:two-step1} \\
& E \left\{ \frac{\partial \tau}{\partial \gamma} (U;\theta^*, \bar \alpha_2, \bar\gamma_2) \right\} =0. \label{eq:two-step2}
\end{align}
If  model (\ref{eq:g-model}) is correctly specified and $\bar\alpha_1 = \alpha^*$, then by comparison of (\ref{eq:DR-EE1}) and (\ref{eq:two-step2}), $\bar\alpha_2 =\alpha^*$, and
hence (\ref{eq:two-step2}) and (\ref{eq:two-step1}) yield (\ref{eq:CAL1}) and (\ref{eq:CAL2}) respectively for $(\bar\alpha_2,\bar\gamma_2)$.
If  model (\ref{eq:f-model}) is correctly specified and $\bar\gamma_1 = \gamma^*$, then by comparison of (\ref{eq:DR-EE2}) and (\ref{eq:two-step1}), $\bar\gamma_2 =\gamma^*$,
and by (\ref{eq:DR-EE2}),
\begin{align}
& E \left\{ \frac{\partial \tau}{\partial \alpha} (U;\theta^*, \bar\alpha_2, \gamma^*) \right\} =0. \label{eq:two-step3}
\end{align}
In this case, (\ref{eq:two-step2}) and (\ref{eq:two-step3}) lead to (\ref{eq:CAL1}) and (\ref{eq:CAL2}) respectively for $(\bar\alpha_2,\bar\gamma_2)$.
\end{prf}

Algorithm~\ref{alg:two-step} is a sample version of the population calibration algorithm with two steps, using regularized estimation with Lasso penalties to deal with
high-dimensional data. The initial estimators $(\hat\alpha_1, \hat\gamma_1)$ can be Lasso-regularized maximum likelihood (or quasi-likelihood) estimators in generalized linear models
associated with (\ref{eq:g-model})--(\ref{eq:f-model}). The two-step estimators, $(\hat\alpha_{\mbox{\tiny RCAL}}, \hat\gamma_{\mbox{\tiny RCAL}})=(\hat\alpha_2,\hat\gamma_2)$,
serves as an adjustment to the usual estimators $(\hat\alpha_1,\hat\gamma_1)$,
such that calibration equations (\ref{eq:CAL1})--(\ref{eq:CAL2}) are satisfied if either model (\ref{eq:g-model}) or (\ref{eq:f-model}) is correct.

\subsection{Theoretical analysis} \label{sec:theo}

We provide high-dimensional analysis of the two-step estimators $(\hat\alpha_{\mbox{\tiny RCAL}}, \hat\gamma_{\mbox{\tiny RCAL}})=(\hat\alpha_2,\hat\gamma_2)$
and the resulting estimator $\hat\theta_{\mbox{\tiny RCAL}}=\hat\theta ( \hat\alpha_2 ,\hat\gamma_2)$.
Throughout this section, we assume that either model (\ref{eq:g-model}) or (\ref{eq:f-model}), but not necessarily both, is correctly specified.

Our main result, summarized as Proposition~\ref{pro:main}, can be deduced from Theorems~\ref{thm:gamma2}--\ref{thm:theta2} later.
For initial estimators $(\hat\alpha_1,\hat\gamma_1)$ defined as Lasso-regularized maximum likelihood (or quasi-likelihood) estimators,
the rates of convergence in Assumption~\ref{ass:gamma2-basic}(iv) later are satisfied under suitable conditions with $M_0 = O(1) (|S_{\bar\alpha_1}| + |S_{\bar\gamma_1}|)$,
where $|S_{\bar\alpha_1}|$ or $|S_{\bar\gamma_1}|$ denotes the number of nonzero coefficients of the target value $\bar\alpha_1$ or $\bar\gamma_1$ respectively
(\citealt{buhlmann2011statistics, negahban2012unified}).
For the two-step estimators $(\hat\alpha_2,\hat\gamma_2)$, denote by
$|S_{\bar\alpha_2}|$ or $|S_{\bar\gamma_2}|$ denotes the number of nonzero coefficients of the target value $\bar\alpha_2$ or $\bar\gamma_2$ respectively.
Suppose that the Lasso tuning parameters are specified as $\lambda_1 =A_1^\dag r_0$ and  $\lambda_2 = A_2^\dag r_0$ for sufficiently large constants $A_1^\dag$ and $A_2^\dag$,
where $r_0 = \{\log(\me p)/n\}^{1/2}$.

\begin{pro}\label{pro:main}
Suppose that Assumptions \ref{ass:convex}--\ref{ass:theta2-rate} hold, and
$ (M_0 + |S_{\bar\alpha_2}| + |S_{\bar \gamma_2}|) r_0^2 = o(n^{-1/2})$, i.e., $ (M_0 + |S_{\bar\alpha_2}| + |S_{\bar \gamma_2}|) \log(\me p) = o(n^{1/2})$,
If  model (\ref{eq:g-model}) with $g(x;\alpha)$ is correctly specified or model (\ref{eq:f-model}) is correctly specified with
$f(x;\gamma)$, then $\hat\theta_{\mbox{\tiny RCAL}}=\hat\theta ( \hat\alpha_2, \hat\gamma_2)$ satisfies
\begin{align}
\hat \theta_{\mbox{\tiny RCAL}} - \theta^* = - E^{-1} \left( \frac{\partial \tau}{\partial\theta} \right)
\tilde E (\tau) \Big|_{ (\theta,\alpha,\gamma) =(\theta^*,\bar \alpha_2,\bar\gamma_2)}+ o_p (n^{-1/2}) .  \label{eq:expansion-rcal}
\end{align}
Furthermore, the following results hold in either case: \vspace{-.1in}
\begin{itemize} \addtolength{\itemsep}{-.1in}
\item[(i)] $\sqrt{n}(\hat{\theta}_{\mbox{\tiny RCAL}} - \theta^*) \overset{\mathcal D}{\rightarrow} \N(0, V)$, where
$V= \var (\tau)/ E^2 (\partial\tau/\partial \theta) \big|_{ (\theta,\alpha,\gamma) =(\theta^*,\bar \alpha_2,\bar\gamma_2)}$;
\item[(ii)] A consistent estimator $\hat{V}$ of $V$ is
\begin{align*}
\hat V = \tilde E (\tau^2) / \tilde E^2 (\partial\tau/\partial \theta) \Big|_{ (\theta,\alpha,\gamma) =(\hat\theta_{\mbox{\tiny RCAL}} ,\hat\alpha_2, \hat\gamma_2)};
\end{align*}
\item[(iii)] An asymptotic $(1-c)$ confidence interval for $\theta^*$ is $\hat{\theta}_{\mbox{\tiny RCAL}} \pm z_{c/2} \sqrt{\hat{V}/n}$, where $z_{c/2}$ is the $(1-c/2)$ quantile of $\N(0,1 )$.
\end{itemize} \vspace{-.1in}
Hence a doubly robust confidence interval for $\theta^*$ is obtained.
\end{pro}

In the remainder of Section~\ref{sec:theo}, we present several formal results underlying Proposition~\ref{pro:main}. Our analysis of the estimators $(\hat\alpha_2,\hat\gamma_2)$,
while building on the existing literature on Lasso penalized $M$-estimation (\citealt{buhlmann2011statistics, negahban2012unified}),
needs to tackle the dependency of $\hat\gamma_2$ on $(\hat\theta_1,\hat\alpha_1)$ and subsequently that of $\hat\alpha_2$ on $(\hat\theta_1,\hat\gamma_2)$.
The situation is more general and more complicated than studied in \cite{tan2020model}.
We develop a technical strategy to control such dependency through use of the $L_1$ norm, so that the usual rates of convergence are obtained. See Lemma~\ref{lem:remove-hat} in the Supplement.

We first discuss theoretical analysis of $\hat\gamma_2$, with the Lasso tuning parameter $\lambda_1 = A_1 \lambda_0$ for a constant $A_1$, where  $\lambda_0 = \{\log( p/\epsilon)/n\}^{1/2}$.
The loss function for defining $\hat\gamma_2$ is
$L_2(\gamma; \hat\theta_1, \hat\alpha_1) = \tilde E \{ \ell_2( U; \hat\theta_1, \hat\alpha_1, \gamma) \} $, where $\ell_2$ is from Assumption~\ref{ass:convex}.
As $L_2(\gamma; \hat\theta_1,\hat\alpha_1)$ is convex in $\gamma$, the corresponding Bregman divergence is defined as
\begin{align*}
& D_2 (\gamma^\prime, \gamma; \hat\theta_1, \hat\alpha_1) = L_2(\gamma^\prime; \hat\theta_1, \hat\alpha_1) - L_2(\gamma; \hat\theta_1, \hat\alpha_1) - (\gamma^\prime-\gamma)^\T \frac{\partial L_2}{\partial \gamma}(\gamma; \hat\theta_1,\hat\alpha_1).
\end{align*}
The symmetrized Bregman divergence is easily shown to be
\begin{align}
& D_2^\dag (\gamma^\prime, \gamma; \hat\theta_1, \hat\alpha_1) =  (\gamma^\prime-\gamma)^\T \left\{\frac{\partial L_2}{\partial \gamma}(\gamma^\prime; \hat\theta_1,\hat\alpha_1) -
\frac{\partial L_2}{\partial \gamma}(\gamma; \hat\theta_1,\hat\alpha_1)\right\} \nonumber \\
& =  (\gamma^\prime-\gamma)^\T \tilde E \left[ \xi \left\{ \frac{\partial\tau}{\partial \eta_g} (U; \hat\theta_1, \hat\alpha_1,\gamma^\prime) -
\frac{\partial\tau}{\partial \eta_g} (U; \hat\theta_1, \hat\alpha_1,\gamma) \right\} \right]. \label{eq:sym-bregman}
\end{align}
The target value $\bar\gamma_2$ is defined as a solution to $E \{ (\partial\tau /\partial\alpha) (U; \theta^*, \bar \alpha_1,\gamma)\}=0$ or equivalently
a minimizer of the expected loss $E \{ \ell_2( U; \theta^*, \bar\alpha_1, \gamma) \} $,
where $(\theta^*,\bar\alpha_1)$ are the target values (i.e., probability limits) of the initial estimators $(\hat\theta_1,\hat\alpha_1)$.
After statement of the assumptions required, Theorem~\ref{thm:gamma2} establishes the convergence of $\hat\gamma_2$ to $\bar\gamma_2$ in the
both $L_1$ norm $\|\hat\gamma_2 - \bar\gamma_2\|_1$ and the symmetrized Bregman divergence $D_2^\dag(\hat\gamma_2, \bar\gamma_2; \hat\theta_1, \hat\alpha_1)$.

A variable $Y$ is said to be sub-exponential with parameter $(B_{01}, B_{02})$ if $E( |Y-E(Y)|^k ) \le \frac{k!}{2} B_{01}^2 B_{02}^{k-2} $ for each $k\ge 2$.
For a $p\times p$ matrix $\Sigma$,
a compatibility condition (\citealt{buhlmann2011statistics}) is said to hold with a subset $S \in \{1,\ldots,p\}$ and constants $\nu_1 >0$ and $\mu_1>1$ if
$\nu_1^2  (\sum_{j\in S} |b_j|)^2 \le |S| ( b^\T \Sigma b )$
for any vector $b=(b_1,\ldots,b_k)^\T \in \bbR^k$ satisfying $\sum_{j\not\in S} |b_j| \le \mu_1 \sum_{j\in S} |b_j|$.
Throughout, $|S|$ denotes the size of a set $S$.

\begin{ass} \label{ass:gamma2-basic}
Suppose that the following conditions are satisfied. \vspace{-.1in}
\begin{itemize} \addtolength{\itemsep}{-.1in}
\item[(i)] $\max_{j=1,\ldots,p} |\xi_j(X)| \le C_0$ almost surely for a constant $C_0 >0$.

\item[(ii)] The variable $\frac{\partial\tau}{\partial \eta_g} (U; \theta^*, \bar\alpha_1,\bar\gamma_2) $ is sub-exponential with parameter $(B_{01}, B_{02})$.

\item[(iii)] The compatibility condition holds for $\Sigma_\gamma= E\{ \xi\xi^\T \frac{\partial^2 \tau}{\partial \eta_g \partial\eta_f}(U;\theta^*, \bar\alpha_1, \bar\gamma_2)\}$
with the subset $S_{\bar\gamma_2}= \{j: (\bar\gamma_2)_j \not=0, j=1,\ldots,p\}$ and some constants $\nu_1>0$ and $\mu_1 >1$.

\item[(iv)] For some constants $c_0>0$ and $M_0 \ge 1$, possibly depending on $(\bar\alpha_1,\bar\gamma_1)$, and any small $\epsilon>0$, it holds with probability at least $1-c_0 \epsilon$ that
$ (\hat\alpha_1 - \bar\alpha_1)^\T \tilde \Sigma_0 (\hat\alpha_1 - \bar\alpha_1) \le M_0 \lambda_0^2 $, $\| \hat\alpha_1 - \bar\alpha_1 \|_1 \le M_0\lambda_0$,
and $|\hat\theta_1 - \theta^* |\le M_0^{1/2} \lambda_0$, where  $\lambda_0 = \{\log( p/\epsilon)/n\}^{1/2}$,
and $\bar\alpha_1=\alpha^*$ if model (\ref{eq:g-model}) is correctly specified or
$\bar\gamma_1 = \gamma^*$ if model (\ref{eq:f-model}) is correctly specified.
\end{itemize}
\end{ass}

\begin{ass} \label{ass:gamma2-hess}
There exist positive constants $c_1$, $c_2$, $B_{11}$, $B_{12}$, $C_1$, $C_2$, $\varrho_0$, and $\varrho_1$ such that the following conditions are satisfied,
where $\mathcal N_1 = \{(\theta,\alpha): |\theta-\theta^*|\le c_1, \|\alpha-\bar\alpha_1\|_1 \le c_1\}$. \vspace{-.1in}
\begin{itemize} \addtolength{\itemsep}{-.1in}
\item[(i)] The variables
$T_{\eta_g^2}^{(1)} (U;\theta^*, \bar\alpha_1, \bar\gamma_2)  = \sup_{(\theta,\alpha)\in \mathcal N_1} |\frac{\partial^2 \tau}{\partial \eta_g^2}(U;\theta,\alpha,\bar\gamma_2) |$ and
$T_{\eta_g\theta}^{(1)} (U;\theta^*, \bar\alpha_1, \bar\gamma_2) =   \sup_{(\theta,\alpha)\in \mathcal N_1} |\frac{\partial^2 \tau}{\partial \eta_g \partial\theta}(U;\theta,\alpha,\bar\gamma_2) |$
are sub-exponential with parameter $(B_{11}, B_{12})$, and
$E \{ T_{\eta_g^2}^{(1)}$ $ (U;\theta^*, \bar\alpha_1, \bar\gamma_2) |X \} \le C_1$ and $E \{ T_{\eta_g\theta}^{(1)} (U;\theta^*, \bar\alpha_1,\bar\gamma_2) |X  \} \le C_1$ almost surely.

\item[(ii)] The variable $\frac{\partial^2 \tau}{\partial \eta_g \partial\eta_f}(U;\theta^*, \bar\alpha_1, \bar\gamma_2)$ is sub-exponential with parameter $(B_{11},B_{12})$, and
$ E \{\frac{\partial^2 \tau}{\partial \eta_g \partial\eta_f}(U;\theta^*, \bar\alpha_1, \bar\gamma_2) | X \} \ge c_2 $ almost surely.

\item[(iii)] For any $(\theta,\alpha) \in \mathcal N_1$ and $\gamma \in \bbR^p$, it holds that almost surely
\begin{align*}
\frac{\partial^2 \tau}{\partial \eta_g \partial\eta_f}(U; \theta, \alpha, \gamma) \le
\frac{\partial^2 \tau}{\partial \eta_g \partial\eta_f}(U;\theta^*, \bar\alpha_1, \bar\gamma_2)  \me^{-C_2 (|\theta-\theta^*|+|(\alpha-\bar\alpha_1)^\T\xi|+ |(\gamma -\bar\gamma_2)^\T \xi |) }.
\end{align*}
\item[(iv)] $M_0 \lambda_0 \le \varrho_0 \,(\le c_1)$ and $|S_{\bar\gamma_2}| \lambda_0 \le \varrho_1$ such that
$ \varrho_2 = \nu_1^{-2} (1+\mu_1)^2 \varrho_1 B_{15} < 1 $,
$ \varrho_3 = C_0 C_2 A_{11}^{-1} \mu_{12}^2 \nu_{11}^{-2} \varrho_1 \me^{\varrho_5 } < 1$, and
$ \varrho_4 = C_0 C_2 A_{11}^{-1} \mu_{11}^{-2} C_{12} \varrho_0 \me^{\varrho_5 }  <1$,
where $\varrho_5 = C_2(1+C_0)\varrho_0$,  $A_{11} = A_1 - B_0 - C_{13}$, $\mu_{11} = 1- 2 A_1 / \{(\mu_1+1) A_{11}\} \in (0,1]$,
$\mu_{12} = (\mu_1 +1) A_{11}$,
$\nu_{11} =  \nu_1 (1-\varrho_2)^{1/2}$,
$B_0 = C_0 (B_{02} + \sqrt{2} B_{01}) $,
$B_{15}$ is defined in  Lemma~\ref{lem:prob-hess} depending on $(C_0,C_1,B_{11},B_{12})$,
and $(C_{12},C_{13})$ are defined in Lemma~\ref{lem:remove-hat} depending on $(\varrho_0, c_2,C_0,C_1,B_{11},B_{12})$.
\end{itemize}
\end{ass}

\begin{thm} \label{thm:gamma2}
Suppose that Assumptions~\ref{ass:convex}--\ref{ass:gamma2-hess} hold and $\lambda_0\le 1$. Then
for $\lambda_1 = A_1 \lambda_0$ and $A_1 > (B_0+C_{13}) (\mu_1+1)/(\mu_1-1)$, we have with probability at least $1-(c_0+10)\epsilon$,
\begin{align}
& D^\dag_2 ( \hat \gamma_2, \bar \gamma_2; \hat\theta_1 , \hat\alpha_1) + A_{11} \lambda_0 \| \hat\gamma_2 - \bar\gamma_2  \|_1 \nonumber \\
& \le \left\{ \me^{\varrho_5 }(1-\varrho_3)^{-1} \mu_{12}^2 \nu_{11}^{-2} ( |S_{\bar\gamma_2}| \lambda_0^2) \right\} \vee
 \left\{ \me^{\varrho_5 }(1-\varrho_4)^{-1} \mu_{11}^{-2} C_{12} (M_0 \lambda_0^2 ) \right\}, \label{eq:gamma2-expan}
\end{align}
where $\vee$ denotes the maximum between two numbers, and $(\mu_{11},\mu_{12},\nu_{11}, \varrho_3, \varrho_4, \varrho_5, A_{11}, B_0, C_{12},$ $C_{13})$ are defined in Assumption~\ref{ass:gamma2-hess}(iv).
\end{thm}

\begin{rem} \label{rem:gamma2}
Assumptions~\ref{ass:gamma2-basic}(iii) and \ref{ass:gamma2-hess}(iii) are standard in high-dimensional analysis of $M$-estimation (e.g., \citealt{buhlmann2011statistics,tan2020regularized}).
Assumptions~\ref{ass:gamma2-hess}(i)--(ii) are used to control the deviation of $(\hat\theta_1,\hat\alpha_1)$ from $(\theta^*,\bar\alpha_1)$ in the basic inequality.
Given Assumption~\ref{ass:gamma2-hess}(ii), the compatibility condition on $\Sigma_\gamma$ in Assumption~\ref{ass:gamma2-basic}(iii) can be equivalently replaced by
a compatibility condition on the matrix $\Sigma_0 = E( \xi\xi^\T)$, independent of $(\theta^*,\bar\alpha_1,\bar\gamma_2)$.
\end{rem}

\begin{rem} \label{rem:gamma2-theta1}
Assumption~\ref{ass:gamma2-basic}(iv) is concerned with the convergence of the initial estimators $(\hat\theta_1,\hat\alpha_1,\hat\gamma_1)$.
In fact, $\hat\theta_1$ is required to converge to $\theta^*$ at rate $M_0^{1/2} \lambda_0$ if model (\ref{eq:g-model}) or (\ref{eq:f-model}) is correctly specified.
Hence $\hat\theta_1$ is pointwise doubly robust, although it does not in general admit doubly robust confidence intervals.
For $\hat\theta_1 = \hat\theta(\hat\alpha_1, \hat\gamma_1)$ in Algorithm~\ref{alg:two-step},
the required convergence for $\hat\theta_1$ can be deduced from the stated rates of convergence for $(\hat\alpha_1, \hat\gamma_1)$ under
suitable conditions, similar to Assumptions~\ref{ass:theta2-consistency}--\ref{ass:theta2-rate} for Theorem~\ref{thm:theta2} later.
For simplicity, the convergence of $\hat\theta_1$ is included as part of Assumption~\ref{ass:gamma2-basic}(iv).
This formulation also allows Theorem~\ref{thm:gamma2} to be applied with other possible choices of $\hat\theta_1$. See the proof of Corollary~\ref{cor:debiased}.
\end{rem}

The following corollary provides a bound on the prediction $L_2$ norm (in the scale of linear predictors $\eta_f$),
$\tilde E[ \{(\hat\gamma_2 - \bar\gamma_2)^T \xi\}^2 ] = ( \hat\gamma_2 - \bar\gamma_2 )^\T \tilde \Sigma_0 ( \hat\gamma_2 - \bar\gamma_2 ) $,
where $\tilde\Sigma_0 = \tilde E (\xi\xi^\T)$.

\begin{cor} \label{cor:gamma2}
In the setting of Theorem~\ref{thm:gamma2}, with probability at least $1-(c_0+10)\epsilon$, we have, in addition to (\ref{eq:gamma2-expan}),
\begin{align}
& ( \hat\gamma_2 - \bar\gamma_2 )^\T \tilde \Sigma_0 ( \hat\gamma_2 - \bar\gamma_2 ) \nonumber \\
& \le \left\{ c_2^{-1} \me^{\varrho_5 }(1-\varrho_3 \vee \varrho_4)^{-1} + (1+c_2^{-1}) B_1 A_{11}^{-2} C_3 (\varrho_0 \vee \varrho_1) \right\} C_3 ( |S_{\bar\gamma_2}| \vee M_0) \lambda_0^2 , \label{eq:gamma2-expan2}
\end{align}
where $B_1 = (4C_0^2) \vee B_{15}$, and $C_3$ is a constant such that the right hand side of (\ref{eq:gamma2-expan}) is upper bounded by $C_3 ( |S_{\bar\gamma_2}| \vee M_0) \lambda_0^2$.
\end{cor}

From Theorem~\ref{thm:gamma2} and Corollary~\ref{cor:gamma2}, let $M_1 \,(\ge M_0)$ be a constant such that
the right hand side of (\ref{eq:gamma2-expan}) is upper bounded by $A_{11} M_1 \lambda_0^2$
and that of (\ref{eq:gamma2-expan2}) is upper bounded by $M_1 \lambda_0^2$.
Then with probability at least $1-(c_0+10)\epsilon$,  we have
\begin{align}
(\hat\gamma_2 - \bar\gamma_2)^\T \tilde \Sigma_0 (\hat\gamma_2 - \bar\gamma_2) \le M_1 \lambda_0^2, \quad
\| \hat\gamma_2 - \bar\gamma_2 \|_1 \le M_1 \lambda_0 . \label{eq:gamma2-summary}
\end{align}
These bounds can be used to justify a rate condition on the convergence of $\hat\gamma_2$ corresponding to Assumption~\ref{ass:gamma2-basic}(iv),
and to obtain a similar result to Theorem~\ref{thm:gamma2} about the convergence of $\hat\alpha_2$ to a target value $\bar\alpha_2$, which is defined as a solution to
$E \{ (\partial\tau /\partial\gamma) (U; \theta^*, \alpha, \bar \gamma_2)\}=0$ or equivalently
a minimizer of the expected loss $E \{ \ell_1( U; \theta^*, \alpha, \bar\gamma_2) \} $.

\begin{ass} \label{ass:alpha2-basic}
Suppose that the conditions (ii)--(iii) in Assumption~\ref{ass:gamma2-basic} hold, with $(\bar\alpha_1, \bar\gamma_2)$ replaced by $(\bar\gamma_2, \bar\alpha_2)$,
$\partial\tau /\partial \eta_g$ by $\partial\tau / \partial\beta_f$, and $(B_{01},B_{02},\mu_1,\nu_1)$ replaced by some alternative constants throughout.
\end{ass}

\begin{ass} \label{ass:alpha2-hess}
Suppose that the conditions (i)--(iv) in Assumption~\ref{ass:gamma2-hess} hold,
with $(\bar\alpha_1, \bar\gamma_2)$ replaced by $(\bar\gamma_2, \bar\alpha_2)$,
$(\partial^2 \tau / \partial\eta_g^2, \partial^2\tau/(\partial\eta_g\partial\theta))$ by $(\partial^2 \tau / \partial\eta_f^2,  \partial^2\tau/(\partial\eta_f\partial\theta))$,
$M_0$ by $M_1$, and $(c_1,c_2,B_{11},B_{12},C_1,C_2,\varrho_0,\varrho_1)$ by some alternative constants throughout.
\end{ass}

\begin{thm} \label{thm:alpha2}
In the setting of Theorem~\ref{thm:gamma2}, suppose that Assumptions~\ref{ass:alpha2-basic}--\ref{ass:alpha2-hess} also hold.
Then for $\lambda_2 = A_2 \lambda_0$ and sufficiently large $A_2$,
we have with probability at least $1-(c_0+18)\epsilon$, in addition to (\ref{eq:gamma2-summary}),
\begin{align}
(\hat\alpha_2 - \bar\alpha_2)^\T \tilde \Sigma_0 (\hat\alpha_2 - \bar\alpha_2) \le M_2 \lambda_0^2, \quad
\| \hat\alpha_2 - \bar\alpha_2 \|_1 \le M_2 \lambda_0 , \label{eq:alpha2-summary}
\end{align}
where $M_2 \,(\ge M_1)$ is a constant determined similarly as $M_1$ in (\ref{eq:gamma2-summary}).
\end{thm}

With the preceding results about $(\hat\alpha_2,\hat\gamma_2)$, we are ready to study the convergence of $\hat\theta_2 = \hat\theta(\hat\alpha_2,\hat\gamma_2)$.
As convergence in probability is of main interest, the high-probability bounds (\ref{eq:gamma2-summary}) and (\ref{eq:alpha2-summary})
can be used to deduce the following in-probability statements:
$ (\hat\alpha_2 - \bar\alpha_2)^\T \tilde \Sigma_0 (\hat\alpha_2 - \bar\alpha_2)  = O_p( M_2 r_0^2)  $, $\| \hat\alpha_2 - \bar\alpha_2 \|_1 = O_p(M_2 r_0)$,
$ (\hat\gamma_2 - \bar\gamma_2)^\T \tilde \Sigma_0 (\hat\gamma_2 - \bar\gamma_2) = O_p( M_2 r_0^2) $, and
$\| \hat\gamma_2 - \bar\gamma_2 \|_1 = O_p (M_2 r_0 )$, where $r_0 = \{\log(\me p)/n\}^{1/2}$.
After statement of assumptions required, Theorem~\ref{thm:theta2} establishes the desired convergence result for $\hat\theta_2$.

\begin{ass} \label{ass:theta2-consistency}
Suppose that the following conditions are satisfied.\vspace{-.1in}
\begin{itemize}\addtolength{\itemsep}{-.1in}
\item[(i)]  $E \left\{ \tau(U; \theta^*, \bar\alpha_2, \bar\gamma_2 )\right\} =0$ and
$\inf_{\theta\in\Theta: |\theta-\theta^*|\ge \delta} \left| E \left\{ \tau(U; \theta, \bar\alpha_2, \bar\gamma_2 )\right\} \right|>0$ for each $\delta>0$.

\item[(ii)] $E \left\{ \sup_{\theta\in\Theta} \left| \tau (U; \theta, \bar\alpha_2, \bar\gamma_2) \right| \right\} < \infty$.

\item[(iii)] There exists a neighborhood $\mathcal N_2=\{(\alpha,\gamma): \|\alpha -\bar\alpha_2\|_1 \le c_3, \|\gamma -\bar\gamma_2\|_1 \le c_3\}$ for a constant $c_3>0$ such that
$E \{ T^{(2) 2}_{\eta_g} (U; \bar\alpha_2,\bar\gamma_2) \} < \infty$ and $E\{T^{(2) 2}_{\eta_f} (U; \bar\alpha_2,\bar\gamma_2)\}<\infty$, where
$ T^{(2)}_{\eta_g} (U; \bar\alpha_2,\bar\gamma_2) = \sup_{\theta\in\Theta, (\alpha,\gamma)\in \mathcal N_2} |\frac{\partial\tau}{\partial\eta_g} (U; \theta, \alpha, \gamma) |$ and
$ T^{(2)}_{\eta_f} (U; \bar\alpha_2,\bar\gamma_2) = \sup_{\theta\in\Theta, (\alpha,\gamma)\in \mathcal N_2} |\frac{\partial\tau}{\partial\eta_f}$ $ (U; \theta, \alpha, \gamma) | $.
\end{itemize}
\end{ass}

\begin{ass} \label{ass:theta2-rate}
There exist positive constants $c_4$ and $C_4$ such that
the following conditions are satisfied, where $\mathcal N_3 = \{(\theta,\alpha,\gamma): |\theta-\theta^*|\le c_4, \|\alpha-\bar\alpha_2\|_1 \le c_4, \|\gamma-\bar\gamma_2\| \le c_4 \}$.\vspace{-.1in}
\begin{itemize}\addtolength{\itemsep}{-.1in}
\item[(i)] $E \{ \sup_{(\theta,\alpha,\gamma)\in \mathcal N} \tau^2 (U; \theta,\alpha,\gamma) \} < \infty$.

\item[(ii)]  $H= E  \{ \frac{\partial \tau}{\partial\theta}(U; \theta^*, \bar\alpha_2, \bar\gamma_2)  \} \not=0 $ and
$E  \{ \sup_{(\theta,\alpha,\gamma)\in \mathcal N_3} | \frac{\partial \tau}{\partial\theta}(U; \theta,\alpha,\gamma)  |  \} < \infty$.

\item[(iii)] The variables $\frac{\partial\tau}{\partial\eta_g}(U;\theta^*, \bar\alpha_2,\bar\gamma_2)$ and $\frac{\partial\tau}{\partial\eta_f}(U;\theta^*, \bar\alpha_2,\bar\gamma_2)$
are sub-exponential.

\item[(iv)] The variables
$ T^{(2)}_{\eta_g^2}(U;\theta^*, \bar\alpha_2, \bar\gamma_2)  = \sup_{(\theta,\alpha,\gamma)\in \mathcal N_3} |\frac{\partial^2 \tau}{\partial \eta_g^2}(U;\theta,\alpha, \gamma) |$,
$ T^{(2)}_{\eta_f^2}(U;\theta^*, \bar\alpha_2, \bar\gamma_2)  = \sup$ $_{(\theta,\alpha,\gamma)\in \mathcal N_3} |\frac{\partial^2 \tau}{\partial \eta_f^2}(U;\theta,\alpha, \gamma) |$, and
$ T^{(2)}_{\eta_g\eta_f} (U;\theta^*, \bar\alpha_2, \bar\gamma_2) =   \sup_{(\theta,\alpha,\gamma)\in \mathcal N_3} |\frac{\partial^2 \tau}{\partial \eta_g \partial\eta_f}(U;\theta,\alpha, \gamma) |$
are sub-exponential, and
$E \{ T^{(2)}_{\eta_g^2}(U;\theta^*, \bar\alpha_2, \bar\gamma_2) |X \} \le C_4$, $E \{ T^{(2)}_{\eta_f^2}(U;\theta^*, \bar\alpha_2, \bar\gamma_2) |X \} \le C_4$,
and $E \{ T^{(2)}_{\eta_g\eta_f} (U;\theta^*, \bar\alpha_2, \bar\gamma_2) |X  \} \le C_4$ almost surely.
\end{itemize}
\end{ass}

\begin{thm} \label{thm:theta2}
In the setting of Theorem~\ref{thm:alpha2}, suppose that Assumption \ref{ass:theta2-consistency} and \ref{ass:theta2-rate} hold and  $M_2 r_0 =o(1)$.
Then $\hat\theta_2$ is consistent for $\theta^*$ and admits the asymptotic expansion
\begin{align}
\hat\theta_2 - \theta^*
= - H^{-1} \tilde E \left\{ \tau(U; \theta^*, \bar\alpha_2, \bar\gamma_2 ) \right\} + O_p( M_2 r_0^2 ), \label{eq:theta2-expan}
\end{align}
where $ H = E \{ \frac{\partial \tau}{\partial\theta}(U; \theta^*, \bar\alpha_2, \bar\gamma_2 ) \}$.
Moreover, a consistent estimator of
$V= \var \{\tau(U; \theta^*, \bar\alpha_2, \bar\gamma_2 ) \}$ $ / H^2$ is
$\hat V = \tilde E \{\tau^2(U;\hat\theta_2,\hat\alpha_2,\hat\gamma_2 ) \} / \hat H^2$,
where $\hat H = \tilde E \{ \frac{\partial \tau}{\partial\theta}(U; \hat\theta_2, \hat\alpha_2, \hat\gamma_2 ) \}$.
\end{thm}

\begin{rem} \label{rem:theta2}
Assumption~\ref{ass:theta2-consistency} is involved to show the consistency of $\hat\theta_2$ for $\theta^*$. Assumptions~\ref{ass:theta2-consistency}(i)--(ii) are standard
for showing consistency if $\tau(U; \theta, \bar\alpha_2, \bar\gamma_2 )$ were employed as an estimating functinon in $\theta$ (e.g., \citealt{van2000asymptotic}).
Assumption~\ref{ass:theta2-consistency}(iii) is used to control the deviation of $(\hat\alpha_2,\hat\gamma_2)$ from the target values, with unrestricted $\theta\in\Theta$.
Moreover, Assumption~\ref{ass:theta2-rate} is involved to show the asymptotic expansion (\ref{eq:theta2-expan}).
Assumption~\ref{ass:theta2-rate}(i)--(ii) is adapted from classical asymptotic theory for maximum likelihood estimation (e.g., \citealt{ferguson1996course}).
Assumption~\ref{ass:theta2-rate}(iv) is used to control the deviation of  $(\hat\theta_2,\hat\alpha_2,\hat\gamma_2)$.
\end{rem}

Combining Theorems~\ref{thm:gamma2}--\ref{thm:theta2} leads to Proposition~\ref{pro:main} provided
$M_2 r_0^2 = o(n^{-1/2})$, i.e., the remainder term in (\ref{eq:theta2-expan}) reduces to $o_p(n^{-1/2})$.
As motivated in Section~\ref{sec:cal} and made explicit in the proofs, the primary reason for $\hat\theta_2$ to achieve asymptotic expansion (\ref{eq:theta2-expan})
is that the two-step estimators $(\hat\alpha_2,\hat\gamma_2)$ are constructed such that according to Proposition~\ref{pro:two-step},
the target values $(\bar\alpha_2, \bar\gamma_2)$ satisfy the calibration equations (\ref{eq:CAL1})--(\ref{eq:CAL2})
if model (\ref{eq:g-model}) or (\ref{eq:f-model}) is correctly specified. In this case,
both the linear and quadratic terms in $(\hat\alpha_2-\bar\alpha_2, \hat\gamma_2-\bar\gamma_2)$ are $O_p(M_2 r_0^2)$ from a Taylor expansion argument. Otherwise, the linear
term would in general be $O_p(M_2^{1/2} r_0)$, as reflected in the convergence rate for the initial estimator $\hat\theta_1$ in Assumption~\ref{ass:gamma2-basic}(iv).

\section{Applications} \label{sec:applications}

\begin{eg} \label{eg:PLM-RCAL}
Return to Examples~\ref{eg:PLM} and \ref{eg:PLM-EE} with a partially linear model (\ref{eq:PLM}).
For $ g(x;\alpha) = \alpha^\T \xi$ and $f(x; \gamma) = \psi_f( \gamma^\T \xi)$, models (\ref{eq:g-model}) and (\ref{eq:f-model}) can be stated as
\begin{align}
& E ( Y | Z, X) = \theta Z + \alpha^\T \xi, \label{eq:g-PLM} \\
& E ( Z| X ) = \psi_f ( \gamma^\T \xi). \label{eq:f-PLM}
\end{align}
For estimating function $\tau$ in (\ref{eq:PLM-tau}) and any estimators $(\hat\alpha,\hat\gamma)$,
$\hat\theta=\hat\theta(\hat\alpha,\hat\gamma)$ as a solution to
$\tilde E\{ \tau(U; \theta, \hat\alpha, \hat\gamma)\}=0$ is of closed form:
\begin{align*}
\hat\theta (\hat\alpha, \hat\gamma) = \frac{\tilde E \{ (Y- \hat\alpha^\T \xi) ( Z - \psi_f(\hat\gamma^\T \xi))\} } {\tilde E\{ Z( Z - \psi_f(\hat\gamma^\T \xi))\} }.
\end{align*}
For initial estimation, let $(\hat\theta_0, \hat\alpha_1)$ be Lasso regularized least-squares estimators in model (\ref{eq:g-PLM}),
$\hat \gamma_1$ be a Lasso regularized quasi-likelihood estimator in model (\ref{eq:f-PLM}), and $\hat\theta_1=\hat\theta(\hat\alpha_1,\hat\gamma_1)$.
For second-step estimation,
the regularized calibrated estimator $\hat\gamma_2$ is defined with a Lasso penalty and the loss function 
\begin{align}
L_2 ( \gamma ) = \tilde E \{\ell_2(U;\gamma)\}  = \tilde E \left\{ - Z \gamma^\T \xi + \Psi_f (\gamma^\T \xi)  \right\} , \label{eq:PLM-loss-gamma}
\end{align}
and $\hat\alpha_2$ is defined with a Lasso penalty and the loss function
\begin{align}
L_1 (\alpha; \hat\theta_1,\hat\gamma_2 ) =  \tilde E \{\ell_1(U; \hat\theta_1, \alpha,\hat\gamma_2)\}  = \tilde E \left\{ \psi_f^\prime(\hat\gamma_2^\T \xi) (Y - \hat\theta_1 Z - \alpha^\T \xi)^2  \right\} , \label{eq:PLM-loss-alpha}
\end{align}
where $\Psi_f(t) = \int_0^t \psi_f(u)\,\dif u$ and $\psi_f^\prime$ is the derivative of $\psi_f$, and $\ell_1$ and $\ell_2$ are determined from (\ref{eq:convex-loss}),
with $(\partial\tau/\partial\alpha, \partial\tau/\partial\gamma)$ in (\ref{eq:PLM-EE1})--(\ref{eq:PLM-EE2}).
The estimator $\hat\gamma_2$ coincides with the usual estimator $\hat\gamma_1$ with a canonical link in (\ref{eq:f-PLM}),
whereas $\hat\alpha_2$ can be interpreted as a regularized weighted least squares estimator.
The resulting estimator of $\theta$ is then $\hat\theta_2 = \hat\theta(\hat\alpha_2,\hat\gamma_2)$.
\end{eg}

We stress that the loss (\ref{eq:PLM-loss-alpha}) is for estimation of $\alpha$ with $(\hat\theta_1,\hat\gamma_2)$ fixed,
and $\hat\theta_1$ is determined as $\hat\theta(\hat\alpha_1,\hat\gamma_1)$ and hence pointwise doubly robust (see Remark~\ref{rem:gamma2-theta1}).
In other words, for $\hat\theta_2$ to admit doubly robust confidence intervals as in Proposition~\ref{pro:main},
it is in general incorrect to (i) replace $\hat\theta_1$ in (\ref{eq:PLM-loss-alpha}) by $\hat\theta_0$ computed from the first step,
or (ii) to redefine $(\hat\theta_1, \hat\alpha_2)$ jointly as a regularized weighted least squares estimator for $Y | (Z,X)$, with weight $\psi_f^\prime(\hat\gamma_2^\T \xi)$.
Nevertheless, these simple options become valid in the special situation where $\psi_f()$ is an identity function, i.e., (\ref{eq:f-PLM}) is a linear model.
In this case, $\hat\gamma_2$ can be taken the same as $\hat\gamma_1$ because (\ref{eq:PLM-loss-gamma}) becomes the usual least-squares loss,
and then either option (i) or (ii) can be shown to yield
$\hat\theta_2$ identical to the first-step estimator $\hat\theta(\hat\alpha_1,\hat\gamma_1)$,
provided that the same Lasso tuning parameter is used in computing $\hat\alpha_2$
as in computing $(\hat\theta_0,\hat\alpha_1)$. See the proof of Corollary~\ref{cor:debiased}.
Moreover, $\hat\theta(\hat\alpha_1,\hat\gamma_1)$ can be expressed as
a debiased Lasso estimator of $\theta$ in linear regression (\ref{eq:g-PLM}) (\citealt{zhang2014confidence, SVD, javanmard2014confidence}):
\begin{align*}
\hat\theta_{\mbox{\tiny DB}} & = \hat\theta (\hat\alpha_1,\hat\gamma_1) = \frac{\tilde E \{ (Y- \hat\alpha_1^\T \xi) ( Z - \hat\gamma_1^\T \xi )\} } {\tilde E\{ Z( Z - \hat\gamma_1^\T \xi) \} }
= \hat\theta_0 +  \frac{\tilde E \{ (Y- \hat\theta_0 Z - \hat\alpha_1^\T \xi) ( Z - \hat\gamma_1^\T \xi )\} } {\tilde E\{ Z( Z - \hat\gamma_1^\T \xi) \} },
\end{align*}
where $(\hat\theta_0,\hat\alpha_1)$ are jointly Lasso estimators  in linear regression of $Y|(Z,X)$, and
$\hat\gamma_1$ is that in linear regression of $Z|X$.
Suppose that the Lasso tuning parameters are sufficiently large, of order $O( \{\log(\me p)/n\}^{1/2} )$.
The following result can be deduced from Proposition~\ref{pro:main}.

\begin{cor} \label{cor:debiased}
Suppose that Assumption~\ref{ass:gamma2-basic}(i) and a compatibility condition holds for $\Sigma_0 = E(\xi\xi^\T)$,
$Y - \theta^*Z - \bar\alpha_1^\T \xi$ and $Z - \bar \gamma_1^\T \xi$ are sub-exponential,
$V = E \{ (Y - \theta^*Z - \bar\alpha_1^\T \xi)^2 (Z - \bar \gamma_1^\T \xi)^2 \} <\infty$,
and $H = - E \{ Z (Z - \bar \gamma_1^\T \xi)\} \not= 0$.
If model (\ref{eq:g-PLM}) or model (\ref{eq:f-PLM}) with $\psi_f \equiv 1$ is correctly specified, then
the conclusions in Proposition~\ref{pro:main} are valid for $\hat\theta_{\mbox{\tiny DB}}=\hat\theta(\hat\alpha_1,\hat\gamma_1)$,
where
\begin{align*}
\hat V = \tilde E\{ (Y - \hat\theta_{\mbox{\tiny DB}} Z - \hat\alpha_1^\T \xi)^2 (Z - \hat \gamma_1^\T \xi)^2 \} / E^2 \{ Z (Z - \hat \gamma_1^\T \xi)\}.
\end{align*}
Hence a doubly robust confidence interval for $\theta^*$ is obtained in partially linear model (\ref{eq:PLM}).
\end{cor}

From Corollary~\ref{cor:debiased}, the debiased Lasso estimator $\hat\theta_{\mbox{\tiny DB}}$ in linear regression (\ref{eq:g-PLM}) can be used to obtain doubly robust
confidence intervals for $\theta^*$ in a partially linear model. This finding appears new and gives a high-dimensional extension of
the double robustness (including pointwise and confidence intervals) of least-squares estimation in low-dimensional settings (Example~\ref{eg:PLM-LS}).
It is helpful to make several comments.
First, although $\hat\theta_{\mbox{\tiny DB}}$ is the same point estimator, the variance estimator $\hat V$ differs from
those originally in debiased Lasso, in the context of linear regression with a constant error variance, which then needs to be estimated
(\citealt{zhang2014confidence, SVD, javanmard2014confidence}).

Second, \cite{buhlmann2015high} studied debiased Lasso in possibly misspecified linear regression.
They employed the same point estimator $\hat\theta_{\mbox{\tiny DB}}$ and proposed a variance estimator similar to $\hat V$,
\begin{align*}
\hat V_{\mbox{\tiny BG}} =  \tilde E\{ ( \hat\varepsilon \hat Z - \tilde E(\hat\varepsilon\hat Z))^2 \} / E^2(Z \hat Z ),
\end{align*}
where $\hat \varepsilon = Y - \hat\theta_0 Z - \hat\alpha_1^\T \xi$ and $\hat Z = Z - \hat \gamma_1^\T \xi$.
Specifically, $\hat V_{\mbox{\tiny BG}}$ can be obtained from $\hat V$ by replacing $ \hat\theta_{\mbox{\tiny DB}}$ with $\hat\theta_0$  and
the sample second-moment of the product $(Y - \hat\theta_0 Z - \hat\alpha_1^\T \xi) (Z - \hat \gamma_1^\T \xi)$ with  the sample variance.
\cite{buhlmann2015high} showed that under suitable conditions, $\hat\theta_{\mbox{\tiny DB}} \pm z_{c/2} \sqrt{ \hat V_{\mbox{\tiny BG}}}$ is
a $(1-c)$ confidence interval for $\bar\theta_0$, defined such that
\begin{align*}
(\bar\theta_0, \bar\alpha_1) = \argmin_{ (\theta, \alpha)} E \left\{ (Y - \theta Z - \alpha^\T \xi)^2 \right\},
\end{align*}
with possible misspecfication of linear model (\ref{eq:g-PLM}).
This result is compatible with ours, because, from the proof of Corollary~\ref{cor:debiased},
$\bar\theta_0$ identifies $\theta^*$ in partially linear model (\ref{eq:PLM}) if
linear model (\ref{eq:g-PLM}) is misspecified but a linear model for $E(Z|X)$ is correctly specified.

Finally, for a nonlinear model (\ref{eq:f-model}) with $\psi_f$ a non-identity function (for example when $Z$ is binary or nonnegative),
our estimator $\hat\theta_2$ and associated confidence intervals are distinct from debiased Lasso including \cite{buhlmann2015high}.
Although $\hat\theta_{\mbox{\tiny DB}} \pm z_{c/2} \sqrt{ \hat V_{\mbox{\tiny BG}}}$ remains a $(1-c)$ confidence interval for $\bar\theta_0$ under suitable conditions,
the target value $\bar\theta_0$ may in general differ from $\theta^*$ in partially linear model  (\ref{eq:PLM}) even if
model (\ref{eq:g-PLM}) is misspecified but a nonlinear model for $E(Z|X)$ is correctly specified.

\begin{eg} \label{eg:PLM-log-RCAL}
Return to Examples~\ref{eg:PLM-log} and \ref{eg:PLM-log-EE} with a partially log-linear model (\ref{eq:PLM-log}).
For $ g(x;\alpha) = \alpha^\T \xi$ and $f(x; \gamma) = \psi_f( \gamma^\T \xi)$, models (\ref{eq:g-model}) and (\ref{eq:f-model}) can be stated as
\begin{align}
& E ( Y | Z, X) = \exp( \theta Z + \alpha^\T \xi ), \label{eq:g-PLM-log} \\
& E ( Z| X ) = \psi_f ( \gamma^\T \xi). \label{eq:f-PLM-log}
\end{align}
For estimating function $\tau$ in (\ref{eq:PLM-log-tau}) and any estimators $(\hat\alpha,\hat\gamma)$,
$\hat\theta=\hat\theta(\hat\alpha,\hat\gamma)$ is a solution to
\begin{align}
0 = \tilde E\{ \tau(U; \theta, \hat\alpha, \hat\gamma)\}= \tilde E \left\{ (Y\me^{-\theta Z} - \me^{\hat\alpha^\T \xi} ) ( Z - \psi_f(\hat\gamma^\T \xi)) \right\}  . \label{eq:theta-PLM-log}
\end{align}
For initial estimation, let $(\hat\theta_0, \hat\alpha_1)$ be Lasso regularized quasi-likelihood estimators in model (\ref{eq:g-PLM-log}),
$\hat \gamma_1$ be that in model (\ref{eq:f-PLM-log}), and $\hat\theta_1=\hat\theta(\hat\alpha_1,\hat\gamma_1)$.
For second-step estimation,
the regularized calibrated estimator $\hat\gamma_2$ is defined with a Lasso penalty and the loss function
\begin{align}
L_2 ( \gamma; \hat\alpha_1 ) = \tilde E\{\ell_2(U;\hat\alpha_1, \gamma)\}  = \tilde E \left[ \me^{\hat\alpha_1^\T\xi} \left\{ - Z \gamma^\T \xi + \Psi_f (\gamma^\T \xi)  \right\} \right], \label{eq:PLM-log-loss-gamma}
\end{align}
and $\hat\alpha_2$ is defined with a Lasso penalty and the loss function
\begin{align*}
L_1 (\alpha; \hat\theta_1,\hat\gamma_2 ) =  \tilde E \{\ell_1(U; \hat\theta_1, \alpha,\hat\gamma_2)\}  =
\tilde E \left\{ \psi_f^\prime(\hat\gamma_2^\T \xi) \left( -Y \me^{-\hat\theta_1 Z} \alpha^\T \xi + \me^{\alpha^\T \xi} \right) \right\} , 
\end{align*}
where $\ell_1$ and $\ell_2$ are determined from (\ref{eq:convex-loss}),
with $(\partial\tau/\partial\alpha, \partial\tau/\partial\gamma)$ in (\ref{eq:PLM-log-EE1})--(\ref{eq:PLM-log-EE2}).
Unlike (\ref{eq:PLM-loss-gamma}), the loss (\ref{eq:PLM-log-loss-gamma}) in $\gamma$ depends on $\hat\alpha_1$.
The resulting estimator of $\theta$ is then $\hat\theta_2 = \hat\theta(\hat\alpha_2,\hat\gamma_2)$.
\end{eg}

In contrast with Example~\ref{eg:PLM-RCAL}, our method is distinct from debiased Lasso, even when $\psi_f \equiv 1$, i.e., (\ref{eq:f-PLM-log}) is a linear model.
Similarly as in our method, let $(\hat\theta_0,\hat\alpha_1)$ be the Lasso estimators associated with the loss
$ \tilde E \{ -Y (\theta Z + \alpha^\T \xi) + \me^{\theta Z +\alpha^\T \xi} \}$, and $\hat\gamma_1$ be that associated with the loss
$ \tilde E \{ \me^{\hat \theta_0 Z +\hat \alpha_1^\T \xi} (Z - \gamma^\T \xi)^2 \}$.
The debiased Lasso estimator in \cite{SVD}, also called the one-step estimator in \cite{ning2017general}, is
\begin{align*}
\hat\theta_{\mbox{\tiny DB}} = \hat\theta_0 + \frac{ \tilde E \left\{ (Y - \me^{\hat \theta_0 Z +\hat \alpha_1^\T \xi} )(Z-\hat\gamma_1^\T \xi) \right\} }
{ \tilde E \left\{ \me^{\hat \theta_0 Z +\hat \alpha_1^\T \xi}  Z( Z -\hat\gamma_1^\T \xi) \right\} }.
\end{align*}
A variation of debiased Lasso in \cite{neykov2018unified} is to define $\hat\theta_{\mbox{\tiny DB}2}$ as a solution to
\begin{align}
 \tilde E \left\{ (Y - \me^{\theta Z +\hat \alpha_1^\T \xi} )(Z-\hat\gamma_1^\T \xi) \right\} = 0 . \label{eq:DB-PLM-log}
\end{align}
Equation (\ref{eq:DB-PLM-log}) is somewhat similar to (\ref{eq:theta-PLM-log}) with $\psi_f\equiv 1$ , but there is an important difference.
Equation (\ref{eq:theta-PLM-log}) is doubly robust: its limit version, with $(\hat\alpha_1,\hat\gamma_1)$ replaced by their limit values and $\tilde E()$ replaced by $E()$,
holds at $\theta= \theta^*$ in partially log-linear model (\ref{eq:PLM-log}) if
either model (\ref{eq:g-PLM-log}) or model (\ref{eq:f-PLM-log}) is correctly specified.
In contrast, (\ref{eq:DB-PLM-log}) is not doubly robust: its limit version in general holds at $\theta=\bar\theta_0$, defined such that
\begin{align*}
(\bar\theta_0, \bar\alpha_1 )= \argmin_{\theta,\alpha}  E \left\{ -Y (\theta Z + \alpha^\T \xi) + \me^{\theta Z +\alpha^\T \xi} \right\}.
\end{align*}
The target value $\bar\theta_0$ coincides with $\theta^*$ if model (\ref{eq:g-PLM-log}) is correctly specified, but in general
may differ from $\theta^*$ otherwise including when model (\ref{eq:f-PLM-log}) with $\psi_f\equiv 1$ is correctly specified.

\begin{eg} \label{eg:PLM-logit-RCAL}
Return to Examples~\ref{eg:PLM-logit} and \ref{eg:PLM-logit-EE} with a partially log-linear model (\ref{eq:PLM-logit}).
For $ g(x;\alpha) = \alpha^\T \xi$ and $f(x; \gamma) = \psi_f( \gamma^\T \xi)$, models (\ref{eq:g-model}) and (\ref{eq:f-model}) can be stated as
\begin{align}
& E ( Y | Z, X) = \expit( \theta Z + \alpha^\T \xi ), \label{eq:g-PLM-logit} \\
& E ( Z| Y=0, X ) = \psi_f ( \gamma^\T \xi). \label{eq:f-PLM-logit}
\end{align}
For estimating function $\tau$ in (\ref{eq:PLM-logit-tau}) and any estimators $(\hat\alpha,\hat\gamma)$,
$\hat\theta=\hat\theta(\hat\alpha,\hat\gamma)$ is a solution to
\begin{align}
0 = \tilde E\{ \tau(U; \theta, \hat\alpha, \hat\gamma)\}= \tilde E \left\{ \me^{-\theta ZY} (Y - \expit(\hat\alpha^\T \xi)) ( Z - \psi_f(\hat\gamma^\T \xi)) \right\}  . \label{eq:theta-PLM-logit}
\end{align}
For initial estimation, let $(\hat\theta_0, \hat\alpha_1)$ be Lasso likelihood estimators in model (\ref{eq:g-PLM-logit}),
$\hat \gamma_1$ be a Lasso quasi-likelihood estimator in model (\ref{eq:f-PLM-logit}), and $\hat\theta_1=\hat\theta(\hat\alpha_1,\hat\gamma_1)$.
For second-step estimation,
the regularized calibrated estimator $\hat\gamma_2$ is defined with a Lasso penalty and the loss function
\begin{align*}
L_2 ( \gamma; \hat\theta_1,\hat\alpha_1 ) = \tilde E\{\ell_2(U;\hat\theta_1,\hat\alpha_1,\gamma)\}  = \tilde E \left[ \me^{-\hat\theta_1 ZY}\expit_2(\hat\alpha_1^\T \xi)
\left\{ - Z \gamma^\T \xi + \Psi_f (\gamma^\T \xi)  \right\} \right], 
\end{align*}
and $\hat\alpha_2$ is defined with a Lasso penalty and the loss function
\begin{align*}
L_1 (\alpha; \hat\theta_1,\hat\gamma_2 ) =  \tilde E \{\ell_1(U; \hat\theta_1, \alpha,\hat\gamma_2)\}  =
\tilde E \left[ \me^{-\hat\theta_1 ZY}\psi_f^\prime(\hat\gamma_2^\T \xi) \left\{ -Y\alpha^\T \xi + \log ( 1+\me^{\alpha^\T \xi} ) \right\} \right] ,
\end{align*}
where $\ell_1$ and $\ell_2$ are determined from (\ref{eq:convex-loss}),
with $(\partial\tau/\partial\alpha, \partial\tau/\partial\gamma)$ in (\ref{eq:PLM-logit-EE1})--(\ref{eq:PLM-logit-EE2}).
The resulting estimator of $\theta$ is then $\hat\theta_2 = \hat\theta(\hat\alpha_2,\hat\gamma_2)$.
\end{eg}

Our method in Example~\ref{eg:PLM-logit-RCAL} differs from debiased Lasso even more substantially than in Examples~\ref{eg:PLM-RCAL}--\ref{eg:PLM-log-RCAL}.
The debiased Lasso estimator in van de Geet et al.~(2014) is
\begin{align*}
\hat\theta_{\mbox{\tiny DB}} = \hat\theta_0 + \frac{ \tilde E \left\{ (Y - \expit(\hat \theta_0 Z +\hat \alpha_1^\T \xi) )(Z-\tilde \gamma_1^\T \xi) \right\} }
{ \tilde E \left\{ \expit_2(\hat \theta_0 Z +\hat \alpha_1^\T \xi)  Z( Z -\tilde \gamma_1^\T \xi) \right\} },
\end{align*}
and a variation $\hat\theta_{\mbox{\tiny DB}2}$ in \cite{neykov2018unified} is a solution to
\begin{align*}
 \tilde E \left\{ (Y - \expit(\theta Z +\hat \alpha_1^\T \xi) )(Z-\tilde \gamma_1^\T \xi) \right\} = 0 ,
\end{align*}
where $(\hat\theta_0,\hat\alpha_1)$ are Lasso estimators in model (\ref{eq:g-PLM-logit}) as in our method, but
$\tilde\gamma_1$, different from $\hat\gamma_1$, is a Lasso estimator associated with the loss
$\tilde E \{ \expit_2( \hat\theta_0 Z + \hat\alpha_1^\T \xi) ( Z - \gamma^\T \xi)^2 \}$,  corresponding to a model $E(Z|X) = \gamma^\T \xi$
instead of model (\ref{eq:f-PLM-logit}) in our method.
Confidence intervals based on $\hat\theta_{\mbox{\tiny DB}}$ or $\hat\theta_{\mbox{\tiny DB2}}$ would not be valid for $\theta^*$ in partially linear model
(\ref{eq:PLM-logit}) if model (\ref{eq:g-PLM-logit}) is misspecified, irrespective of whether model (\ref{eq:f-PLM-logit}) is correctly specified.

\begin{eg} \label{eg:ATE-RCAL}
Return to Examples~\ref{eg:ATE} and \ref{eg:ATE-EE}.
For $ g(x;\alpha) = \psi_g(\alpha^\T \xi)$ and $f(x; \gamma) = \psi_f( \gamma^\T \xi)$, models (\ref{eq:g-model}) and (\ref{eq:f-model}) can be stated as
\begin{align}
& E ( Y | Z=1, X) = \psi_g(\alpha^\T \xi ), \label{eq:g-ATE} \\
& P(Z=1 | X ) = \psi_f ( \gamma^\T \xi). \label{eq:f-ATE}
\end{align}
For estimating function $\tau$ in (\ref{eq:ATE-tau}) and any estimators $(\hat\alpha,\hat\gamma)$,
$\hat\theta=\hat\theta(\hat\alpha,\hat\gamma)$ is of closed form
\begin{align}
\hat\theta(\hat\alpha, \hat\gamma) = \tilde E \left[\frac{Z Y}{\psi_f(\hat\gamma^\T \xi)} - \left\{ \frac{Z}{\psi_f(\hat\gamma^\T \xi)} -1 \right\} \psi_g(\hat\alpha^\T \xi) \right]. \label{eq:theta-ATE}
\end{align}
For initial estimation, let $\hat\alpha_1$ be a Lasso quasi-likelihood estimator in model (\ref{eq:g-ATE}),
$\hat \gamma_1$ be that in model (\ref{eq:f-ATE}), and $\hat\theta_1=\hat\theta(\hat\alpha_1,\hat\gamma_1)$.
For second-step estimation,
the regularized calibrated estimator $\hat\gamma_2$ is defined with a Lasso penalty and the loss function
\begin{align*}
L_2 ( \gamma; ,\hat\alpha_1 ) = \tilde E\{\ell_2(U; \hat\alpha_1,\gamma)\}  = \tilde E \left[ \psi_g^\prime(\hat\alpha_1^\T \xi)
\left\{ - Z \Psi_f(\gamma^\T \xi) + \gamma^\T \xi) \right\} \right],
\end{align*}
and $\hat\alpha_2$ is defined with a Lasso penalty and the loss function
\begin{align*}
L_1 (\alpha; \hat\gamma_2 ) =  \tilde E \{\ell_1(U; \alpha,\hat\gamma_2)\}  =
\tilde E \left[ \frac{\psi^\prime_f(\hat\gamma_2^\T \xi)}{\psi_f^2(\hat\gamma_2^\T \xi)} Z \left\{ -Y\alpha^\T \xi + \Psi_g(\alpha^\T \xi) \right\} \right] ,
\end{align*}
where
$\Psi_f (u) = \int_0^u \psi_f^{-1}(t)\,\dif t$,
$\Psi_g (u) =\int_0^u \psi_g (t)\,\dif t$, and
$\ell_1$ and $\ell_2$ are determined from (\ref{eq:convex-loss}),
with $(\partial\tau/\partial\alpha, \partial\tau/\partial\gamma)$ in (\ref{eq:ATE-EE1})--(\ref{eq:ATE-EE2}).
The resulting estimator of $\theta$ is then $\hat\theta_2 = \hat\theta(\hat\alpha_2,\hat\gamma_2)$.
\end{eg}

By Proposition~\ref{pro:main}, valid confidence intervals based on $\hat\theta_2$ can be obtained for $\theta^*=E(Y)$ if either model (\ref{eq:g-ATE}) or (\ref{eq:f-ATE}) is correctly specified.
Hence our work extends \cite{tan2020model}, where doubly robust confidence intervals are obtained for $\theta^*$ only with linear outcome model (\ref{eq:g-ATE}).
With a nonlinear outcome model, valid confidence intervals are obtained in \cite{tan2020model}, depending on propensity score model (\ref{eq:f-ATE}) being correctly specified.

\section{Simulation studies}

Consider the settings of partially linear, log-linear, and logistic models (Examples~\ref{eg:PLM-RCAL}--\ref{eg:PLM-logit-RCAL}).
Assume that the covariate of interest $Z$ is binary (for example a treatment variable), and hence
the coefficient $\theta^*$ represents some homogeneous treatment effect.
The link function for $Z$ given $X$ is taken to be logistic: $\psi_f = \expit(\cdot).$

We investigate the performance of our two-step estimator $\hat{\theta}_2$,
compared with the debiased Lasso estimator $\hat{\theta}_{\mbox{\tiny DB}}$ and the initial estimator $\hat{\theta}_1$
using regularized likelihood (or quasi-likelihood) estimation, as described in Section~\ref{sec:applications}.
For all point estimators, the Lasso tuning parameters are selected via 5-fold cross validation.
Wald confidence intervals based on $\hat\theta_2$ are obtained by Proposition~\ref{pro:main}.
For comparison, confidence intervals based on $\hat\theta_1$ are computed in a similar manner, with $(\hat\alpha_1,\hat\gamma_1,\hat\theta_1)$
in place of $(\hat\alpha_2, \hat\gamma_2, \hat\theta_2)$.
Wald confidence intervals based on $\hat{\theta}_{\mbox{\tiny DB}}$ are computed using a robust variance estimator,
which, for linear modeling, is defined as $\hat V_{\mbox{\tiny BG}}$ in Section~\ref{sec:applications}.
See the Supplement for further implementation details.

\subsection{Partially linear modeling} \label{sec:linear-sim}

Consider the following data-generating configurations for $(Z,X,Y)$, where $\theta^*=3$.
\begin{itemize}\addtolength{\itemsep}{-.1in}
	\item[(C1)] Generate $Z$ as Bernoulli with $P(Z=1)=0.5$ and $X$ given $Z=0$ or 1 as multivariate normal
with means $\mu_0\not=\mu_1$ and variance matrices $\Sigma_0 = \Sigma_1$, such that $Z$ given $X$ is Bernoulli with
$P(Z=1|X) = \expit ( -0.4297-0.25X_1 + 0.5X_2 +  0.75X_3 + X_4 + 1.25X_5)$.
Then $Y$ given $(Z,X)$ is generated as normal with variance 0.5 and mean $E(Y|Z,X) =\theta^*Z +0.25X_1 + 1.5X_2 + 1.75X_3 + 5X_4$.

    \item[(C2)] Generate $Z$ as Bernoulli with $P(Z=1)=0.5$ and $X$ given $Z=0$ or 1 as multivariate normal
with means $\mu_0\not=\mu_1$ and variance matrices $\Sigma_0 \not= \Sigma_1$, such that
$Z$ given $X$ as Bernoulli with $P(Z=1|X) = \expit (- 0.4687 + \frac{p}{2}\ln 2 - 0.25X_1 + 0.5X_2 +0.75X_3 + X_4 + 0.5 X^\T  X)$.
Then $Y$ given $(Z,X)$ is generated as in (C1).

    \item[(C3)] Generate $Z$ given $X$ as in (C1) and then $Y$ given $(Z,X)$ as normal with variance 0.5 and mean $E(Y|Z,X) =
     \theta^*Z + \expit (0.5X_1 + X_2) + 4(X_3 - 0.75) + 2(X_4 - 1)^2$.
\end{itemize}
See the Supplement for details of $(\mu_0,\mu_1)$ and $(\Sigma_0,\Sigma_1)$ and the derivation of $P(Z=1|X)$ stated above, related to Fisher's discrimination analysis.

Consider models (\ref{eq:g-PLM}) for $E(Y|Z,X)$ and (\ref{eq:f-PLM}) for $P(Z=1|X)$, with the regressor vector $\xi = (1,X^\T)^\T$.
Then the two models are both correctly specified in (C1),
model (\ref{eq:g-PLM}) is correctly specified but model (\ref{eq:f-PLM}) is misspecified in (C2),
and model (\ref{eq:g-PLM}) is misspecified and model (\ref{eq:f-PLM}) is correctly specified in (C3).
Hence the simulation settings (C1), (C2), and (C3) are labeled as ``Cor Cor", ``Cor Mis", and ``Mis Cor" respectively.

For $n=400$ and $p=800$, Table~\ref{tab:linear-p800} summarizes the results for estimation of $\theta^*$
and Figure~\ref{fig:linear-p800} shows the QQ plots of $t$-statistics. Additional results with $p=100$ or 200
are included in the Supplement. In settings (C1) and (C2) with model (\ref{eq:g-PLM}) correctly specified for $E(Y|Z,X)$,
the three methods using $\hat\theta_{\mbox{\tiny DB}}$, $\hat\theta_1$, and $\hat\theta_2$
perform similarly to each other. In theory, all the methods in such settings deliver valid confidence intervals.
In setting (C3) with model (\ref{eq:g-PLM}) misspecified for $E(Y|Z,X)$ but model (\ref{eq:f-PLM}) correctly specified for $P(Z=1|X)$,
there are important differences between the three methods.
The debiased Lasso estimator $\hat\theta_{\mbox{\tiny DB}}$ becomes inconsistent for $\theta^*$,
as seen from a large bias and poor coverage proportion. The initial estimator $\hat\theta_1$ is, in theory, consistent but
does not yield valid confidence intervals.
Our calibrated estimator $\hat\theta_2$ shows the best performance, with a small bias and close to 95\% coverage.
The improvement of $\hat\theta_2$ over $\hat\theta_{\mbox{\tiny DB}}$ and $\hat\theta_1$ is also confirmed in Figure~\ref{fig:linear-p800},
where the QQ plot of $t$-statistics from $\hat\theta_2$ is much better aligned with standard normal.

\begin{table}[H]
\caption{Summary of results for partially linear modeling ($n=400, p=800$)} \label{tab:linear-p800}  \vspace{-.2in}
\small
\begin{center}
		\begin{tabular}{cccc|ccc|cccc}
            \toprule
			\multirow{4}{*}{} &
			\multicolumn{3}{c|}{(C1) Cor Cor}&
			\multicolumn{3}{c|}{(C2) Cor Mis}&
			\multicolumn{3}{c}{(C3) Mis Cor}\\
			& $\hat\theta_{\mbox{\tiny DB}}$ & $\hat\theta_1$ & $\hat\theta_2$ & $\hat\theta_{\mbox{\tiny DB}}$ & $\hat\theta_1$ & $\hat\theta_2$ & $\hat\theta_{\mbox{\tiny DB}}$& $\hat\theta_1$& $\hat\theta_2$&\\
			\midrule
			\multirow{1} {*}{Bias}
			& 0.006 & 0.006 & 0.007 & 0.013 & 0.012 & 0.012 & 0.283 & 0.082 & 0.004 \\ 			
			
			\multirow{1}{*}{$\sqrt{\text{Var}}$}
			& 0.058 & 0.057 & 0.057 & 0.057 & 0.057 & 0.058 & 0.322 & 0.301 & 0.299 \\

			\multirow{1}{*}{$\sqrt{\text{Evar}}$}
			& 0.055 & 0.056 & 0.058 & 0.059 & 0.058 & 0.056 & 0.334 & 0.330 & 0.328 \\

			\multirow{1}{*}{Cov95}
			& 0.922 & 0.928 & 0.928 & 0.921 & 0.931 & 0.941 & 0.872 & 0.929 & 0.941 \\
			\bottomrule
		\end{tabular} \\[.1in]
\parbox{1\textwidth}{\scriptsize Note: Bias and Var are the Monte Carlo bias and variance of the point estimator,  EVar
is the mean of the variance estimator, and Cov95 is the coverage proportion of nominal 95\% confidence intervals, based on 1000 repeated simulations.}
\end{center}
\end{table}

\captionsetup[subfigure]{labelformat=empty}

\vspace{-.5in}
\begin{figure}[H]
	\centering
	\subfloat[]{{\includegraphics[width=70mm]{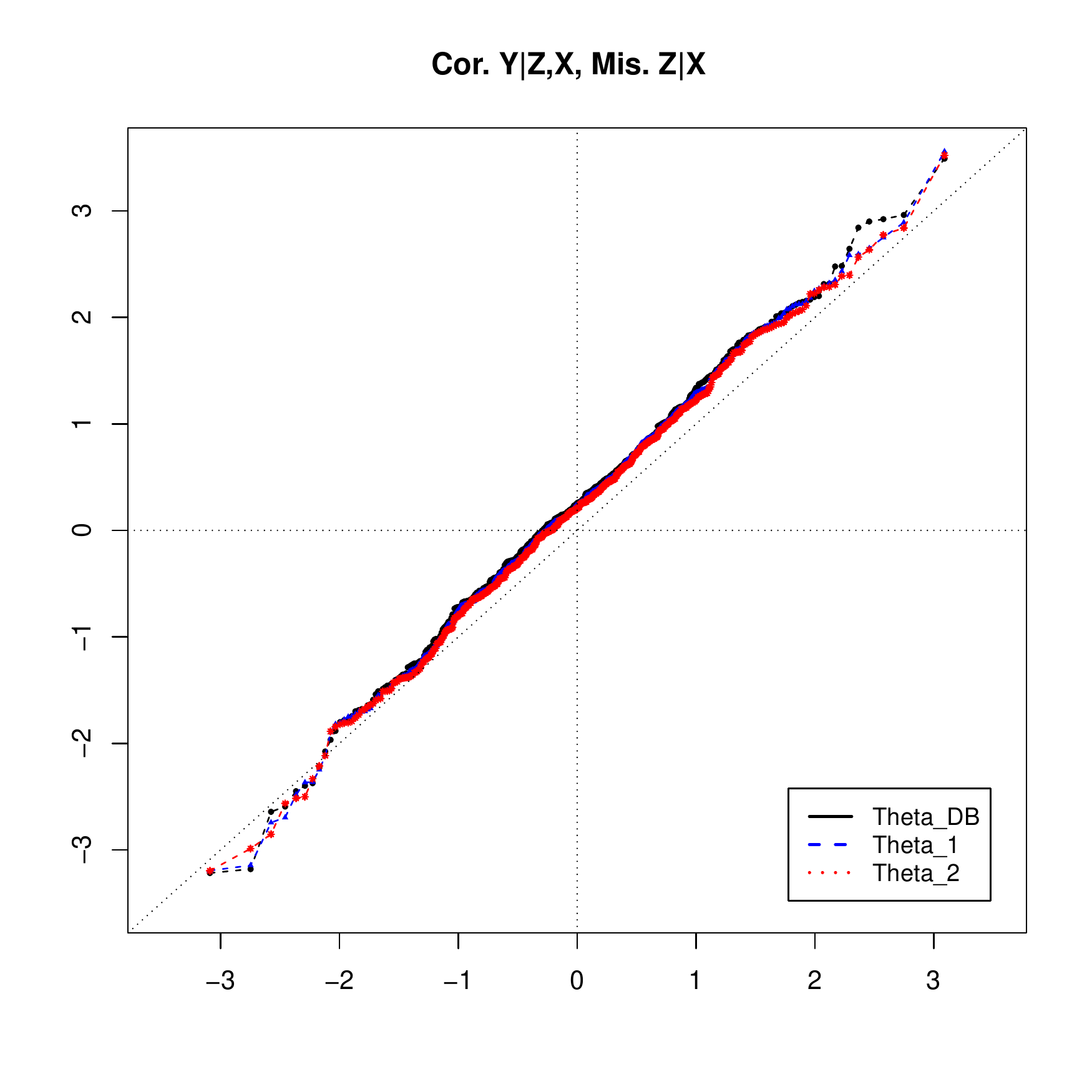} }}%
	\qquad
	\subfloat[]{{\includegraphics[width=70mm]{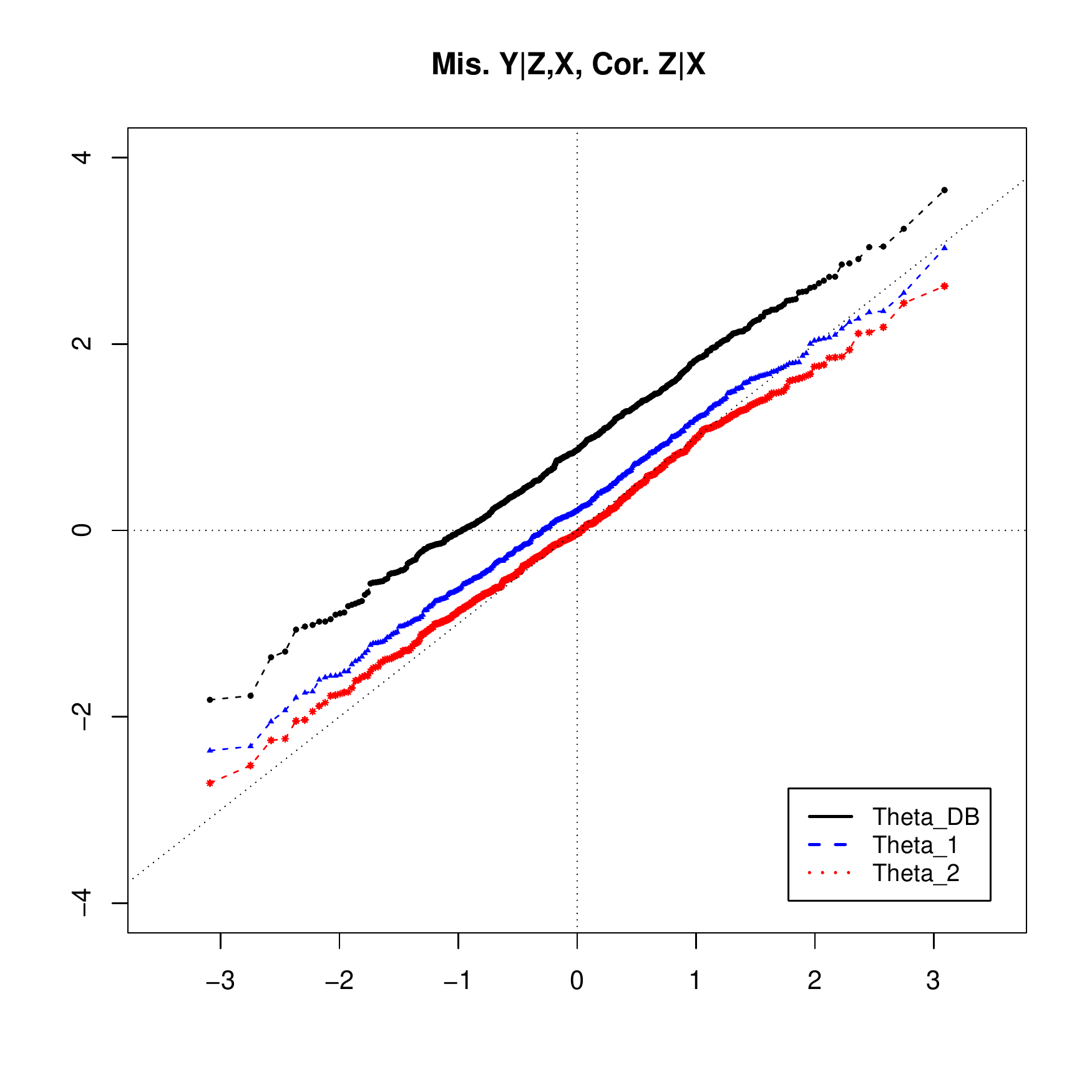} }}%
	\vspace*{-.5in} \caption{QQ plots of $t$-statistics for partially linear modeling ($n=400, p=800$)}%
\label{fig:linear-p800}
\end{figure}

\subsection{Partially log-linear modeling}

Consider the following data-generating configurations for $(Z,X,Y)$, where $\theta^*=2$.
\begin{itemize}\addtolength{\itemsep}{-.1in}
	\item[(C4)] Generate $(Z,X)$ as in (C1) in Section~\ref{sec:linear-sim} and then
$Y$ given $(Z,X)$ as Poisson with mean $\exp(\theta^*Z + 0.1X_1 + 0.25X_2 + 0.5X_3 + 0.75X_4)$.

    \item[(C5)] Generate $(Z,X)$ as in (C2) in Section~\ref{sec:linear-sim} and then
$Y$ as in (C4).

    \item[(C6)] Generate $(Z,X)$ as in (C1) in Section~\ref{sec:linear-sim} and then
    $Y$ given $(Z,X)$ as Poisson with mean $\exp(\theta^*Z + X_1 + 0.1 X_2^2 + 0.2X_3^2)$.
\end{itemize}   \vspace{-.1in}
Consider models (\ref{eq:g-PLM-log}) for $E(Y|Z,X)$ and (\ref{eq:f-PLM-log}) for $P(Z=1|X)$, with the regressor vector $\xi = (1,X^\T)^\T$.
Then the two models are both correctly specified in (C4),
model (\ref{eq:g-PLM-log}) is correctly specified but model (\ref{eq:f-PLM-log}) is misspecified in (C5),
and model (\ref{eq:g-PLM-log}) is misspecified and model (\ref{eq:f-PLM-log}) is correctly specified in (C6).

For $n=600$ and $p=800$, Table~\ref{tab:log-p800} summarizes the results for estimation of $\theta^*$
and Figure~\ref{fig:log-p800} shows the QQ plots of $t$-statistics. Additional results with $p=100$ or 200
are included in the Supplement.
The three methods perform similarly to each other in setting (C4).
However, unlike in Section~\ref{sec:linear-sim}, our calibrated method achieves the best performance in both settings (C5) and (C6),
with a smaller bias, closer to 95\% coverage, and better aligned $t$-statistics with standard normal than the other methods.

\begin{table}[H]
\caption{Summary of results for partially log-linear modeling ($n=600, p=800$)} \label{tab:log-p800}  \vspace{-.2in}
\small
\begin{center}
	\begin{tabular}{cccc|ccc|cccc}
		\toprule
		\multirow{4}{*}{} &
		\multicolumn{3}{c|}{(C4) Cor Cor}&
		\multicolumn{3}{c|}{(C5) Cor Mis}&
		\multicolumn{3}{c}{(C6) Mis Cor}\\
		& $\hat{\theta}_\DB$ & $\hat{\theta}_1$ & $\hat{\theta}_2$ & $\hat{\theta}_\DB$ & $\hat{\theta}_1$ & $\hat{\theta}_2$ & $\hat{\theta}_\DB$ & $\hat{\theta}_1$ & $\hat{\theta}_2$&\\
		\midrule
		\multirow{1} {*}{Bias}
		& 0.008 & 0.009 & 0.005 & -0.023 & -0.015 & 0.005 & -0.093 & -0.021 & 0.010 &\\ 			
		
		\multirow{1}{*}{$\sqrt{\text{Var}}$}
		& 0.043 & 0.046 & 0.048 & 0.045 & 0.045 & 0.048 & 0.075 & 0.078 & 0.081 &\\

		\multirow{1}{*}{$\sqrt{\text{Evar}}$}
		& 0.043 & 0.048 & 0.045 & 0.046 & 0.045 & 0.047 & 0.078 & 0.077 & 0.080 &\\

		\multirow{1}{*}{Cov95}
		& 0.938 & 0.934 & 0.941 & 0.857 & 0.913 & 0.923 & 0.701 & 0.933 & 0.941 &\\
		\bottomrule
		\end{tabular} \\[.1in]
\parbox{1\textwidth}{\scriptsize Note: See the footnote of Table~\ref{tab:linear-p800}.}
\end{center}
\end{table}

\vspace{-.6in}
\begin{figure}[H]
	\centering
	\subfloat[]{{\includegraphics[width=70mm]{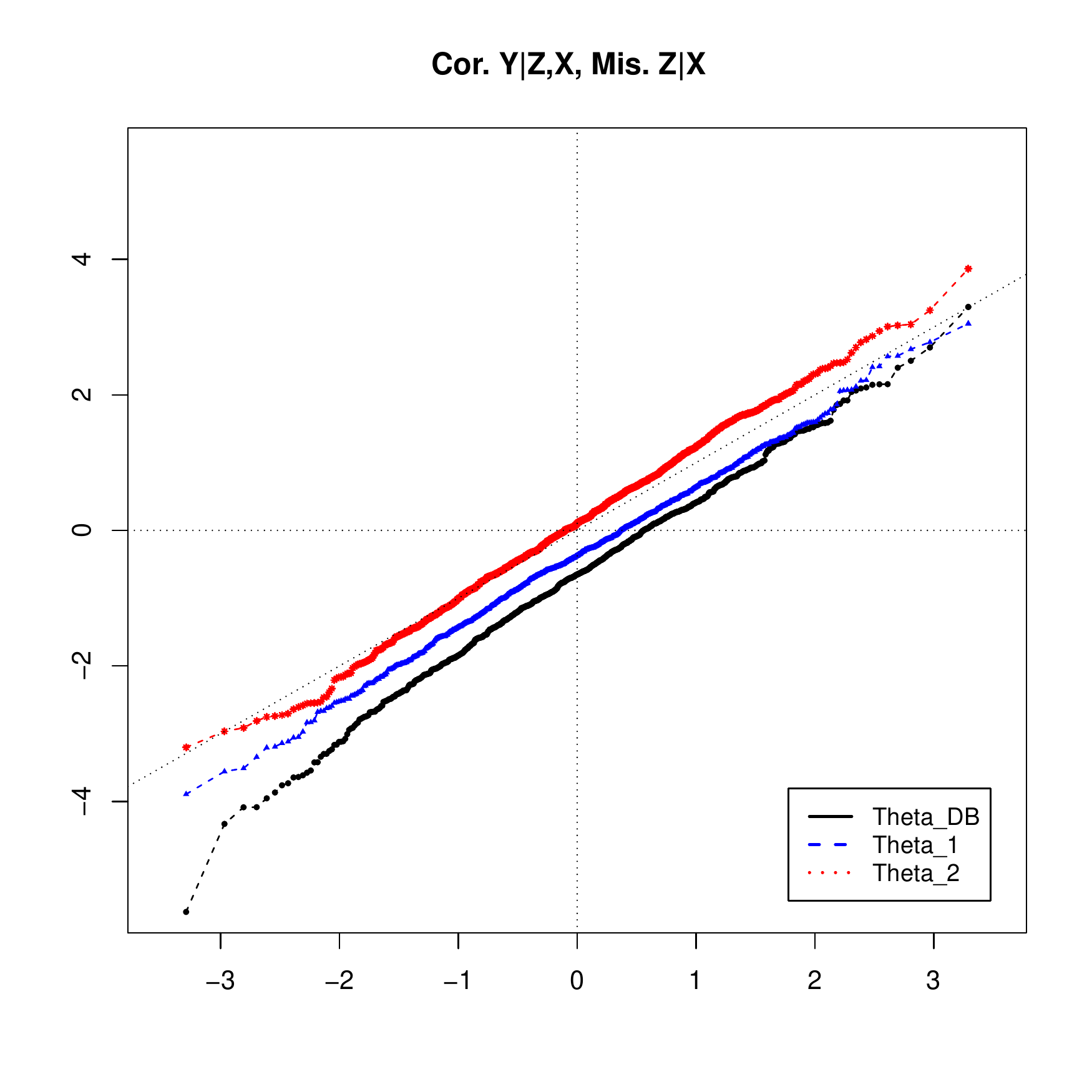} }}%
	\qquad
	\subfloat[]{{\includegraphics[width=70mm]{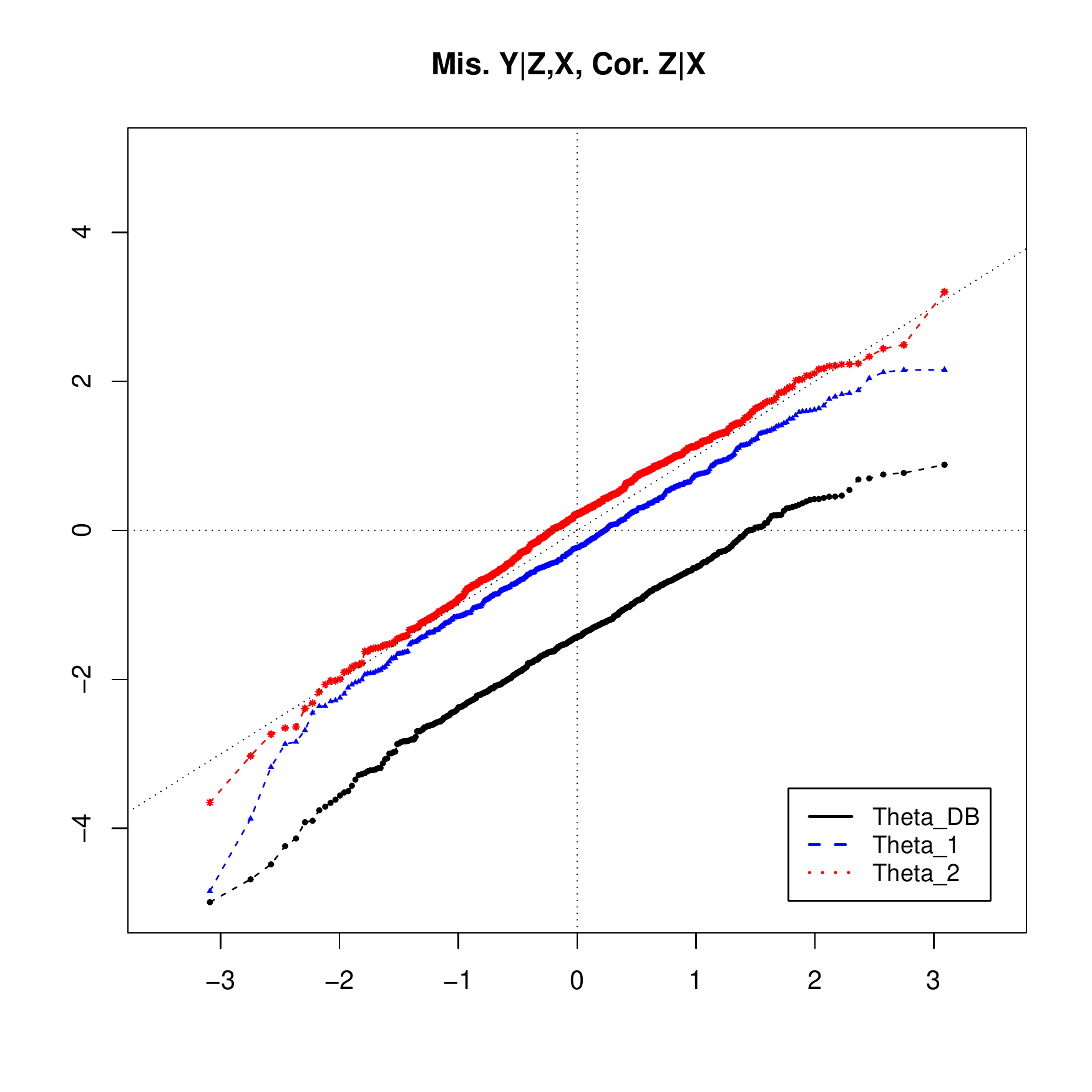} }}%
\vspace*{-.6in}  \caption{QQ plots of $t$-statistics for partially log-linear modeling ($n=600, p=800$)}%
	\label{fig:log-p800}
\end{figure}

\subsection{Partially logistic modeling}

The covariates $X= (X_1,\ldots,X_p)^\T$ are generated as multivariate normal with means 0 and $\cov(X_j, X_k) = 2^{-|j-k|}$ for $1\le j,k\le p$.
Then $(Z,Y)$ given $X$ are generated jointly (rather than sequentially) such that the following configurations are obtained, where $\theta^*=2$.
\begin{itemize}\addtolength{\itemsep}{-.1in}
	\item[(C7)] $Z$ given $Y=0$ and $X$ is Bernoulli with $P(Z = 1|Y = 0, X) = \expit(0.25 - 0.125X_1 + 0.125X_2 + 0.25X_3 + 0.375X_4)$ and
$Y$ given $(Z,X)$ is Bernoulli with $P(Y = 1|Z, X) = \expit(\theta^*Z - 0.125X_1 + 0.125X_2 + 0.25X_3 + 0.375X_4 )$.

    \item[(C8)] $Z$ given $Y=0$ and $X$ is Bernoulli with $P(Z =1|Y = 0, X) =\expit(\theta^*Z - 0.25 + 0.25X_1 + 0.8X_2 + \expit(X_3) )$ and
$Y$ given $(Z,X)$ is the same as in (C7).

    \item[(C9)] $Z$ given $Y=0$ and $X$ is the same as in (C7) and
$Y$ given $(Z,X)$ is Bernoulli with $P(Y = 1|Z, X) =\expit(\theta^*Z - 0.25 + 0.25X_1 + 0.8X_2 + \expit(X_3) )$.

\end{itemize}
See the Supplement for details of data generation, related to the odds ratio model in \cite{chen2007semiparametric}.
Consider models (\ref{eq:g-PLM-logit}) for $E(Y|Z,X)$ and (\ref{eq:f-PLM-logit}) for $P(Z=1|Y=0,X)$, with the regressor vector $\xi = (1,X^\T)^\T$.
Then the two models are both correctly specified in (C7),
model (\ref{eq:g-PLM-logit}) is correctly specified but model (\ref{eq:f-PLM-logit}) is misspecified in (C8),
and model (\ref{eq:g-PLM-logit}) is misspecified and model (\ref{eq:f-PLM-logit}) is correctly specified in (C9).

\vspace{.2in}
\begin{table}[H]
\caption{Summary of results for partially logistic modeling ($n=600, p=800$)} \label{tab:logit-p800}  \vspace{-.2in}
\small
\begin{center}
        \begin{tabular}{cccc|ccc|cccc}
		\toprule
		\multirow{4}{*}{} &
		\multicolumn{3}{c|}{(C7) Cor Cor}&
		\multicolumn{3}{c|}{(C8) Cor Mis}&
		\multicolumn{3}{c}{(C9) Mis Cor}\\
		& $\hat{\theta}_\DB$ & $\hat{\theta}_1$ & $\hat{\theta}_2$ & $\hat{\theta}_\DB$ & $\hat{\theta}_1$ & $\hat{\theta}_2$ & $\hat{\theta}_\DB$ & $\hat{\theta}_1$ & $\hat{\theta}_2$&\\
		\midrule
		\multirow{1} {*}{Bias}
		& 0.059 & 0.065 & 0.049 & 0.045 & 0.068 & 0.046 & 0.244 & 0.0524 & 0.045 &\\ 			
		
		\multirow{1}{*}{$\sqrt{\text{Var}}$}
		& 0.232 & 0.245 & 0.239 & 0.273 & 0.315 & 0.298 & 0.278 & 0.339 & 0.332 &\\

		\multirow{1}{*}{$\sqrt{\text{Evar}}$}
		& 0.226 & 0.241 & 0.238 & 0.296 & 0.287 & 0.289 & 0.371 & 0.322 & 0.326 &\\

		\multirow{1}{*}{Cov95}
		& 0.936 & 0.930 & 0.936 & 0.945 & 0.937 & 0.949 & 0.900 & 0.929 & 0.938 &\\
		\bottomrule
		\end{tabular} \\[.1in]
\parbox{1\textwidth}{\scriptsize Note: See the footnote of Table~\ref{tab:linear-p800}.}
\end{center}
\end{table}

\vspace{-.6in}
\begin{figure}[H]
	\centering
	\subfloat[]{{\includegraphics[width=70mm]{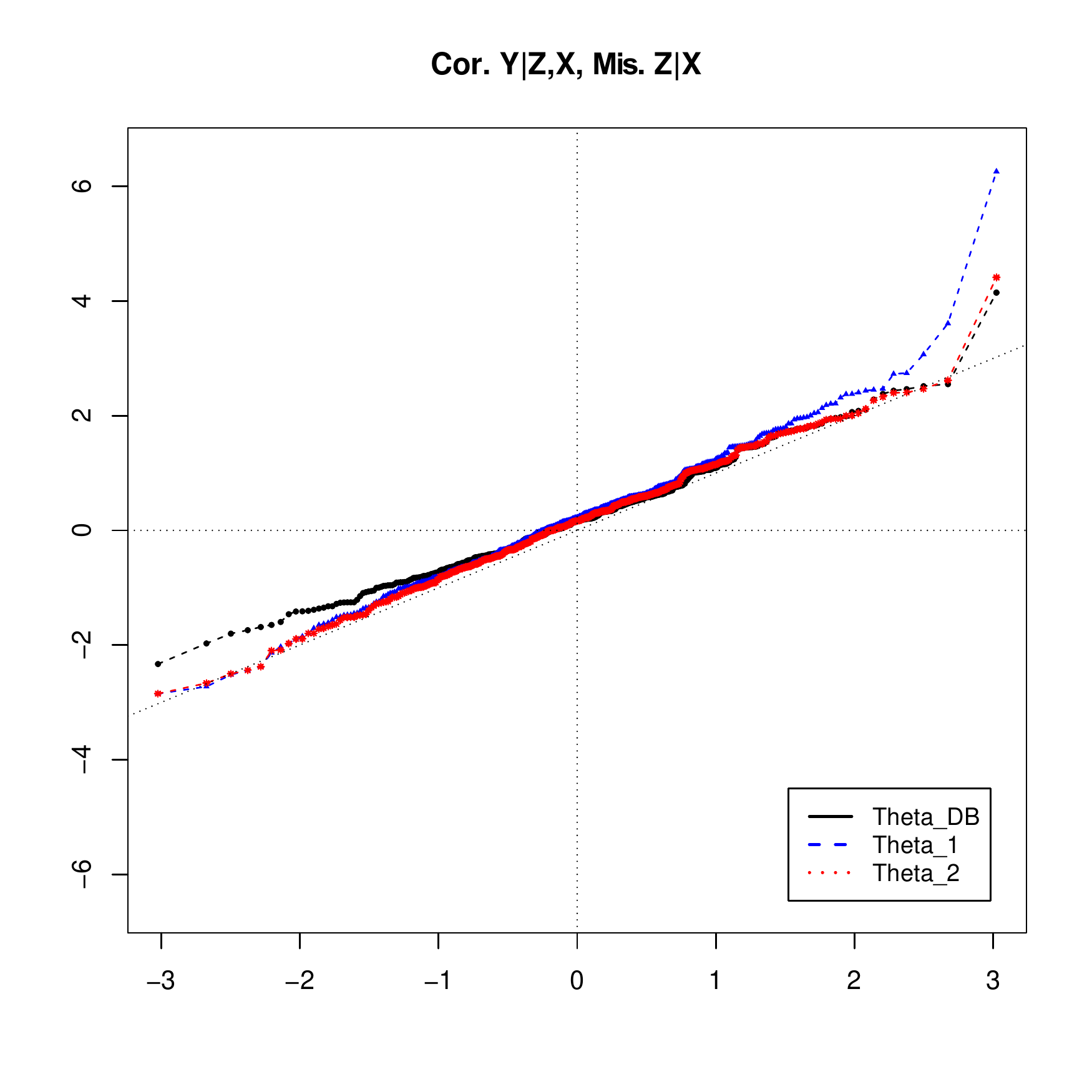} }}%
	\qquad
	\subfloat[]{{\includegraphics[width=70mm]{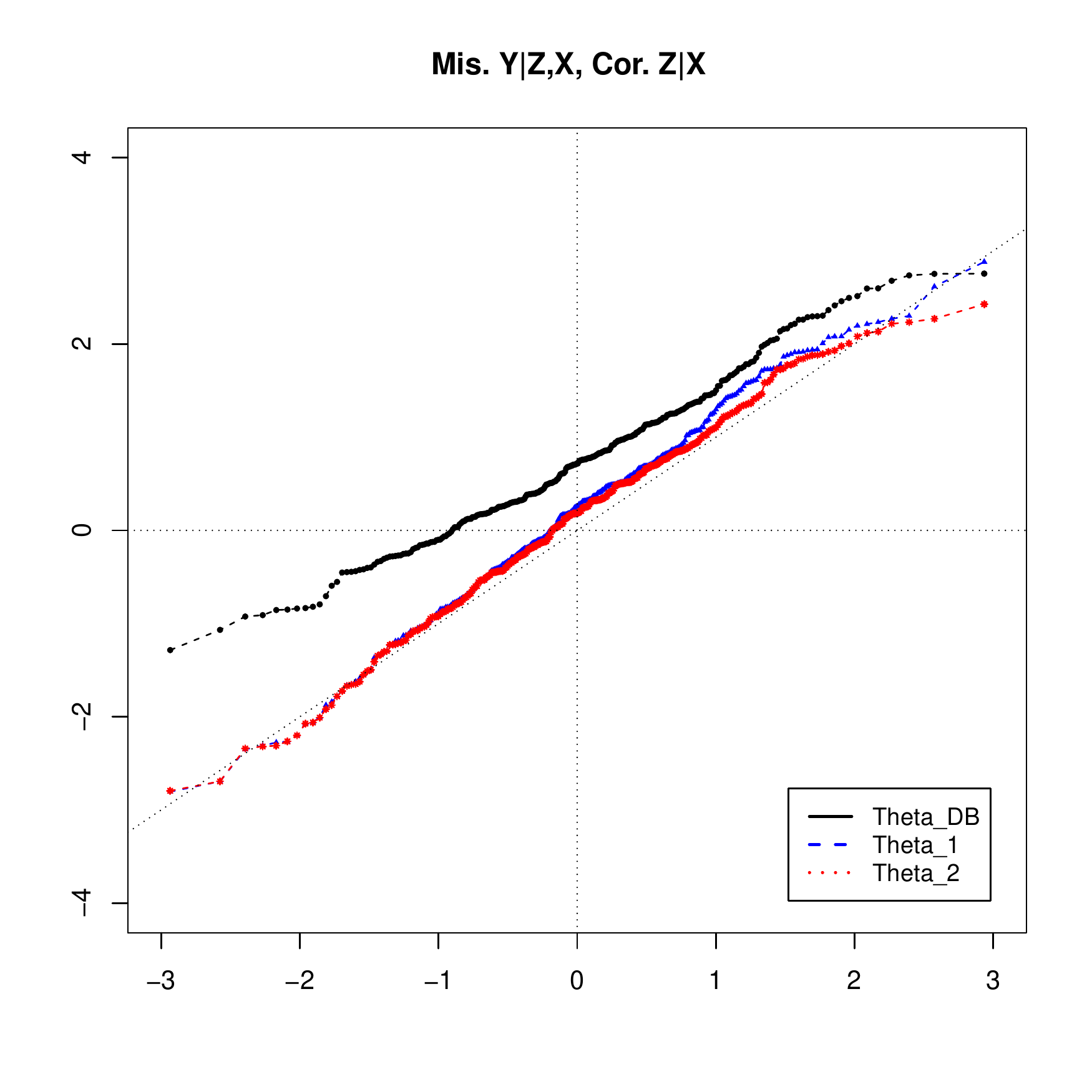} }}%
	\vspace*{-.6in}  \caption{QQ plots of $t$-statistics for partially logistic modeling ($n=600, p=800$)}%
	\label{fig:logit-p800}
\end{figure} \vspace{-.1in}

For $n=600$ and $p=800$, Table~\ref{tab:logit-p800} summarizes the results for estimation of $\theta^*$
and Figure~\ref{fig:logit-p800} shows the QQ plots of $t$-statistics. Additional results with $p=100$ or 200
are included in the Supplement.
While the three methods perform similarly to each other in settings (C7) and (C8),
our calibrated method achieves substantially better performance in setting (C9) than the other methods, similarly as in Section~\ref{sec:linear-sim}.

\vspace{-.1in}
\section{Conclusion}
We develop regularized calibrated estimation as a general method for obtaining doubly robust confidence intervals in high-dimensional settings, provided a doubly robust estimating function is available. While various applications of the method can be pursued,
there are interesting topics which warrant further investigation. As an alternative to the two-step algorithm,
it is of interest to study the generalized Dantzig selector mentioned in Section~\ref{sec:rcal},
including development of practical algorithms and theoretical analysis without a convex loss. This approach has a potential benefit in
producing valid confidence intervals centered about the target value in $\theta$ even if both working models are misspecified, similarly as discussed in Remark 9 of \cite{tan2020model}. Moreover, it is helpful to incorporate sample splitting and cross fitting for
our method and study whether both rate and model double robustness can generally be achieved.
A related question is raised in \cite{smucler2019unifying} about construction of such desired estimators beyond bilinear influence functions.

\bibliographystyle{myapa}
\bibliography{References-0921}

\appendix



\clearpage

\setcounter{page}{1}

\setcounter{section}{0}
\setcounter{equation}{0}

\setcounter{figure}{0}
\setcounter{table}{0}

\renewcommand{\theequation}{S\arabic{equation}}
\renewcommand{\thesection}{\Roman{section}}

\renewcommand\thefigure{S\arabic{figure}}
\renewcommand\thetable{S\arabic{table}}

\setcounter{lem}{0}
\renewcommand{\thelem}{S\arabic{lem}}

\setcounter{algorithm}{0}
\renewcommand\thealgorithm{S\arabic{algorithm}}

\begin{center}
{\Large Supplementary Material for}

{\Large ``Doubly Robust Semiparametric Inference Using Regularized Calibrated Estimation with High-dimensional Data"}

\vspace{.1in} {\large Satyajit Ghosh \& Zhiqiang Tan}

\end{center}

\section{Technical details}

\subsection{Probability lemmas}

Denote by $\Omega_0$ the event that
$ (\hat\alpha_1 - \bar\alpha_1)^\T \tilde \Sigma_0 (\hat\alpha_1 - \bar\alpha_1) \le M_0 \lambda_0^2 $, $\| \hat\alpha_1 - \bar\alpha_1 \|_1 \le M_0\lambda_0$,
$ (\hat\gamma_1 - \bar\gamma_1)^\T \tilde \Sigma_0 (\hat\gamma_1 - \bar\gamma_1) \le M_0 \lambda_0^2 $, $\| \hat\gamma_1 - \bar\gamma_1 \|_1 \le M_0 \lambda_0 $,
and $|\hat\theta_1 - \theta^* |\le M_0^{1/2} \lambda_0 $.
Then Assumption~\ref{ass:gamma2-basic}(iv) says that $P(\Omega_0) \ge 1-\epsilon$.

\begin{lem}  \label{lem:prob-grad}
Denote by $\Omega_1$ the event that
\begin{align*}
\sup_{j=1,\ldots,p} \left| \tilde E \left\{ \xi_j(X) \frac{\partial\tau}{\partial \eta_g} (U; \theta^*, \bar\alpha_1,\bar\gamma_2) \right\} \right| \le B_0 \lambda_0,
\end{align*}
where $B_0 = C_0 (B_{02} + \sqrt{2} B_{01}) $.
Under Assumption~\ref{ass:gamma2-basic}(i)--(ii), if $\lambda_0\le 1$, then $P(\Omega_1 ) \ge 1- 2\epsilon$.
\end{lem}

\begin{prf}
The variable $ \frac{\partial\tau}{\partial \eta_g} (U; \theta^*, \bar\alpha_1,\bar\gamma_2)$ has mean 0 (because $\xi$ includes a constant)
and is sub-exponential with parameters $(B_{01}, B_{02})$.
For $j=1,\ldots,p$, the variable $\xi_j(X) \frac{\partial\tau}{\partial \eta_g} (U; \theta^*, \bar\alpha_1,\bar\gamma_2)$ has mean 0 and is sub-exponential with parameter
$(C_0 B_{01}, C_0 B_{02})$.
By Bernstein's inequality (\citealt{buhlmann2011statistics}, Lemma 14.9; \citealt{tan2020model}, Lemma 16),
\begin{align*}
P \left\{ |\tilde E (Y_j)| \ge C_0 B_{02} t + C_0 B_{01} \sqrt{2t} \right\} \le 2 \frac{\epsilon}{p}
\end{align*}
where $t = \log( p/\epsilon )/n = \lambda_0^2$. The result then follows from the union bound.
\end{prf}

\begin{lem} \label{lem:prob-hess}
Denote by $\Omega_{21}, \Omega_{22}, \Omega_{23}$ respectively the events that
\begin{align*}
& \sup_{j,k=1,\ldots,p} \left| (\tilde E -E) \left\{ \xi_j \xi_k T_{\eta_g^2}^{(1)} (U;\theta^*, \bar\alpha_1, \bar\gamma_2) \right\} \right| \le B_{15} \lambda_0, \\
& \sup_{j,k=1,\ldots,p} \left| (\tilde E -E) \left\{ \xi_j \xi_k T_{\eta_g\theta}^{(1)} (U;\theta^*, \bar\alpha_1, \bar\gamma_2)\right\} \right| \le B_{15} \lambda_0, \\
& \sup_{j,k=1,\ldots,p} \left| (\tilde E -E) \left\{ \xi_j \xi_k \frac{\partial^2 \tau}{\partial \eta_g \partial\eta_f}(U;\theta^*, \bar\alpha_1, \bar\gamma_2) \right\} \right| \le B_{15} \lambda_0,
\end{align*}
where $B_{15} = 2(B_{14}+B_{13})$, $B_{13} = (2 C_0^2 B^2_{11} + 8 C_0^4 C_1^2)^{1/2}$, and $B_{14} = 2B_{12} + 2C_0^2 C_1$.
Under Assumptions~\ref{ass:gamma2-basic}(i) and \ref{ass:gamma2-hess}(i)--(ii), if $\lambda_0\le 1$, then $P(\Omega_{21} ) \ge 1- 2\epsilon^2$, $P(\Omega_{22}) \ge 1-2 \epsilon^2$, and  $P(\Omega_{23}) \ge 1-2 \epsilon^2$.
\end{lem}

\begin{prf}
First, we show that for $j,k=1,\ldots,p$, $\xi_j \xi_k T_{\eta_g^2}^{(1)} (U;\theta^*, \bar\alpha_1, \bar\gamma_2) $ is sub-exponential with parameter
$(B_{13}, B_{14})$.
Denote $T_{\eta_g^2}^{(1)} = T_{\eta_g^2}^{(1)} (U;\theta^*, \bar\alpha_1, \bar\gamma_2)$. Then for $k \ge 2$,
\begin{align*}
& E \left| \xi_j \xi_k T_{\eta_g^2}^{(1)} - E\left\{ \xi_j \xi_k T_{\eta_g^2}^{(1)} \right\} \right|^k \\
& \le 2^{k-1} \left[ E \left| \xi_j \xi_k \left\{ T_{\eta_g^2}^{(1)} - E( T_{\eta_g^2}^{(1)}) \right\} \right|^k  +
E \left| \xi_j \xi_k E( T_{\eta_g^2}^{(1)}) - E(\xi_j \xi_k T_{\eta_g^2}^{(1)})  \right|^k \right] \\
& \le 2^{k-1} \left\{ C_0^2 \frac{k!}{2} B_{11}^2 B_{12}^{k-2} + (2 C_0^2 C_1 )^k \right\}\
\le \frac{k!}{2} (2 C_0^2 B^2_{11} + 8 C_0^4 C_1^2) (2B_{12} + 2C_0^2 C_1)^{k-2} .
\end{align*}
Applying Bernstein's inequality to
$\xi_j \xi_k T_{\eta_g^2}^{(1)} (U;\theta^*, \bar\alpha_1, \bar\gamma_2) $ yields
\begin{align*}
P \left\{ \left| (\tilde E -E) \left\{ \xi_j \xi_k T_{\eta_g^2}^{(1)} (U;\theta^*, \bar\alpha_1, \bar\gamma_2) \right\} \right|
\ge B_{14} t + B_{13} \sqrt{2t} \right\} \le 2 \frac{\epsilon^2}{p^2}
\end{align*}
where  $t = \log( p^2/\epsilon^2 )/n = 2 \lambda_0^2$. Then $P(\Omega_{21}) \ge 1-2 \epsilon^2$ by the union bound.
Similarly,  $P(\Omega_{22}) \ge 1-2 \epsilon^2$ and  $P(\Omega_{23}) \ge 1-2 \epsilon^2$.
\end{prf}

\begin{lem} \label{lem:prob-hess0}
Denote by $\Omega_{24}$ the event that
\begin{align*}
& \sup_{j,k=1,\ldots,p} \left| (\tilde E -E) (\xi_j \xi_k ) \right| \le 4 C_0^2 \lambda_0.
\end{align*}
Under Assumption~\ref{ass:gamma2-basic}(i), $P(\Omega_{24}) \ge 1- 2\epsilon^2$.
\end{lem}

\begin{prf}
For $j,k=1,\ldots,p$, by Hoeffding's inequality (\citealt{buhlmann2011statistics}, Lemma 14.11; \citealt{tan2020model}, Lemma 14),
\begin{align*}
P \left\{ \left| (\tilde E -E) (\xi_j \xi_k  )\right| \ge 2 C_0^2 (\sqrt{2}t) \right\} \le 2 \frac{\epsilon^2}{p^2}
\end{align*}
where $|\xi_j\xi_k - E(\xi_j \xi_k) | \le 2 C_0^2$ and $t= \{\log( p^2/\epsilon^2 )/n\}^{1/2} = \sqrt{2}\lambda_0$.
The result then follows from the union bound.
\end{prf}

\begin{lem} \label{lem:prob-hess2}
Denote $\Omega_2 = \Omega_{21} \cap \Omega_{22} \cap \Omega_{23} \cap \Omega_{24}$.
Under Assumptions~\ref{ass:gamma2-basic}(i) and \ref{ass:gamma2-hess}(i)--(ii), if $\lambda_0\le 1$, then $P(\Omega_2) \ge 1- 8\epsilon^2$.
Moreover, in the event $\Omega_2$, we have
for any vector $b\in\bbR^p$,
\begin{align*}
& \tilde E \left\{ T_{\eta_g^2}^{(1)}(U;\theta^*, \bar\alpha_1, \bar\gamma_2) (b^\T \xi)^2 \right\} \le
 C_1 \tilde E \left\{ (b^\T \xi)^2   \right\} +  (1+C_1) B_1 \lambda_0 \|b\|_1^2 , \\
& \tilde E \left\{  T_{\eta_g\theta}^{(1)} (U;\theta^*, \bar\alpha_1, \bar\gamma_2) (b^\T \xi)^2 \right\} \le
 C_1  \tilde E \left\{ (b^\T \xi)^2  \right\} +  (1+C_1) B_1 \lambda_0 \|b\|_1^2 , \\
& \tilde E \left\{\frac{\partial^2 \tau}{\partial \eta_g \partial\eta_f} (U; \theta^*, \bar \alpha_1,\bar\gamma_2) (b^\T\xi)^2 \right\}\ge
c_2  \tilde E \left\{ (b^\T \xi)^2  \right\} - (1+c_2) B_1 \lambda_0 \|b\|_1^2 ,
\end{align*}
where $B_1 = \max(4 C_0^2, B_{15})$.
\end{lem}

\begin{prf}
Combining Lemmas~\ref{lem:prob-hess}--\ref{lem:prob-hess0} shows that $P(\Omega_2) \ge 1-\epsilon^2$.
In the event $\Omega_2$, simple manipulation yields
\begin{align*}
& \left| (\tilde E-E) \left\{ T_{\eta_g^2}^{(1)} (U;\theta^*, \bar\alpha_1, \bar\gamma_2) (b^\T \xi)^2 \right\} \right| \le B_{15} \lambda_0 \|b\|_1^2 ,\\
& \left| (\tilde E-E) \left\{ (b^\T \xi)^2 \right\} \right| \le 4 C_0^2 \lambda_0 \|b\|_1^2.
\end{align*}
By the law of iterated expectations and Assumption~\ref{ass:gamma2-hess}(i),
\begin{align*}
& E \left\{ T_{\eta_g^2}^{(1)} (U;\theta^*, \bar\alpha_1, \bar\gamma_2) (b^\T \xi)^2 \right\} \\
& =  E \left[ \left\{ T_{\eta_g^2}^{(1)} (U;\theta^*, \bar\alpha_1, \bar\gamma_2) \Big| X\right\} (b^\T \xi)^2 \right] \le C_1 E \left\{ (b^\T \xi)^2 \right\}
\end{align*}
Combining the preceding three inequalities yields the result on  $T_{\eta_g^2}^{(1)}$. Similarly, the results on $T_{\eta_g\theta}^{(1)} $ and
$\frac{\partial^2 \tau}{\partial \eta_g \partial\eta_f}$ can be shown.
\end{prf}

\subsection{Proofs of Theorem \ref{thm:gamma2}, Corollary \ref{cor:gamma2}, and Theorem \ref{thm:alpha2}}

We split the proof of Theorem \ref{thm:gamma2} into a series of lemmas.
The first one is usually called a basic inequality for $\hat\gamma_2$, but depending on the first-step estimators $(\hat\theta_1,\hat\alpha_1)$.

\begin{lem} \label{lem:basic-ineq-hat}
For any vector $\gamma \in \bbR^p$, we have
\begin{align}
D^\dag_2 ( \hat \gamma_2, \gamma; \hat\theta_1 , \hat\alpha_1) + A_1 \lambda_0 \| \hat\gamma_2 \|_1 \le
(\hat\gamma_2 - \gamma)^\T \tilde E \left\{  \xi   \frac{\partial\tau}{\partial \eta_g} (U; \hat\theta_1, \hat\alpha_1,\gamma)\right\} + A_1 \lambda_0 \| \gamma \|_1. \label{eq:basic-ineq-hat}
\end{align}
\end{lem}

\begin{prf}
For any $u \in (0,1]$, the definition of $\hat\gamma_2$ implies
\begin{align*}
& L_2 (\hat \gamma_2 ; \hat\theta_1, \hat\alpha_1) + A_1 \lambda_0  \|\hat\gamma_2 \|_1  \\
& \le L_2 \{(1-u)\hat\gamma_2 + u \gamma ; \hat\theta_1, \hat\alpha_1 \} +  A_1 \lambda_0  \|(1-u)\hat\gamma_2 + u \gamma \|_1 ,
\end{align*}
which, by the convexity of $\|\cdot\|_1$, gives
\begin{align*}
L_2 (\hat\gamma_2 ; \hat\theta_1, \hat\alpha_1) - L_2 \{ (1-u) \hat\gamma_2 + u \gamma ; \hat\theta_1, \hat\alpha_1 \} + A_1 \lambda_0 u \| \hat\gamma_2 \|_1 \le A_1 \lambda_0 u \| \gamma \|_1 .
\end{align*}
Dividing both sides of the preceding inequality by $u$ and letting $u \to 0+$ yields
\begin{align*}
- (\hat\gamma_2 - \gamma)^\T \tilde E\left\{ \xi \frac{\partial\tau}{\partial\eta_f} (U; \hat\theta_1,\hat\alpha_1, \hat\gamma_2) \right\} + A_1 \lambda_0 \|\hat \gamma_2 \|_1
\le A_1 \lambda_0 \| \gamma \|_1 ,
\end{align*}
which leads to (\ref{eq:basic-ineq-hat}) after a simple rearrangement using (\ref{eq:sym-bregman}).
\end{prf}

The second lemma deals with the dependency on $(\hat\theta_1,\hat\alpha_1)$ in the upper bound from the basic inequality (\ref{eq:basic-ineq-hat}).
Denote
\begin{align*}
& Q_2 (\hat\gamma_2,\bar\gamma_2; \theta^*, \bar\alpha_1)
= \tilde E \left\{\frac{\partial^2 \tau}{\partial \eta_g \partial\eta_f} (U; \theta^*, \bar \alpha_1,\bar\gamma_2)  (\hat\gamma_2^\T \xi - \bar\gamma_2^\T \xi )^2 \right\} \\
& = (\hat\gamma_2- \bar\gamma_2)^\T  \tilde \Sigma_\gamma  (\hat\gamma_2- \bar\gamma_2).
\end{align*}
where $\tilde \Sigma_\gamma = \tilde E \{ \frac{\partial^2 \tau}{\partial \eta_g \partial\eta_f}(U; \theta^*, \bar\alpha_1, \bar\gamma_2) \xi \xi^\T  \}$.

\begin{lem} \label{lem:remove-hat}
Suppose that Assumptions~\ref{ass:gamma2-hess}(i) and \ref{ass:gamma2-hess}(iv) hold.
In the event $\Omega_0 \cap \Omega_2$, we have
\begin{align*}
& (\hat\gamma_2 - \bar\gamma_2)^\T \tilde E \left\{ \xi  \frac{\partial\tau}{\partial \eta_g} (U; \hat\theta_1, \hat\alpha_1,\gamma)\right\}  \\
& \le (\hat\gamma_2 - \bar\gamma_2)^\T \tilde E \left\{ \xi \frac{\partial\tau}{\partial \eta_g} (U; \theta^*, \bar \alpha_1,\bar\gamma_2)\right\} \\
& \quad +(C_{12} M_0 \lambda_0^2)^{1/2} \{ Q_2 (\hat\gamma_2,\bar\gamma_2; \theta^*, \bar\alpha_1) \}^{1/2}
+ C_{13} \lambda_0 \|\hat\gamma_2-\bar\gamma_2\|_1  ,
\end{align*}
where $C_{12} = 4 c_2^{-1} C_1 (C_1+C_{11} \varrho_0) $, $C_{13} = \{C_{11}^{1/2} +(1+c_2^{-1})^{1/2} B_1^{1/2} C_1^{1/2} \} \{4(C_1+C_{11} \varrho_0 ) \varrho_0\}^{1/2}$,
and $C_{11} = (1+C_1) B_1$ with $B_1$ from Lemma~\ref{lem:prob-hess2}.
\end{lem}

\begin{prf}
Consider the following decomposition
\begin{align*}
& (\hat\gamma_2 - \bar\gamma_2)^\T \tilde E \left\{ \frac{\partial\tau}{\partial \eta_g} (U; \hat\theta_1, \hat\alpha_1,\bar\gamma_2) \xi  \right\} \\
& =(\hat\gamma_2 - \bar\gamma_2)^\T \tilde E \left\{ \xi \frac{\partial\tau}{\partial \eta_g} (U; \theta^*, \bar\alpha_1,\bar\gamma_2)  \right\} + \Delta_1 + \Delta_2,
\end{align*}
where
\begin{align*}
& \Delta_1 = (\hat\gamma_2 - \bar\gamma_2)^\T \tilde E \left[ \xi \left\{ \frac{\partial\tau}{\partial \eta_g} (U; \hat\theta_1, \hat\alpha_1,\bar\gamma_2)
-\frac{\partial\tau}{\partial \eta_g} (U; \theta^*, \hat\alpha_1,\bar\gamma_2)\right\} \right], \\
& \Delta_2 = (\hat\gamma_2 - \bar\gamma_2)^\T \tilde E \left[ \xi \left\{ \frac{\partial\tau}{\partial \eta_g} (U; \theta^*, \hat\alpha_1,\bar\gamma_2)
-\frac{\partial\tau}{\partial \eta_g} (U; \theta^*, \bar\alpha_1,\bar\gamma_2)\right\} \right].
\end{align*}
In the event $\Omega_0$, $(\hat\theta_1,\hat\alpha_1) \in \mathcal N_1$ by Assumption \ref{ass:gamma2-hess}(iv).
By the mean value theorem and the Cauchy--Schwartz inequality, and Assumption \ref{ass:gamma2-hess}(i),
\begin{align*}
& |\Delta_2| =  \left| \tilde E \left\{ (\hat\gamma_2 - \bar\gamma_2)^\T  \xi \frac{\partial^2 \tau}{\partial \eta_g^2} (U; \theta^*, \tilde\alpha ,\bar\gamma_2) (\hat\alpha_1 - \bar\alpha_1)^\T \xi \right\} \right|\\
& \le \tilde E^{1/2} \left\{ T^{(1)}_{\eta_g^2}(U;\theta^*, \bar\alpha_1, \bar\gamma_2) (\hat\gamma_2^\T \xi - \bar\gamma_2^\T \xi )^2 \right\}
\tilde E^{1/2} \left\{ T^{(1)}_{\eta_g^2}(U;\theta^*, \bar\alpha_1, \bar\gamma_2)(\hat\alpha_1^\T \xi - \bar\alpha_1^\T \xi)^2 \right\} ,
\end{align*}
where $\tilde\alpha$ lies between $\hat\alpha_1$ and $\bar\alpha_1$. Hence in the event $\Omega_0\cap\Omega_2$ by Lemma~\ref{lem:prob-hess2},
\begin{align*}
& |\Delta_2|  \le \left[ C_1 \tilde E \left\{ (\hat\gamma_2^\T \xi - \bar\gamma_2^\T \xi )^2 \right\} + C_{11} \lambda_0 \|\hat\gamma_2-\bar\gamma_2\|_1^2  \right]^{1/2}
\left[ C_1 \tilde E \left\{ (\hat\alpha_1^\T \xi - \bar\alpha_1^\T \xi)^2 \right\}   + C_{11} \lambda_0 \|\hat\alpha_1-\bar\alpha_1\|_1^2 \right]^{1/2},
\end{align*}
where $C_{11} = (1+C_1)B_1$.
Similarly, in the event $\Omega_0\cap\Omega_2$ by Lemma~\ref{lem:prob-hess2},
\begin{align*}
& |\Delta_1| = \left|\tilde E \left\{ (\hat\gamma_2 - \bar\gamma_2)^\T  \xi \frac{\partial^2 \tau}{\partial \eta_g \partial\theta}(U; \tilde\theta, \hat\alpha_1,\bar\gamma_2) (\hat\theta_1 - \theta^*) \right\}  \right|\\
& \le \tilde E^{1/2} \left\{ T^{(1)}_{\eta_g\theta} (U;\theta^*, \bar\alpha_1, \bar\gamma_2)( \hat\gamma_2^\T \xi - \bar\gamma_2^\T \xi )^2 \right\}
 \tilde E^{1/2} \left\{T^{(1)}_{\eta_g\theta} (U;\theta^*, \bar\alpha_1, \bar\gamma_2)(\hat\theta_1 - \theta^*)^2 \right\}  \\
& \le \left[ C_1 \tilde E \left\{ (\hat\gamma_2^\T \xi - \bar\gamma_2^\T \xi )^2 \right\} + C_{11} \lambda_0 \|\hat\gamma_2-\bar\gamma_2\|_1^2  \right]^{1/2}
\left[ (C_1+ B_{15} \lambda_0) (\hat\theta_1-\theta^*)^2 \right]^{1/2},
\end{align*}
where $\tilde\theta$ lies between $\hat\theta_1$ and $\theta^*$. Hence in the event $\Omega_0 \cap\Omega_2$,
\begin{align*}
& |\Delta_1| + |\Delta_2|
\le 2 \left\{ (C_1+C_{11}M_0\lambda_0) M_0\lambda_0^2 \right\}^{1/2}
\left[ C_1^{1/2} \tilde E^{1/2} \left\{ (\hat\gamma_2^\T \xi - \bar\gamma_2^\T \xi )^2 \right\} + C_{11}^{1/2} \lambda_0^{1/2} \|\hat\gamma_2-\bar\gamma_2\|_1 \right] \\
& = (M_{01} \lambda_0^2)^{1/2} \tilde E^{1/2} \left\{ (\hat\gamma_2^\T \xi - \bar\gamma_2^\T \xi )^2 \right\} + M_{02} \lambda_0 \|\hat\gamma_2-\bar\gamma_2\|_1  ,
\end{align*}
where $M_{01} = 4C_1 (C_1+C_{11}M_0\lambda_0) M_0$ and $M_{02} = \{4C_{11}(C_1+C_{11}M_0\lambda_0) M_0 \lambda_0\}^{1/2}$.
Furthermore, in the event $\Omega_2$ by Lemma~\ref{lem:prob-hess2},
\begin{align*}
& \tilde E \left\{ (\hat\gamma_2^\T \xi - \bar\gamma_2^\T \xi )^2 \right\} \\
& \le c_2^{-1} \tilde E \left\{\frac{\partial^2 \tau}{\partial \eta_g \partial\eta_f} (U; \theta^*, \bar \alpha_1,\bar\gamma_2)  (\hat\gamma_2^\T \xi - \bar\gamma_2^\T \xi )^2 \right\}
+ (1+c_2^{-1}) B_1 \lambda_0 \|\hat\gamma_2-\bar\gamma_2\|_1^2 .
\end{align*}
Combining the preceding inequalities shows that in event $\Omega_0 \cap \Omega_2$,
\begin{align*}
& |\Delta_1| + |\Delta_2| \le
(M_{03} \lambda_0^2)^{1/2} \tilde E^{1/2} \left\{\frac{\partial^2 \tau}{\partial \eta_g \partial\eta_f} (U; \theta^*, \bar \alpha_1,\bar\gamma_2)  (\hat\gamma_2^\T \xi - \bar\gamma_2^\T \xi )^2 \right\}
+ M_{04} \lambda_0 \|\hat\gamma_2-\bar\gamma_2\|_1  ,
\end{align*}
where $M_{03} = c_2^{-1} M_{01}$ and $M_{04} = M_{02} + \{(1+c_2^{-1}) B_1 \tilde C_{12} M_0\lambda_0\}^{1/2} = \{C_{11}^{1/2} +(1+c_2^{-1})^{1/2} B_1^{1/2} C_1^{1/2} \} $ $\times\{4(C_1+C_{11}M_0\lambda_0)M_0 \lambda_0\}^{1/2}$. Using $M_0 \lambda_0 \le \varrho_0$ by Assumption~\ref{ass:gamma2-hess}(iv) yields the desired result.
\end{prf}

The third lemma derives an implication of the basic inequality (\ref{eq:basic-ineq-hat}) using the triangle inequality for the $L_1$ norm,
while incorporating the bound from Lemma~\ref{lem:remove-hat}.

\begin{lem} \label{lem:basic-ineq-bar}
Denote $b = \hat \gamma_2 - \bar \gamma_2$. In the event $\Omega_0 \cap \Omega_1 \cap \Omega_2$, (\ref{eq:basic-ineq-hat}) implies that
\begin{align}
& D^\dag_2 ( \hat \gamma_2, \bar \gamma_2; \hat\theta_1 , \hat\alpha_1) + A_{11} \lambda_0 \| b \|_1 \nonumber \\
& \le 2 A_1 \lambda_0 \sum_{j\in S_{\bar\gamma_2}} |b_j|  +(C_{12} M_0 \lambda_0^2)^{1/2} \{ Q_2 (\hat\gamma_2,\bar\gamma_2; \theta^*, \bar\alpha_1) \}^{1/2} , \label{eq:basic-ineq-bar}
\end{align}
where $A_{11} = A_1 -B_0 -C_{13}$, with $B_0$ from Lemma~\ref{lem:prob-grad}.
\end{lem}

\begin{prf}
In the event $\Omega_1$ from Lemma~\ref{lem:prob-grad}, we have
\begin{align*}
b^\T \tilde E \left\{ \xi \frac{\partial\tau}{\partial \eta_g} (U; \theta^*, \bar\alpha_1,\bar\gamma_2) \right\} \le B_0 \lambda_0 \|b\|_1 .
\end{align*}
From (\ref{eq:basic-ineq-hat}), the preceding bound, and Lemma~\ref{lem:remove-hat}, we have in the event $\Omega_0\cap\Omega_1\cap\Omega_2$,
\begin{align*}
& D^\dag_2 ( \hat \gamma_2, \bar \gamma_2; \hat\theta_1 , \hat\alpha_1) + A_1 \lambda_0 \| \hat\gamma_2 \|_1 \\
& \le B_0 \lambda_0 \|b\|_1 +  A_1 \lambda_0 \| \bar\gamma_2 \|_1 +(C_{12} M_0 \lambda_0^2)^{1/2} \{ Q_2 (\hat\gamma_2,\bar\gamma_2; \theta^*, \bar\alpha_1) \}^{1/2}
+ C_{13} \lambda_0 \| b \|_1  .
\end{align*}
Using  the identity $ | (\hat\gamma_2) _j| = | (\hat\gamma_2-\bar\gamma_2)_j|$ for $j \not\in S_{\bar\gamma_2}$
and the triangle inequality $ | (\hat\gamma_2) _j| \ge |(\bar\gamma_2) _j| - | (\hat\gamma_2-\bar\gamma_2)_j|$ for $j \in S_{\bar\gamma_2}$
and rearranging the result yields (\ref{eq:basic-ineq-bar}).
\end{prf}

The following lemma provides a desired bound relating the Bregman divergence $D_2^\dag (\gamma, \bar\gamma_2; $ $\hat\theta_1, \hat\alpha_1)$
with the quadratic function $(\gamma - \bar\gamma_2)^\T \tilde \Sigma_\gamma (\gamma - \bar\gamma_2)$.

\begin{lem}  \label{lem:local-quad}
Suppose that Assumptions~\ref{ass:gamma2-basic}(i) and \ref{ass:gamma2-hess}(iii)--(iv) hold. In the event $\Omega_0$, we have for any $\gamma\in\bbR^p$,
\begin{align*}
D_2^\dag (\gamma, \bar\gamma_2; \hat\theta_1, \hat\alpha_1)  \ge \me^{-\Delta} \frac{1- \me^{-C_0C_2 \|\gamma-\bar\gamma_2\|_1}}{C_0C_2\|\gamma-\bar\gamma_2\|_1}
(b^\T \tilde \Sigma_\gamma b),
\end{align*}
where $ b =\gamma - \bar\gamma_2$ and $\Delta = C_2(|\hat\theta_1 -\theta^*|+ C_0 \|\hat\alpha_1 -\bar\alpha_1\|_1)$. Throughout, set $(1-\me^{-c})/c=1$ for $c=0$.
\end{lem}

\begin{prf}
By direct calculation, we have
\begin{align*}
& D_2^\dag (\gamma, \bar\gamma_2; \hat\theta_1, \hat\alpha_1) =  (\gamma-\bar\gamma_2)^\T \tilde E \left[ \xi \left\{\frac{\partial \tau}{\partial \eta_g}(U;\hat\theta_1, \hat\alpha_1,\gamma) -
\frac{\partial \tau}{\partial \eta_g}(U;  \hat\theta_1, \hat\alpha_1, \gamma)\right\} \right] \\
& =  \tilde E \left[ \left\{ \int_0^1 \frac{\partial^2 \tau}{\partial \eta_g \partial\eta_f}(U;  \hat\theta_1, \hat\alpha_1,\gamma_u) \,\dif u\right\} (\gamma^\T \xi - \bar\gamma_2^\T \xi)^2 \right],
\end{align*}
where $\gamma_u = \bar\gamma_2 + u (\gamma-\bar\gamma_2)$.
In the event $\Omega_0$, $(\hat\theta_1,\hat\alpha_1) \in \mathcal N_1$ by Assumption \ref{ass:gamma2-hess}(iv).
Then by Assumption~\ref{ass:gamma2-basic}(i) and \ref{ass:gamma2-hess}(iii), we have
\begin{align*}
& D_2^\dag (\gamma, \bar\gamma_2; \hat\theta_1, \hat\alpha_1) \\
& \ge \tilde E \left[\left\{ \int_0^1  \me^{-C_2 (|\hat\theta_1 -\theta^*|+|(\hat\alpha_1 -\bar\alpha_1)^\T\xi| + u|(\gamma-\bar\gamma_2)^\T\xi|) } \,\dif u\right\}
\frac{\partial^2 \tau}{\partial \eta_g \partial\eta_f}(U; \theta^*, \bar\alpha_1,\gamma)
(\gamma^\T \xi - \bar\gamma_2^\T \xi)^2 \right] \\
& \ge \left\{ \int_0^1 \me^{-C_2 (|\hat\theta_1 -\theta^*|+ C_0 \|\hat\alpha_1 -\bar\alpha_1\|_1 + u C_0 \|\gamma-\bar\gamma_2\|_1)} \,\dif u\right\}
\tilde E \left\{ \frac{\partial^2 \tau}{\partial \eta_g \partial\eta_f}(U; \theta^*, \bar\alpha_1,\gamma) (\gamma^\T \xi - \bar\gamma_2^\T \xi)^2 \right\}.
\end{align*}
The desired result follows because  $\int_0^1 \me^{-cu}\,\dif u = (1-\me^{-c})/c$ for $c \ge 0$.
\end{prf}

The following lemma shows that Assumption~\ref{ass:gamma2-basic}(iii), a theoretical compatibility condition for $\Sigma_\gamma$, implies
an empirical compatibility condition for $\tilde\Sigma_\gamma$.

\begin{lem} \label{lem:emp-compat}
Suppose that Assumption~\ref{ass:gamma2-hess}(iv) holds. In the event $\Omega_2$, Assumption \ref{ass:gamma2-basic} (iii) implies that
for any vector $b \in \bbR^p$ such that $\sum_{j\not\in S_{\bar\gamma_2}} |b_j|  \le \mu_1 \sum_{j \in S_{\bar\gamma_2}} |b_j| $, we have
\begin{align*}
 \nu_{11}^2 \left( \sum_{j \in S_{\bar\gamma_2}} |b_j| \right)^2 \le |S_{\bar\gamma_2}| \left( b^\T \tilde \Sigma_\gamma b \right),
\end{align*}
where $\nu_{11} = \nu_1 \{1- \nu_1^{-2} (1+\mu_1)^2 \varrho_1 B_{15} \}^{1/2} = \nu_1 (1-\varrho_2)^{1/2}$.
\end{lem}

\begin{prf}
In the event $\Omega_2$, we have $ |b^\T (\tilde \Sigma_\gamma - \Sigma_\gamma) b | \le B_{15} \lambda_0 \|b\|_1^2$ from Lemma~\ref{lem:prob-hess}.
Then Assumption~\ref{ass:gamma2-basic}(iii) implies that for any $b=(b_1,\ldots,b_p)^\T $ satisfying $\sum_{j\not \in S_{\bar\gamma_2}} |b_j| \le \mu_1 \sum_{j\in S_{\bar\gamma_2}} |b_j|$,
\begin{align*}
&\nu_1^2 \|b_{S_{\bar\gamma_2}} \|_1^2 \le |S_{\bar\gamma_2}| (b^\T \Sigma_\gamma b) \le |S_{\bar\gamma_2}| \left(b^\T \tilde \Sigma_\gamma b + B_{15}\lambda_0 \|b\|_1^2 \right) \\
& \le |S_{\bar\gamma_2}| (b^\T \tilde \Sigma_\gamma b ) + B_{15} |S_{\bar\gamma_2}| \lambda_0 (1+\mu_1)^2 \|b_{S_{\bar\gamma_2}} \|_1^2 ,
\end{align*}
where $\|b_{S_{\bar\gamma_2}} \|_1=\sum_{j\in S_{\bar\gamma_2}} |b_j|$. The last inequality uses $\|b\|_1 \le (1+\mu_1) \| b_{S_{\bar\gamma_2}}\|_1$.
The desired result follows because $ |S_{\bar\gamma_2}| \lambda_0 \le \varrho_1$ and $\varrho_2=\nu_1^{-2} B_{15} (1+\mu_1)^2 \varrho_1 < 1$ by Assumption~\ref{ass:gamma2-hess}(iv).
\end{prf}

The final lemma completes the proof of Theorem~\ref{thm:gamma2}, because $P(\Omega_0\cap\Omega_1\cap\Omega_2) \ge 1-(c_0+10)\epsilon$ by
Assumption~\ref{ass:gamma2-basic}(iv) and Lemmas~\ref{lem:prob-grad} and \ref{lem:prob-hess2},

\begin{lem} \label{lem:invert}
Suppose that Assumptions~\ref{ass:gamma2-basic}--\ref{ass:gamma2-hess} hold and $\lambda_0 \le 1$. Then for $A_1 > (B_0+C_{13}) (\mu_1+1)/(\mu_1 -1)$,
inequality (\ref{eq:gamma2-expan}) holds in the event $\Omega_0 \cap \Omega_1 \cap \Omega_2$:
\begin{align*}
& D^\dag_2 ( \hat \gamma_2, \bar \gamma_2; \hat\theta_1 , \hat\alpha_1) + A_{11} \lambda_0 \| \hat\gamma_2 - \bar\gamma_2  \|_1 \\
& \le \left\{ \me^{\varrho_5 }(1-\varrho_3)^{-1} \mu_{12}^2 \nu_{11}^{-2} ( |S_{\bar\gamma_2}| \lambda_0^2) \right\} \vee
 \left\{ \me^{\varrho_5 }(1-\varrho_4)^{-1} \mu_{11}^{-2} C_{12} (M_0 \lambda_0^2 ) \right\} .
\end{align*}
\end{lem}

\begin{prf}
Denote  $b = \hat \gamma_2 - \bar \gamma_2$, $D^\dag_2 = D^\dag_2 ( \hat \gamma_2, \bar \gamma_2; \hat\theta_1 , \hat\alpha_1)$,
$D^\ddag_2 =  D^\dag_2( \hat \gamma_2, \bar \gamma_2; \hat\theta_1 , \hat\alpha_1)+ A_{11} \lambda_0 \| b \|_1 $,
$Q_2 = Q_2( \hat \gamma_2, \bar \gamma_2; \hat\theta_1 , \hat\alpha_1) =b^\T \tilde \Sigma_\gamma b$.
In the event $\Omega_0\cap\Omega_1\cap\Omega_2$, inequality (\ref{eq:basic-ineq-bar}) from Lemma~\ref{lem:basic-ineq-bar} leads to two possible cases:
\begin{align}
\mu_{11} D^\ddag_2 \le (C_{12} M_0 \lambda_0^2)^{1/2} Q_2^{1/2}, \label{eq:invert-case1}
\end{align}
or $ (1-\mu_{11}) D^\ddag_2 \le  2 A_1 \lambda_0 \sum_{j \in S_{\bar\gamma_2}} |b_j|$, that is,
\begin{align}
D^\ddag_2 \le (\mu_1 +1) A_{11} \lambda_0 \sum_{j \in S_{\bar\gamma_2}} |b_j| = \mu_{12} \lambda_0 \sum_{j \in S_{\bar\gamma_2}} |b_j|, \label{eq:invert-case2}
\end{align}
where $\mu_{11} = 1- 2 A_1 / \{(\mu_1+1) A_{11}\} \in (0,1]$ because $A_1 > (B_0 +C_{13}) (\mu_1+1) / (\mu_1-1)$,
and $\mu_{12} = (\mu_1 +1) A_{11}$. We deal with the two cases separately as follows.

In the case where (\ref{eq:invert-case2}) holds,  $\sum_{j \not\in S_{\bar\gamma_2}} |b_j| \le \mu_1 \sum_{j \in S_{\bar\gamma_2}} |b_j|$. Then
by Lemma \ref{lem:emp-compat},
\begin{align}
\sum_{j \in S_{\bar\gamma_2}} |b_j| \le \nu_{11}^{-1} |S_{\bar\gamma_2}|^{1/2}  \left(b^\T \tilde \Sigma_\gamma b \right)^{1/2}. \label{eq:emp-compat}
\end{align}
By  Lemma \ref{lem:local-quad}, we have
\begin{align}
& D^\dag_2 \ge  \me^{-\Delta} \frac{1- \me^{-C_0C_2 \|b\|_1}}{C_0C_2\|b\|_1} \left(b^\T \tilde \Sigma_\gamma b \right)  , \label{eq:local-quad}
\end{align}
where $\Delta = C_2(|\hat\theta_1 -\theta^*|+ C_0 \|\hat\alpha_1 -\bar\alpha_1\|_1)$.
Combining (\ref{eq:invert-case1}), (\ref{eq:emp-compat}), and (\ref{eq:local-quad}) yields
\begin{align}
D^\ddag_2 \le  \mu_{12}^2 \nu_{11}^{-2} |S_{\bar\gamma_2}| \lambda_0^2 \me^{\Delta} \frac{C_0C_2\|b\|_1}{1- \me^{-C_0C_2 \|b\|_1}} . \label{eq:case2-prf}
\end{align}
But $A_{11}\lambda_0 \|b\|_1 \le D^\ddag_2$. Then (\ref{eq:case2-prf}) along with $|\Delta| \le C_2(1+C_0)\varrho_0 $ implies that
$1 - \me^{-C_0C_2 \|b\|_1} \le C_0 C_2 A_{11}^{-1} \mu_{12}^2 \nu_{11}^{-2} |S_{\bar\gamma_2}| \lambda_0 \me^{C_2(1+C_0)\varrho_0 }  \le \varrho_3 \, (<1)$ by Assumption~\ref{ass:gamma2-hess}(iv).
As a result, $C_0C_2 \|b \|_1 \le -\log ( 1-\varrho_3)$ and hence
\begin{align}
 \frac{1- \me^{-C_0C_2 \|b\|_1}}{C_0C_2\|b\|_1} = \int_0^1 \me^{-C_0C_2 \|b\|_1 u}\,\dif u \ge \me^{-C_0 C_2 \|b\|_1} \ge 1- \varrho_3. \label{eq:case2-prf2}
\end{align}
From this bound, (\ref{eq:case2-prf}) leads to $D_2^\ddag \le  \me^{C_2(1+C_0)\varrho_0 }(1-\varrho_3)^{-1} \mu_{12}^2 \nu_{11}^{-2} |S_{\bar\gamma_2}| \lambda_0^2 $.

In the first case where (\ref{eq:invert-case1}) holds, simple manipulation using (\ref{eq:local-quad}) yields
\begin{align}
D^\ddag_2 \le  \mu_{11}^{-2} (C_{12} M_0\lambda_0^2)  \me^{\Delta} \frac{ C_0C_2\|b\|_1} {1- \me^{ - C_0C_2 \|b\|_1}} . \label{eq:case1-prf}
\end{align}
Similarly as above, using $A_{11}\lambda_0 \|b\|_1 \le D^\ddag_2$ and (\ref{eq:case1-prf}) along with $|\Delta| \le C_2(1+C_0)\varrho_0 $, we find
$1 - \me^{-C_0C_2 \|b\|_1} \le C_0 C_2 A_{11}^{-1} \mu_{11}^{-2} C_{12} M_0 \lambda_0 \me^{C_2(1+C_0)\varrho_0 }  \le \varrho_4\,(<1)$ by Assumption~\ref{ass:gamma2-hess}(iv).
As a result, $C_0C_2 \|b \|_1 \le -\log ( 1-\varrho_4)$ and hence
\begin{align}
 \frac{1- \me^{-C_0C_2 \|b\|_1}}{C_0C_2\|b\|_1} = \int_0^1 \me^{-C_0C_2 \|b\|_1 u}\,\dif u \ge \me^{-C_0 C_2 \|b\|_1} \ge 1- \varrho_4,  \label{eq:case1-prf2}
\end{align}
From this bound, (\ref{eq:case1-prf}) leads to $D_2^\ddag \le  \me^{C_2(1+C_0)\varrho_0 }(1-\varrho_4)^{-1} \mu_{11}^{-2} C_{12} M_0 \lambda_0^2 $.
Therefore, (\ref{eq:gamma2-expan}) holds through (\ref{eq:invert-case1}) and (\ref{eq:invert-case2}) in the event $\Omega_0\cap\Omega_1\cap\Omega_2$.
\end{prf}

\vspace{.1in}
{\noindent \textbf{Proof of Corollary~\ref{cor:gamma2}.}
Return to the proof of Lemma~\ref{lem:invert}, where (\ref{eq:invert-case1}) or (\ref{eq:invert-case2}) holds in the event $\Omega_0\cap\Omega_1\cap\Omega_2$.
If (\ref{eq:invert-case2}) holds, then we have, by (\ref{eq:local-quad}) and (\ref{eq:case2-prf2}),
$ b^\T \tilde \Sigma_\gamma b \le \me^{ \varrho_5 }(1-\varrho_3)^{-1} D_2^\dag$.
If (\ref{eq:invert-case1}) holds, then we have, by (\ref{eq:local-quad}) and (\ref{eq:case1-prf2}),
$ b^\T \tilde \Sigma_\gamma b \le \me^{ \varrho_5 }(1-\varrho_4)^{-1} D_2^\dag$.
Hence
\begin{align*}
b^\T \tilde \Sigma_\gamma b \le \me^{ \varrho_5 }(1-\varrho_3 \vee \varrho_4)^{-1} D_2^\dag \le  \me^{ \varrho_5 }(1-\varrho_3 \vee \varrho_4)^{-1} C_3 ( |S_{\bar\gamma_2}| \vee M_0) \lambda_0^2 .
\end{align*}
Moreover, by (\ref{eq:gamma2-expan}), we have $\|b\|_1 \le A_{11}^{-1} C_3 ( |S_{\bar\gamma_2}| \vee M_0) \lambda_0$
and hence $ \lambda_0 \|b\|_1^2 \le A_{11}^{-2} C_3^2 (\varrho_0 \vee \varrho_1) ( |S_{\bar\gamma_2}| \vee M_0)\lambda_0^2 $, because
$ ( |S_{\bar\gamma_2}| \vee M_0) \lambda_0 \le  \varrho_0 \vee \varrho_1$ by Assumption \ref{ass:gamma2-hess}(iv).
Then (\ref{eq:gamma2-expan2}) follows from the third inequality in Lemma~\ref{lem:prob-hess2}.
\hfill $\Box$

\vspace{.1in}
{\noindent \textbf{Proof of Theorem~\ref{thm:alpha2}.}
The proof follows from similar steps as in that of Theorem~\ref{thm:gamma2}. The probability decreases from $1-(c_0+10)\epsilon$ to $1-(c_0+18) \epsilon$,
due to additional restriction to the events similar to $\Omega_1$, $\Omega_{21}$, $\Omega_{22}$, and $\Omega_{23}$, while $\Omega_{24}$ is unchanged.
\hfill $\Box$

\subsection{Proof of Theorem~\ref{thm:theta2}}

We split the proof into three lemmas. The first one shows the consistency of $\hat\theta_2$ for $\theta^*$.

\begin{lem} \label{lem:theta2-consistency}
In the setting of Theorem~\ref{thm:alpha2}, suppose that Assumption \ref{ass:theta2-consistency} holds and  $M_2 r_0  =o(1)$.
Then $\hat\theta_2$ is consistent for $\theta^*$, i.e., $|\hat \theta_2 - \theta^*| = o_p(1)$.
\end{lem}

\begin{prf}
By Theorem~\ref{thm:alpha2} and $M_2 r_0 =o(1)$,  we have
$\|\hat\alpha_2-\bar\alpha_2\|_1 = o_p(1)$ and $\|\hat\gamma_2 - \bar\gamma_2\|_1 = o_p(1)$.
Hence for any small $\epsilon>0$, $(\hat\alpha_2,\hat\gamma_2)\in\mathcal N_2$ with probability at least $1-\epsilon$ for all sufficiently large $n$.
In the following, we restrict analysis within this event.

To show $|\hat \theta_2 - \theta^*| = o_p(1)$, by standard consistency arguments (e.g., \citealt{van2000asymptotic}) using Assumption~\ref{ass:theta2-consistency}(i)--(ii),
it suffices to show that
$\tilde E  \{ \tau(U; \hat\theta_2, \bar\alpha_2, \bar\gamma_2 )\}= o_p(1)$. Because
$\tilde E  \{ \tau(U; \hat\theta_2, \hat\alpha_2, \hat\gamma_2 ) \} =0$ by definition of $\hat\theta_2$, consider the decomposition
\begin{align*}
&  \tilde E \left\{ \tau(U; \hat\theta_2, \hat\alpha_2, \hat\gamma_2 ) \right\}  - \tilde E \left\{ \tau(U; \hat\theta_2, \bar\alpha_2, \bar\gamma_2 ) \right\} \\
& = - \tilde E \left\{ \tau(U; \hat\theta_2, \bar\alpha_2, \bar\gamma_2 ) \right\}  = \Delta_1 + \Delta_2 ,
\end{align*}
where
\begin{align*}
& \Delta_1 = \tilde E \left\{ \tau(U; \hat\theta_2, \hat\alpha_2, \hat\gamma_2 ) \right\} - \tilde E \left\{ \tau(U; \hat\theta_2, \bar\alpha_2, \hat\gamma_2 ) \right\} , \\
& \Delta_2 = \tilde E \left\{ \tau(U; \hat\theta_2, \bar\alpha_2, \hat\gamma_2 ) \right\} - \tilde E \left\{ \tau(U; \hat\theta_2, \bar\alpha_2, \bar\gamma_2 ) \right\} .
\end{align*}
By the mean value theorem, the Cauchy--Schwartz inequality and Assumption \ref{ass:theta2-consistency} (iii),
\begin{align*}
& |\Delta_1 |= \left| (\hat\alpha_2-\bar\alpha_2)^\T  \tilde E \left\{ \xi \frac{\partial\tau}{\partial \alpha} (U; \hat\theta_2, \tilde\alpha, \hat\gamma_2 ) \right\} \right| \\
& \le \tilde E^{1/2} \left\{ (\hat\alpha_2^\T \xi -\bar\alpha_2^\T \xi)^2  \right\} \tilde E^{1/2} \{  T^{(2)2}_{\eta_g} (U; \bar\alpha_2,\bar\gamma_2) \} = O_p( M_2^{1/2} r_0), \\
& |\Delta_2 |= \left| (\hat\gamma_2-\bar\gamma_2)^\T  \tilde E \left\{ \xi \frac{\partial\tau}{\partial \gamma} (U; \hat\theta_2, \bar\alpha_2, \tilde\gamma ) \right\} \right| \\
& \le \tilde E^{1/2} \left\{ (\hat\gamma_2^\T \xi -\bar\gamma_2^\T \xi)^2  \right\} \tilde E^{1/2} \{  T^{2(2)}_{\eta_f} (U; \bar\alpha_2,\bar\gamma_2) \}= O_p( M_2^{1/2} r_0),
\end{align*}
where $\tilde\alpha$ lies between $\hat\alpha_2$ and $\bar\alpha_2$, and $\tilde\gamma$ lies between $\hat\gamma_2$ and $\bar\gamma_2$.
Hence $|\Delta_1| + |\Delta_2| = o_p(1)$ because  $M_2^{1/2} r_0 \le M_2 r_0 =o(1)$ with $M_2 \ge M_0 \ge 1$.
\end{prf}

The following lemma establishes the asymptotic expansion (\ref{eq:theta2-expan}) for $\hat\theta_2$.

\begin{lem} \label{lem:theta2-expan}
In the setting of Theorem~\ref{thm:alpha2}, suppose that Assumption \ref{ass:theta2-consistency} and \ref{ass:theta2-rate}(ii)--(iv) hold and  $M_2 r_0 =o(1)$.
Then $\hat\theta_2$ admits the asymptotic expansion (\ref{eq:theta2-expan}),
\begin{align*}
\hat\theta_2 - \theta^*
= - H^{-1} \tilde E \left\{ \tau(U; \theta^*, \bar\alpha_2, \bar\gamma_2 ) \right\} + O_p( M_2 r_0^2 ),
\end{align*}
where $ H = E \{ \frac{\partial \tau}{\partial\theta}(U; \theta^*, \bar\alpha_2, \bar\gamma_2 ) \}$.
\end{lem}

\begin{prf}
By Theorem~\ref{thm:alpha2}, Lemma~\ref{lem:theta2-consistency}, and $M_2 r_0 =o(1)$,  we have
$\|\hat\alpha_2-\bar\alpha_2\|_1 = o_p(1)$, $\|\hat\gamma_2 - \bar\gamma_2\|_1 = o_p(1)$, and $|\hat\theta_2 - \theta^*| = o_p(1)$.
Hence for any small $\epsilon>0$, $(\hat\theta_2, \hat\alpha_2,\hat\gamma_2)\in\mathcal N_3$ with probability at least $1-\epsilon$ for all sufficiently large $n$.
In the following, we restrict analysis within this event.
Consider the decomposition
\begin{align}
&  \tilde E \left\{ \tau(U; \hat\theta_2, \hat\alpha_2, \hat\gamma_2 ) \right\}  - \tilde E \left\{ \tau(U; \theta^*, \bar\alpha_2, \bar\gamma_2 ) \right\} \nonumber \\
& = - \tilde E \left\{ \tau(U; \theta^*, \bar\alpha_2, \bar\gamma_2 ) \right\}  = \Delta_3 + \Delta_4, \label{eq:expansion-prf}
\end{align}
where
\begin{align*}
& \Delta_3 = \tilde E \left\{ \tau(U; \hat\theta_2, \hat\alpha_2, \hat\gamma_2 ) \right\}  - \tilde E \left\{ \tau(U; \theta^*, \hat\alpha_2, \hat\gamma_2 ) \right\} ,\\
& \Delta_4 = \tilde E \left\{ \tau(U; \theta^*, \hat\alpha_2, \hat\gamma_2 ) \right\}  - \tilde E \left\{ \tau(U; \theta^*, \bar\alpha_2, \bar\gamma_2 ) \right\}.
\end{align*}
We deal with the two terms $\Delta_3$ and $\Delta_4$ respectively.

By a Taylor expansion, $\Delta_4 = \Delta_{41} + \Delta_{42}$ with
\begin{align*}
\Delta_{41} & = (\hat\alpha_2 - \bar\alpha_2)^\T \tilde E \left\{ \xi \frac{\partial\tau}{\partial\eta_g} (U; \theta^*, \bar\alpha_2, \bar\gamma_2 ) \right\}
+ (\hat\gamma_2 - \bar\gamma_2)^\T \tilde E \left\{ \xi \frac{\partial\tau}{\partial\eta_f} (U; \theta^*, \bar\alpha_2, \bar\gamma_2 ) \right\}  \\
\Delta_{42} &= \frac{1}{2} (\hat\alpha_2 - \bar\alpha_2)^\T \tilde E \left\{ \xi \frac{\partial^2 \tau}{\partial\eta_g^2} (U; \theta^*, \tilde\alpha, \tilde\gamma ) \xi^\T \right\} (\hat\alpha_2 - \bar\alpha_2) \\
& \quad + \frac{1}{2} (\hat\gamma_2 - \bar\gamma_2)^\T \tilde E \left\{ \xi \frac{\partial^2 \tau}{\partial\eta_f^2} (U; \theta^*, \tilde\alpha, \tilde\gamma ) \xi^\T \right\} (\hat\gamma_2 - \bar\gamma_2) \\
& \quad +  (\hat\alpha_2 - \bar\alpha_2)^\T \tilde E \left\{ \xi \frac{\partial^2 \tau}{\partial\eta_g \partial \eta_f} (U; \theta^*, \tilde\alpha, \tilde\gamma ) \xi^\T \right\} (\hat\gamma_2 - \bar\gamma_2)
\end{align*}
where $(\tilde\alpha,\tilde\gamma)$ lie between $(\hat\alpha_2,\hat\gamma_2)$ and $(\bar\alpha_2,\bar\gamma_2)$.
As model (\ref{eq:g-model}) or (\ref{eq:f-model}) is correctly specified,  Proposition~\ref{pro:two-step} implies that
calibration equations (\ref{eq:CAL1})--(\ref{eq:CAL2}) are satisfied by $(\alpha,\gamma) = (\bar\alpha_2,\bar\gamma_2)$,
that is, the variables $\xi_j \frac{\partial\tau}{\partial\eta_g} (U; \theta^*, \bar\alpha_2, \bar\gamma_2 )$
and $\xi_j \frac{\partial\tau}{\partial\eta_g} (U; \theta^*, \bar\alpha_2, \bar\gamma_2 )$ have mean 0 for $j=1,\ldots,p$.
By Assumption~\ref{ass:theta2-rate}(iii) and similar reasoning as in Lemma~\ref{lem:prob-grad}, we have
\begin{align*}
\sup_j | \tilde E  \{ \xi_j \frac{\partial\tau}{\partial\eta_g} (U; \theta^*, \bar\alpha_2, \bar\gamma_2 ) \} | = O_p (r_0), \quad
\sup_j | \tilde E  \{ \xi_j \frac{\partial\tau}{\partial\eta_f} (U; \theta^*, \bar\alpha_2, \bar\gamma_2 ) \} | = O_p (r_0) .
\end{align*}
By Theorem~\ref{thm:alpha2}, $ \| \hat\alpha_2 - \bar\alpha_2\|_1 = O_p( M_2 r_0)$ and $\| \hat\gamma_2 - \bar\gamma_2\|_1  = O_p( M_2 r_0)$. Hence
\begin{align*}
& \left| (\hat\alpha_2 - \bar\alpha_2)^\T \tilde E \left\{ \xi \frac{\partial\tau}{\partial\eta_g} (U; \theta^*, \bar\alpha_2, \bar\gamma_2 ) \right\}  \right|
= O_p(r_0) \| \hat\alpha_2 - \bar\alpha_2\|_1 = O_p( M_2 r_0^2), \\
& \left| (\hat\gamma_2 - \bar\gamma_2)^\T \tilde E \left\{ \xi \frac{\partial\tau}{\partial\eta_f} (U; \theta^*, \bar\alpha_2, \bar\gamma_2 ) \right\}  \right|
= O_p(r_0) \| \hat\gamma_2 - \bar\gamma_2\|_1  = O_p( M_2 r_0^2),
\end{align*}
and $|\Delta_{41}| =  O_p( M_2 r_0^2)$.
Moreover, by Assumption~\ref{ass:theta2-rate}(iv) and similar reasoning as in Lemma~\ref{lem:prob-hess2}, we have
\begin{align*}
& \left| (\hat\alpha_2 - \bar\alpha_2)^\T \tilde E \left\{ \xi \frac{\partial^2 \tau}{\partial\eta_g^2} (U; \theta^*, \tilde\alpha_2, \tilde\gamma_2 ) \xi^\T \right\} (\hat\alpha_2 - \bar\alpha_2) \right| \\
& \le (\hat\alpha_2 - \bar\alpha_2)^\T \tilde E \left\{ \xi T^{(2)}_{\eta_g^2} (U; \theta^*, \tilde\alpha_2, \tilde\gamma_2 ) \xi^\T \right\} (\hat\alpha_2 - \bar\alpha_2) \\
& \le C_4 \tilde E  \left\{ (\hat\alpha_2^\T \xi -\bar\alpha_2^\T \xi)^2  \right\} + (1+C_4) O_p(r_0) \| \hat\alpha_2 - \bar\alpha_2 \|_1^2 = O_p(M_2 r_0^2),
\end{align*}
where $\tilde E \{ (\hat\alpha_2^\T \xi -\bar\alpha_2^\T \xi)^2  \} = O_p( M_2 r_0^2)$ and $O_p(r_0) O_p( M_2^2 r_0^2 ) = o_p(M_2 r_0^2)$ because $M_2 r_0 = o(1)$.
Similarly, we have
\begin{align*}
& \left| (\hat\gamma_2 - \bar\gamma_2)^\T \tilde E \left\{ \xi \frac{\partial^2 \tau}{\partial\eta_f^2} (U; \theta^*, \tilde\alpha, \tilde\gamma ) \xi^\T \right\} (\hat\gamma_2 - \bar\gamma_2) \right|
= O_p(M_2 r_0^2).
\end{align*}
and by the Cauchy--Schwartz inequality,
\begin{align*}
& \left|  (\hat\alpha_2 - \bar\alpha_2)^\T \tilde E \left\{ \xi \frac{\partial^2 \tau}{\partial\eta_g \partial \eta_f} (U; \theta^*, \tilde\alpha, \tilde\gamma ) \xi^\T \right\} (\hat\gamma_2 - \bar\gamma_2) \right| \\
& \le \left[(\hat\alpha_2 - \bar\alpha_2)^\T \tilde E \left\{ \xi T^{(2)}_{\eta_g \eta_f} (U; \theta^*, \tilde\alpha_2, \tilde\gamma_2 ) \xi^\T \right\} (\hat\alpha_2 - \bar\alpha_2)\right]^{1/2} \\
& \quad \times \left[(\hat\gamma_2 - \bar\gamma_2)^\T \tilde E \left\{ \xi T^{(2)}_{\eta_g \eta_f} (U; \theta^*, \tilde\gamma_2, \tilde\gamma_2 ) \xi^\T \right\} (\hat\gamma_2 - \bar\gamma_2)\right]^{1/2}
= O_p(M_2 r_0^2).
\end{align*}
Hence $|\Delta_{42} | = O_p( M_2 r_0^2)$ and $|\Delta_4| = O_p(M_2 r_0^2)$.

Next, by the mean value theorem, we have
\begin{align*}
\Delta_3 =  (\hat\theta_2 - \theta^*) \tilde E \left\{ \frac{\partial \tau}{\partial\theta}(U; \tilde \theta, \hat\alpha_2, \hat\gamma_2 ) \right\},
\end{align*}
where $\tilde\theta$ lies between $\hat\theta_2$ and $\theta^*$. Consider the decomposition
\begin{align}
& \tilde E \left\{ \frac{\partial \tau}{\partial\theta}(U; \tilde \theta, \hat\alpha_2, \hat\gamma_2 ) \right\} =
E \left\{ \frac{\partial \tau}{\partial\theta}(U; \theta^*, \bar\alpha_2, \bar\gamma_2 ) \right\} + \Delta_{31} + \Delta_{32}, \label{eq:H-prf}
\end{align}
where
\begin{align*}
& \Delta_{31}  =  \tilde E \left\{ \frac{\partial \tau}{\partial\theta}(U; \tilde \theta, \hat\alpha_2, \hat\gamma_2 ) \right\} -
E \left\{ \frac{\partial \tau}{\partial\theta}(U; \tilde\theta, \hat\alpha_2, \hat \gamma_2 ) \right\}, \\
& \Delta_{32} =  E \left\{ \frac{\partial \tau}{\partial\theta}(U; \tilde\theta, \hat\alpha_2, \hat \gamma_2 ) \right\} -
E \left\{ \frac{\partial \tau}{\partial\theta}(U; \theta^*, \bar\alpha_2, \bar\gamma_2 ) \right\} .
\end{align*}
By Assumption~\ref{ass:theta2-rate}(ii) and the uniform law of large numbers (\citealt{ferguson1996course}, Theorem 16),
$|\Delta_{31}| \le \sup_{ (\theta,\alpha,\gamma)\in\mathcal N_3} |(\tilde E-E) \{ \frac{\partial \tau}{\partial\theta}(U;\theta,\alpha,\gamma )\} = o_p(1)$.
Moreover, by $|\tilde\theta - \theta^*| \le |\hat\theta_2 - \theta^*| =o_p(1)$ and the continuous mapping theorem, $|\Delta_{32}| =o_p(1)$.
Hence $\tilde E  \{ \frac{\partial \tau}{\partial\theta}(U; \tilde \theta, \hat\alpha_2, \hat\gamma_2 )  \}  = H  + o_p(1)$
and $\Delta_3 =(\hat\theta_2 - \theta^*)  \{ H + o_p(1)\}$.

Finally, from the preceding analysis, (\ref{eq:expansion-prf}) yields
\begin{align*}
 - \tilde E \left\{ \tau(U; \theta^*, \bar\alpha_2, \bar\gamma_2 ) \right\}  = (\hat\theta_2 - \theta^*)  \{ H + o_p(1)\} + O_p( M_2 r_0^2).
\end{align*}
The desired result then follows because $H \not=0$.
\end{prf}

The following lemma establishes the consistency of $\hat V$ for $V$.

\begin{lem} \label{lem:theta2-var-est}
In the setting of Theorem~\ref{thm:alpha2}, suppose that Assumption \ref{ass:theta2-consistency} and \ref{ass:theta2-rate} hold and  $M_2 r_0 =o(1)$.
Then a consistent estimator of
$V= \var \{\tau(U; \theta^*, \bar\alpha_2, \bar\gamma_2 ) \}/ H^2$ is
\begin{align*}
\hat V = \tilde E \{\tau^2(U;\hat\theta_2,\hat\alpha_2,\hat\gamma_2 ) \} / \hat H^2,
\end{align*}
where $\hat H = \tilde E \{ \frac{\partial \tau}{\partial\theta}(U; \hat\theta_2, \hat\alpha_2, \hat\gamma_2 ) \}$.
\end{lem}

\begin{prf}
First, $\hat H = H + o_p(1)$ can be shown similarly as $\tilde E  \{ \frac{\partial \tau}{\partial\theta}(U; \tilde \theta, \hat\alpha_2, \hat\gamma_2 )  \}  = H  + o_p(1)$
in the proof of Lemma~\ref{lem:theta2-expan}. Next, we show that $\hat G = G + o_p(1)$, where
$G = E \{\tau^2 (U; \theta^*, \bar\alpha_2, \bar\gamma_2 ) \}$ and $\hat G = \tilde E \{\tau^2(U;\hat\theta_2,\hat\alpha_2,\hat\gamma_2 ) \}$.
Similarly as (\ref{eq:H-prf}), consider the decomposition
\begin{align*}
&   \tilde E \left\{ \tau^2 (U; \hat\theta_2, \hat\alpha_2, \hat\gamma_2 ) \right\} = E \left\{\tau^2 (U; \theta^*, \bar\alpha_2, \bar\gamma_2 ) \right\} + \Delta_{51} + \Delta_{52},
\end{align*}
where
\begin{align*}
& \Delta_{51}  =  \tilde E \left\{ \tau^2(U; \hat\theta_2, \hat\alpha_2, \hat\gamma_2 ) \right\} -
E \left\{ \tau^2 (U; \hat\theta_2, \hat\alpha_2, \hat \gamma_2 ) \right\}, \\
& \Delta_{52} =  E \left\{ \tau^2 (U; \hat\theta_2, \hat\alpha_2, \hat \gamma_2 ) \right\} -
E \left\{ \tau^2 (U; \theta^*, \bar\alpha_2, \bar\gamma_2 ) \right\} .
\end{align*}
By Assumption~\ref{ass:theta2-rate}(i) and the uniform law of large numbers (\citealt{ferguson1996course}, Theorem 16),
$|\Delta_{51}| \le \sup_{ (\theta,\alpha,\gamma)\in\mathcal N_3} |(\tilde E-E) \{ \tau^2 (U;\theta,\alpha,\gamma )\} = o_p(1)$.
Moreover, by $ |\hat\theta_2 - \theta^*| =o_p(1)$ and the continuous mapping theorem, $|\Delta_{52}| =o_p(1)$.
Hence $\hat G = G  + o_p(1)$.
\end{prf}

\subsection{Proof of Corollary~\ref{cor:debiased}}

Assume that $\psi_f \equiv 1$ in model (\ref{eq:f-PLM}) and $\hat\gamma_2=\hat\gamma_1$. First, we show that option (i) or (ii) in the discussion preceding Corollary~\ref{cor:debiased}
yields $(\hat\theta_1,\hat\alpha_2) = (\hat\theta_0,\hat\alpha_1)$
and hence $\hat\theta_2 = \hat\theta(\hat\alpha_1,\hat\gamma_1)$, provided that the same Lasso tuning parameter is used in computing $\hat\alpha_2$
as in computing $(\hat\theta_0,\hat\alpha_1)$.
Suppose that $(\hat\theta_0,\hat\alpha_1)$ are Lasso least square estimators as
\begin{align}
(\hat\theta_0,\hat\alpha_1) = \argmin_{(\theta,\alpha)} \left[ \tilde E \left\{ (Y - \theta Z - \alpha^\T \xi)^2 \right\} + \lambda |\theta| + \lambda \|\alpha \|_1 \right]. \label{eq:lasso-LS}
\end{align}
For option (ii), if $(\hat\theta_1,\hat\alpha_2)$ are redefined as Lasso least square estimators with the same tuning parameter $\lambda$,
then $(\hat\theta_1,\hat\alpha_2) = (\hat\theta_0,\hat\alpha_1)$ by definition.
For option (i), $\hat\theta_1$ is replaced by $\hat\theta_0$ in (\ref{eq:PLM-loss-alpha}). If $\hat\alpha_2$ is redefined as follows, with the same tuning parameter $\lambda$ as in (\ref{eq:lasso-LS}),
\begin{align}
\hat\alpha_2  = \argmin_{\alpha} \left[ \tilde E \left\{ (Y - \hat \theta_0 Z - \alpha^\T \xi)^2 \right\} + \lambda \|\alpha \|_1 \right], \label{eq:lasso-LS2}
\end{align}
then $\hat\alpha_2 = \hat\alpha_1$, because
for any $\alpha$,
\begin{align*}
& \tilde E \left\{ (Y - \hat \theta_0 Z - \alpha^\T \xi)^2 \right\} + \lambda |\hat\theta_0| +
\lambda \|\alpha \|_1 \ge  \tilde E \left\{ (Y - \hat \theta_0 Z -\hat \alpha_1^\T \xi)^2 \right\} + \lambda \|\hat\theta_0| + \lambda \|\hat\alpha_1 \|_1 \\
\Longleftrightarrow \; &
\tilde E \left\{ (Y - \hat \theta_0 Z - \alpha^\T \xi)^2 \right\} + \lambda \|\alpha \|_1 \ge  \tilde E \left\{ (Y - \hat \theta_0 Z -\hat \alpha_1^\T \xi)^2 \right\} + \lambda \|\hat\alpha_1 \|_1,
\end{align*}
and hence $\hat\alpha_1$ is also a minimizer to the objective in (\ref{eq:lasso-LS2}).

Now suppose that option (ii) is used, i.e., $\hat\theta_1$ is redefined as $\hat\theta_0$, in our two-step algorithm.
Then $\hat\alpha_2=\hat\alpha_1$ as shown above.
Proposition~\ref{pro:main} can be applied with $(\hat\theta_1,\hat\alpha_2,\hat\gamma_2)$ replaced by $(\hat\theta_0,\hat\alpha_1,\hat\gamma_1)$
and $\hat\theta_{\mbox{\tiny RCAL}}=\hat\theta_2$ by $\hat\theta(\hat\alpha_1,\hat\gamma_1)$,
because $\hat\theta_0$ can be shown to be pointwise doubly robust and hence Assumption~\ref{ass:gamma2-basic}(iv) is satisfied
under the stated regularity conditions. In fact, the target value (i.e., probability limit) $\bar\theta_1$ for $\hat\theta_0$, by definition, satisfies
$ E\{ (Y -\bar\theta_0 Z- \bar \alpha_1^\T \xi) (Z,\xi^\T)^\T \} =0$, which implies the population doubly robust estimating equation
$ E\{ (Y -\bar\theta_0 Z- \bar \alpha_1^\T \xi) (Z - \bar\gamma_1^\T \xi)\}=0$ or equivalently
\begin{align*}
\bar\theta_0 =  \frac{E\{ (Y - \bar \alpha_1^\T \xi) ( Z - \bar\gamma_1^\T \xi)\} }{ E \{ Z( Z - \bar\gamma_1^\T \xi)\} } .
\end{align*}
Hence $\bar\theta_0$ coincides with $\theta^*$ if model (\ref{eq:g-PLM}) or model (\ref{eq:f-PLM}) with $\psi_f \equiv 1$ is correctly specified.
This reasoning is a sample analogue of that in Example~\ref{eg:PLM-LS}.

\newpage

\section{Additional material for simulation studies}

We provide implementation details and additional simulation results.

\subsection{Partially linear modeling}

We describe the data-generating configurations used for $(Z,X)$, related to Fisher's discrimination analysis.
For setting (C1), we first generate $Z$ such that $P(Z=1)=q$. Next we generate $X |  Z=1 \sim \N(\mu_1, \Sigma)$ and $X |  Z=0 \sim \N (\mu_0, \Sigma)$. Then
\begin{align}
	P(Z =1 |  X) &= \frac{1}{1+\exp (-\beta_0 -\beta_1^{\T}X)} , \label{logit:linear}
\end{align}
where $\beta_0= - \frac{1}{2}\mu_1^{\T}\Sigma^{-1}\mu_1 + \frac{1}{2}\mu_0^{\T}\Sigma^{-1}\mu_0 + \log ( \frac{q}{1-q} )$ and $\beta_1 = \Sigma^{-1}(\mu_1 - \mu_0)$.
In our experiments, we choose $q=0.5$, $\Sigma= I\, (\text{identity matrix})$, $\mu_0=0$ and $\mu_1$ a sparse $p\times 1 $ vector with first 5 components being $(-0.25, 0.5, 0.75, 1, 1.25)$,
which leads to $\beta_0=-\frac{1}{2}\mu_1^{\T}\mu_1 = -0.4297$ and  $\beta_1= \mu_1 = (-0.25, 0.5, 0.75, 1, 1.25, 0,\ldots, 0)^\T$ in (\ref{logit:linear}).
This gives the stated expression of $P(Z=1 | X)$ in setting (C1).

For setting (C2), we  first generate $Z$ such that $P(Z=1)=q$. Next we generate $X |  Z=1 \sim \N (\mu_1, \Sigma_1)$ and $X |  Z=0 \sim \N(\mu_0, \Sigma_0)$. Then
\begin{align}
	P(Z  |  X) &= \frac{1}{1+\exp ( - \beta_0 - \beta_1^{\T}X - X^{\T}\Omega X)} \label{quadratic:logit}
\end{align}
where $\beta_0  = - \frac{1}{2} \mu_1^{\T}\Sigma_1^{-1}\mu_1 + \frac{1}{2} \mu_0^{\T}\Sigma_0^{-1}\mu_0 + \log \left( \frac{q}{1-q} \frac{|\Sigma_0|^{1/2}}{|\Sigma_1|^{1/2}} \right) $,
$\beta_1 = \Sigma_1^{-1}\mu_1 - \Sigma_0^{-1}\mu_0 $,
$\Omega  = \frac{1}{2} (\Sigma_0^{-1} - \Sigma_1^{-1})$.
In our experiments, we choose $q=0.5$, $\Sigma_0= I$, $\Sigma_1^{-1}=2 I$, $\mu_0=0$, and $\mu_1$ a sparse $p\times 1 $ vector with first 4 components being $(-0.25, 0.5, 0.75, 1)/2$,
which leads to $\beta_0= -\mu_1^{\T}\mu_1 + \log(2^{p/2}) = -0.4687 + \frac{p}{2} \log 2$ and  $\beta_1= 2\mu_1 = (-0.25, 0.5, 0.75, 1, 0, \ldots, 0)^\T$ in (\ref{quadratic:logit}).
This gives the stated expression of $P(Z=1 | X)$ in setting (C2).

Our two-step Algorithm~\ref{alg:two-step}, specialized to partially linear modeling, is presented in~Algorithm~\ref{alg:linear},
including associated commands from R package \texttt{glmnet} (\citealtappend{friedman2010regularization}).
Here \texttt{Y} and \texttt{Z} are $n\times 1$ vectors of observations $\{Y_i: i=1,\ldots,n\}$ and $\{Z_i: i=1,\ldots,n\}$,
and \texttt{X} and \texttt{ZX} are the design matrix of dimension $n\times p$ and $n\times (p+1)$ with $i$th row being $X^\T_i$ and $(Z_i, X^\T_i)$ respectively.
In Step 7 of Algorithm~\ref{alg:linear}, \texttt{offset} is a vector with components $\hat{\theta}_1^\T Z_i$.
and \texttt{weights} is a vector with components $\psi^\prime_f (\hat{\gamma}_1^\T X_i) = \expit (\hat{\gamma}_1^\T X_i)(1- \expit (\hat{\gamma}_1^\T X_i))$
for  $\psi_f= \expit (\cdot)$.  The argument \texttt{alpha=1} stands for the $\ell_1$ penalty.

\begin{algorithm}[t] 
	\caption{Two-step algorithm for partially linear modeling}\label{alg:linear}
	\begin{algorithmic}[1]
		\Procedure{Initial estimation}{}
		\State Compute $(\hat{\theta}_0, \hat{\alpha}_1) = \argmin_{\theta, \alpha} \, \left\{ \tilde{E} (Y - \theta Z - \alpha^\T X)^2 + \lambda_1 (\abs{\theta} + \norm{\alpha}_1) \right\} $
		
		using  \texttt{glmnet(ZX, y=Y, alpha=1, family="gaussian")}.
		\State Compute $\hat\gamma_1 = \argmin_{\gamma} \left[\tilde{E} \{ - Z\gamma^T X + \log (1+ e^{\gamma^\T X})\} + \lambda_2 \norm{\gamma}_1 \right] $
		
			using  \texttt{glmnet(X, y=Z, alpha=1, family="binomial")}.
		\State Compute $\hat{\theta}_1 = \frac{\tilde{E} \{ (Y - \hat{\alpha}_1^\T X)(Z- \psi_f (\hat{\gamma}_1^\T X)) \}}{\tilde{E} \{ Z(Z - \psi_f (\hat{\gamma}_1^\T X)) \} }$.
		\EndProcedure
		\Procedure{Calibrated estimation}{}
		\State Compute $\hat{\alpha}_2 = \argmin_{\alpha} \, \left\{\tilde{E} \psi'_f (\hat{\gamma}_1^\T X) (Y - \hat{\theta}_1 Z - \alpha^\T X)^2 + \lambda_3 \norm{\alpha}_1 \right\}$
		
		using  \texttt{glmnet(X, y=Y, alpha=1, offset=theta1*Z, weights)}.
		\State Compute $\hat{\theta}_2 = \frac{\tilde{E} \{(Y - \hat{\alpha}_2^\T X)(Z- \psi_f (\hat{\gamma}_1^\T X)) \}}{\tilde{E} \{ Z(Z - \psi_f (\hat{\gamma}_1^\T X)) \}}$,

and $\hat{\mathrm{V}}(\hat{\theta}_2) =\frac{\tilde{E} \{(Y - \hat{\alpha}_2^\T X)^2(Z- \psi_f (\hat{\gamma}_1^\T X))^2 \}}{\tilde{E}^2 \{Z(Z - \psi_f (\hat{\gamma}_1^\T X)) \}}$.
		\EndProcedure
	\end{algorithmic}
\end{algorithm}

\begin{algorithm}[H]
	\caption{Debiased Lasso for linear modeling}\label{alg:linear_deb}
	\begin{algorithmic}[1]
		\Procedure{Linear projection}{}
		\State Compute $(\hat{\theta}_0, \hat{\alpha}_1) = \argmin_{\theta, \alpha} \, \left\{ \tilde{E} (Y - \theta Z - \alpha^\T X)^2 + \lambda_1 (\abs{\theta} + \norm{\alpha}_1) \right\} $
		
		using  \texttt{glmnet(ZX, y=Y, alpha=1, family="gaussian")}.
		\State Compute $\hat\gamma_1 = \argmin_{\gamma} \left[ \tilde{E} \{ (Z - \gamma^\T X)^2\} + \lambda_2 \norm{\gamma}_1 \right] $
		
		using  \texttt{glmnet(X, y=Z, alpha=1, family="gaussian")}.
		\State Compute $\hat{\theta}_{\mbox{\tiny DB}} = {\theta}_0 + \frac{\tilde{E} \{ (Y - \hat{\theta}_0 Z -  \hat{\alpha}_1^\T X)(Z- \hat{\gamma}_1^\T X) \}}{\tilde{E} \{Z(Z - \hat{\gamma}_1^\T X) \}}$

and $\hat{\mathrm{V}}(\hat{\theta}_{\mbox{\tiny DB}} ) =\frac{\tilde{E} \{ (Y - \hat{\alpha}_1^\T X)^2(Z- \psi_f (\hat{\gamma}_1^\T X))^2 \}}{\tilde{E}^2 \{ Z(Z - \psi_f (\hat{\gamma}_1^\T X)) \}}$.
		\EndProcedure
	\end{algorithmic}
\end{algorithm}

The tuning parameters $\lambda_1, \lambda_2, \lambda_3$ are sequentially selected from $5$-fold cross validation, using \texttt{cv.glmnet}() in the R package \texttt{glmnet}.
For linear regression we set \texttt{type.measure="MSE"} and
for logistic or log-linear regression, we set \texttt{type.measure="deviance"}.
By default, there are $100$ values of $\lambda$ in the grid search over $\lambda$ (\citealtappend{friedman2010regularization}).

The debiased Lasso method used in our experiments is shown in Algorithm~\ref{alg:linear_deb}, where robust variance estimation is employed
(\citealt{zhang2014confidence, SVD,buhlmann2015high}).
Step 3 in Algorithm~\ref{alg:linear_deb} involves fitting a linear model of $Z$ given $X$, instead of a logistic model in Algorithm~\ref{alg:linear}.

Table~\ref{Table:linear}
presents  simulation results and Figures~\ref{fig:linear_p=100}--\ref{fig:linear_p=800} show QQ plots of estimates and $t$-statistics for $n=400$ and $p=100$, $200$ as well as $p=800$ (for completeness).
Comparison between the three methods is similar as discussed in the main paper.

\subsection{Partially log-linear modeling}

Our two-step Algorithm~\ref{alg:two-step}, specialized to partially log-linear modeling, is presented in~Algorithm~\ref{alg:log-linear},
including associated commands from R package \texttt{glmnet}.
Because  $Z_i$'s are binary, a closed-form solution can be obtained from the doubly robust estimating equation:
\begin{align}
\me^{-\theta} =\frac{ \displaystyle  \sum_{Z_i=1} e^{\alpha^\T X_i}(1- \expit(\gamma^\T X_i)) +  \displaystyle  \sum_{Z_i=0} (Y_i - e^{\alpha^\T X_i}) \expit(\gamma^\T X_i)}{ \displaystyle  \sum_{Z_i=1} Y_i(1-\expit(\gamma^\T X_i))} . \label{DRE:log-linear}
\end{align}
Steps 7 and 8 are implemented as regularized weighted maximum likelihood estimation, by specifying weights in \texttt{glmnet} as follows:
\texttt{weights1} is a vector  with components $\me^{\hat{\alpha}_1^\T X_i}$
 and \texttt{weights2} is a vector with components $\me^{-\hat{\theta}_1 Z_i} \expit(\hat{\gamma}_2^\T X_i)(1- \expit(\hat{\gamma}_2^\T X_i))$.

The debiased Lasso method used in our experiments is shown in Algorithm~\ref{alg:log-linear_deb}, where robust variance estimation is employed.
Step 3 is implemented as regularized least square estimation, where \texttt{weights3} is a vector with components $\me^{\theta_0 Z_i + \hat{\alpha}_1^\T X_i}$.

Table~\ref{Table:log-linear} presents  simulation results and Figure~\ref{fig:log-linear_p=100}--\ref{fig:log-linear_p=200} show QQ plots of estimates and
$t$-statistics for $n=400$ and $p=100$, $200$ as well as $p=800$ (for completeness).
Comparison between the three methods is similar as discussed in the main paper.

\clearpage
\begin{algorithm}[t] 
	\caption{Two-step algorithm for partially log-linear modeling}\label{alg:log-linear}
	\begin{algorithmic}[1]
		\Procedure{Initial estimation}{}
		\State Compute $(\hat{\theta}_0, \hat{\alpha}_1) = \argmin_{\theta, \alpha} \, \left[\tilde{E} \{ - Y(\theta Z + \alpha^\T X) + e^{\theta Z + \alpha^\T X}\} + \lambda_1 (\abs{\theta} + \norm{\alpha}_1)\right]$
		
			using  \texttt{glmnet(ZX, y=Y, alpha=1, family="poisson")}.
		\State Compute $\hat\gamma_1 = \argmin_{\gamma} \left[\tilde{E} \{ - Z\gamma^T X + \log (1+ e^{\gamma^\T X})\} + \lambda_2 \norm{\gamma}_1 \right]$
		
			using  \texttt{glmnet(X, y=Z, alpha=1, family="binomial")}.
		\State Compute $\hat{\theta}_1$ from  (\ref{DRE:log-linear}) with $\alpha = \hat{\alpha}_1$ and $\gamma = \hat{\gamma}_1$.
		\EndProcedure
		\Procedure{Calibrated estimation}{}
		\State Compute $\hat{\gamma}_2=  \argmin_{\gamma} \, \tilde{E} e^{\hat{\alpha}_1^\T X}\{ - Z\gamma^\T X + \log (1+ e^{\gamma^\T X})\} + \lambda_3 \norm{\gamma}_1 $
		
			using  \texttt{glmnet(X, y=Z, alpha=1, weights1, family="binomial")}.
		\State Compute $\hat{\alpha}_2 =   \argmin_{\alpha} \, \tilde{E} \psi'_f(\hat{\gamma}_2^\T X)e^{-\hat{\theta}_1Z} \{ - Y(\hat{\theta}_1Z + \alpha^\T X) + e^{\hat{\theta}_1Z + \alpha^\T X}\} + \lambda_4 \norm{\alpha}_1$
		
		using \texttt{glmnet(X,y=Y,alpha=1,offset=theta1*Z,weights2,family="poisson")}.
		\State Compute $\hat{\theta}_2$ from (\ref{DRE:log-linear}) with $\alpha = \hat{\alpha}_2$ and $\gamma = \hat{\gamma}_2$

and $\hat{V} (\hat{\theta}_2) = \frac{\tilde{E} \{(Y\me^{-\hat{\theta}_2 Z} - \me^{\hat{\alpha}_2^\T X})^2(Z - \expit (\hat{\gamma}_2^\T X))^2 \}}
{\tilde{E}^2 \{ \me^{-\hat{\theta}_2 Z} YZ (Z - \expit (\hat{\gamma}_2^\T X)) \}} $.
		\EndProcedure
	\end{algorithmic}
\end{algorithm}

\begin{algorithm}[H]
	\caption{Debiased Lasso for log-linear modeling}\label{alg:log-linear_deb}
	\begin{algorithmic}[1]
		\Procedure{Linear projection}{}
		\State Compute $(\hat{\theta}_0, \hat{\alpha}_1) = \argmin_{\theta, \alpha} \, \left[\tilde{E} \{ - Y(\theta Z + \alpha^\T X) + e^{\theta Z + \alpha^\T X}\} + \lambda_1 (\abs{\theta} + \norm{\alpha}_1)\right]$
		
		using  \texttt{glmnet(ZX, y=Y, alpha=1, family="poisson")}.
	\State Compute $\hat\gamma_1 = \argmin_{\gamma} \left[\tilde{E} \{ e^{\hat{\theta}_0 Z + \hat{\alpha}_1 X}(Z - \gamma^\T X)^2\} + \lambda_2 \norm{\gamma}_1 \right] $
	
	using  \texttt{glmnet(X, y=Z, alpha=1, weights3, family="gaussian")}.
	\State Compute $\hat{\theta}_{\mbox{\tiny DB}} = {\hat{\theta}}_0 + \frac{\tilde{E} \{ (Y - \me^{\hat{\theta}_0 Z + \hat{\alpha}_1 X})(Z- \hat{\gamma}_1^\T X) \}}
{\tilde{E} \{ \me^{\hat{\theta}_0 Z + \hat{\alpha}_1 X} Z(Z - \hat{\gamma}_1^\T X) \}}$

and $\hat{V} (\hat{\theta}_{\mbox{\tiny DB}}) = \frac{\tilde{E} \{ (Y - \me^{\hat{\theta}_0 Z + \hat{\alpha}_1 X})^2(Z- \hat{\gamma}_1^\T X)^2 \}}
{\tilde{E}^2 \{ \me^{\hat{\theta}_0 Z + \hat{\alpha}_1 X} Z(Z - \hat{\gamma}_1^\T X) \} }$.
	\EndProcedure
	\end{algorithmic}
\end{algorithm}

\clearpage
\subsection{Partially logistic modeling}

We describe the data-generating configurations used for $(Z,Y,X)$, related to the odds ratio model in \citeappend{chen2007semiparametric}.
We first generate $X \sim \N (0, \Sigma)$, where $\Sigma = \text{Toeplitz}(\rho=0.5)$.
Given $X$, we generate binary variables $(Z,Y)$ according to the probabilities proportional to the entries in the following $2\times 2$ table:
\begin{center}
	\begin{tabular}{l|l|c|c|c}
		\multicolumn{2}{c}{}&\multicolumn{2}{c}{}&\\
		\cline{3-4}
		\multicolumn{2}{c|}{} & $Z=0$ & $Z=1$ &\\
		\cline{2-4}
		& $Y=0$ & $1$ & $e^{\beta_1 + h_1(X)}$ \\
		\cline{2-4}
		\multirow{2}{*}{}& $Y=1$ & $e^{\beta_2 + h_2(X)}$ & $e^{\theta^* + \beta_1 + \beta_2 + h_3(X)}$\\
		\cline{2-4}
	\end{tabular}
\end{center}\vspace{.1in}
Here $\theta^*, \beta_1$ and $\beta_2$ are the true parameter values and $h_1(X), h_2(X)$ and $h_3(X)$ are functions in $X$ such that $h_3(X)= h_1(X) + h_2(X)$.
The implied conditional probabilities  are
	\begin{align}
	 P(Y=1  |  X, Z) &= \expit (\theta^*Z +h_2(X) ), \label{logit-cond-prob1}\\
	 P(Z=1  |  X, Y=0) &= \expit (\beta_1 + h_1(X) ) \label{logit-cond-prob2}.
	\end{align}
In our experiments, we set $\theta^*=2, \beta_1=0.25$, $ \beta_2=-0.25$, both $\alpha$ and $\gamma$ as a sparse vector with first four components being $(-0.25, 0.25, 0.5, 0.75)/2$.
The functions $h_1$ and $h_2$ are chosen differently, depending on settings (C7)--(C9).
\begin{itemize}
	\item [(i)] Taking $h_1(X)=\gamma^\T X$ and $h_2(X)=\alpha^\T X$ in (\ref{logit-cond-prob1})--(\ref{logit-cond-prob2})
leads to the stated expressions for $P(Y =1 | X,Z)$ and $P(Z=1 | X,Y=0)$ in settings (C7).

	\item [(ii)] Taking $h_1(X) = \gamma^\T X$ and $h_2(X) = 0.25 X_{1} + 5 X_{2} + \expit(X_{3})$ in (\ref{logit-cond-prob1})--(\ref{logit-cond-prob2})
leads to the stated expressions for $P(Y =1 | X,Z)$ and $P(Z=1 | X,Y=0)$ in settings (C8).

	\item [(iii)] Taking $h_2(X) = \alpha^\T X$ and $h_1(X) = 0.25 X_{1} + 0.8 X_{2} + \expit(X_{3})$ in (\ref{logit-cond-prob1})--(\ref{logit-cond-prob2})
leads to the stated expressions for $P(Y =1 | X,Z)$ and $P(Z=1 | X,Y=0)$ in settings (C9).
\end{itemize}

\begin{algorithm}[t] 
	\caption{Two-step algorithm for partially logistic modeling}\label{alg:logistic}
	\begin{algorithmic}[1]
		\Procedure{Initial estimation}{}
		\State Compute $(\hat{\theta}_0, \hat{\alpha}_1)=\argmin_{\theta, \alpha} \, \left[\tilde{E} \{-Y(\theta Z + \alpha^\T X) + \log (1+ e^{\theta Z + \alpha^\T X})\} + \lambda_1 (\abs{\theta} + \norm{\alpha}_1)\right]$
		
			using  \texttt{glmnet(XZ, y=Y, alpha=1, family="binomial")}.
		\State Compute $\hat\gamma_1= \argmin_{\gamma} \, \left[\tilde{E}_{Y=0} \{ - Z \gamma^T X + \log (1+ \me^{\gamma^\T X})\} + \lambda_2 \norm{\gamma}_1 \right]$
		
			using  \texttt{glmnet(X0, y=Z0, alpha=1, family="binomial")}.
		\State Compute $\hat{\theta}_1$ from  (\ref{DRE:logistic}) with $\alpha = \hat{\alpha}_1$ and $\gamma = \hat{\gamma}_1$..
		\EndProcedure
		\Procedure{Calibrated estimation}{}
		\State Compute $\hat{\gamma}_2=  \argmin_{\gamma} \, \left[\tilde{E} \me^{-\theta_1 ZY}\expit_2 (\hat{\alpha}_1^\T X) \{ - Z\gamma^\T X + \log (1+ \me^{\gamma^\T X})\} + \lambda_3 \norm{\gamma}_1 \right]$
		
			using  \texttt{glmnet(X, y=Z, alpha=1, weights1, family="binomial")}.
		\State Compute $\hat{\alpha}_2 =   \argmin_{\alpha} \, \left[\tilde{E} \me^{-\theta_1 ZY}\expit_2 (\hat{\gamma}_2^\T X) \{ - Y\alpha^\T X + \log (1 + \me^{\alpha^\T X})\} + \lambda_4 \norm{\alpha}_1 \right]$
		
		 using  \texttt{glmnet(X, y=Y, alpha=1, weights2, family="binomial")}
		\State Compute $\hat{\theta}_2$ from  (\ref{DRE:logistic}) with $\alpha = \hat{\alpha}_2$ and $\gamma = \hat{\gamma}_2$

and $\hat{V}(\hat{\theta}_2) =  \frac{\tilde{E} \{\me^{-2\hat{\theta}_2 ZY}(Y- \expit (\hat{\alpha}_2^\T X))^2(Z - \expit(\hat{\gamma}_2^\T X))^2 \}}
{\tilde{E}^2 \{ ZY \me^{-2\hat{\theta}_2 ZY} (Y- \expit (\hat{\alpha}_2^\T X))(Z - \expit(\hat{\gamma}_2^\T X)) \}}$.
		\EndProcedure
	\end{algorithmic}
\end{algorithm}

\begin{algorithm}[H]
	\caption{Debiased Lasso for logistic modeling}\label{alg:logistic_deb}
	\begin{algorithmic}[1]
		\Procedure{Linear projection}{}
		\State Compute $(\hat{\theta}_0, \hat{\alpha}_1)=\argmin_{\theta, \alpha} \, \left[\tilde{E} \{-Y(\theta Z + \alpha^\T X) + \log (1+ e^{\theta Z + \alpha^\T X})\} + \lambda_1 (\abs{\theta} + \norm{\alpha}_1)\right]$
		
		using  \texttt{glmnet(XZ, y=Y, alpha=1, family="binomial")}.
		\State Compute $\hat\gamma_1 = \argmin_{\gamma} \left[\tilde{E} \{ \expit_2 (\hat{\theta}_0 Z + \hat{\alpha}_1 X)(Z - \gamma^\T X)^2\} + \lambda_2 \norm{\gamma}_1\right] $
		
		using  \texttt{glmnet(X, y=Z, alpha=1, weights3, family="gaussian")}.
		\State Compute $\hat{\theta}_{\mbox{\tiny DB}} = {\theta}_0 + \frac{\tilde{E} \{ (Y - \expit(\hat{\theta}_0 Z + \hat{\alpha}_1 X))(Z- \hat{\gamma}_1^\T X) \}}
{\tilde{E} \{\expit_2 (\hat{\theta}_0 Z + \hat{\alpha}_1 X) Z(Z - \hat{\gamma}_1^\T X) \}}$

and $\hat{V}(\hat{\theta}_{\mbox{\tiny DB}}) = \frac{\tilde{E} \{(Y - \expit(\hat{\theta}_0 Z_i + \hat{\alpha}_1 X))^2(Z - \hat{\gamma}_1^\T X)^2 \}}
{\tilde{E}^2 \{\expit_2 (\hat{\theta}_0 Z_i + \hat{\alpha}_1 X)(Z - \hat{\gamma}_1^\T X) \}}$.
		\EndProcedure
	\end{algorithmic}
\end{algorithm}

Our two-step Algorithm~\ref{alg:two-step}, specialized to partially logistic modeling, is presented in~Algorithm~\ref{alg:logistic},
including associated commands from R package \texttt{glmnet}.
Because  $Z_i$'s are binary, a closed-form solution can be obtained from the doubly robust estimating equation:
\begin{align}
\me^{-\theta} =\frac{- \displaystyle  \sum_{Z_i=0 \text{ or } Y_i=0} (Y_i - \expit(\alpha^\T X_i))(Z_i -\expit(\gamma^\T X_i))}{ \displaystyle  \sum_{Z_i=1 \text{ and } Y_i=1}(1-\expit(\gamma^\T X_i))(1-\expit(\alpha^\T X_i))} . \label{DRE:logistic}
\end{align}
In Step 3,  the sample average $\tilde{E}_{Y=0}( )$ is computed on over the subsample with $Y_i=0$, i.e., $\{(Z_i, X_i): Y_i=0, i=1,\ldots,n\}$. Here \texttt{Z0} denotes $\{Z_i: Y_i=0\}$ and the \texttt{X0} is the design matrix with $i$th row being $\{X_i^\T: Y_i=0, i=1,\ldots,n\}$. Steps 7 and 8 are implemented as regularized weighted maximum likelihood estimation by specifying weights in \texttt{glmnet}
as follows: \texttt{weights1} is a vector with components $\me^{-\hat{\theta}_1 Z_i Y_i} \expit(\hat{\alpha}_1^\T X_i)(1- \expit(\hat{\alpha}_1^\T X_i))$ and
 \texttt{weights2} is a vector with components $\me^{-\hat{\theta}_1 Z_i Y_i} \expit(\hat{\gamma}_2^\T X_i)(1- \expit(\hat{\gamma}_2^\T X_i))$.

The debiased Lasso method used in our experiments is shown in Algorithm~\ref{alg:logistic_deb}, where robust variance estimation is employed.
Step 3 is implemented as regularized least square estimation, where \texttt{weights3} is a vector with components
$\expit_2 (\hat{\theta}_0 Z_i + \hat{\alpha}_1 X_i) =\expit(\hat{\theta}_0 Z_i + \hat{\alpha}_1 X_i)(1-\expit(\hat{\theta}_0 Z_i + \hat{\alpha}_1 X_i))$.

Table~\ref{Table:logistic} presents  simulation results and Figures~\ref{fig:logistic_p=100}--\ref{fig:logistic_p=800}
show QQ plots of estimates and $t$-statistics for $n=400$ and $p=100$, $200$ as well as $p=800$ (for completeness).
Comparison between the three methods is similar as discussed in the main paper.

\newpage
\begin{center}
	\captionsetup{width=1\linewidth}
	\captionof{table} {Summary of results for partially linear modeling} \label{Table:linear}
	\scalebox{1}{
		\begin{tabular}{cccc|ccc|ccc}
			\toprule
			\multirow{4}{*}{} &
			\multicolumn{3}{c|}{(C1) Cor Cor}&
			\multicolumn{3}{c|}{(C2) Cor Miss}&
			\multicolumn{3}{c}{(C3) Mis Cor}\\
			& $\hat{\theta}_\DB$ & $\hat{\theta}_1$ & $\hat{\theta}_2$ & $\hat{\theta}_\DB$ & $\hat{\theta}_1$ & $\hat{\theta}_2$ & $\hat{\theta}_\DB$& $\hat{\theta}_1$& $\hat{\theta}_2$\\
			\midrule
			& \multicolumn{9}{c}{$n=400, p=100$}\\
			\multirow{1} {*}{Bias}
			& 0.004 & 0.006 & 0.005 & -0.003 & -0.004 & -0.003 & 0.103 & 0.049 & 0.003 \\ 			
			
			\multirow{1}{*}{$\sqrt{\text{Var}}$}
			& 0.056 & 0.057 & 0.056 & 0.057 & 0.056 & 0.056 & 0.290 & 0.289 & 0.288 \\

			\multirow{1}{*}{$\sqrt{\text{Evar}}$}
			& 0.053 & 0.053 & 0.052 & 0.054 & 0.054 & 0.053 &0.324 & 0.320 & 0.321 \\

			\multirow{1}{*}{Cov95}
			& 0.945 & 0.943 & 0.944 & 0.946 & 0.942 & 0.940 & 0.921 & 0.946 & 0.948 \\
			\midrule
	    	& \multicolumn{9}{c}{$n=400, p=200$}\\
			\multirow{1} {*}{Bias}
			& 0.006 & 0.004 & 0.004 & 0.009 & 0.008 & 0.009 & 0.156 & 0.061 & 0.004 \\ 			
			
			\multirow{1}{*}{$\sqrt{\text{Var}}$}
			& 0.057 & 0.056 & 0.056 & 0.053 & 0.053 & 0.055 & 0.321 & 0.299 & 0.298 \\
			
			\multirow{1}{*}{$\sqrt{\text{Evar}}$}
			& 0.053 & 0.053 & 0.055 & 0.054 & 0.055 & 0.054 & 0.331 & 0.327 & 0.325\\ 	
			
			\multirow{1}{*}{Cov95}
			& 0.932 & 0.933 & 0.934 & 0.945 & 0.938 & 0.944 & 0.911 & 0.936 & 0.944\\
			\midrule
            & \multicolumn{9}{c}{$n=400, p=800$}\\
			\multirow{1} {*}{Bias}
			& 0.006 & 0.006 & 0.007 & 0.013 & 0.012 & 0.012 & 0.283 & 0.082 & 0.004 \\ 			
			
			\multirow{1}{*}{$\sqrt{\text{Var}}$}
			& 0.058 & 0.057 & 0.057 & 0.057 & 0.057 & 0.058 & 0.322 & 0.301 & 0.299 \\

			\multirow{1}{*}{$\sqrt{\text{Evar}}$}
			& 0.055 & 0.056 & 0.058 & 0.059 & 0.058 & 0.056 & 0.334 & 0.330 & 0.328 \\

			\multirow{1}{*}{Cov95}
			& 0.922 & 0.928 & 0.928 & 0.921 & 0.931 & 0.941 & 0.872 & 0.929 & 0.941 \\
			\bottomrule
		\end{tabular}\vspace{3pt}
	}
\end{center}

\newpage
\begin{center}
	\captionsetup{width=1\linewidth}
	\captionof{table} {Summary of results for partially log-linear modeling} \label{Table:log-linear}
	\scalebox{1}{
		\begin{tabular}{cccc|ccc|cccc}
			\toprule
			\multirow{4}{*}{} &
			\multicolumn{3}{c|}{(C4) Cor Cor}&
			\multicolumn{3}{c|}{(C5) Cor Miss}&
			\multicolumn{3}{c}{(C6) Miss Cor}\\
			& $\hat{\theta}_\DB$ & $\hat{\theta}_1$ & $\hat{\theta}_2$ & $\hat{\theta}_\DB$ & $\hat{\theta}_1$ & $\hat{\theta}_2$ & $\hat{\theta}_\DB$& $\hat{\theta}_1$& $\hat{\theta}_2$&\\
			\midrule
			& \multicolumn{9}{c}{$n=600, p=100$}\\
			\multirow{1} {*}{Bias}
			& 0.003 & 0.001 & 0.005 & -0.013 & 0.004 & 0.003 & -0.024 & -0.005 & 0.002 &\\ 			
			
			\multirow{1}{*}{$\sqrt{\text{Var}}$}
			& 0.041 & 0.046 & 0.047 & 0.039 & 0.042 & 0.045 & 0.072 & 0.079 & 0.082 &\\

			\multirow{1}{*}{$\sqrt{\text{Evar}}$}
			& 0.037 & 0.043 & 0.044 & 0.035 & 0.041 & 0.041 & 0.071 & 0.078 & 0.077 &\\

			\multirow{1}{*}{Cov95}
			& 0.946 & 0.948 & 0.942 & 0.914 & 0.943 & 0.942 & 0.912 & 0.944 & 0.946 &\\
			\midrule
			& \multicolumn{9}{c}{$n=600, p=200$}\\
			\multirow{1} {*}{Bias}
			& 0.007 & 0.003 & 0.005 & -0.018 & -0.005 & 0.005 & -0.064 & -0.011 & 0.008 &\\ 			
			
			\multirow{1}{*}{$\sqrt{\text{Var}}$}
			& 0.042 & 0.046 & 0.046 & 0.041 & 0.043 & 0.047 & 0.074 & 0.078 & 0.083 &\\

			\multirow{1}{*}{$\sqrt{\text{Evar}}$}
			& 0.039 & 0.043 & 0.044 & 0.037 & 0.041 & 0.043 & 0.073 & 0.080 & 0.081 &\\
			
			\multirow{1}{*}{Cov95}
			& 0.948 & 0.939 & 0.941 & 0.882 & 0.925 & 0.934 & 0.855 & 0.941 & 0.944 &\\
			\midrule
			& \multicolumn{9}{c}{$n=600, p=800$}\\
			\multirow{1} {*}{Bias}
			& 0.008 & 0.009 & 0.005 & -0.023 & -0.015 & 0.005 & -0.093 & -0.021 & 0.010 &\\ 			
			
			\multirow{1}{*}{$\sqrt{\text{Var}}$}
			& 0.043 & 0.046 & 0.048 & 0.045 & 0.045 & 0.048 & 0.075 & 0.078 & 0.081 &\\

			\multirow{1}{*}{$\sqrt{\text{Evar}}$}
			& 0.043 & 0.048 & 0.045 & 0.046 & 0.045 & 0.047 & 0.078 & 0.077 & 0.080 &\\

			\multirow{1}{*}{Cov95}
			& 0.938 & 0.934 & 0.941 & 0.857 & 0.913 & 0.923 & 0.701 & 0.933 & 0.941 &\\
			\bottomrule
		\end{tabular}\vspace{3pt}
	}
\end{center}

\newpage
\begin{center}
	\captionsetup{width=1\linewidth}
	\captionof{table} {Summary of results for partially logistic modeling} \label{Table:logistic}
	\scalebox{1}{
		\begin{tabular}{cccc|ccc|cccc}
			\toprule
			\multirow{4}{*}{} &
			\multicolumn{3}{c|}{(C7) Cor Cor}&
			\multicolumn{3}{c|}{(C8) Cor Miss}&
			\multicolumn{3}{c}{(C9) Miss Cor}\\
			& $\hat{\theta}_\DB$ & $\hat{\theta}_1$ & $\hat{\theta}_2$ & $\hat{\theta}_\DB$ & $\hat{\theta}_1$ & $\hat{\theta}_2$ & $\hat{\theta}_\DB$& $\hat{\theta}_1$& $\hat{\theta}_2$&\\
			\midrule
			& \multicolumn{9}{c}{$n=600, p=100$}\\
			\multirow{1} {*}{Bias}
			& 0.035 & 0.029 & 0.025 & 0.011 & 0.012 & 0.011 & 0.122 & 0.032 & 0.015 &\\ 			
			
			\multirow{1}{*}{$\sqrt{\text{Var}}$}
			& 0.232 & 0.244 & 0.238 & 0.264 & 0.295 & 0.288 & 0.375 & 0.339 & 0.330 &\\

			\multirow{1}{*}{$\sqrt{\text{Evar}}$}
			& 0.225 & 0.233 & 0.233 & 0.295 & 0.287 & 0.287 & 0.369 & 0.315 & 0.315 &\\

			\multirow{1}{*}{Cov95}
			& 0.944 & 0.945 & 0.946 & 0.935 & 0.937 & 0.941 & 0.931 & 0.939 & 0.940 &\\
			\midrule
			& \multicolumn{9}{c}{$n=600, p=200$}\\
			\multirow{1} {*}{Bias}
			& 0.047 & 0.053 & 0.037 & 0.042 & 0.063 & 0.043 & 0.267 & 0.091 & 0.035 &\\ 			
			
			\multirow{1}{*}{$\sqrt{\text{Var}}$}
			& 0.233 & 0.241 & 0.239 & 0.266 & 0.300 & 0.291 & 0.281 & 0.343 & 0.329 &\\

			\multirow{1}{*}{$\sqrt{\text{Evar}}$}
			& 0.227 & 0.238 & 0.235 & 0.299 & 0.278 & 0.288 & 0.366 & 0.316 & 0.317 &\\
			
			\multirow{1}{*}{Cov95}
			& 0.940 & 0.934 & 0.944 & 0.950 & 0.931 & 0.948 & 0.901 & 0.920 & 0.938 &\\
			\midrule
			& \multicolumn{9}{c}{$n=600, p=800$}\\
			\multirow{1} {*}{Bias}
			& 0.059 & 0.065 & 0.049 & 0.045 & 0.068 & 0.046 & 0.244 & 0.0524 & 0.045 &\\ 			
			
			\multirow{1}{*}{$\sqrt{\text{Var}}$}
			& 0.232 & 0.245 & 0.239 & 0.273 & 0.315 & 0.298 & 0.278 & 0.339 & 0.332 &\\

			\multirow{1}{*}{$\sqrt{\text{Evar}}$}
			& 0.226 & 0.241 & 0.238 & 0.296 & 0.287 & 0.289 & 0.371 & 0.322 & 0.326 &\\

			\multirow{1}{*}{Cov95}
			& 0.936 & 0.930 & 0.936 & 0.945 & 0.937 & 0.949 & 0.900 & 0.929 & 0.938 &\\
			\bottomrule
		\end{tabular}\vspace{3pt}
	}
\end{center}

\newpage

\begin{figure}
	\caption{\small QQ plots of the estimates (first column) and $t$-statistics (second column) against standard normal ($n=400$, $p=100$) with partially linear modeling}
	\label{fig:linear_p=100}
	\centering
	\subfloat{{\includegraphics[width=70mm]{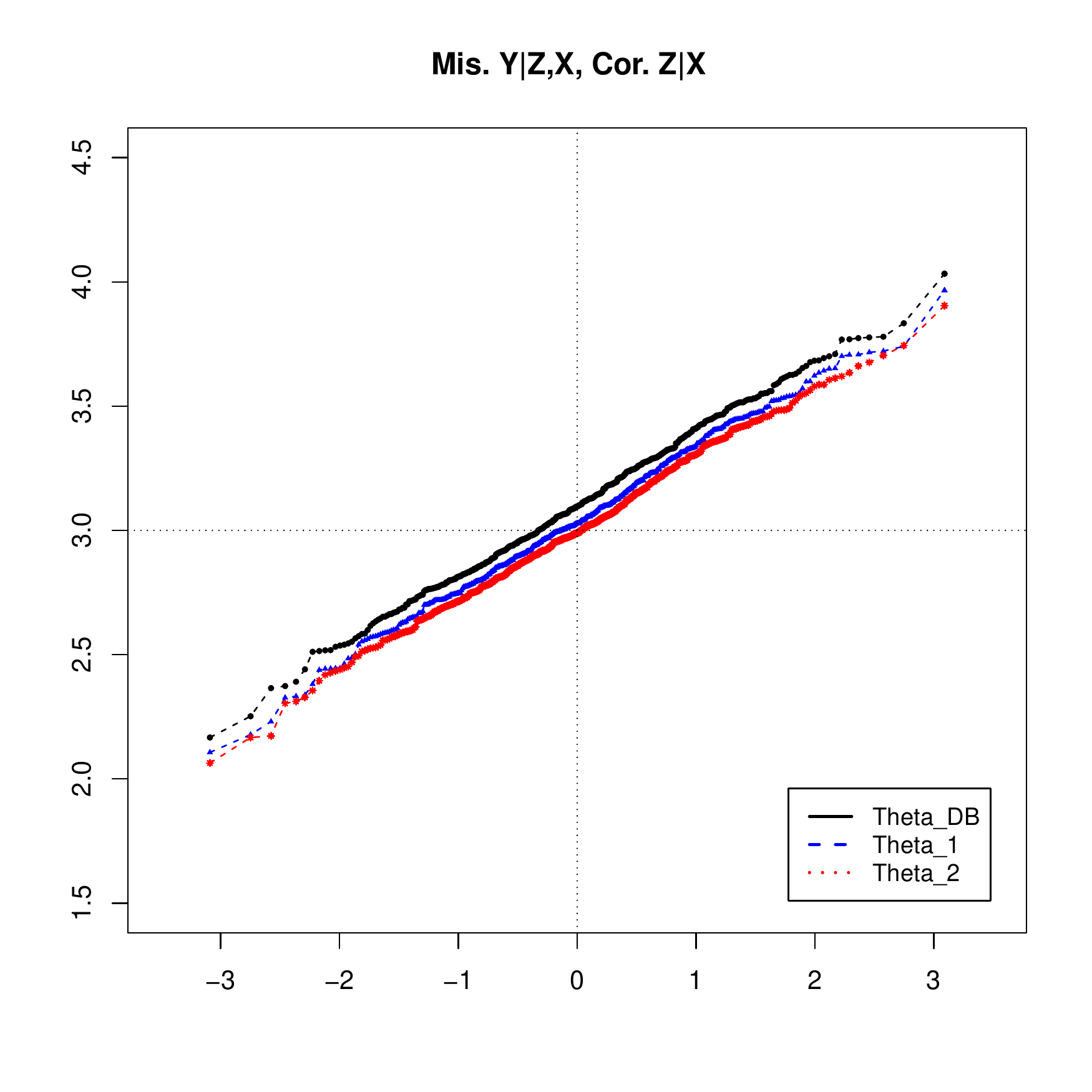} }}%
	\quad
	\subfloat{{\includegraphics[width=70mm]{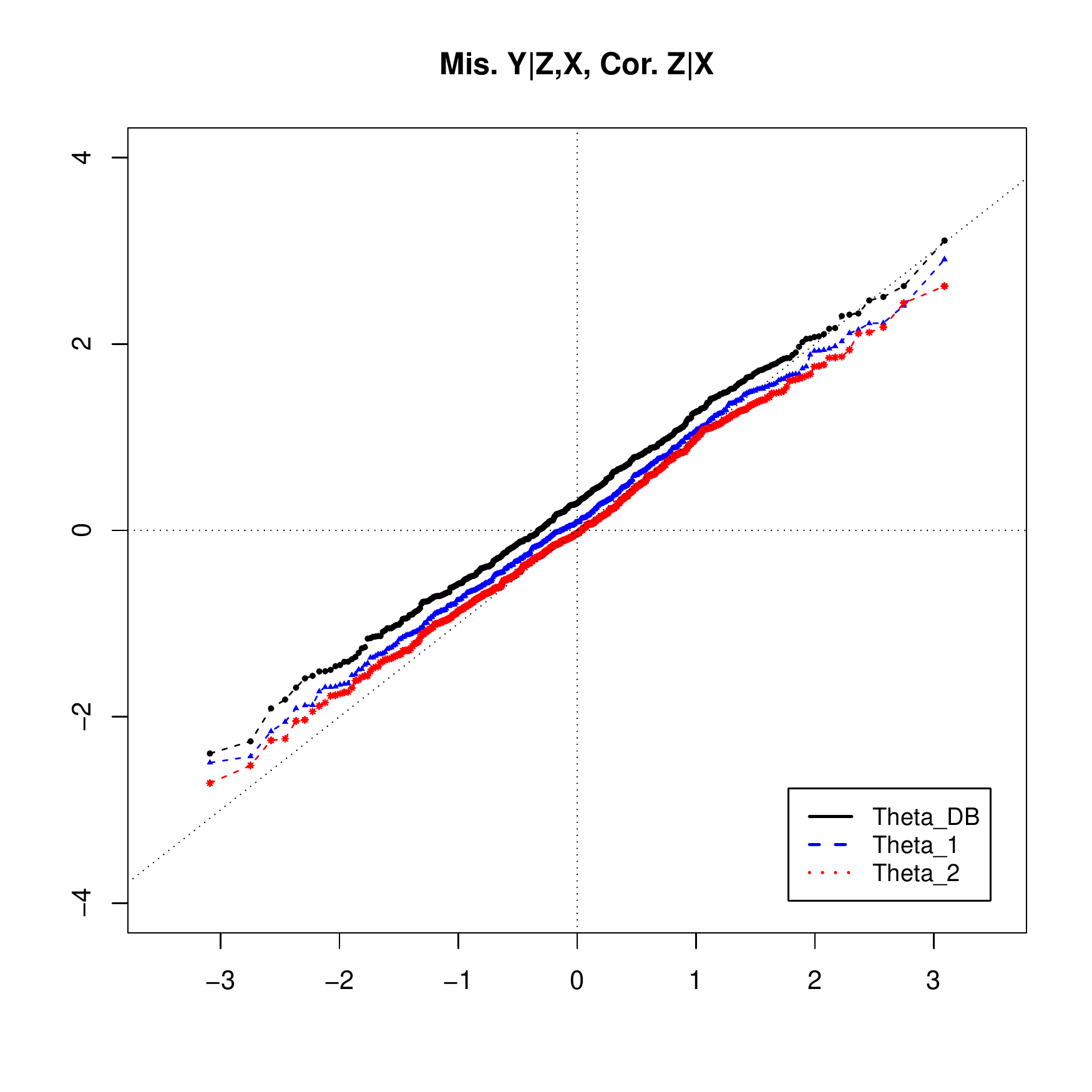} }}%
\vspace{-.25in}
	\subfloat{{\includegraphics[width=70mm]{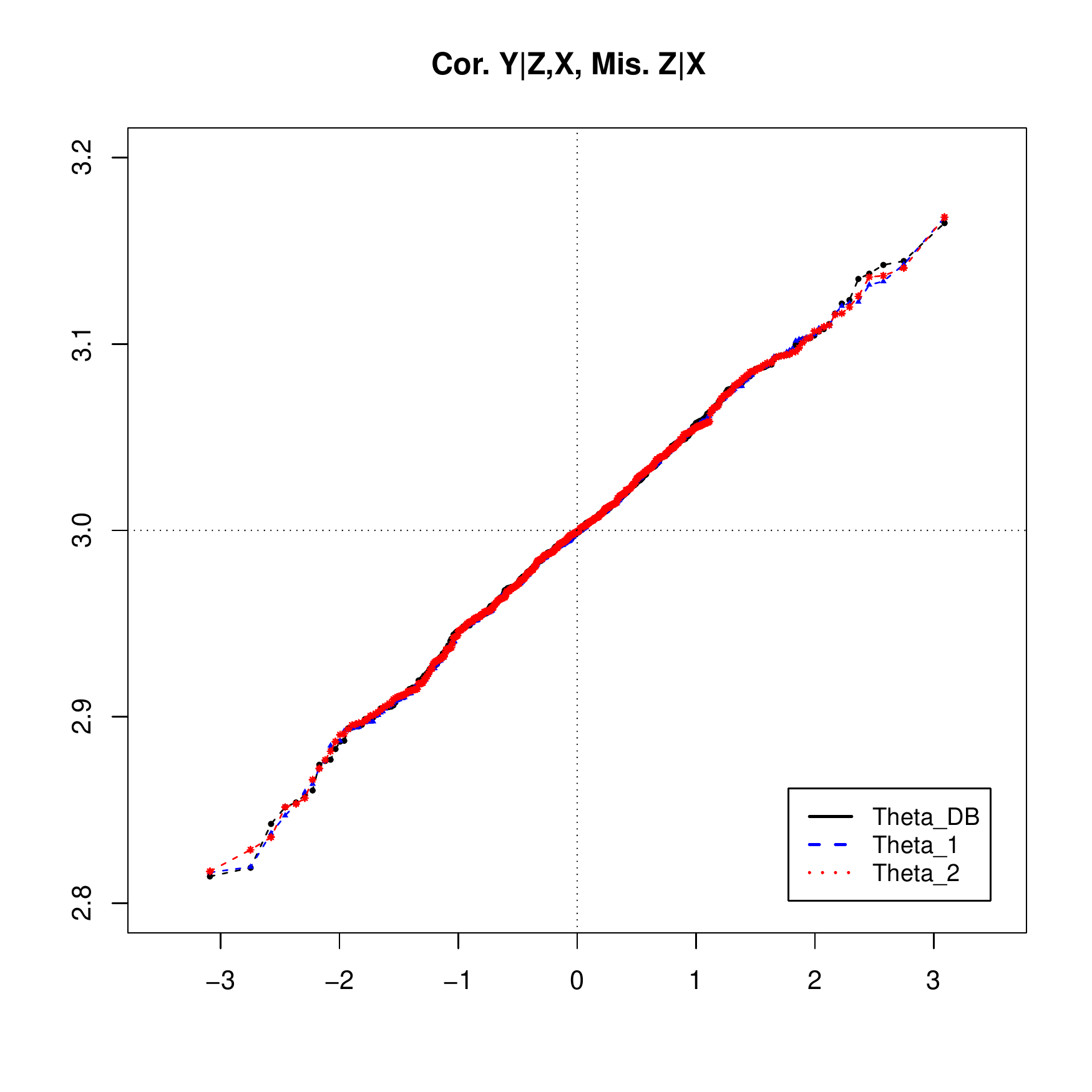} }}%
	\quad
	\subfloat{{\includegraphics[width=70mm]{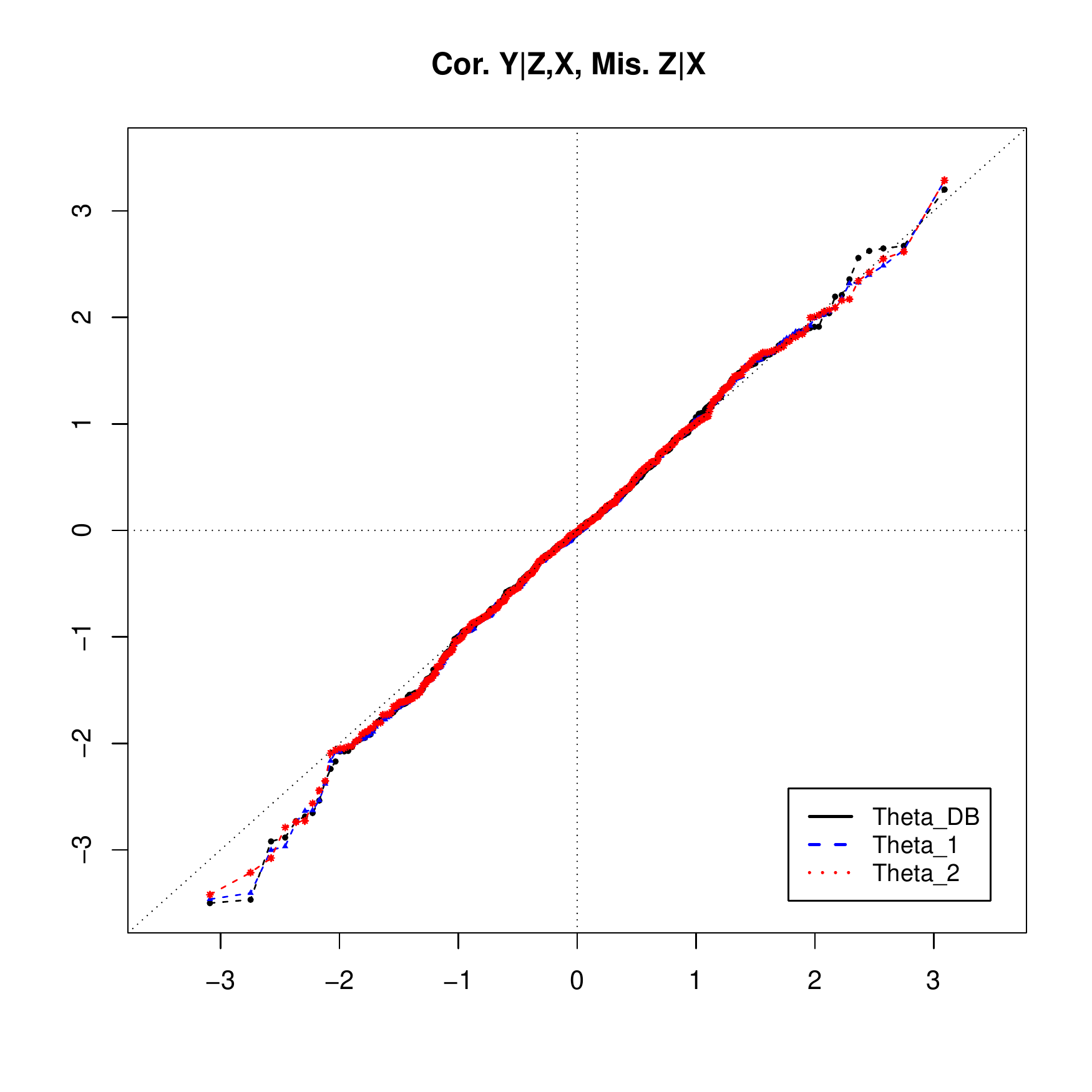} }}%
\vspace{-.25in}
	\subfloat{{\includegraphics[width=70mm]{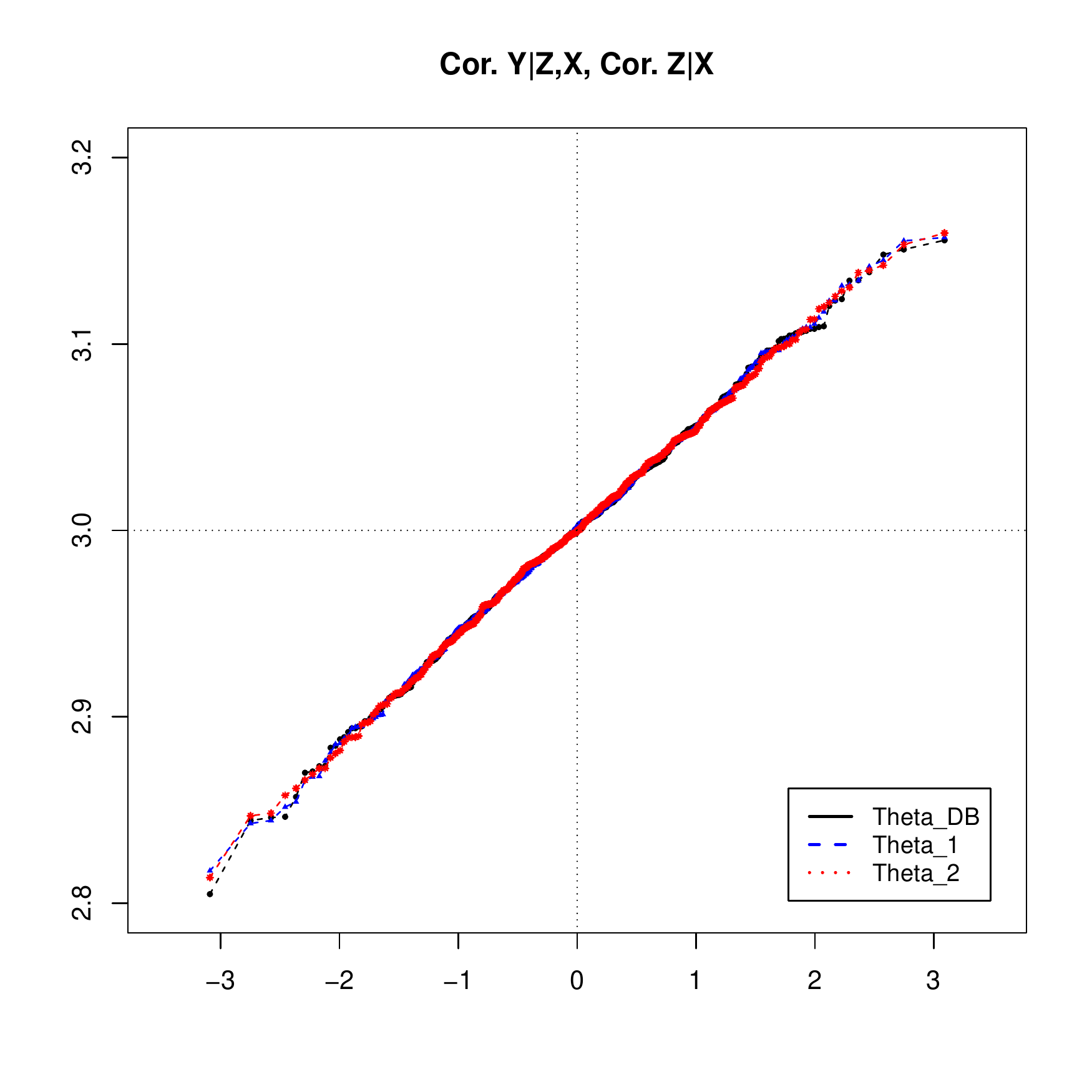} }}%
	\quad
	\subfloat{{\includegraphics[width=70mm]{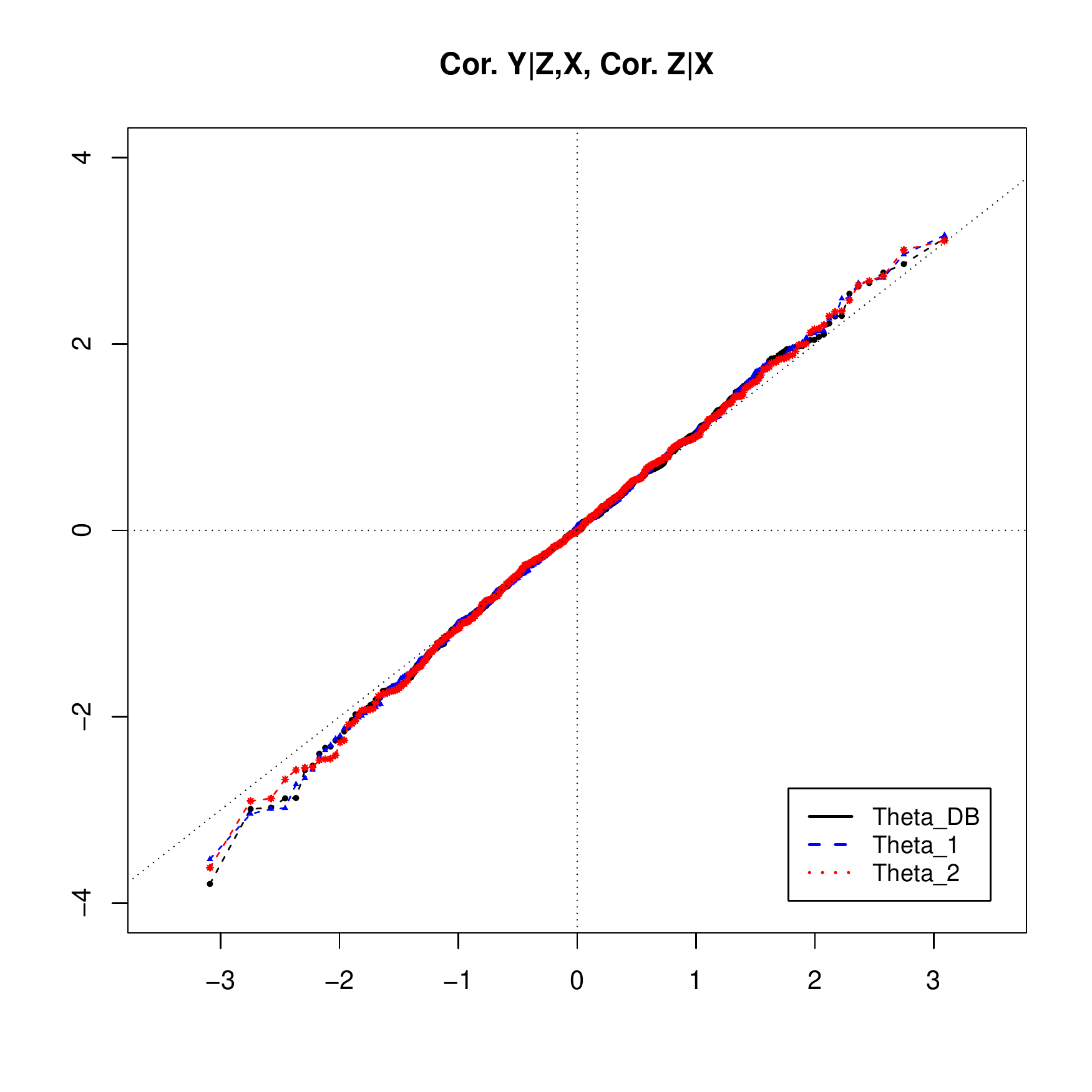} }}%
\end{figure}

\begin{figure}
	\caption{\small QQ plots of the estimates (first column) and $t$-statistics (second column) against standard normal ($n=400$, $p=200$) for partially linear modeling}
	\label{fig:linear_p=200}
	\centering
	\subfloat{{\includegraphics[width=70mm]{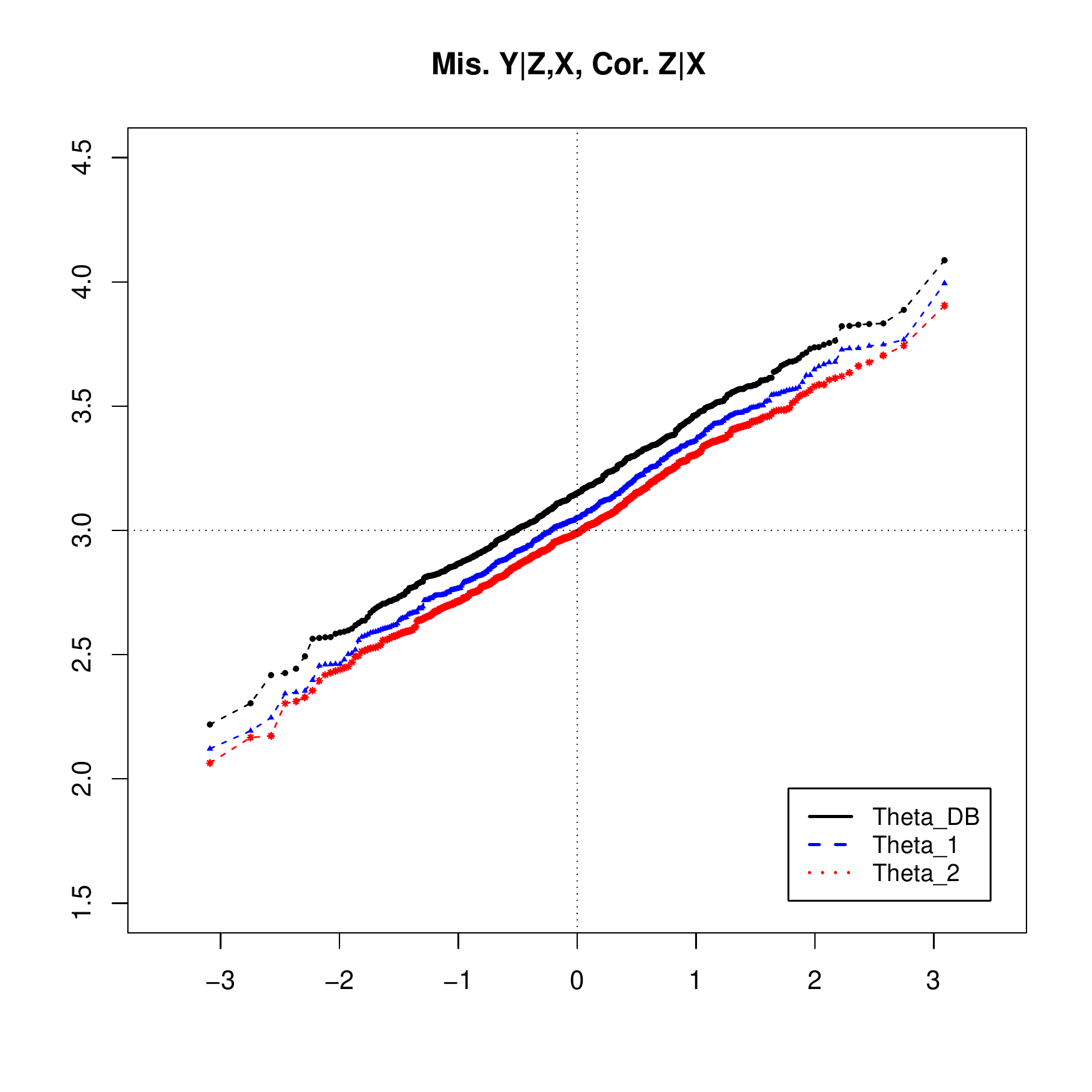} }}%
	\quad
	\subfloat{{\includegraphics[width=70mm]{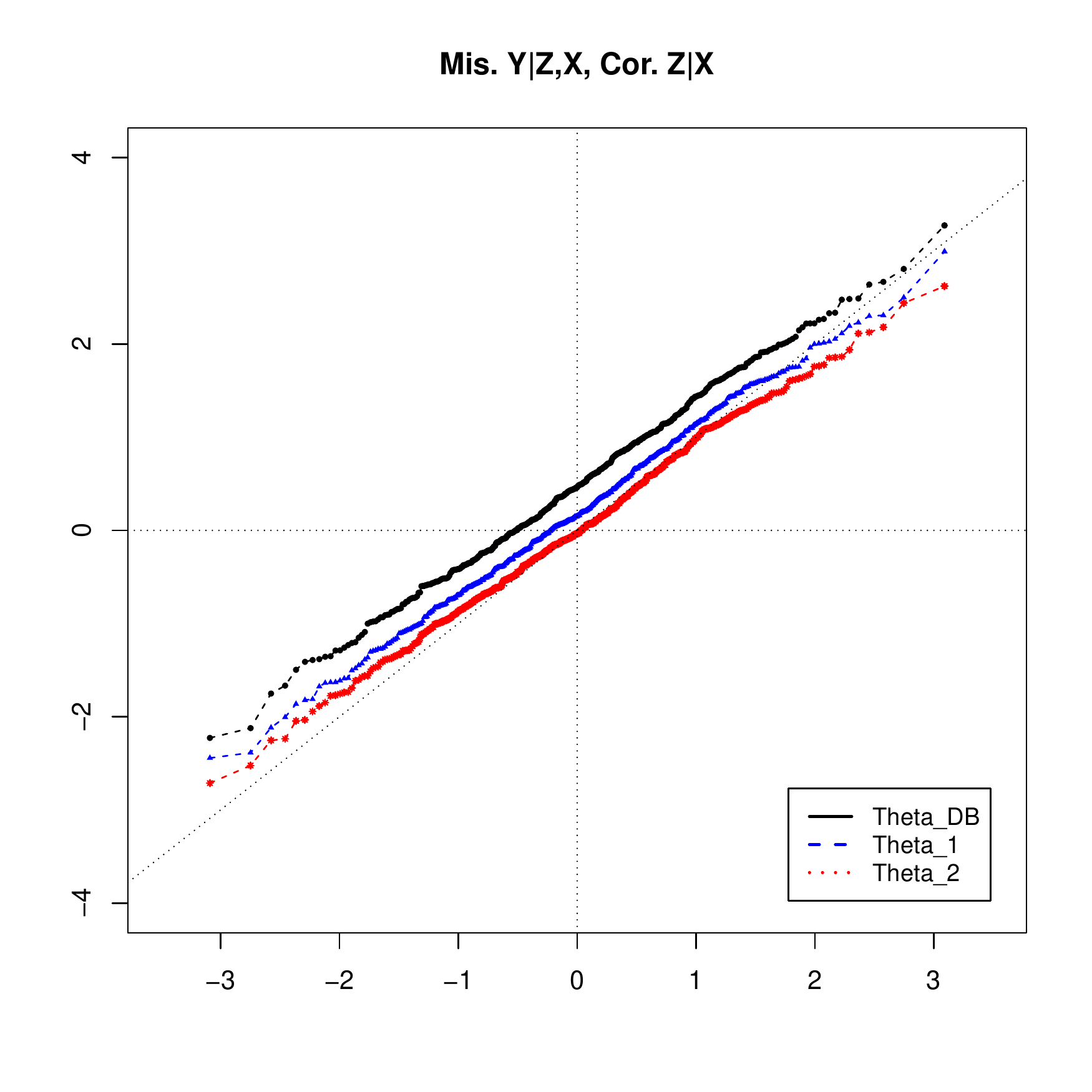} }}%
\vspace{-.25in}
	\subfloat{{\includegraphics[width=70mm]{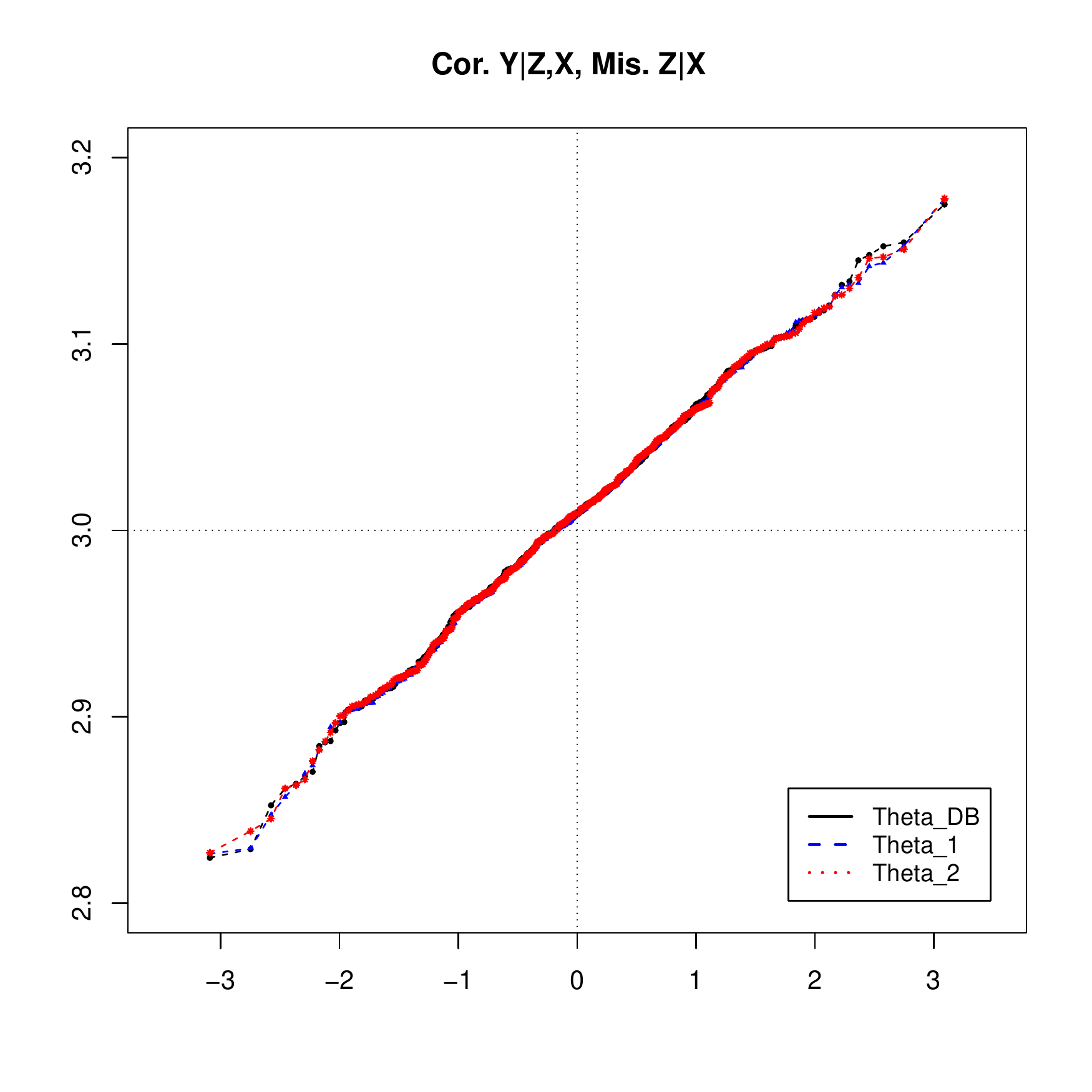} }}%
	\quad
	\subfloat{{\includegraphics[width=70mm]{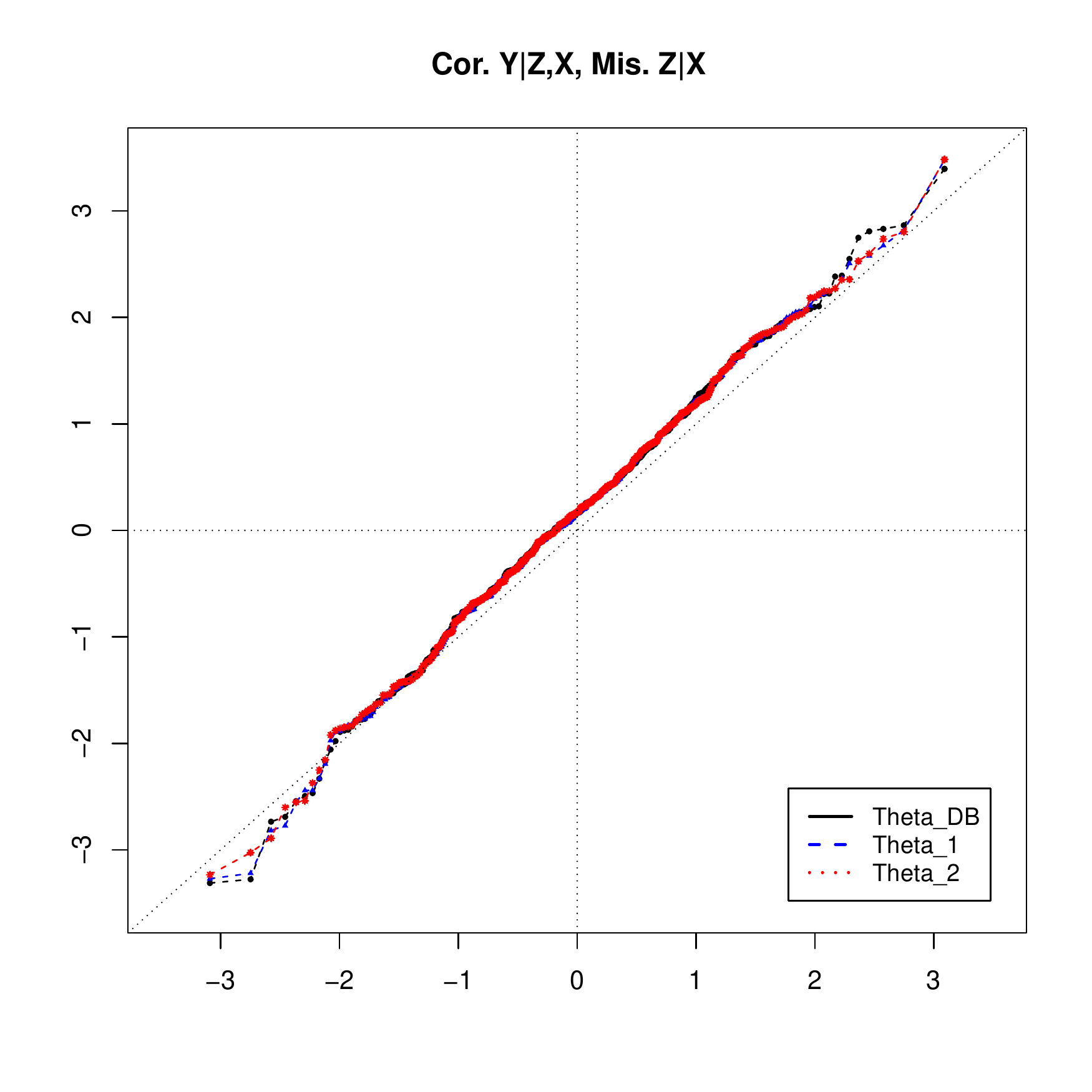} }}%
\vspace{-.25in}
	\subfloat{{\includegraphics[width=70mm]{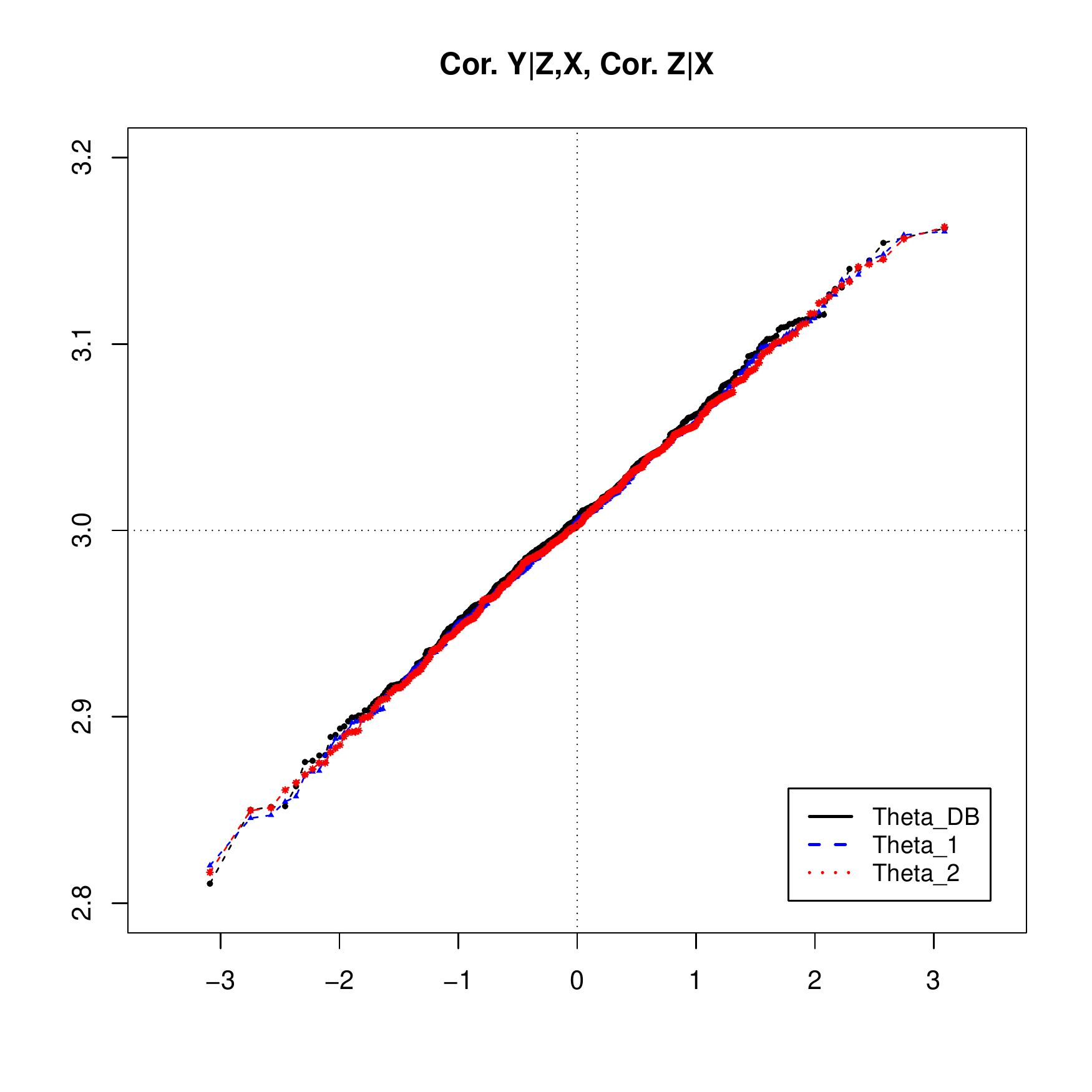} }}%
	\quad
	\subfloat{{\includegraphics[width=70mm]{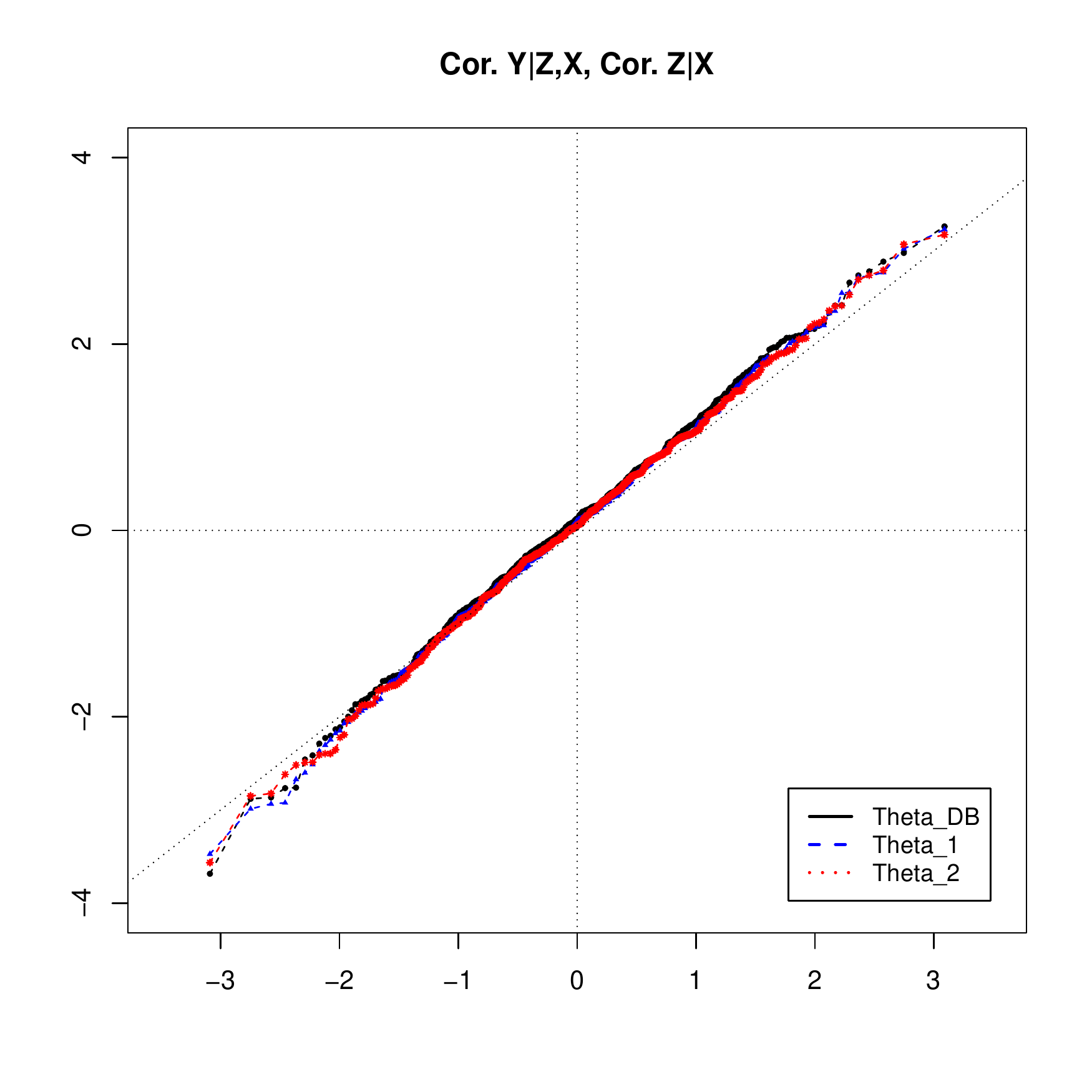} }}%
\end{figure}

\begin{figure}
	\caption{\small QQ plots of the estimates (first column) and $t$-statistics (second column) against standard normal ($n=400$, $p=800$) for partially linear modeling}
	\label{fig:linear_p=800}
\centering
\subfloat{{\includegraphics[width=70mm]{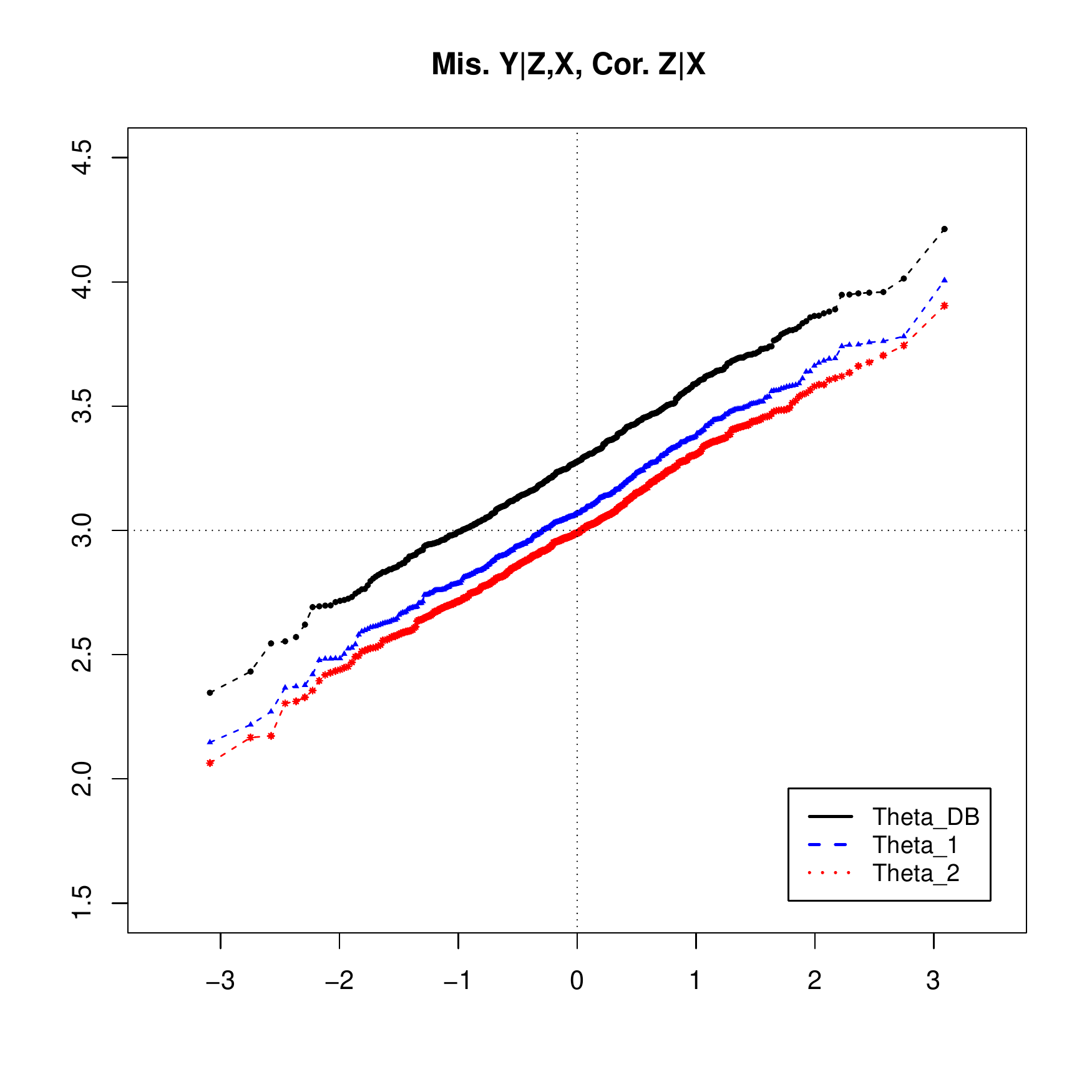} }}%
\quad
\subfloat{{\includegraphics[width=70mm]{QQ/linear_miss_cor_p=800_t.pdf} }}%
\vspace{-.25in}
\subfloat{{\includegraphics[width=70mm]{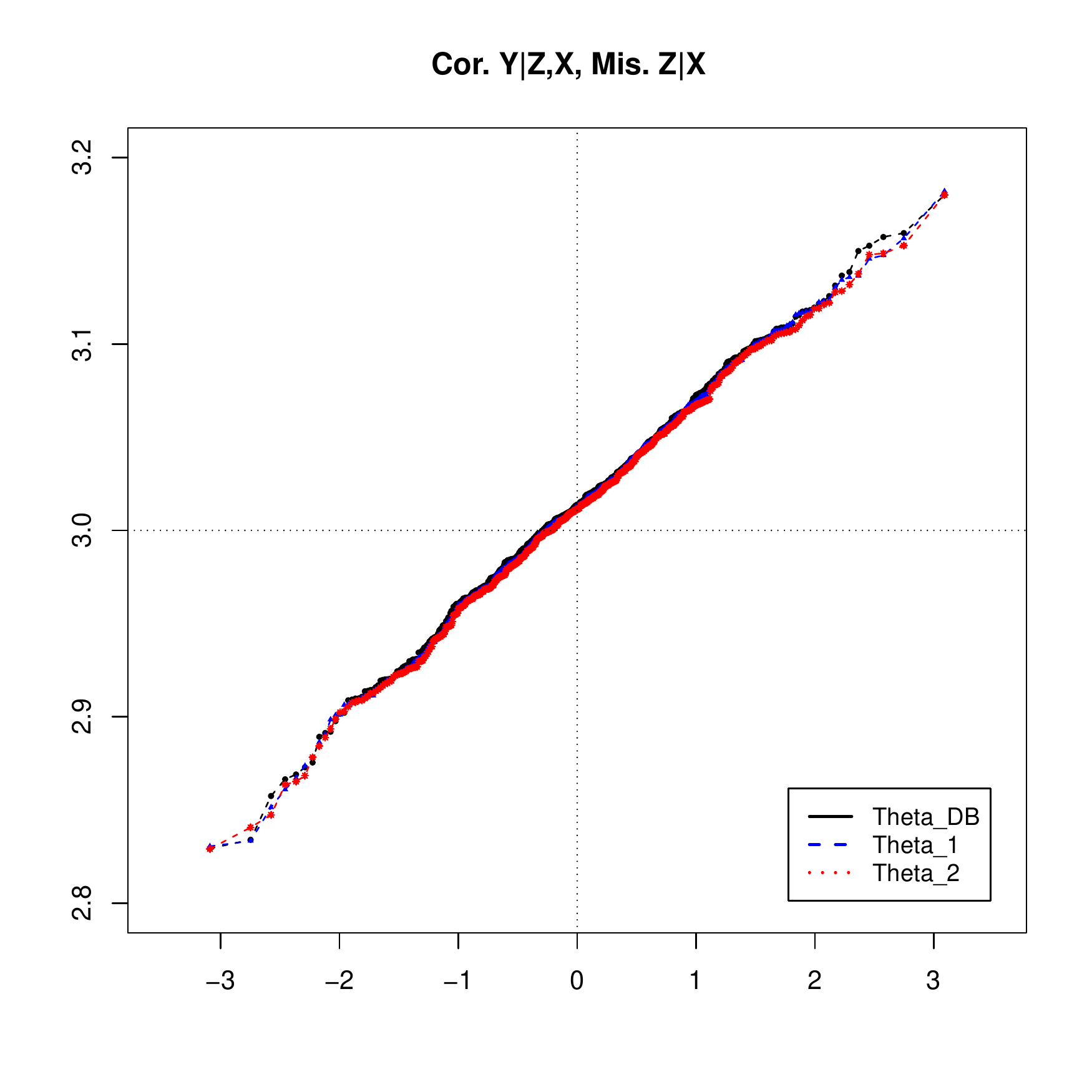} }}%
\quad
\subfloat{{\includegraphics[width=70mm]{QQ/linear_cor_miss_p=800_t.pdf} }}%
\vspace{-.25in}
\subfloat{{\includegraphics[width=70mm]{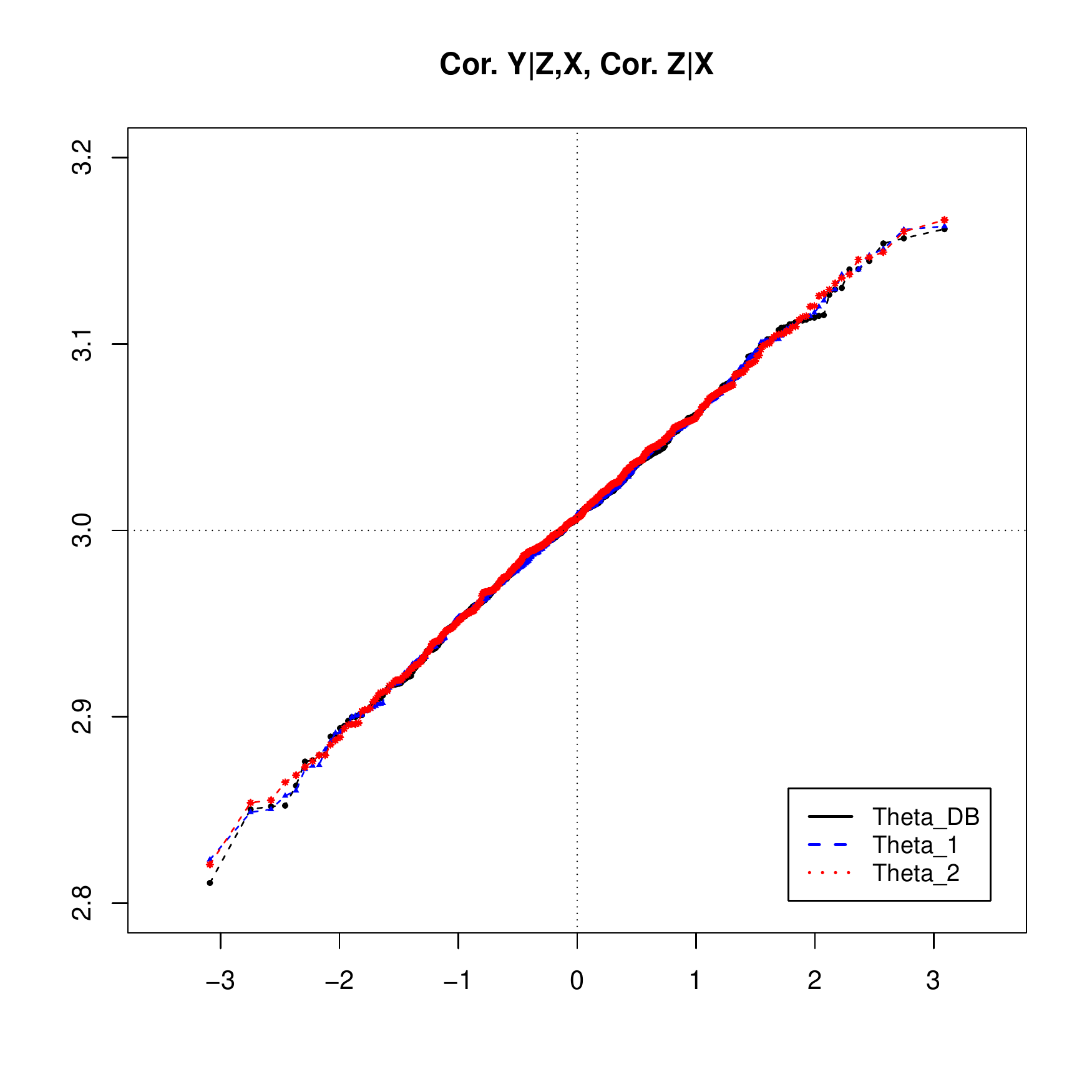} }}%
\quad
\subfloat{{\includegraphics[width=70mm]{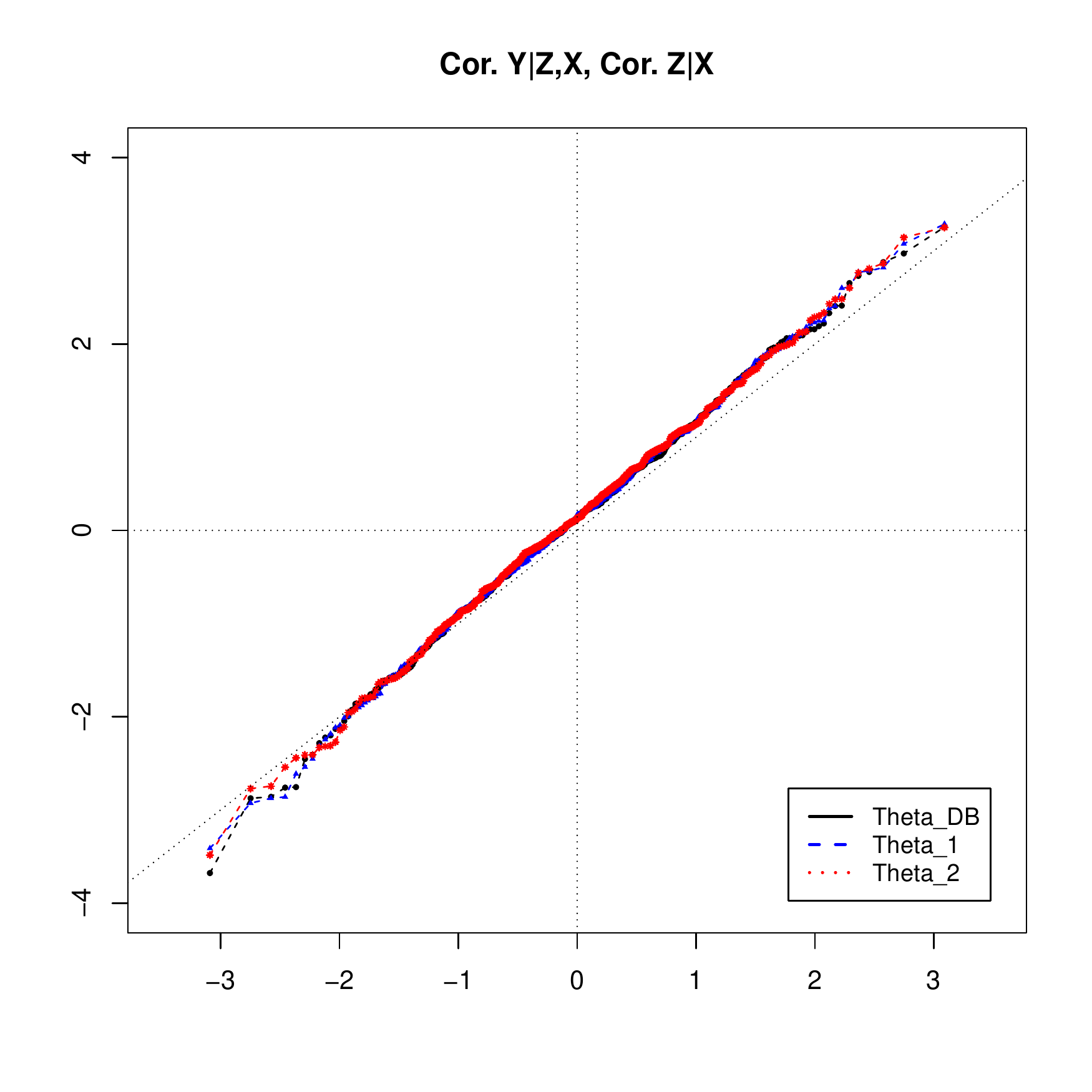} }}%
\end{figure}

\begin{figure}
	\caption{\small QQ plots of the estimates (first column) and $t$-statistics (second column) against standard normal ($n=600$, $p=100$) for partially log-linear modeling}
	\label{fig:log-linear_p=100}
	\centering
	\subfloat{{\includegraphics[width=70mm]{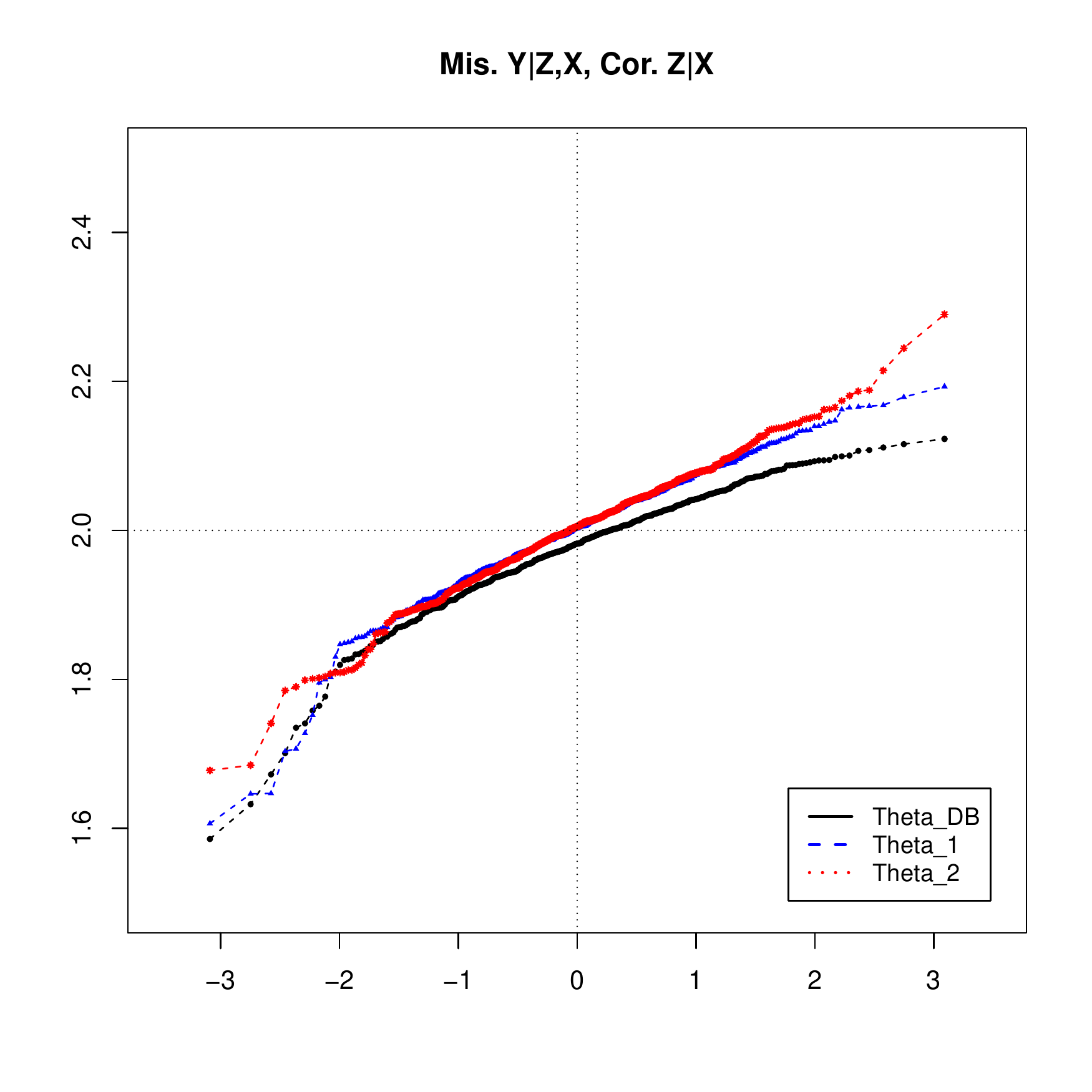} }}%
	\quad
	\subfloat{{\includegraphics[width=70mm]{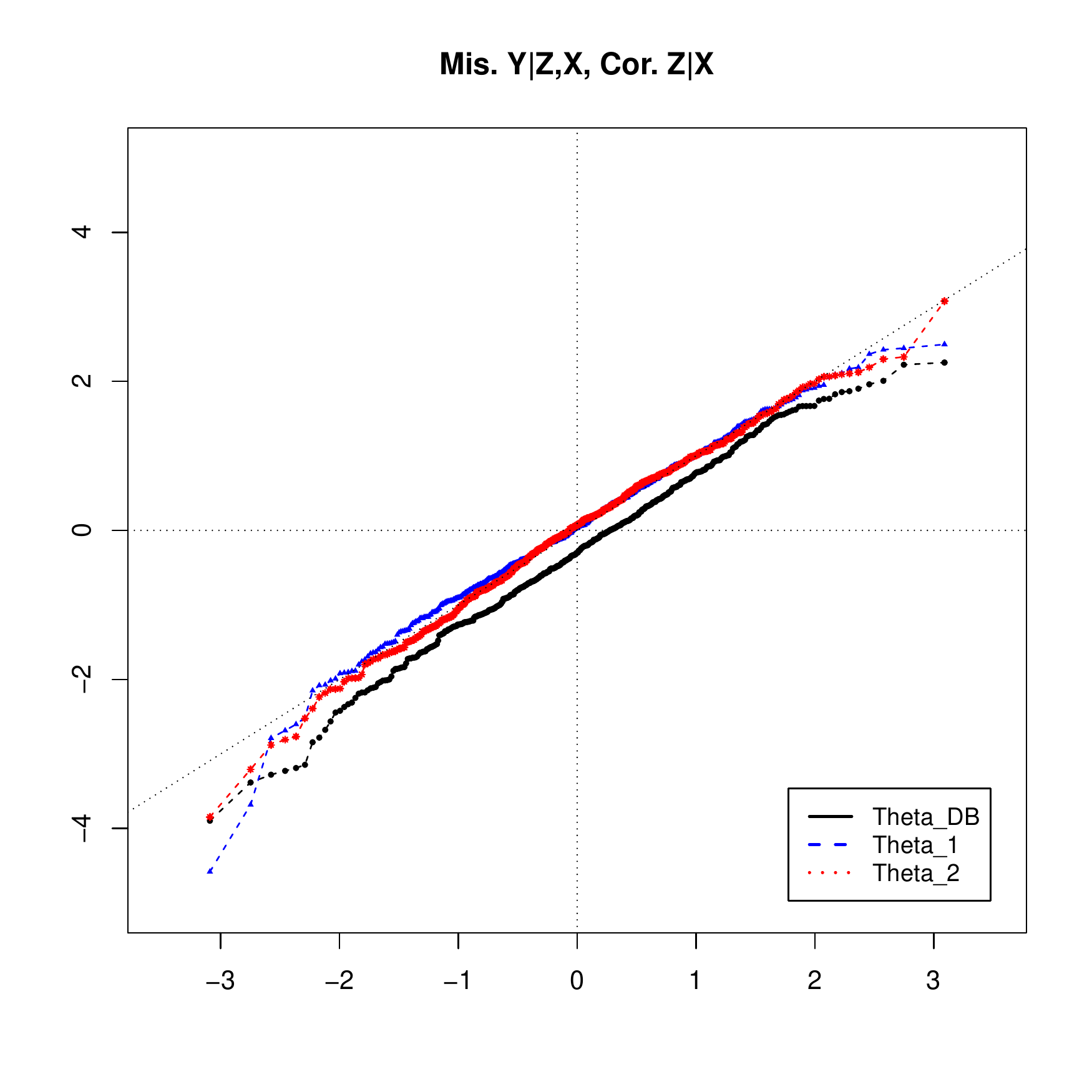} }}%
\vspace{-.25in}
	\subfloat{{\includegraphics[width=70mm]{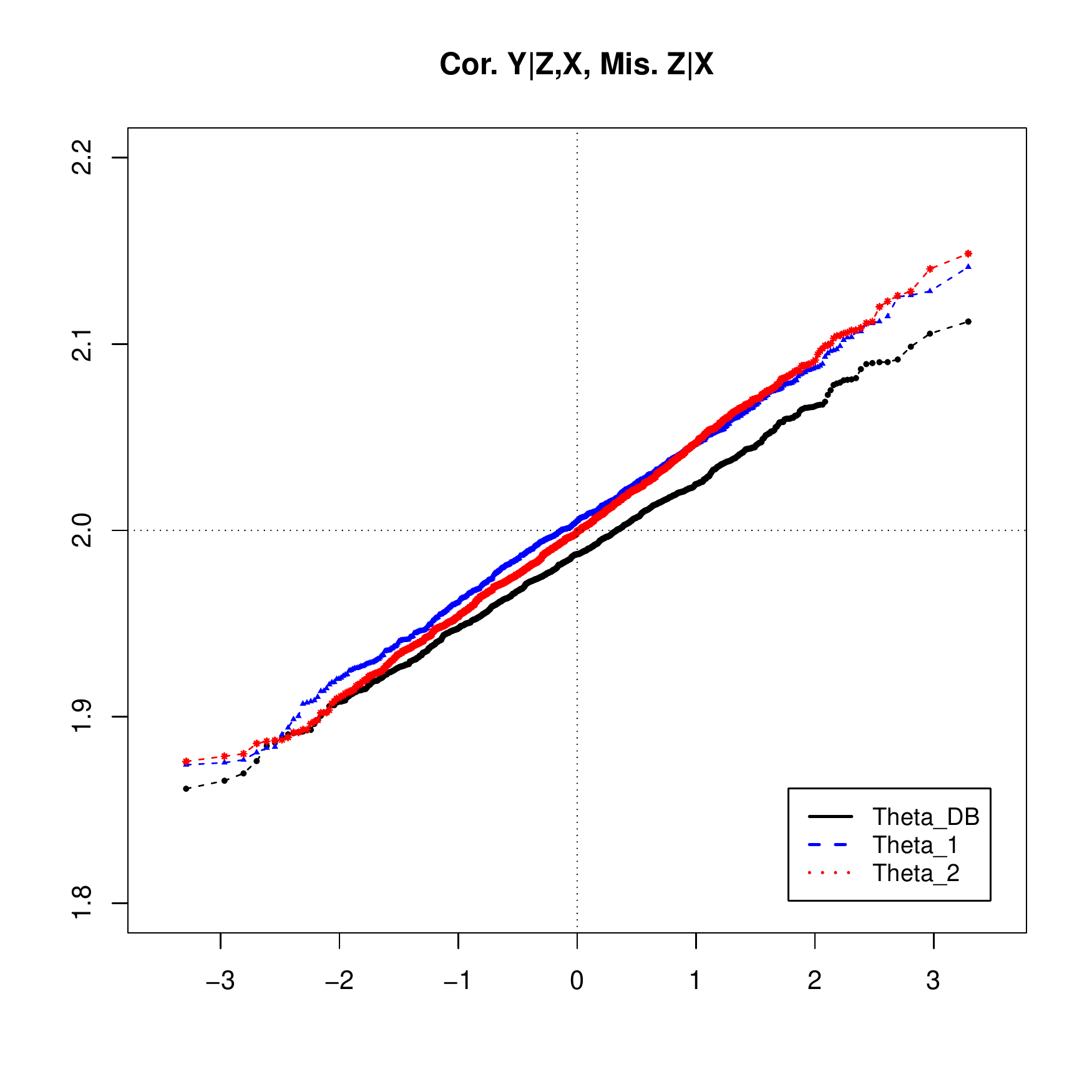} }}%
	\quad
	\subfloat{{\includegraphics[width=70mm]{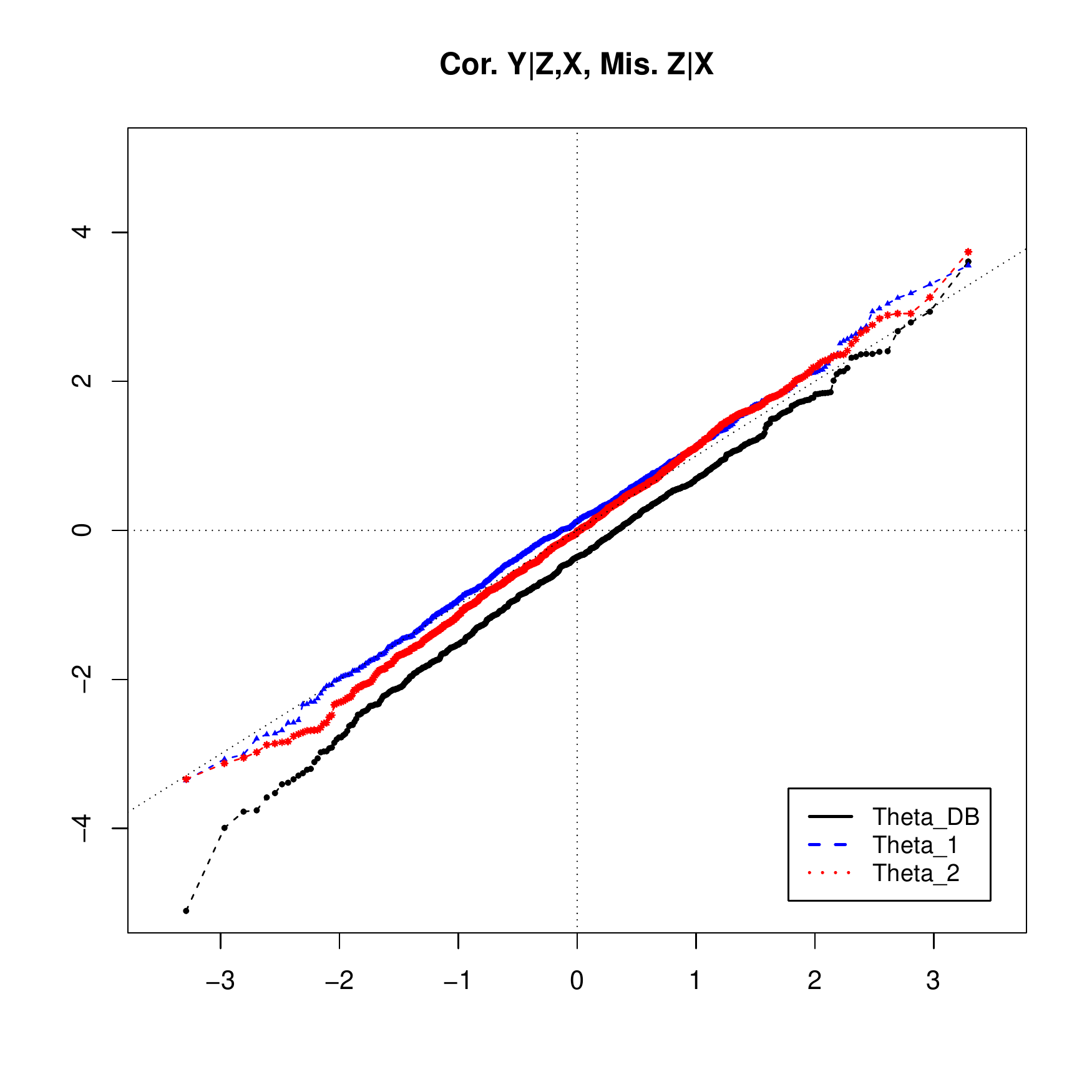} }}%
\vspace{-.25in}
	\subfloat{{\includegraphics[width=70mm]{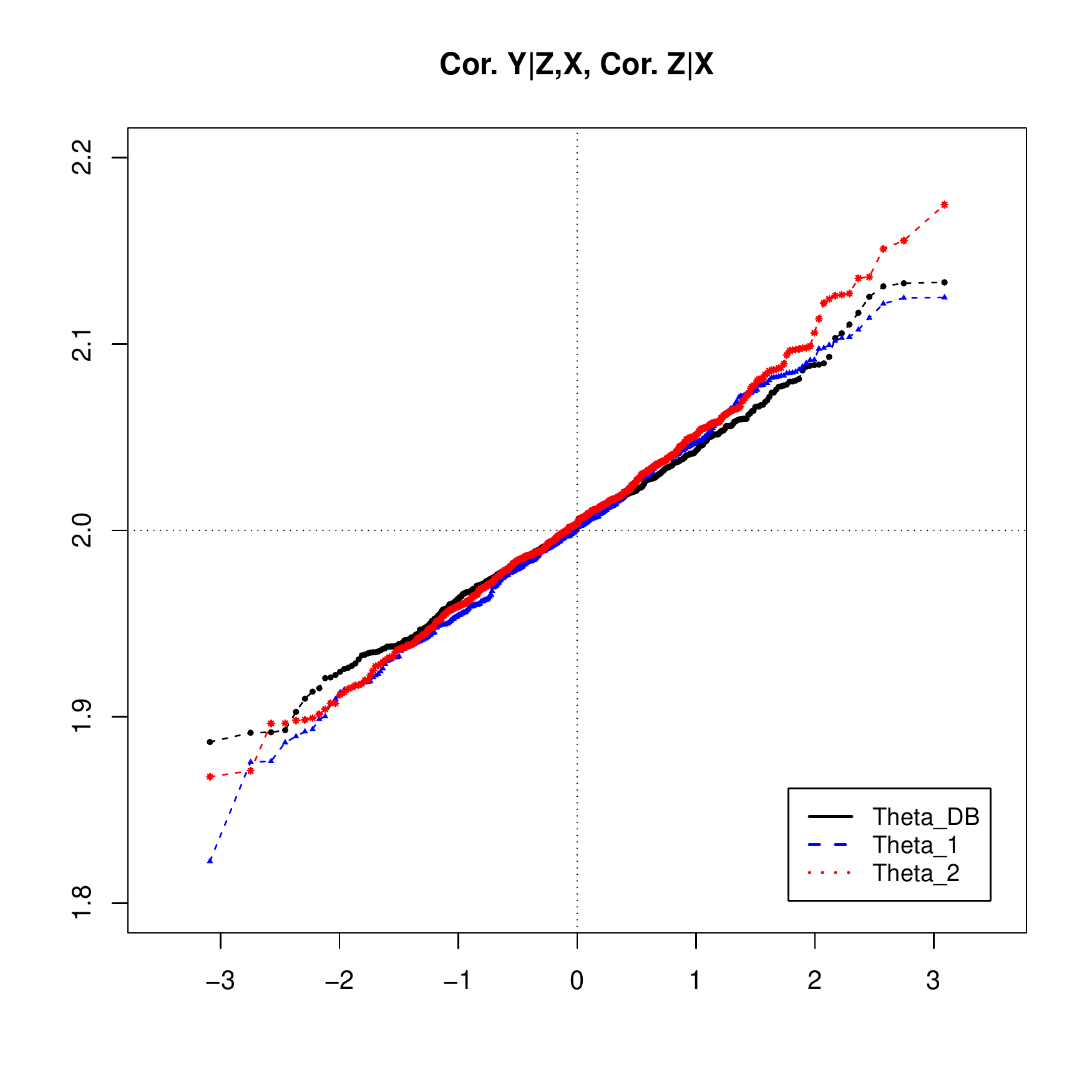} }}%
	\quad
	\subfloat{{\includegraphics[width=70mm]{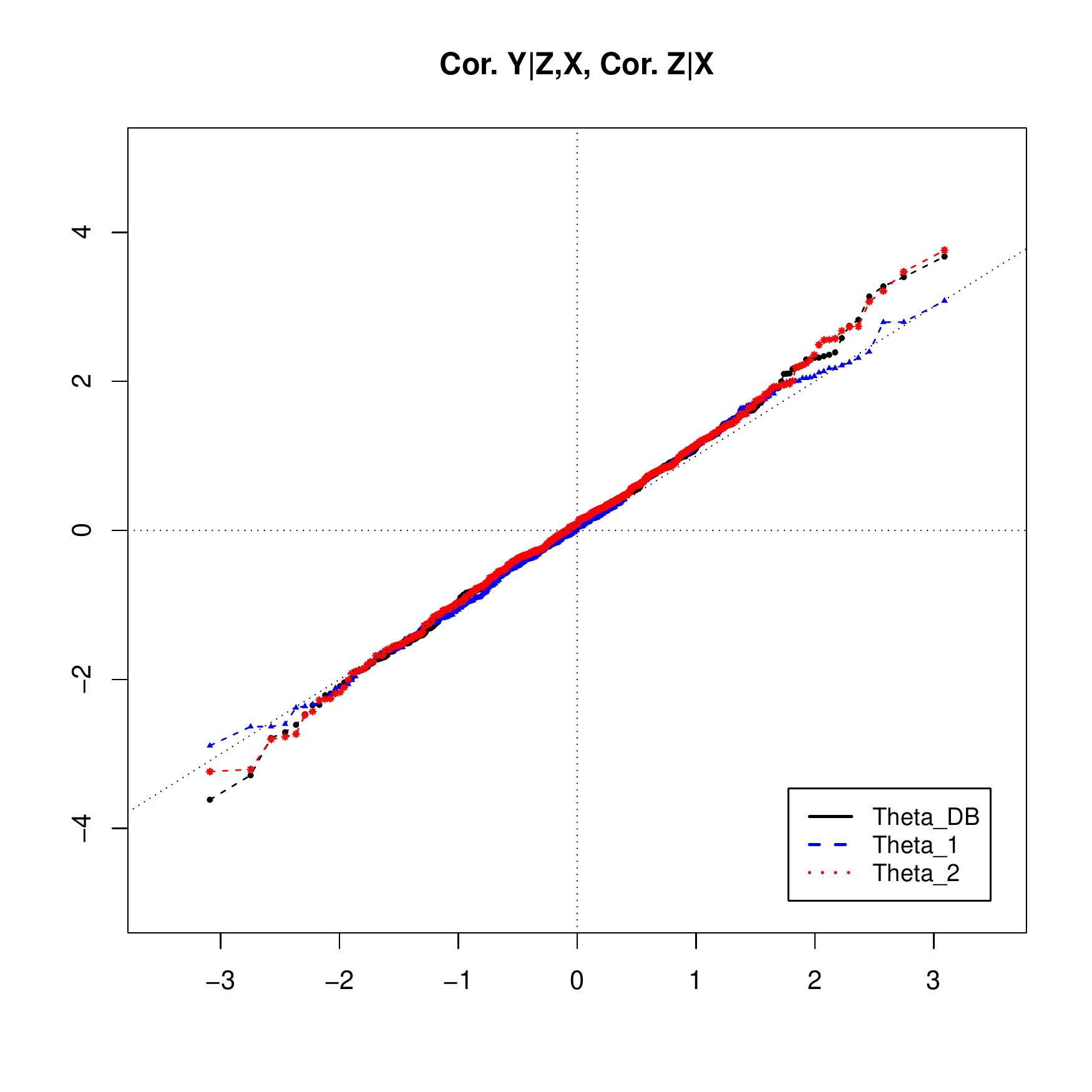} }}%
\end{figure}

\begin{figure}
	\caption{\small QQ plots of the estimates (first column) and $t$-statistics (second column) against standard normal ($n=600$, $p=200$) for partially log-linear modeling}
	\label{fig:log-linear_p=200}
	\centering
	\subfloat{{\includegraphics[width=70mm]{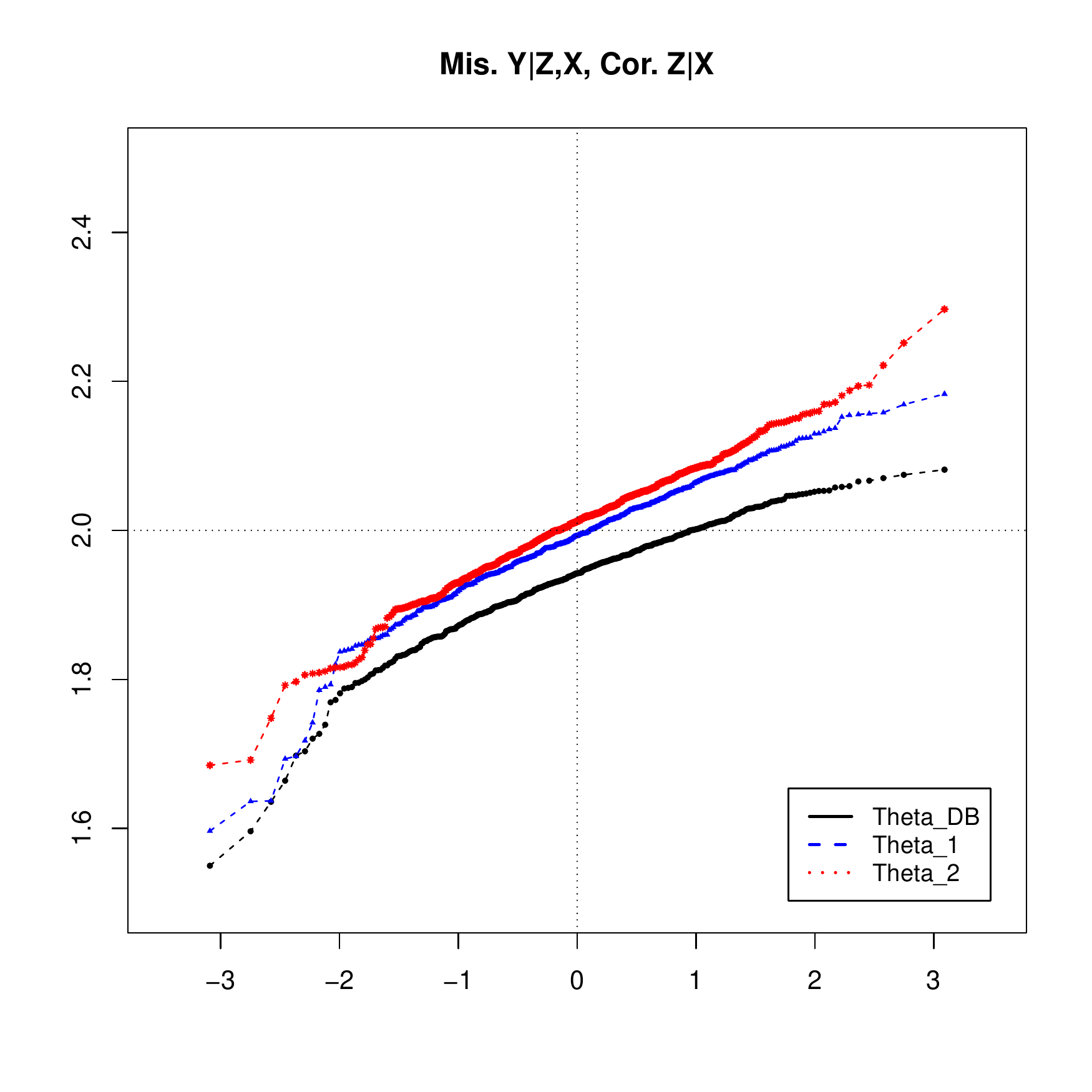} }}%
	\quad
	\subfloat{{\includegraphics[width=70mm]{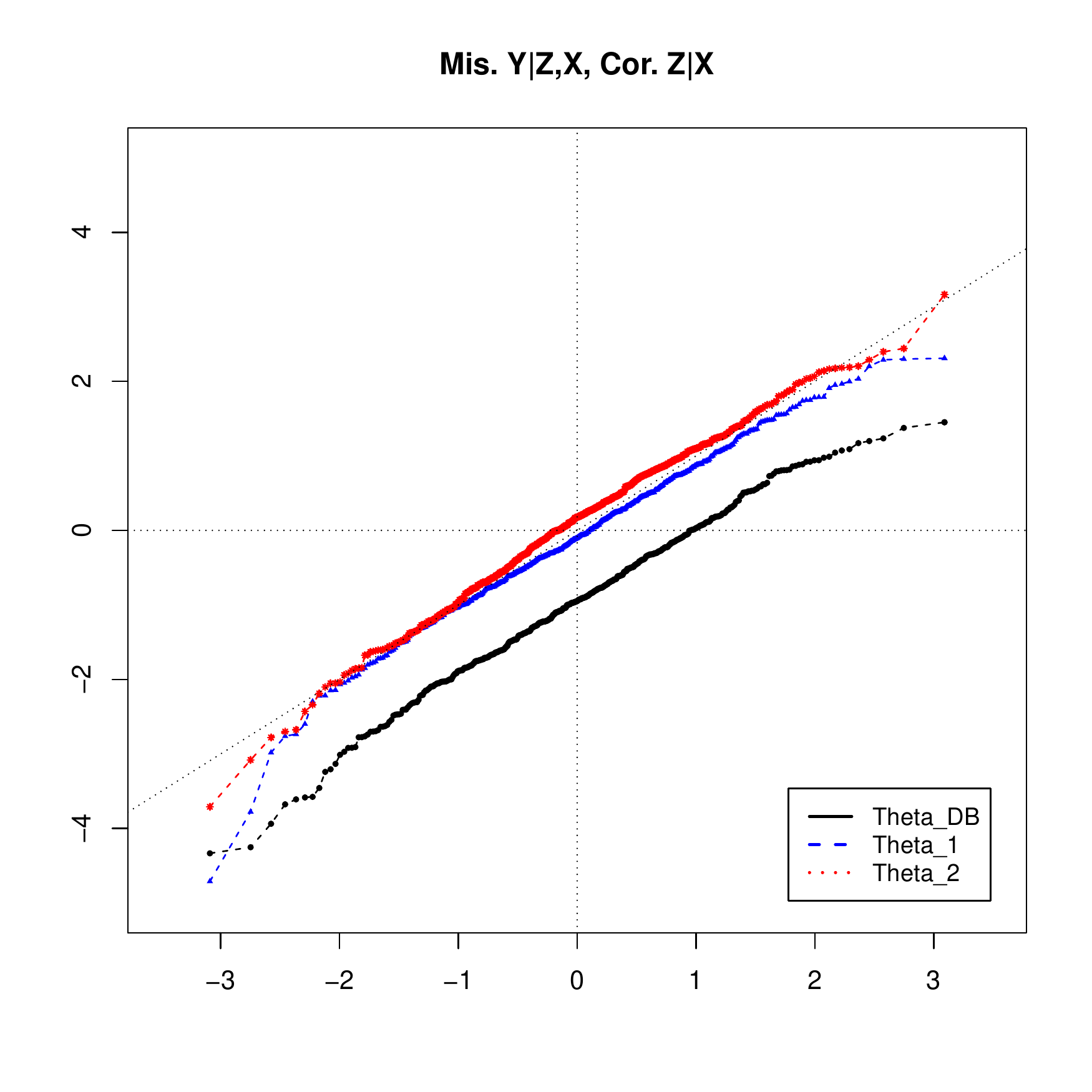} }}%
\vspace{-.25in}
	\subfloat{{\includegraphics[width=70mm]{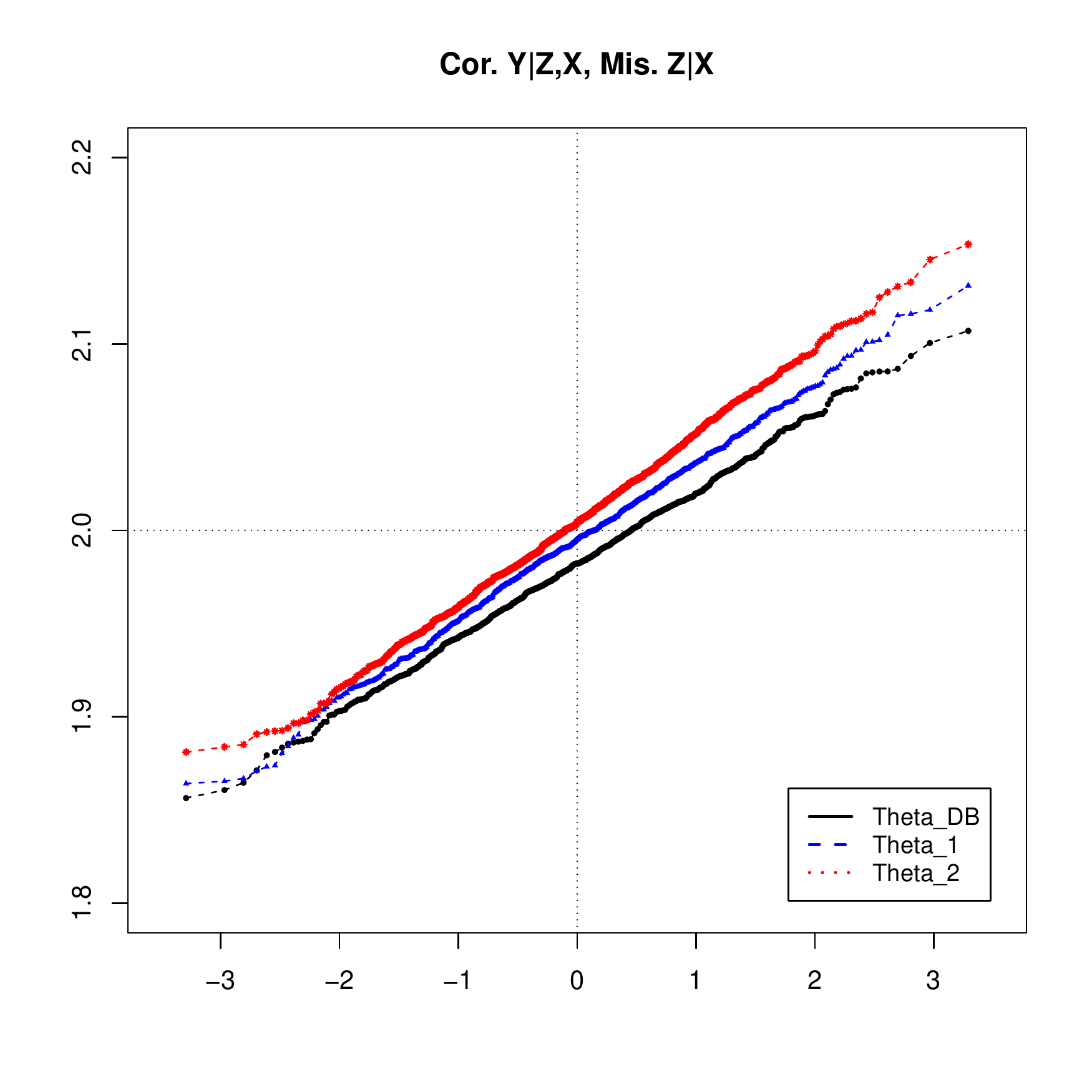} }}%
	\quad
	\subfloat{{\includegraphics[width=70mm]{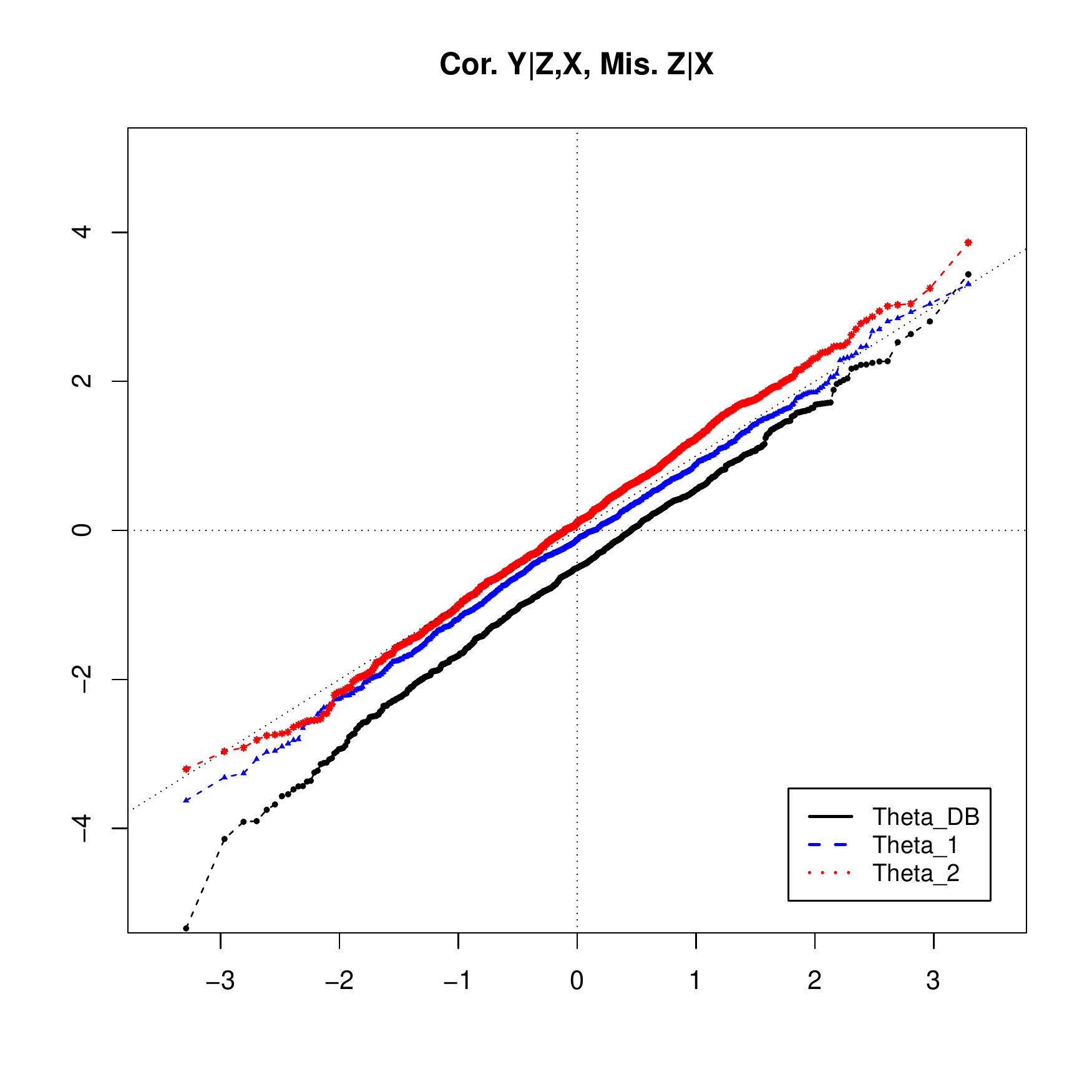} }}%
\vspace{-.25in}
	\subfloat{{\includegraphics[width=70mm]{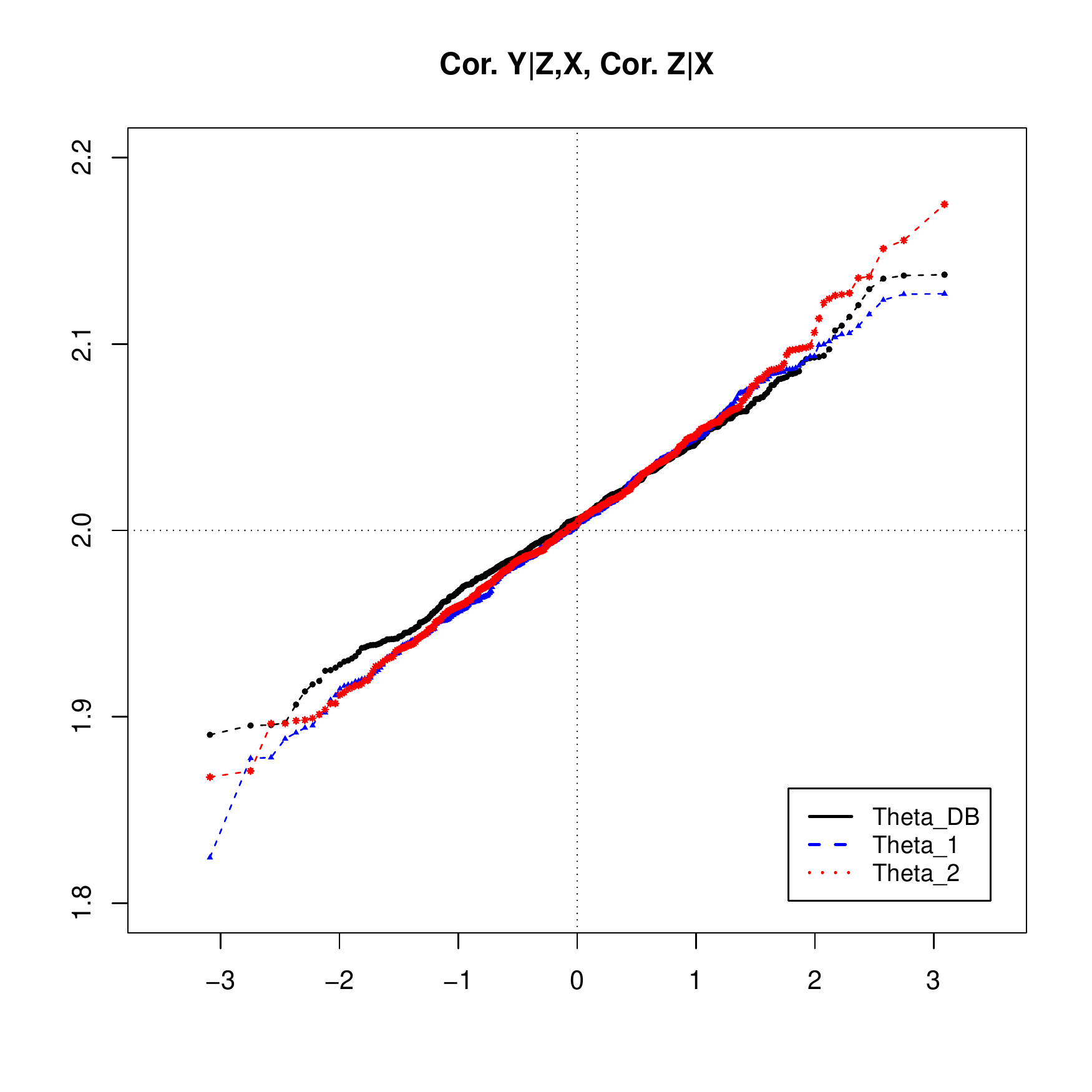} }}%
	\quad
	\subfloat{{\includegraphics[width=70mm]{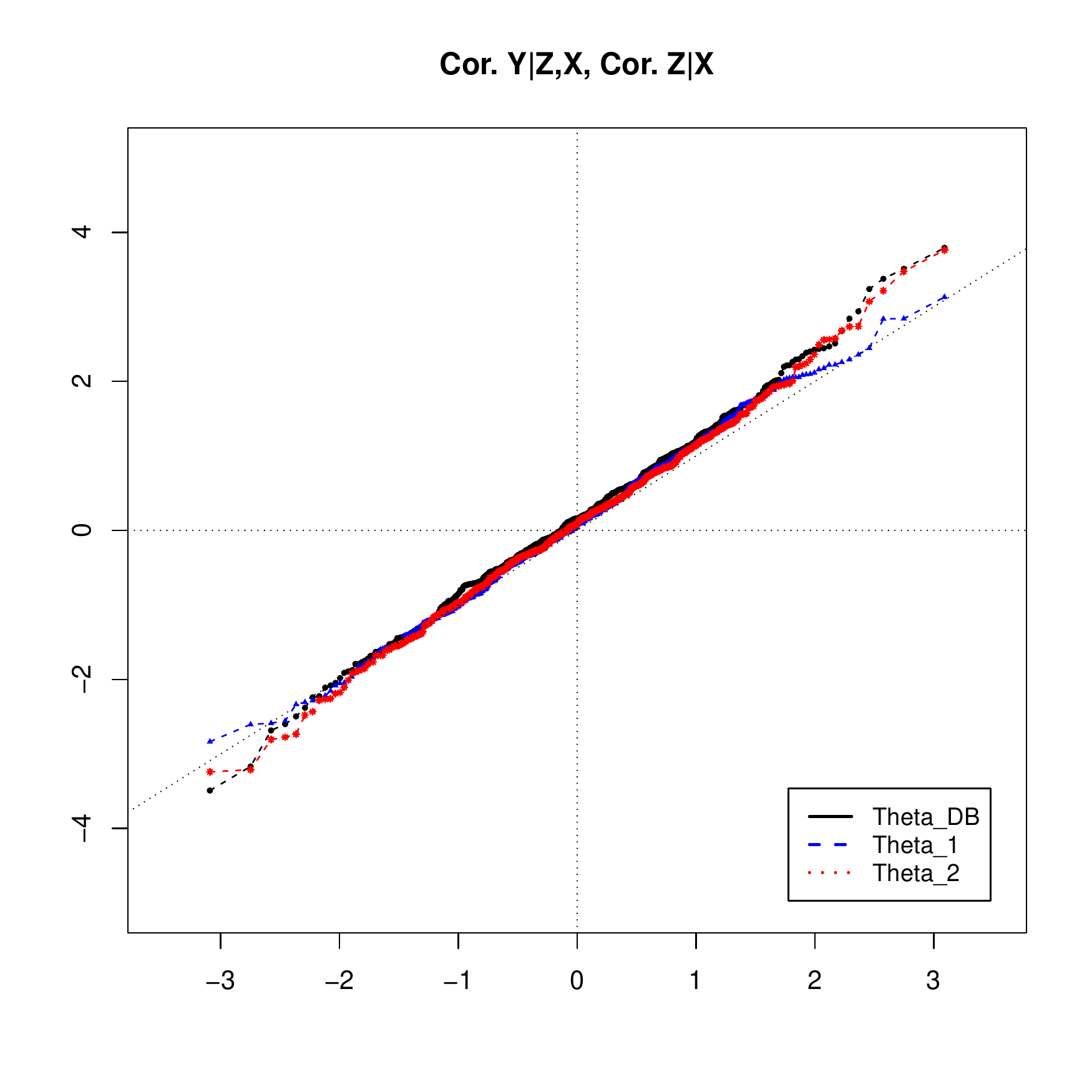} }}%
\end{figure}

\begin{figure}
	\caption{\small QQ plots of the estimates (first column) and $t$-statistics (second column) against standard normal ($n=600$, $p=800$) for partially log-linear modeling}
	\label{fig:log-linear_p=800}
	\centering
	\subfloat{{\includegraphics[width=70mm]{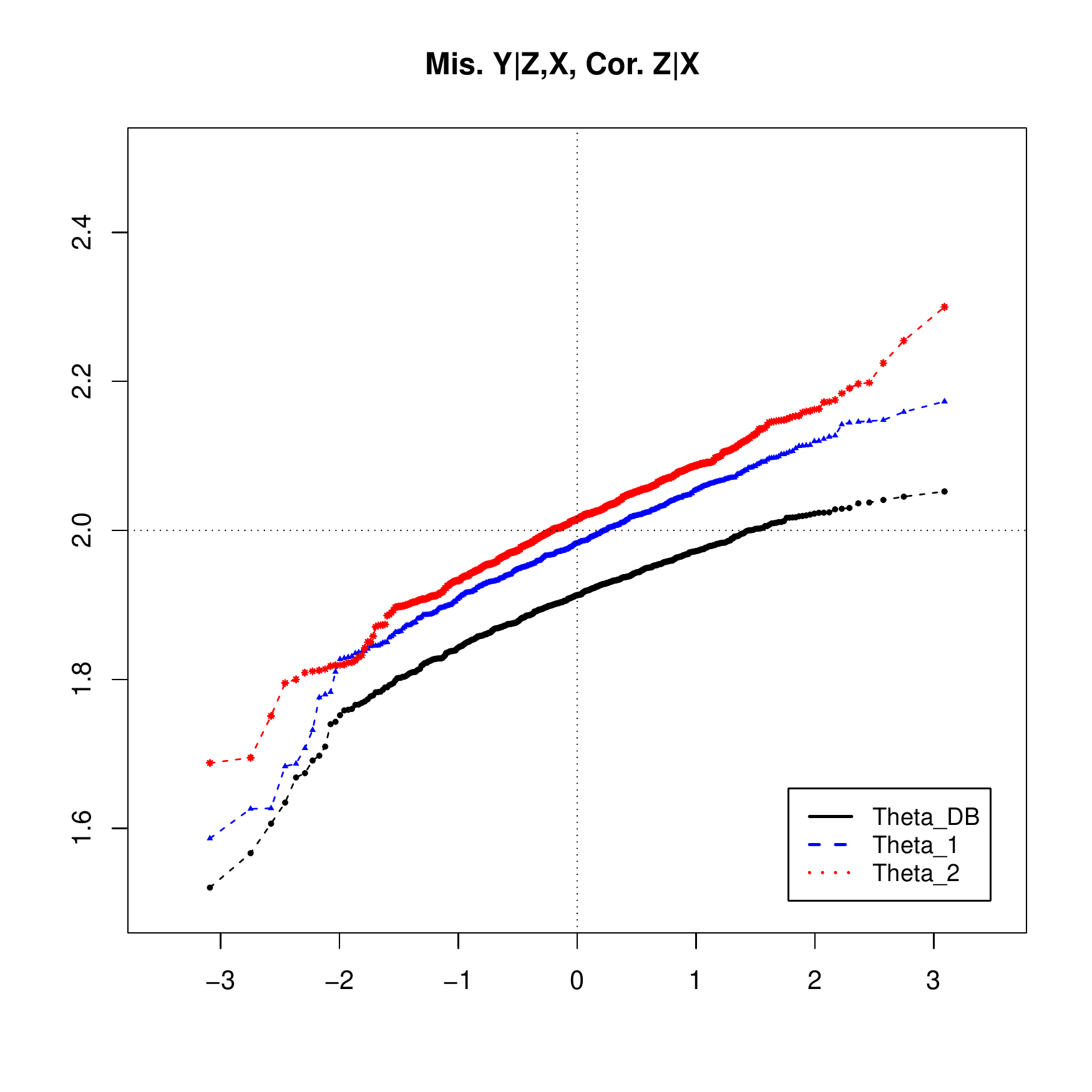} }}%
	\quad
	\subfloat{{\includegraphics[width=70mm]{QQ/log-linear_miss_cor_p=800_t.pdf} }}%
\vspace{-.25in}
	\subfloat{{\includegraphics[width=70mm]{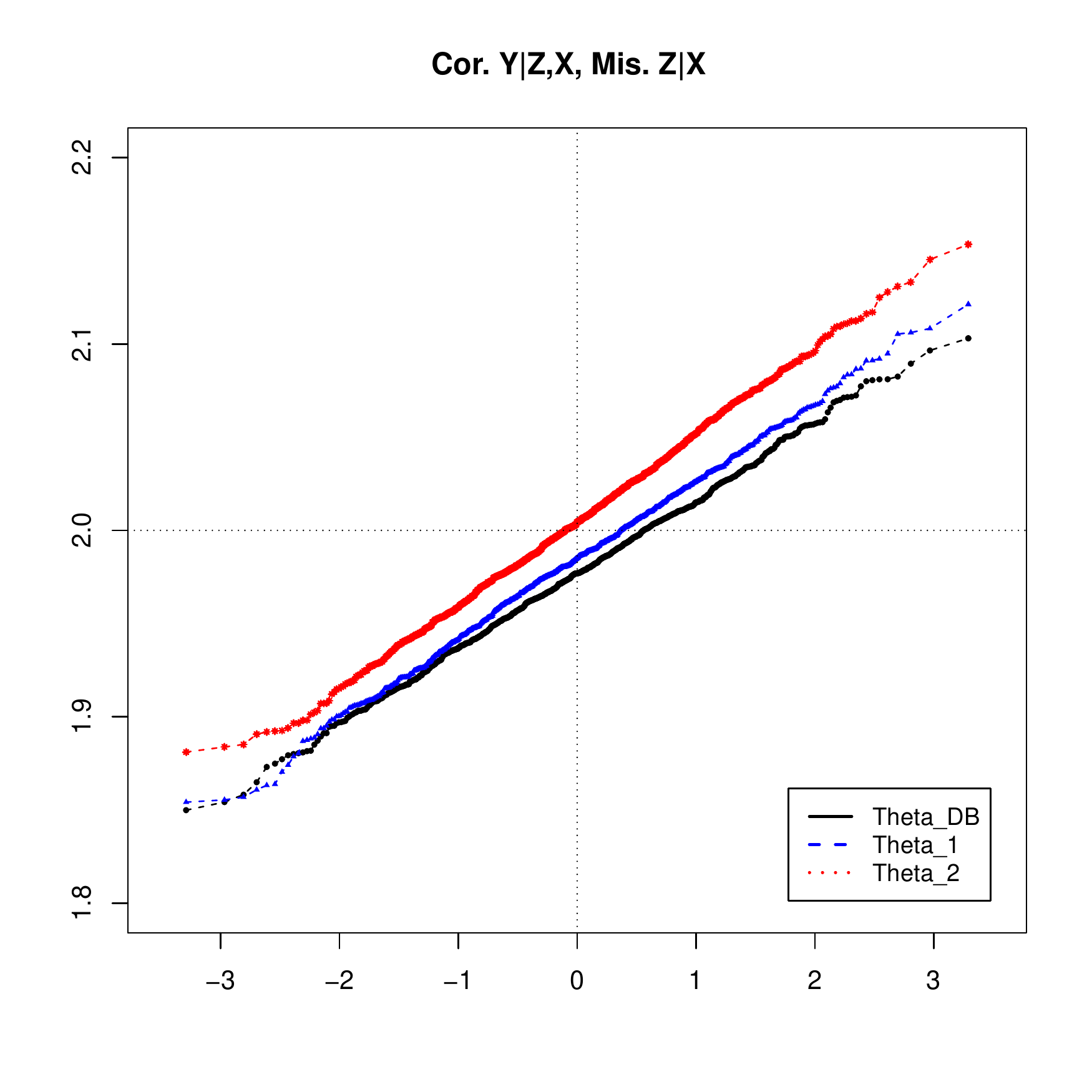} }}%
	\quad
	\subfloat{{\includegraphics[width=70mm]{QQ/log-linear_cor_miss_p=800_t.pdf} }}%
\vspace{-.25in}
	\subfloat{{\includegraphics[width=70mm]{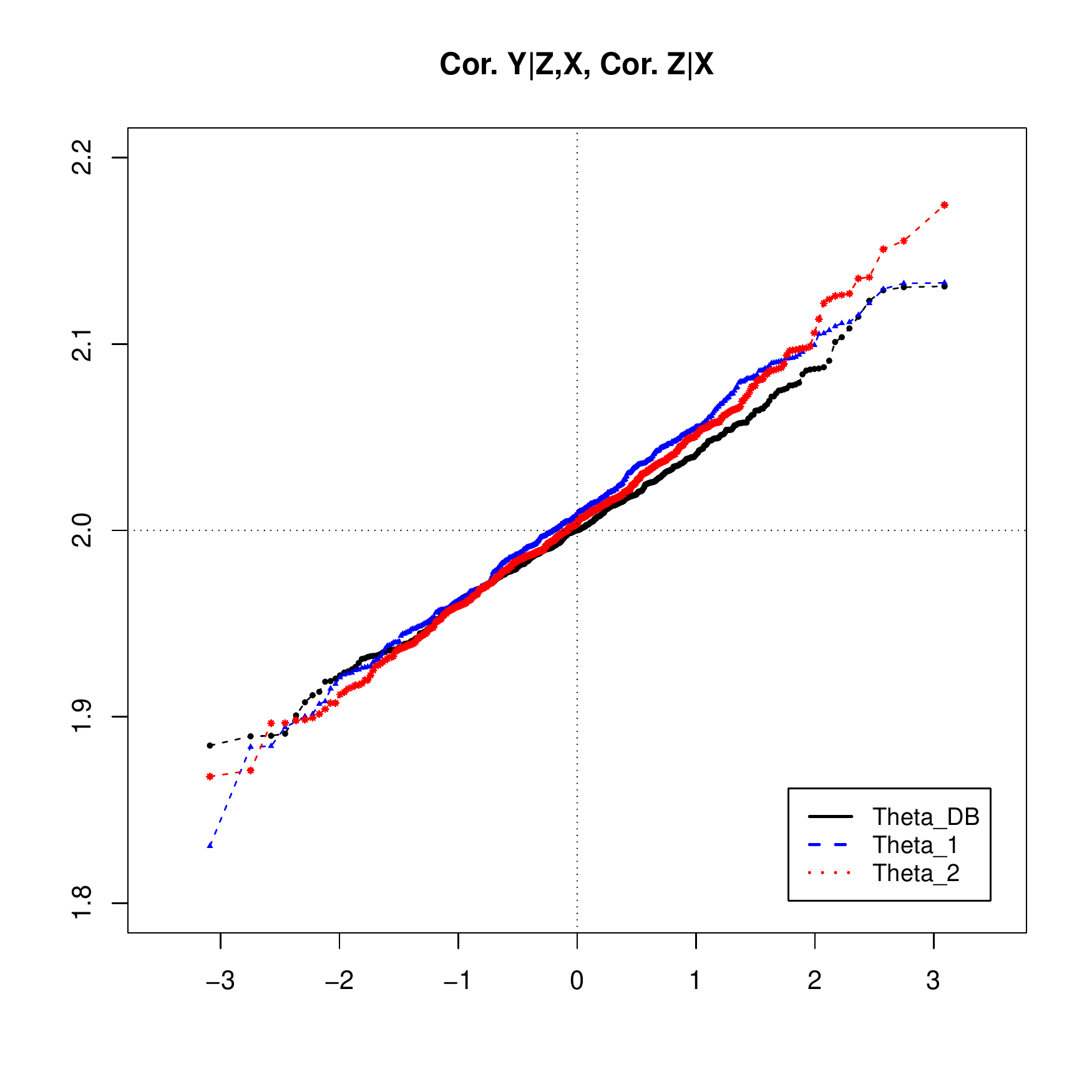} }}%
	\quad
	\subfloat{{\includegraphics[width=70mm]{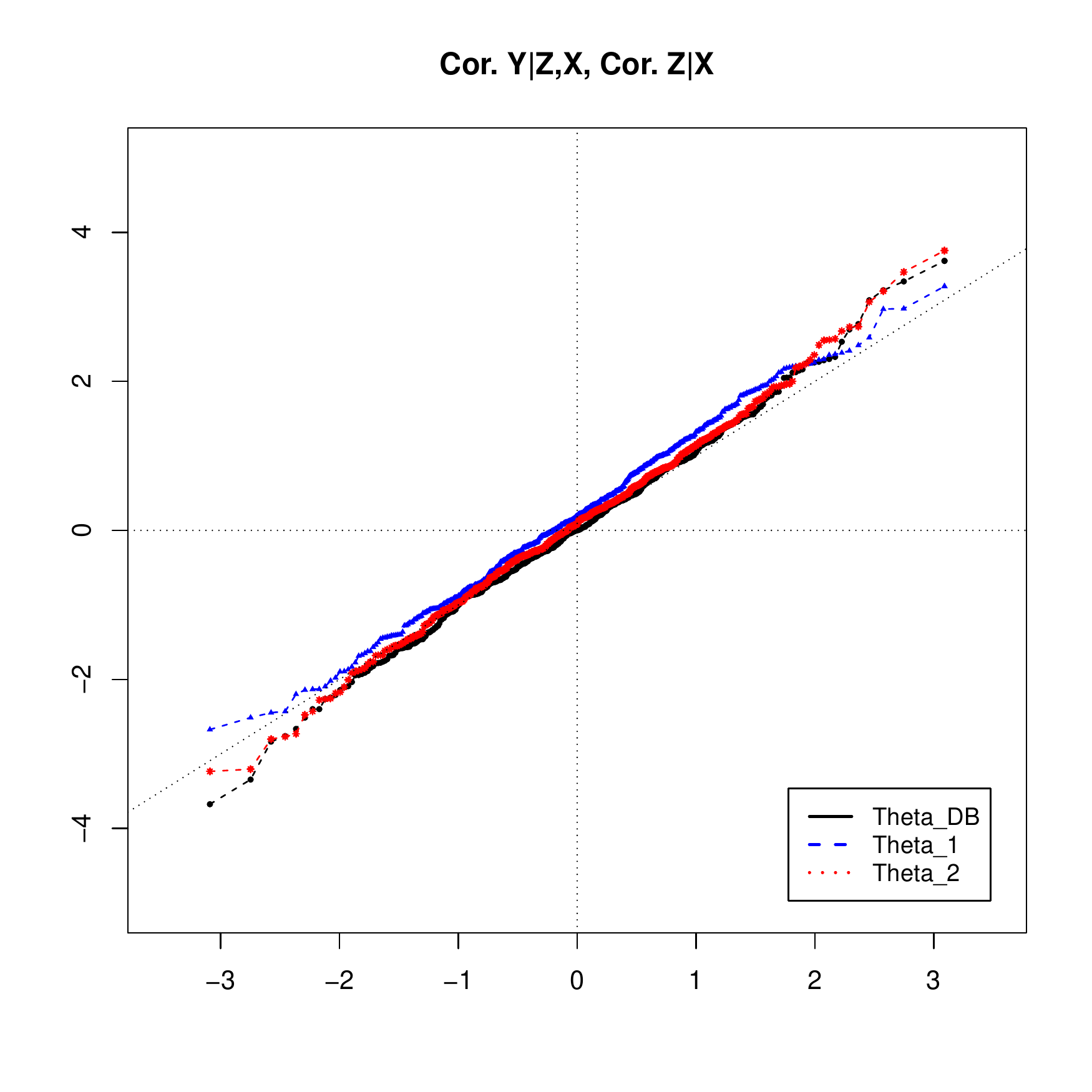} }}%
\end{figure}

\begin{figure}
	\caption{\small QQ plots of the estimates (first column) and $t$-statistics (second column) against standard normal ($n=600$, $p=100$) for partially logistic modeling}
	\label{fig:logistic_p=100}
	\centering
	\subfloat{{\includegraphics[width=70mm]{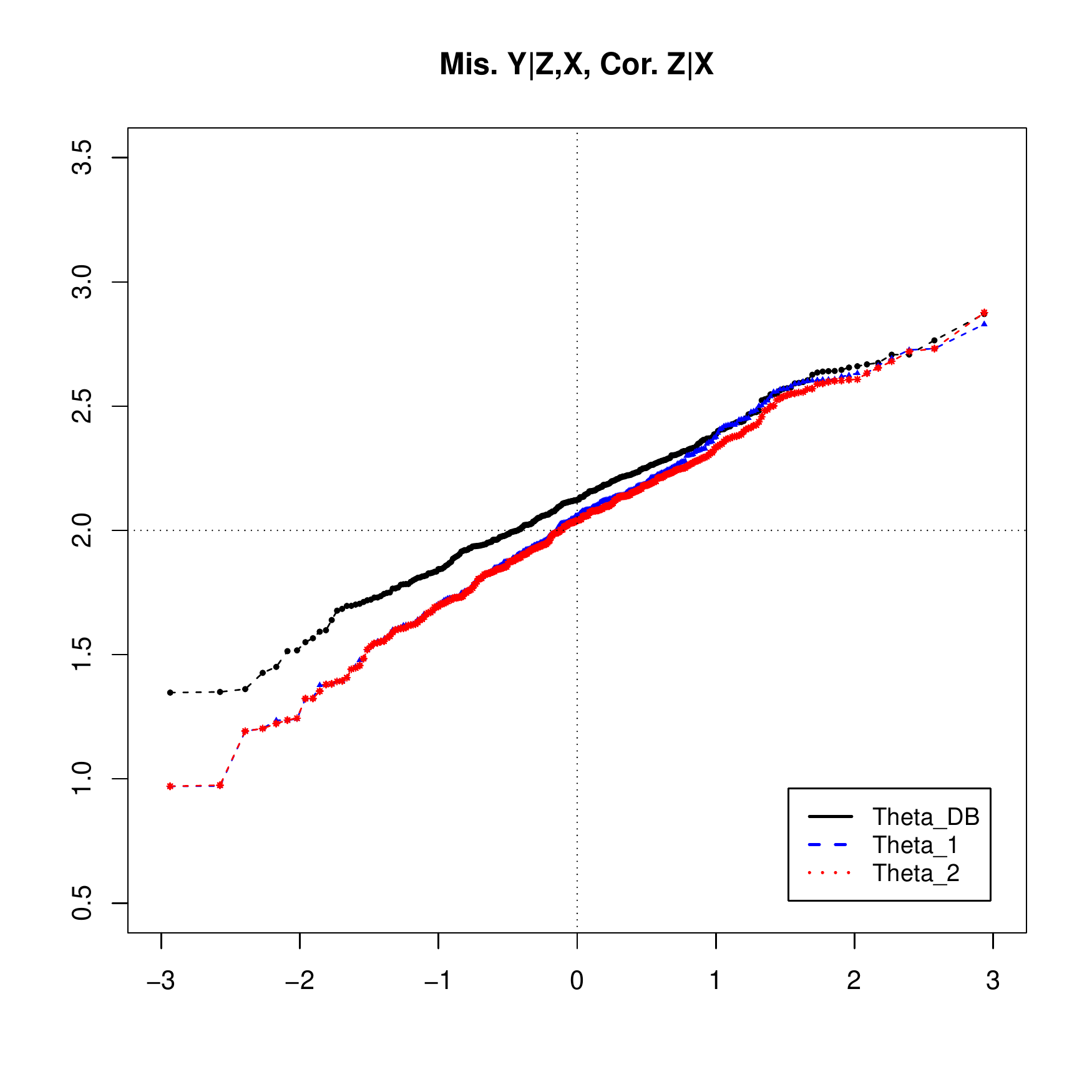} }}%
	\quad
	\subfloat{{\includegraphics[width=70mm]{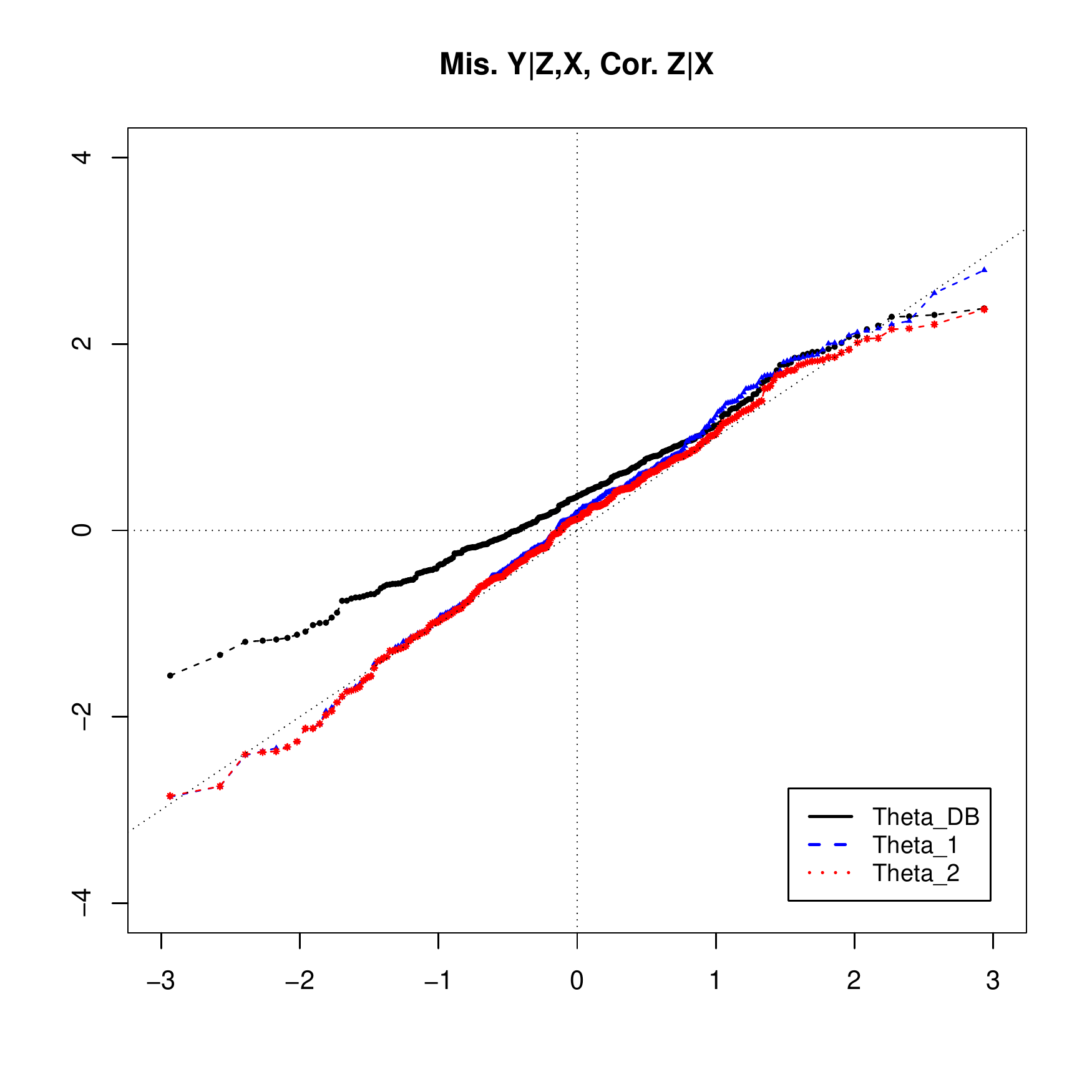} }}%
\vspace{-.25in}
	\subfloat{{\includegraphics[width=70mm]{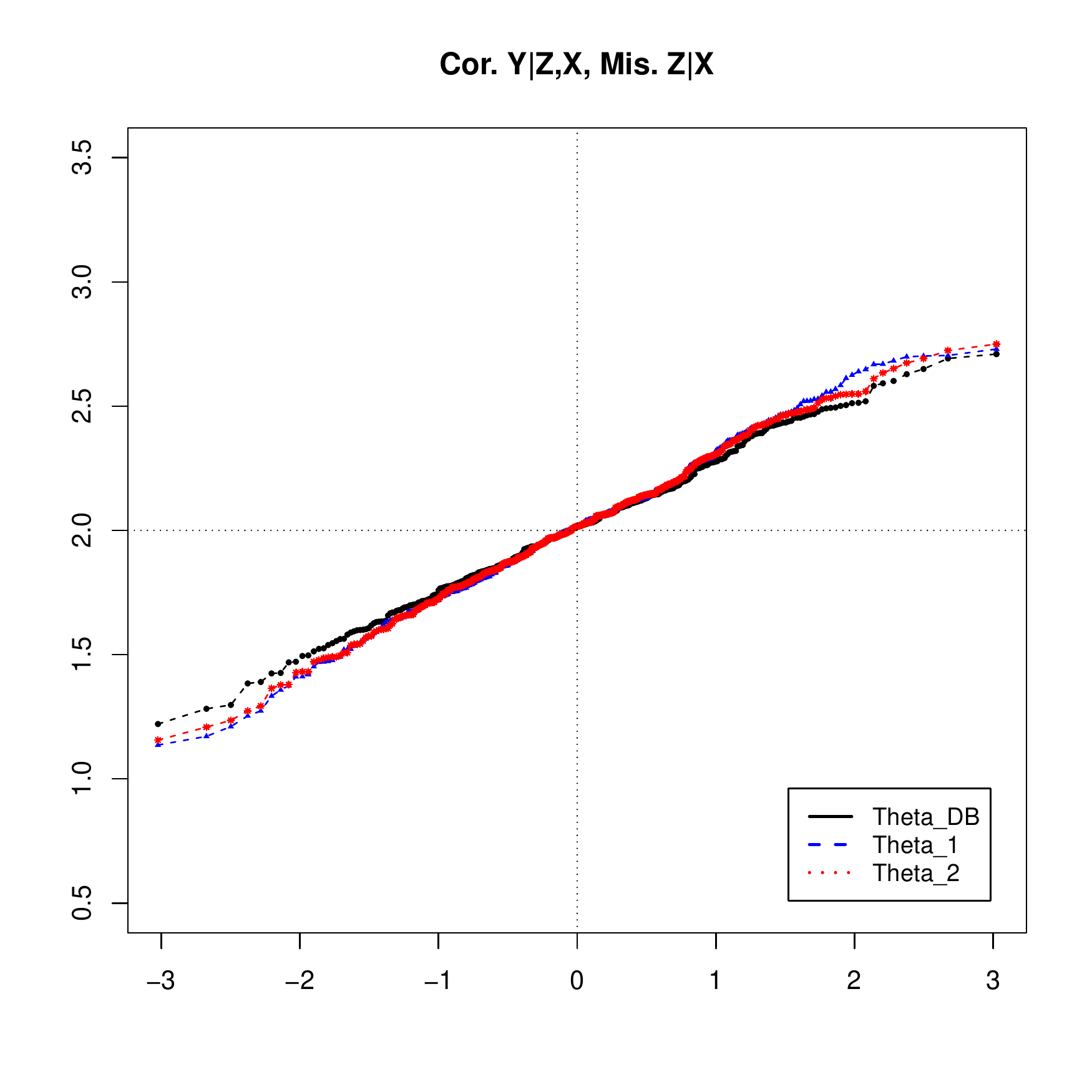} }}%
	\quad
	\subfloat{{\includegraphics[width=70mm]{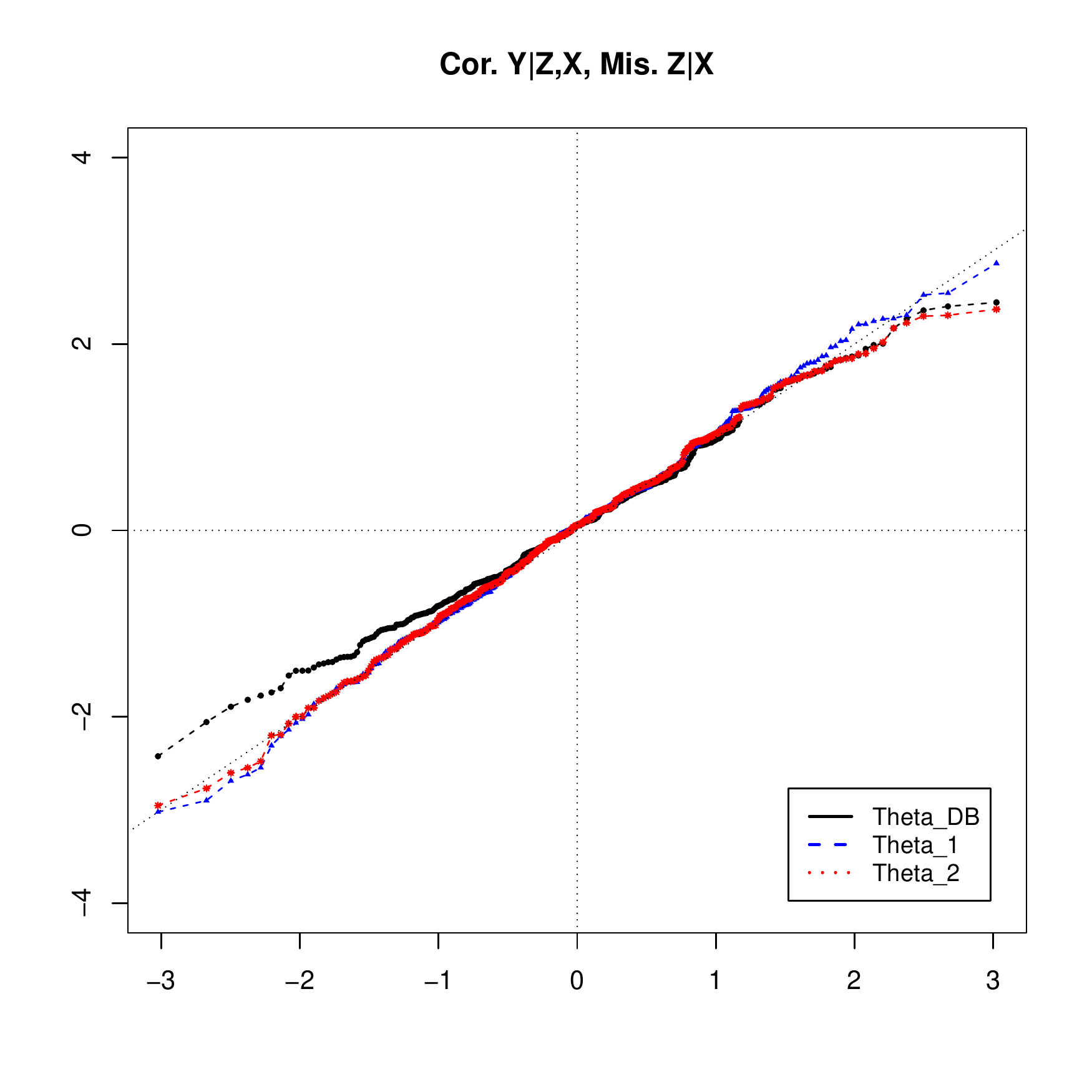} }}%
\vspace{-.25in}
	\subfloat{{\includegraphics[width=70mm]{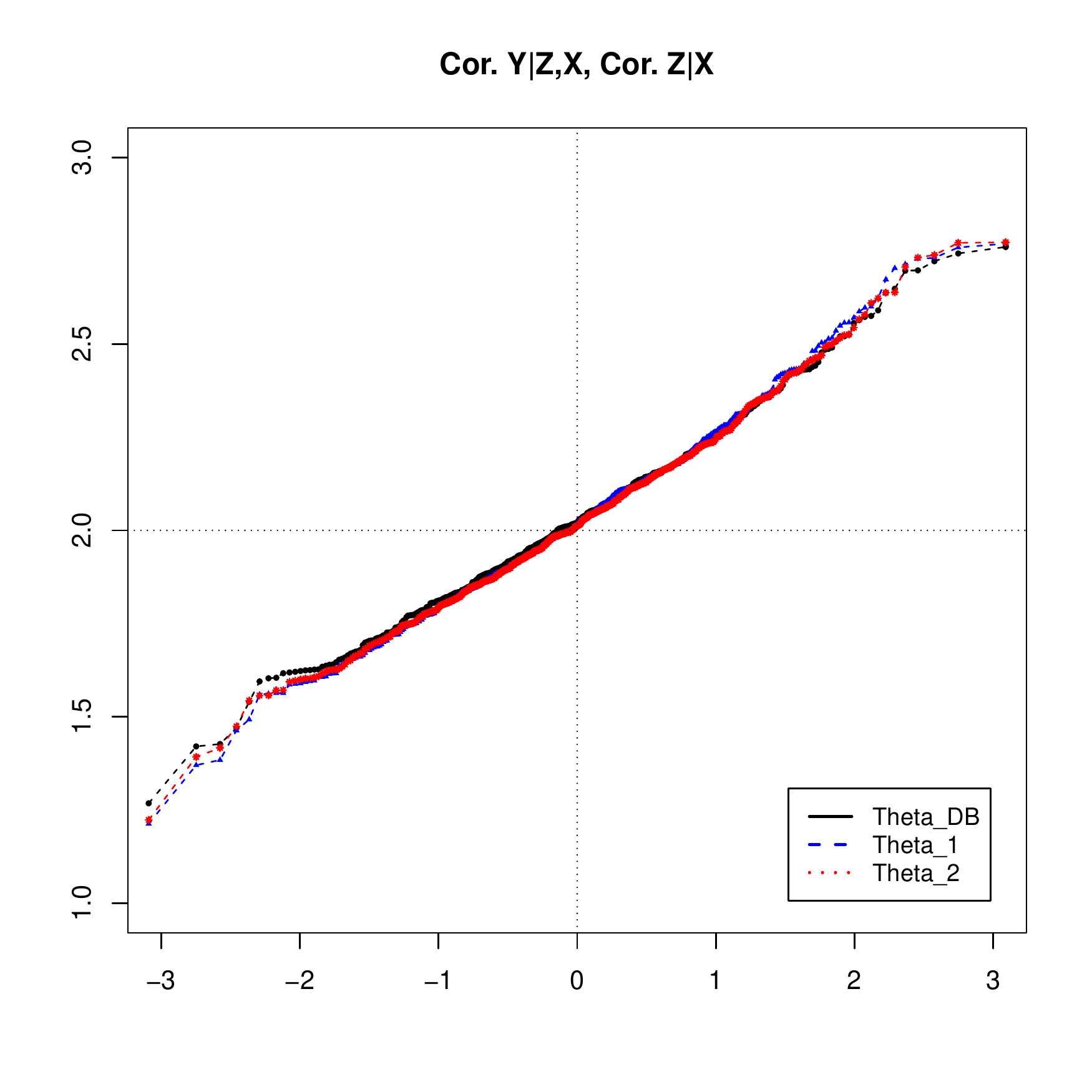} }}%
	\quad
	\subfloat{{\includegraphics[width=70mm]{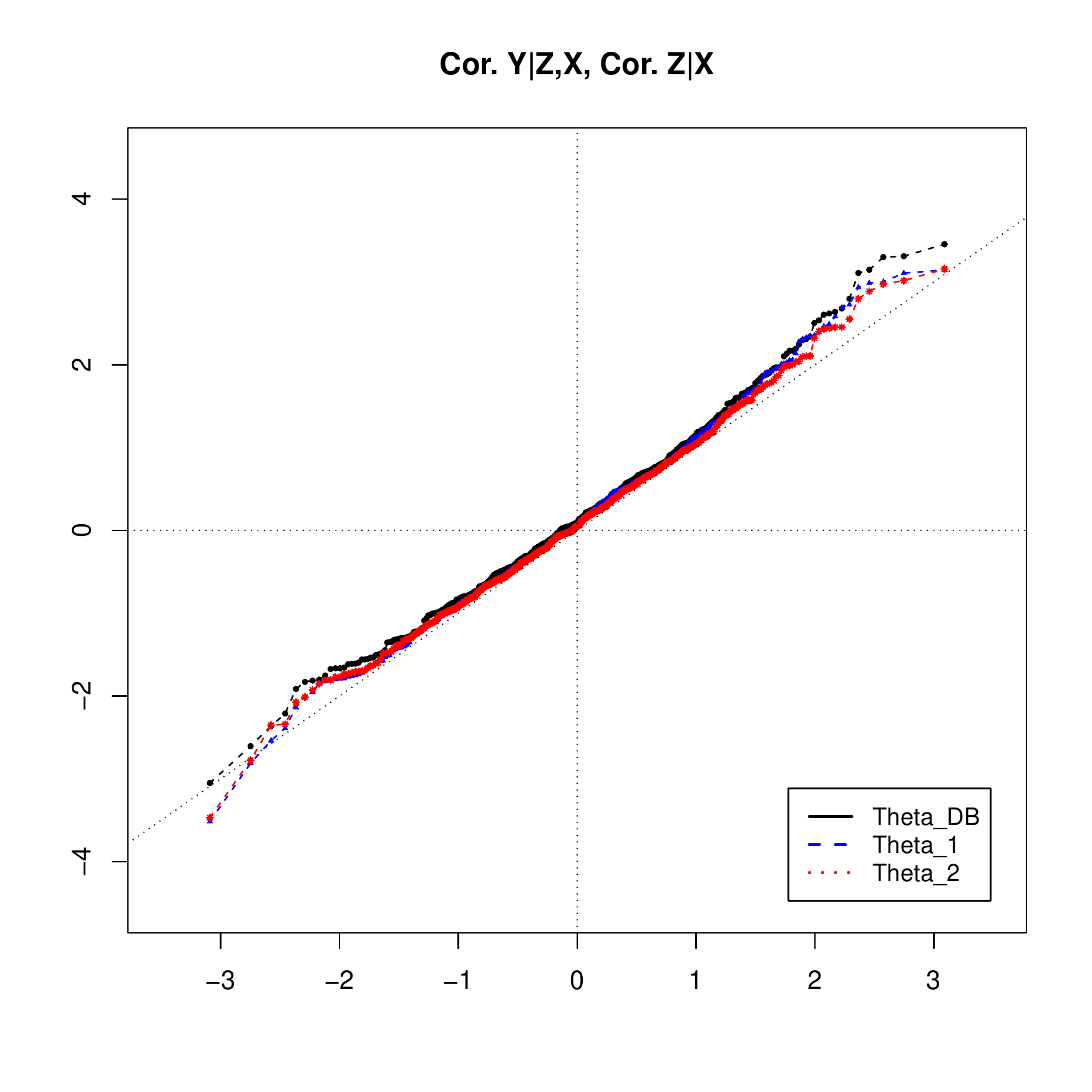} }}%
\end{figure}

\begin{figure}
	\caption{\small QQ plots of the estimates (first column) and $t$-statistics (second column) against standard normal ($n=600$, $p=200$) for partially logistic modeling}
	\label{fig:logistic_p=200}
\centering
\subfloat{{\includegraphics[width=70mm]{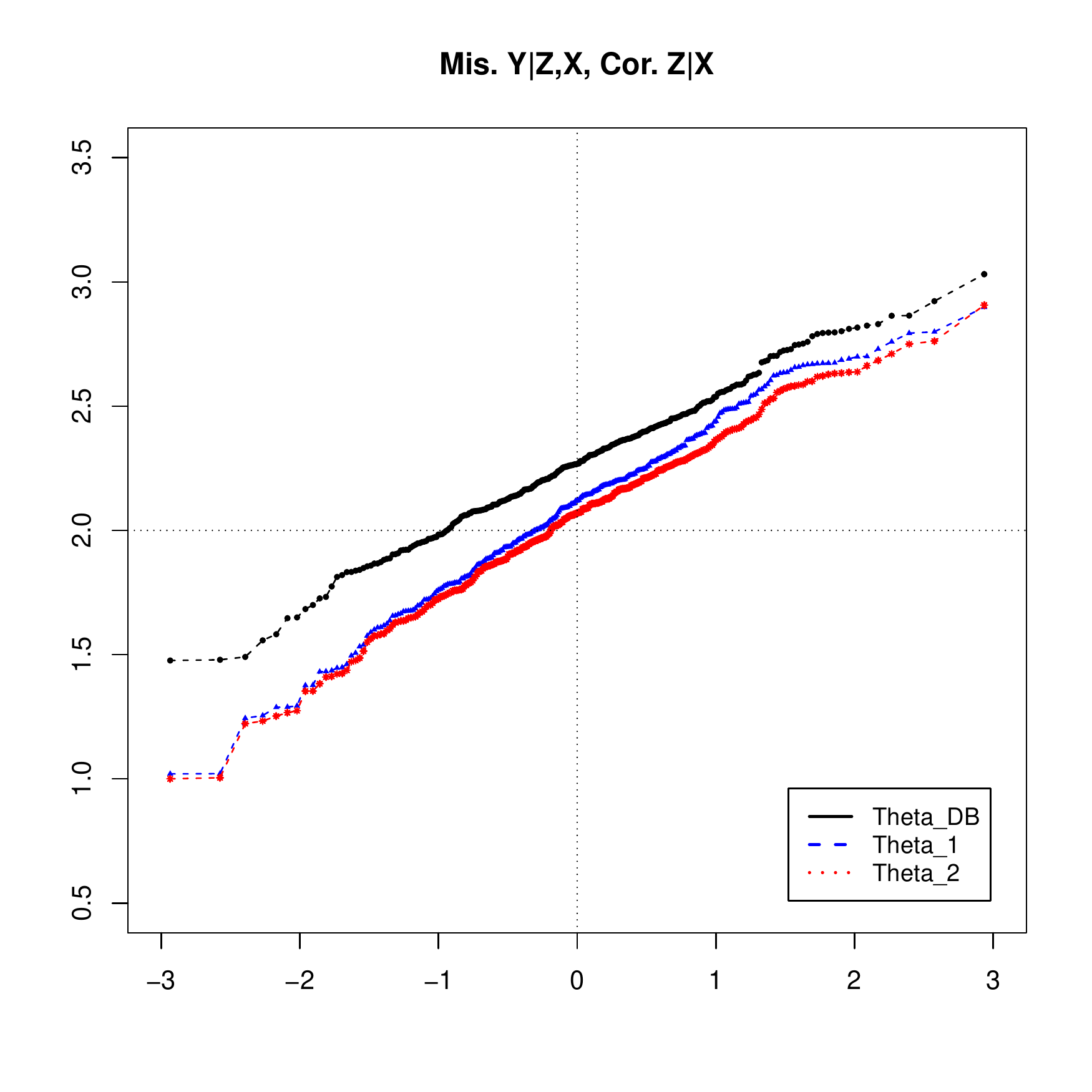} }}%
\quad
\subfloat{{\includegraphics[width=70mm]{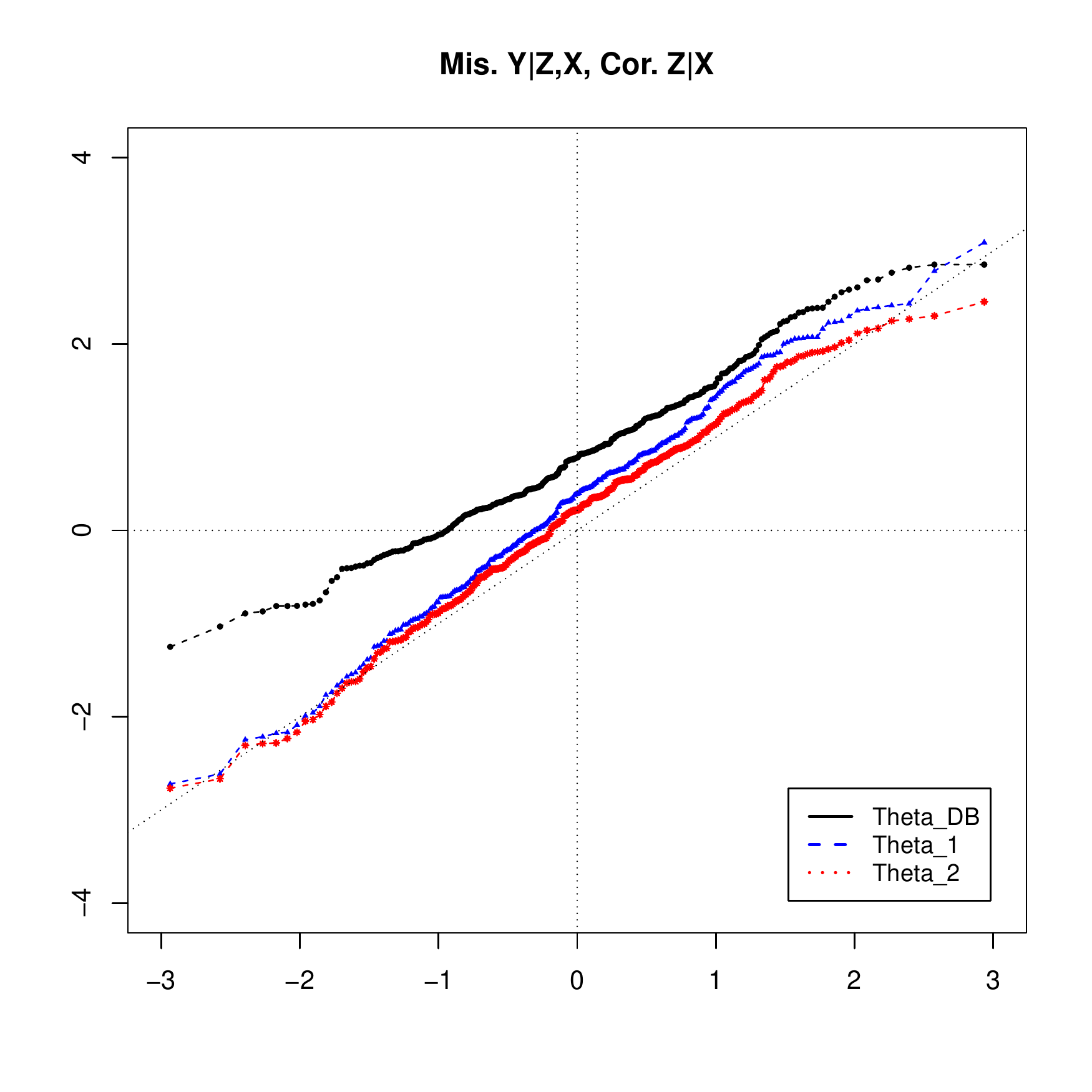} }}%
\vspace{-.25in}
\subfloat{{\includegraphics[width=70mm]{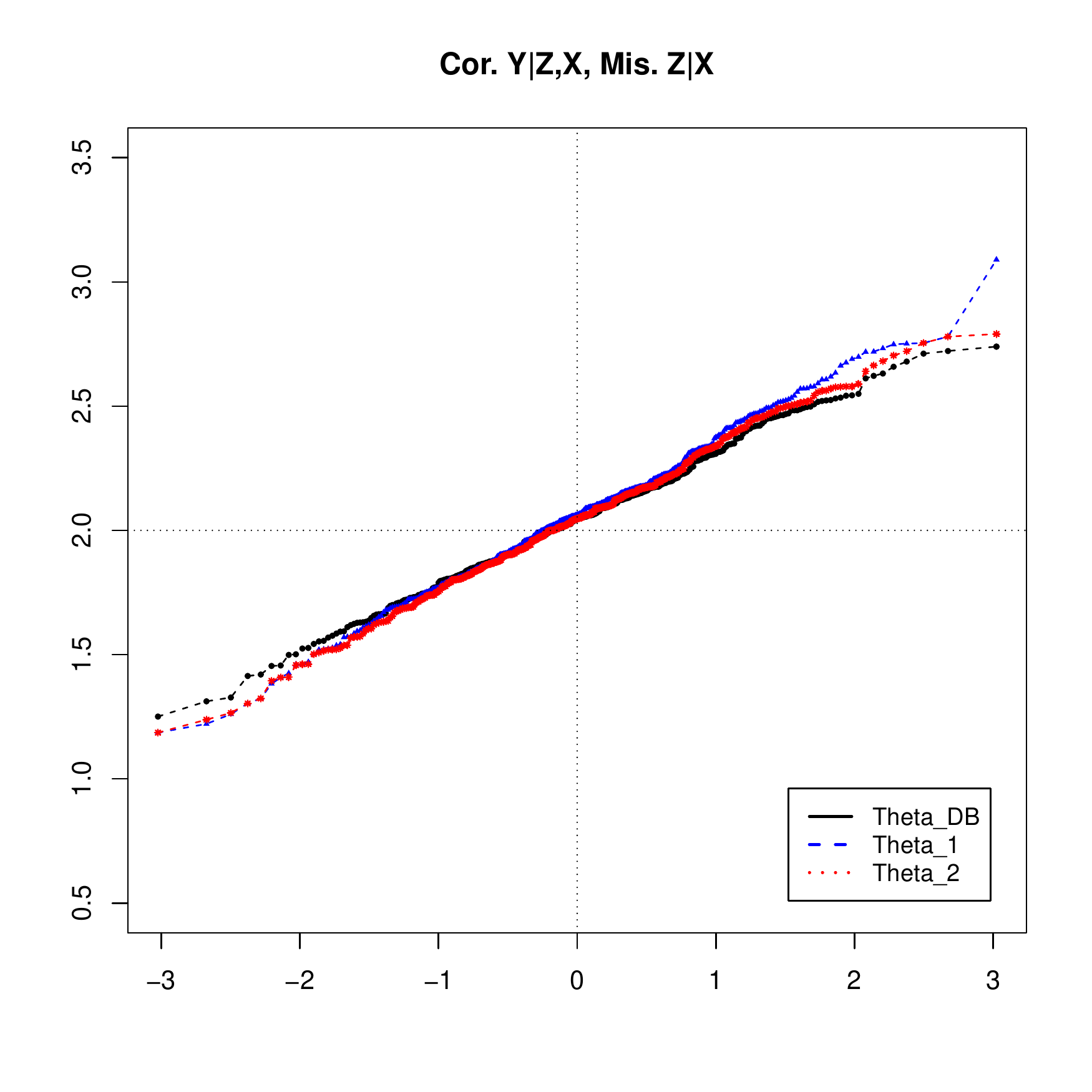} }}%
\quad
\subfloat{{\includegraphics[width=70mm]{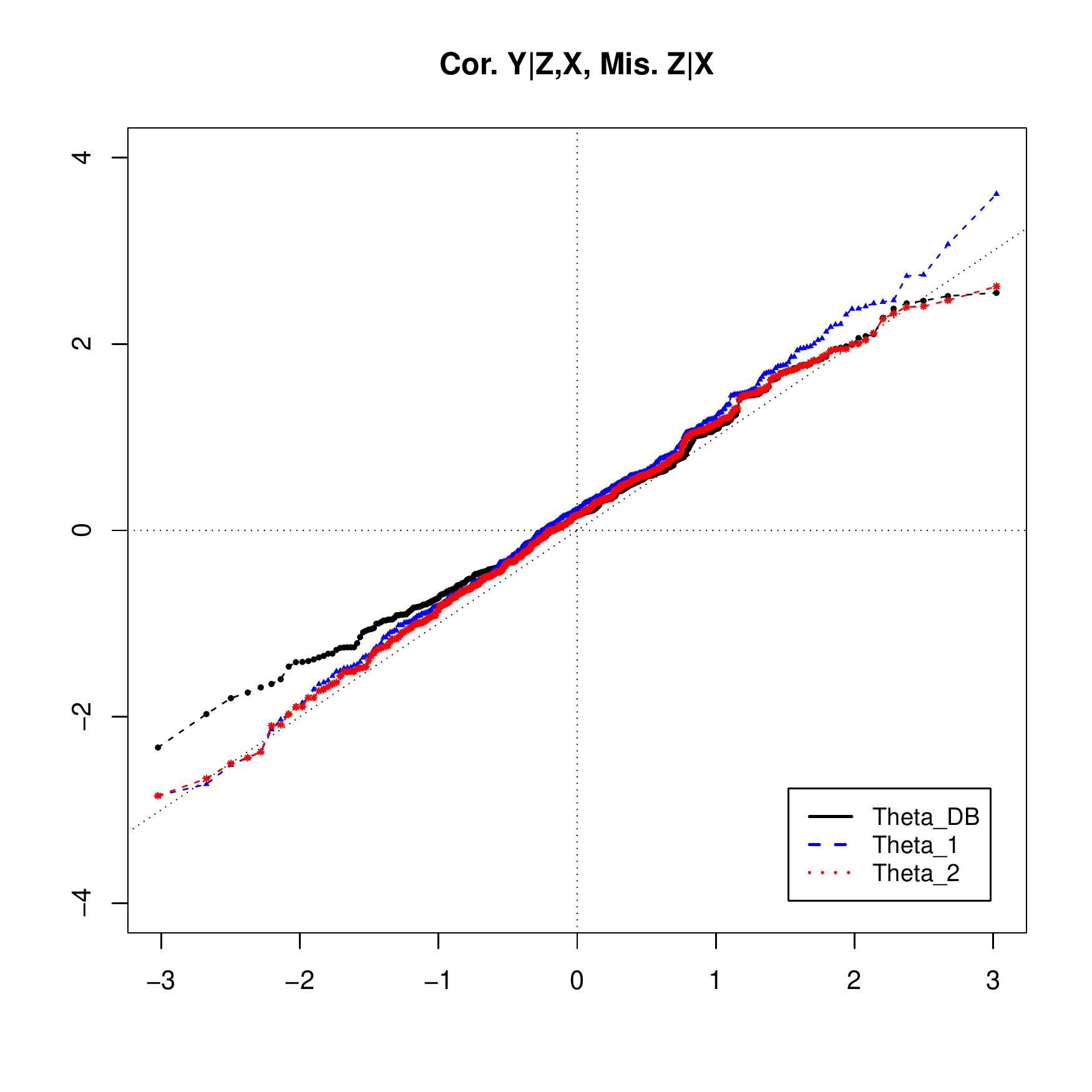} }}%
\vspace{-.25in}
\subfloat{{\includegraphics[width=70mm]{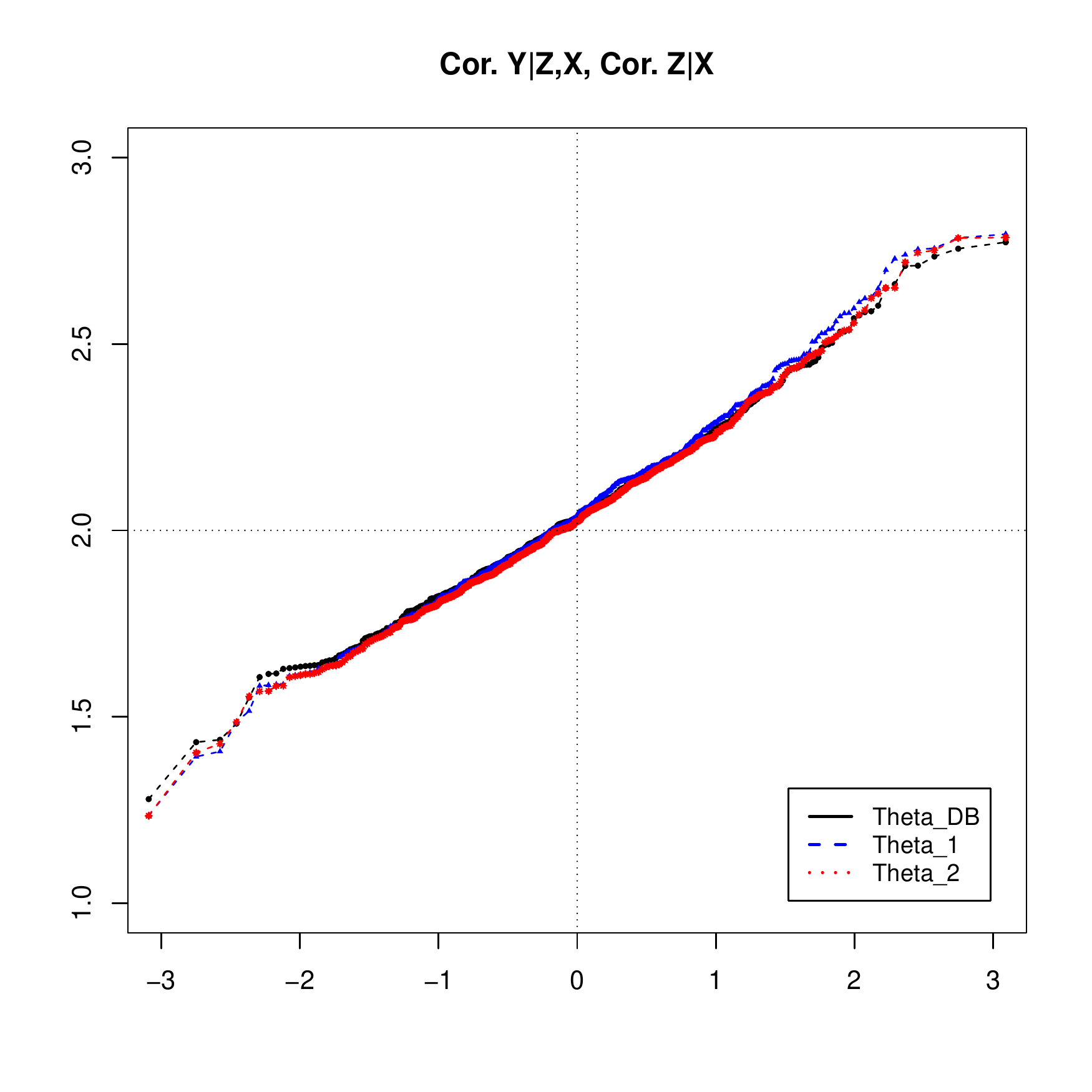} }}%
\quad
\subfloat{{\includegraphics[width=70mm]{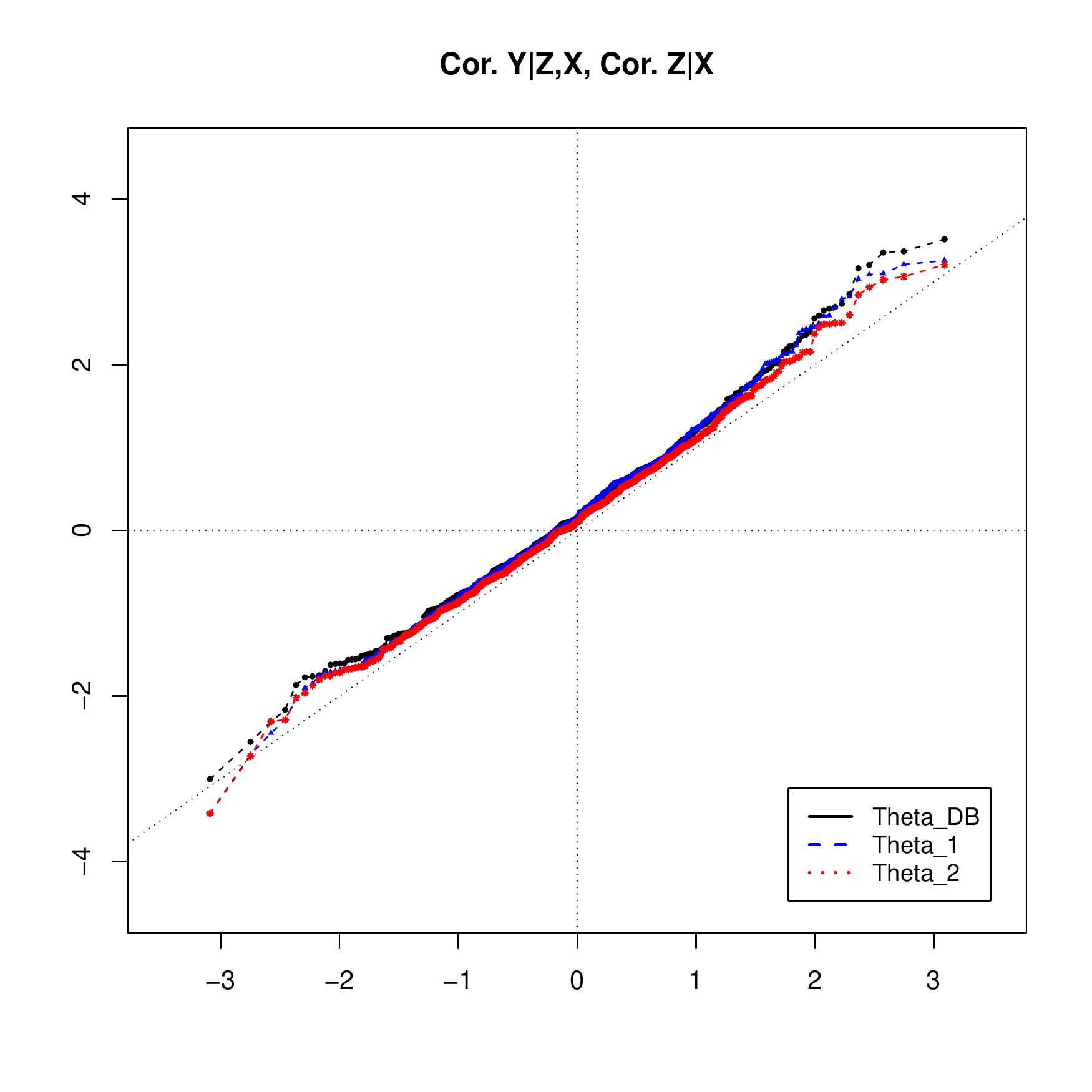} }}%
\end{figure}

\begin{figure}
	\caption{\small QQ plots of the estimates (first column) and $t$-statistics (second column) against standard normal ($n=600$, $p=800$) for partially logistic modeling}
	\label{fig:logistic_p=800}
	\centering
	\subfloat{{\includegraphics[width=70mm]{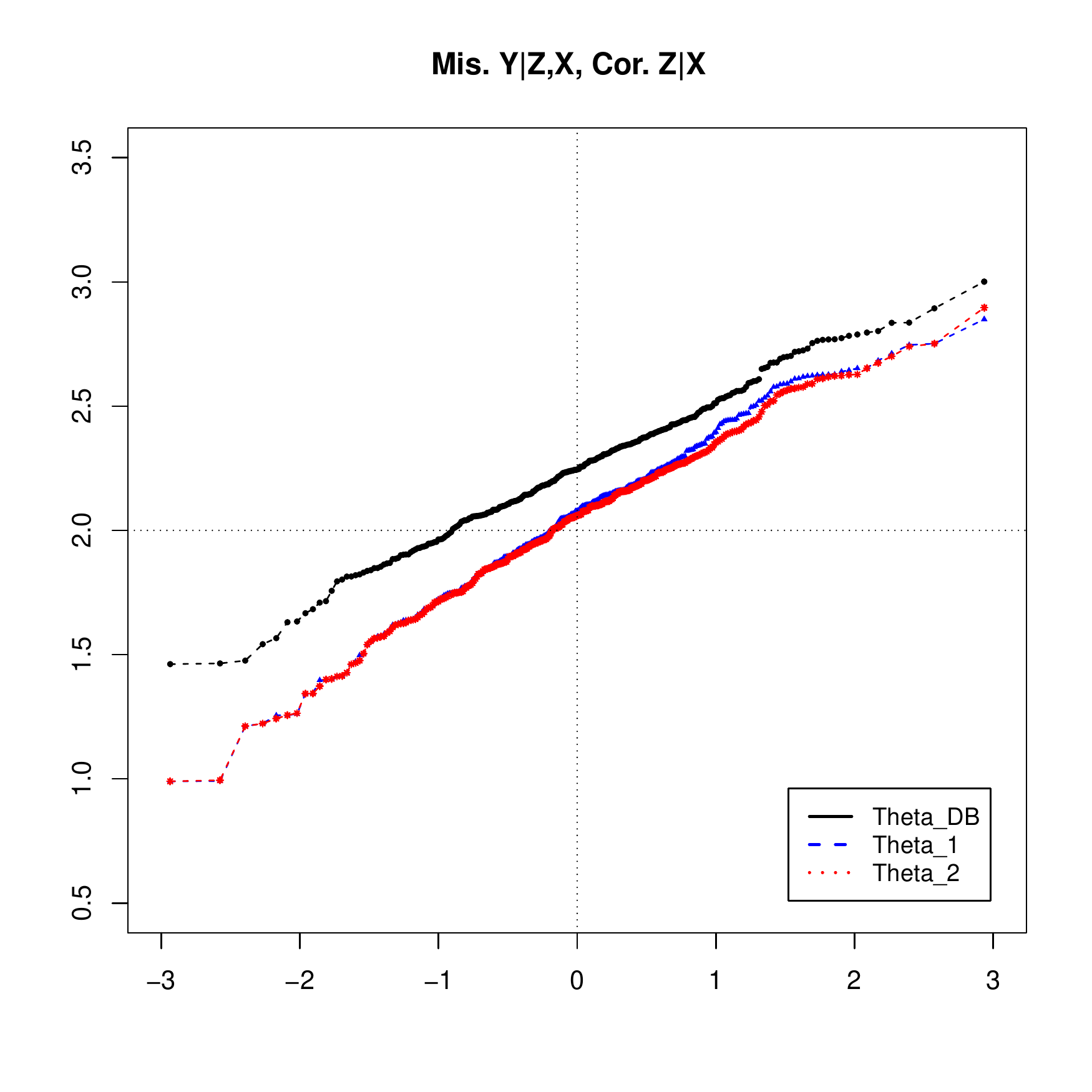} }}%
	\quad
	\subfloat{{\includegraphics[width=70mm]{QQ/logistic_miss_cor_p=800_t.pdf} }}%
    \vspace{-.25in}
	\subfloat{{\includegraphics[width=70mm]{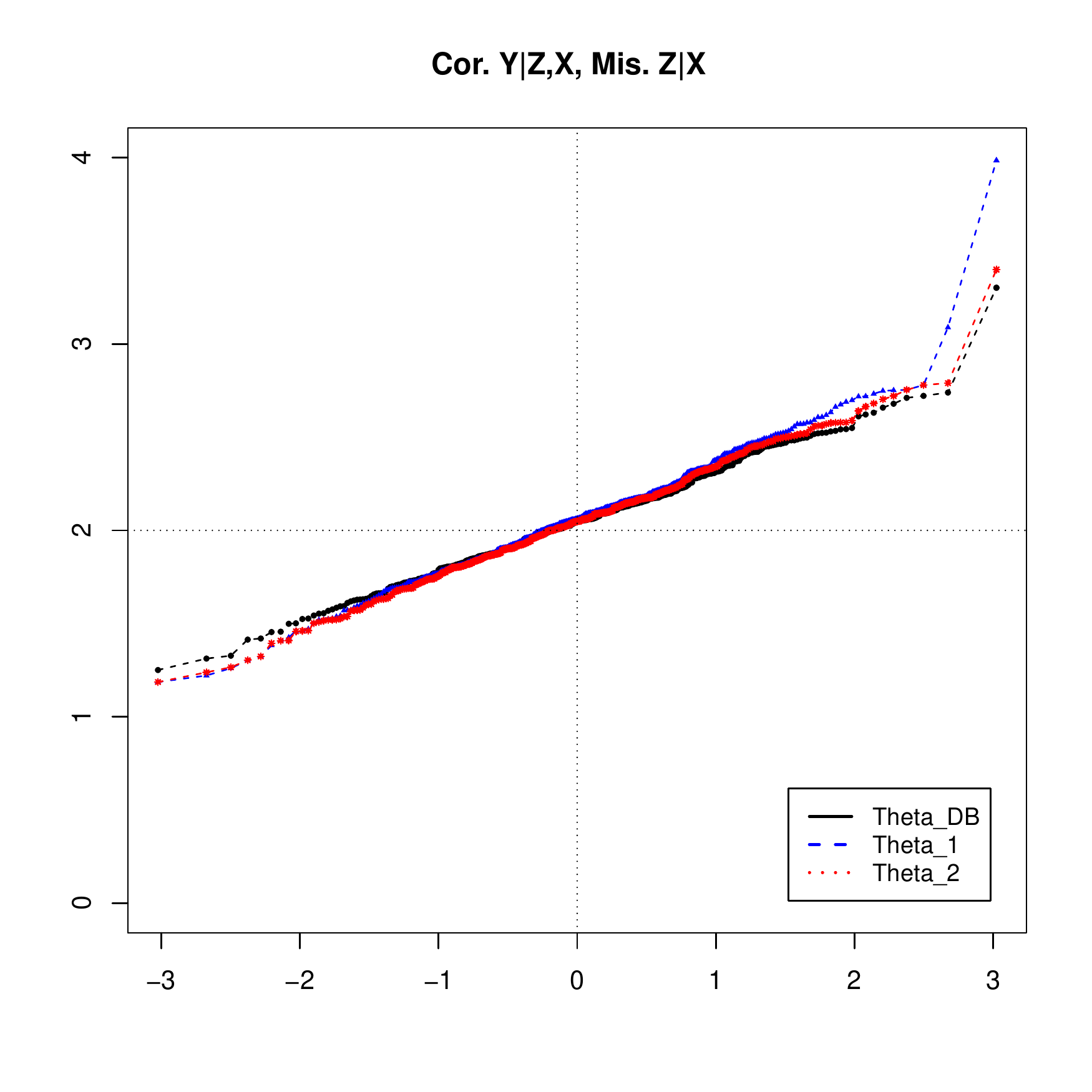} }}%
	\quad
	\subfloat{{\includegraphics[width=70mm]{QQ/logistic_cor_miss_p=800_t.pdf} }}%
	\vspace{-.25in}
	\subfloat{{\includegraphics[width=70mm]{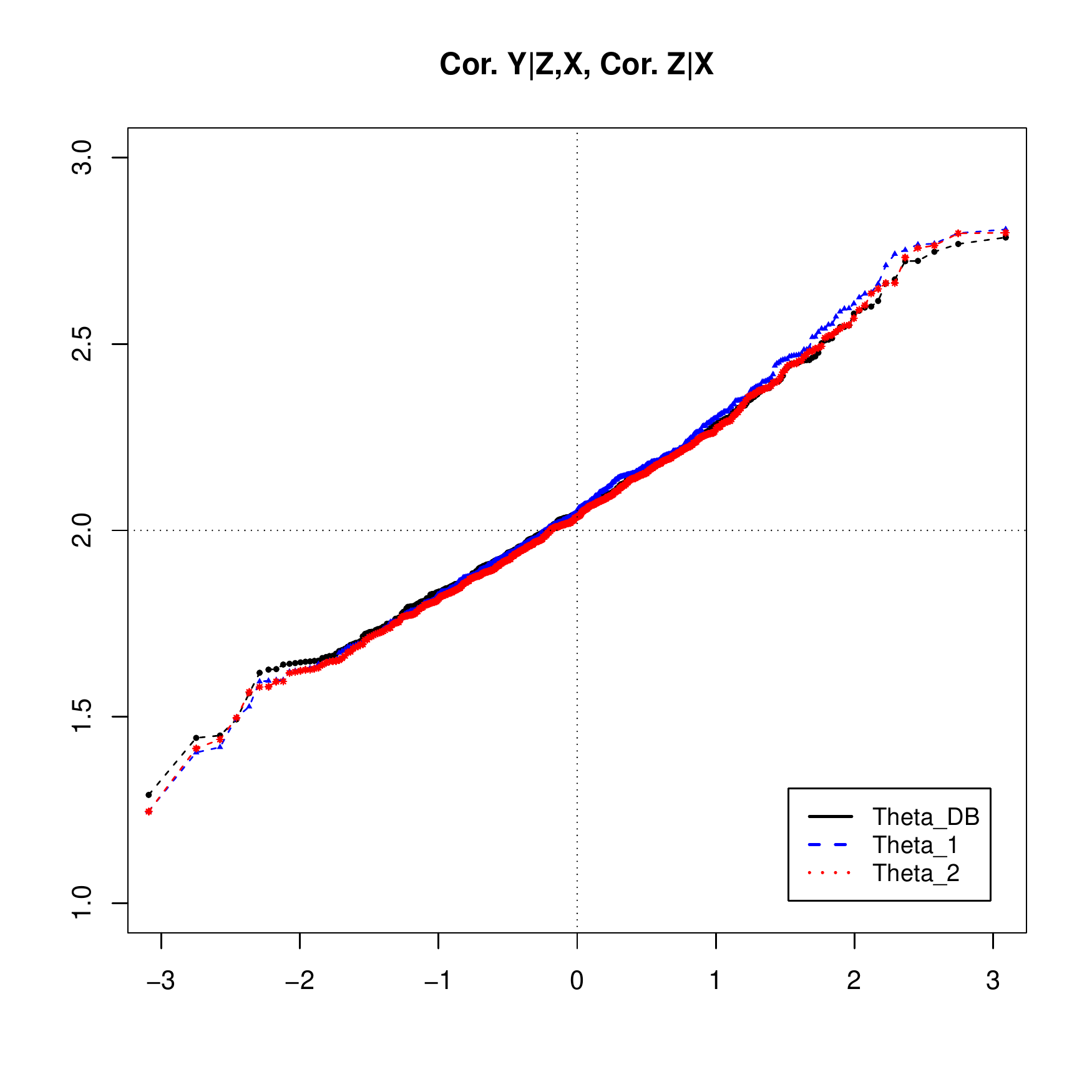} }}%
	\quad
	\subfloat{{\includegraphics[width=70mm]{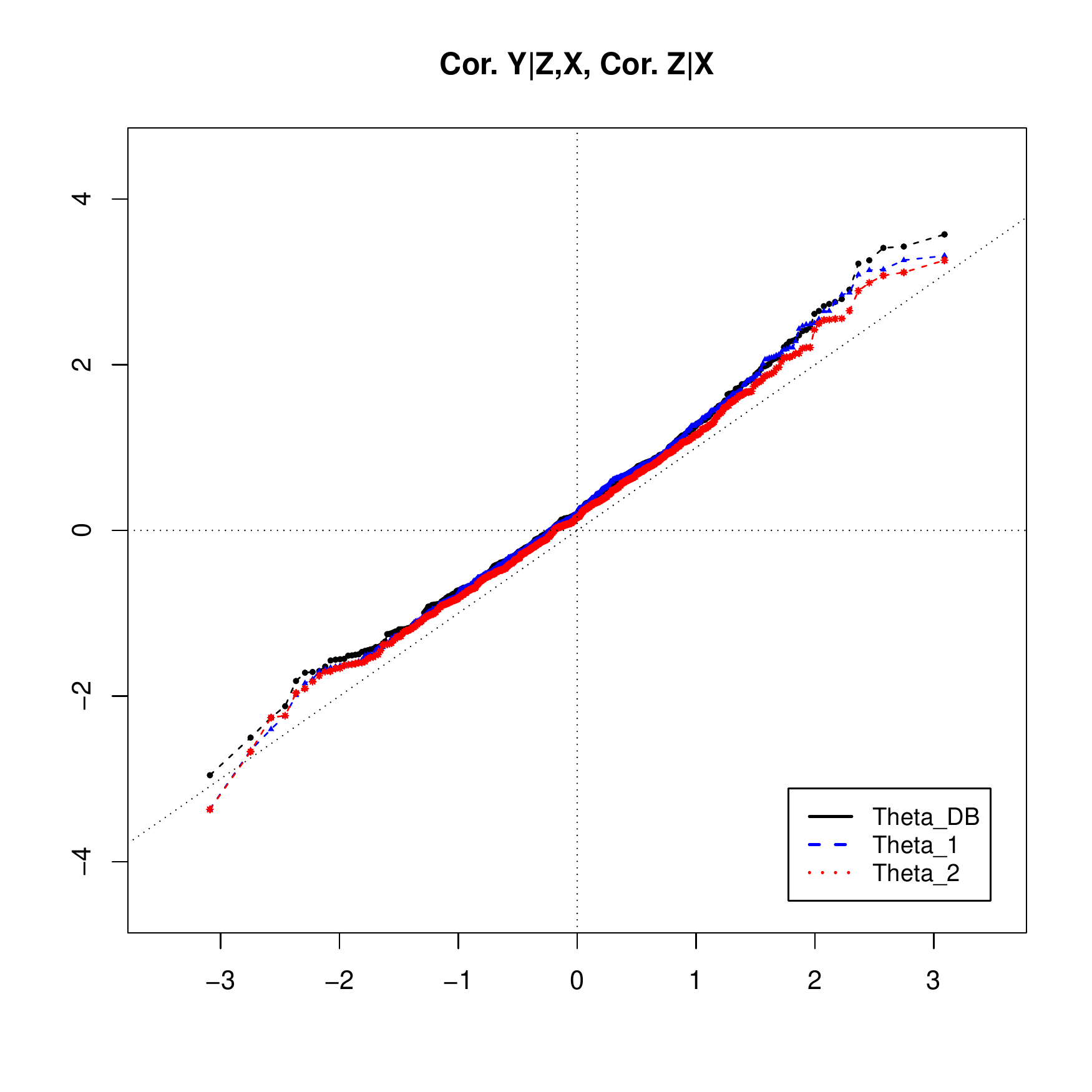} }}%
\end{figure}

\clearpage

\bibliographystyleappend{myapa}
\bibliographyappend{References1-0921}

\end{document}